%% file: main.tex
\documentclass[a4paper,reqno,11pt]{amsart}

\usepackage{amsfonts}
\usepackage{tikz}
\usepackage{tkz-graph}
\usetikzlibrary{graphs}
\usetikzlibrary{graphs.standard}
\usepackage[margin=1.3in]{geometry}
\usepackage{accents}
\usepackage{array,calc} 
\usepackage{mathrsfs}
\usepackage{stmaryrd} 
\usepackage{pgf}
\usetikzlibrary{shapes,cd,arrows,automata,decorations.pathreplacing,angles,quotes,calc}
\usepackage{tkz-euclide}
\usepackage{caption}
\usepackage{subcaption}
\usepackage{amssymb}
\usepackage{amsmath,calc,graphicx}
\usepackage{url}
\usepackage[hidelinks]{hyperref}
\usepackage{amssymb}
\usepackage{amsthm}
\usepackage{multirow}
\usepackage{float}
\usepackage{enumerate}
\usepackage{enumitem}
\usepackage{amsrefs}
 \usepackage[foot]{amsaddr}
\numberwithin{equation}{section}
\numberwithin{figure}{section}
\theoremstyle{plain}

\newtheorem{theorem}{Theorem}[section]
\newtheorem{lemma}[theorem]{Lemma}
\newtheorem{corollary}[theorem]{Corollary}
\newtheorem{proposition}[theorem]{Proposition}
\newtheorem{definition}[theorem]{Definition}

\newtheorem{rhprob}{Riemann-Hilbert Problem}

\theoremstyle{remark}
\newtheorem{remark}[theorem]{Remark}

\newtheorem*{lem*}{\textsc{Lemma}}
\newtheorem*{cor*}{\textsc{Corollary}}
\newtheorem*{exer*}{\textsc{Exercise}}
\newtheorem*{con*}{\textsc{Conjecture}}
\newtheorem*{thm*}{\textsc{Theorem}}
\newcommand{\beq}{\begin{equation}}
\newcommand{\eeq}{\end{equation}}

\usepackage{amsfonts} %% <- also included by amssymb
\DeclareMathSymbol{\shortminus}{\mathbin}{AMSa}{"39}

\newcommand\Psix{\textrm{P}_{\textrm{VI}}}

\newcommand\B[1]{\ensuremath{\scalebox{0.5}{$($}{#1}\scalebox{0.5}{$)$}}}

\title{On $q$-Painlev\'e VI and the geometry of Segre surfaces}
\date{}
\author{Pieter Roffelsen}
\email{pieter.roffelsen@sydney.edu.au}
\address{School of Mathematics and Statistics F07, The University of Sydney, NSW 2006, Australia}

\begin{document}
\begin{abstract}
In the context of $q$-Painlev\'e VI with generic parameter values, the Riemann-Hilbert correspondence induces a one-to-one mapping between solutions of the nonlinear equation and points on an affine Segre surface. Upon fixing a generic point on the surface, we give formulae for the function values of the corresponding solution near the critical points, in the form of complete, convergent, asymptotic expansions. These lead in particular to the solution of the nonlinear connection problem for the general solution of $q$-Painlev\'e VI. We further show that, when the point on the Segre surface is moved to one of the sixteen lines on the surface, one of the asymptotic expansions near the critical points truncates, under suitable parameter assumptions. At intersection points of lines, this  then yields doubly truncated asymptotics at one of the critical points or simultaneous truncation at both.

 \smallskip
\noindent \textit{Keywords: connection problems, Painlev\'e equations, Riemann-Hilbert problems, Segre surfaces, truncated asymptotics.}
\end{abstract}
\maketitle

\tableofcontents

\input{intro}

\input{results}

\input{segre_proofs}

\input{pairs_pants}

\input{rhanalysis}

\input{symmetries}

\appendix

\input{appendix}

\end{document}

%% file: intro.tex
\section{Introduction}\label{sec:intro}
The complete intersection of two quadrics in $\mathbb{CP}^4$ forms a Segre surface. An affine Segre surface is by definition a Segre surface minus a hyperplane section. In the most generic setting, the Segre surface is smooth
and the hyperplane section does not contain any lines.  In this paper, we relate classical algebro-geometric aspects of such affine Segre surfaces, to the asymptotics of solutions of the famous $q$-difference Painlev\'e VI equation with generic parameter values.
% In such case, we call the corresponding affine Segre surface of generic type\footnote{It follows from the classical literature that this definition is intrinsic as it is equivalent to the affine Segre surface having precisely 16 (distinct) lines on it.}.

Assuming $q\in\mathbb{C}$, $0<|q|<1$, and given parameters $\Theta:=(\theta_0,\theta_t,\theta_1,\theta_\infty)\in\mathbb{C}^4$, the latter $q$-difference equation is given by
\begin{equation}\label{eq:qpvi}
q\Psix:\ \begin{cases}\  f\overline{f}&=\dfrac{(\overline{g}-q^{+\theta_0}t)(\overline{g}-q^{-\theta_0}t)}{(\overline{g}-q^{\theta_\infty-1})(\overline{g}-q^{-\theta_\infty})},\\
  \   g\overline{g}&=\dfrac{(f-q^{+\theta_t}t)(f-q^{-\theta_t}t)}{q(f-q^{+\theta_1})(f-q^{-\theta_1})},
    \end{cases} 
\end{equation}
where $f,g:T\rightarrow \mathbb{CP}^1$ are complex functions defined on a domain $T$ invariant under multiplication by $q$ and we have used the abbreviated notation $f=f(t)$, $g=g(t)$, $\overline{f}=f(qt)$, $\overline{g}=g(qt)$, for $t\in T$. For the sake of simplicity, we will assume that $q$ is real an positive throughout the paper.

We will be concerned with asymptotics of solutions of $q\Psix$ along discrete sets of the form $q^\mathbb{Z}t_0$, $t_0\in\mathbb{C}^*$, which we call $q$-lines.
That is, we fix an initial point $t_0\in\mathbb{C}^*$, let time $t$ vary in the $q$-line $q^\mathbb{Z}t_0$, and consider the behaviour of solutions as $t\rightarrow 0$ or $t\rightarrow \infty$ along this $q$-line. Without loss of generality, we may thus restrict our attention to solutions $(f,g)$ whose whole domain is given by the $q$-line $T=q^\mathbb{Z}t_0$ and we will call $(f,g)$ a solution of $q\Psix(\Theta,t_0)$ in this case.

In \cite{roffelsenjoshiqpvi}, the Jimbo-Sakai Lax pair \cite{jimbosakai} of $q\Psix$  was used to construct a bijective mapping from the solution space of $q\Psix(\Theta,t_0)$ to an affine Segre surface. In the current paper, we derive the generic asymptotic behaviours of solutions as $t\rightarrow 0$ and $t\rightarrow \infty$, parametrised explicitly in terms of coordinates on this affine Segre surface. The affine Segre surface has $16$ lines and we show that, on any of these lines, asymptotic expansions around either $t=0$ or $t=\infty$ truncate under suitable parameter assumptions. At intersection points of lines we then have either simultaneous truncation of asymptotic expansions around $t=0$ and $t=\infty$, or double truncation at one the two critical points. In the latter case the asymptotic expansions reduce to
convergent Laurent series, and the corresponding solutions were classified by Ohyama \cite{ohyamamero}.

The asymptotics of solutions to $q\Psix$ have been studied before. 
For each of the two critical points, Mano \cite{manoqpvi} constructed a two complex parameter family of solutions of $q\Psix$ and derived explicit leading order asymptotics around the corresponding critical point, as well as corresponding monodromy. Jimbo et al. \cite{jimbonagoyasakai} further derived conjectural full asymptotic expansions of the corresponding tau functions around one of the two critical points.

One of the important asymptotic problems that was still unsolved, is the corresponding nonlinear connection problem, which entails providing the asymptotic behaviours around both critical points simultaneously, or relating them to one another, explicitly and rigorously.
This problem is solved in the current paper for the general solution of $q\Psix$. It is further shown how, under a formal continuum limit as $q\rightarrow 1$, this solution reduces to the classical solution of the nonlinear connection problem for Painlev\'e VI by Jimbo \cite{jimbo1982}. In particular, this leads to an explicit bi-rational map between the famous Jimbo-Fricke cubic surface and the continuum limit of the Segre surface.

Our approach to the asymptotic analysis of solutions goes via the Riemann-Hilbert problem (RHP) associated with $q\Psix$. It rests on three important technical innovations. The first, are so called Mano decompositions by Ohyama et al. \cite{ohyamaramissualoy}, which decompose the monodromy of rank and degree two linear Fuchsian $q$-difference systems into the monodromies of two  linear Fuchsian $q$-difference systems of rank two and degree one. We will recall this construction in this paper and provide additional formulae.

The second, is that such Mano decompositions can also be realised analytically, by recasting the RHP into a corresponding factorised form. We note that such a factorised form of a RHP for $q\Psix$ first appeared in \cite{jimbonagoyasakai}. 

The third, is that, by quotienting the global solution of this RHP by local parametrices which solve the individual components, the RHP can be recast into one, posed on a single, closed and time-independent contour, with the jump matrix made out of Heine's $q$-hypergeometric functions. The solution of this RHP can be written as perturbation of the identity for small or large time, yielding the asymptotics of the underlying solution of $q\Psix$.
In the context of $\Psix$, this methodology was developed and applied in \cites{itslisovyy,gavlis2018,cafassowidom}. We adapt it to the $q$-difference setting in the current paper.

One of the ways the geometry of the affine Segre surface enters the asymptotic analysis, is that, on the lines, some of factors in the Mano decompositions become reducible. This leads to a triangular reduction of one of the parametrices that consequently truncates the small or large time perturbative solution of the RHP. Similar truncations also occur in the context of $\Psix$, see Jimbo \cite{jimbo1982} and Guzzetti \cites{guzzetti_solving,guzzettireview}.

The classical Cayley–Salmon theorem states that any smooth cubic surface in $\mathbb{CP}^3$ contains precisely $27$ lines. The problem of writing these lines in a useful and explicit form, attracted a lot of interest in the nineteenth and early twentieth century, with various different solutions given over the  years, for example by Schl\"afli \cite{schlafki1858} and B. Segre \cite{segrebook}. The fact that smooth Segre surfaces contain sixteen lines is also classical and due to Corrado Segre himself \cite{Segre}. 
We will show that the interpretation of the Segre surface for $q\Psix$ as a monodromy manifold leads to explicit representations of the $16$ lines. Furthermore, certain natural coordinates, which we call twist parameters, corresponding to different Mano decompositions, trace out the lines when varied and thus lead to explicit parametrisations of the lines.
% Though the problem of writing down the lines on a smooth Segre surfaces was probably solved a long time ago, the author could not find a good reference.

\subsection{Overview of the paper}
Section \ref{sec:results}, starts with introducing the general set up and recalling some of the main results in \cite{roffelsenjoshiqpvi}, that we need in the current paper.
This is followed by Subsections \ref{subsec:geometry}-\ref{subsec:nonlinear_con}, in which all the main results of the paper are given. Near each result, reference is given to the location of the proof in the main body of the text. In Subsection \ref{subsec:continuumlimit}, the continuum limit of some of the results is discussed in relation to $\Psix$, and two forthcoming works are announced.

Section \ref{sec:geometrysegre} is on the geometry of the Segre surface corresponding to $q\Psix$. It contains all the proofs of results on the geometry of the surface, described in Subsection \ref{subsec:geometry}.

In Section \ref{sec:qpairspants}, two Mano decompositions, in the context of the linear problem for $q\Psix$, are described, both algebraically, as factorisations of monodromy, and analytically, as factorisations of the Riemann-Hilbert problem.

In Section \ref{sec:rhanalysis}, the generic asymptotics of solutions near $t=0$ is derived from one of the factorised Riemann-Hilbert problems. Furthermore, truncated asymptotics at $t=0$ on one of the lines of the Segre surface are derived as well as doubly truncated asymptotics at $t=0$, corresponding to an intersection point of lines.

In Section \ref{sec:symmetries}, the actions of several symmetries of $q\Psix$ on monodromy are computed. The results are used to derive all the remaining asymptotic results in the paper. These include, generic asymptotics near $t=\infty$, truncated asymptotics on each of the sixteen lines as well as asymptotics on some of the intersection points of lines.

Finally, in Appendix \ref{appendix:rhdif}, a slightly technical proof is given of Proposition \ref{prop:rhdiffeo}, which says, roughly speaking, that the Riemann-Hilbert correspondence induces a biholomorpism between the initial value space of $q\Psix$ and the corresponding affine Segre surface.

\subsection{Acknowledgements}
This research was supported by Australian Research
Council Discovery Project \#DP210100129 and the Leverhulme Trust Visiting Professorship VP2-2018-013.

I am deeply indepted to Nalini Joshi, for many inspiring discussions and suggestions that helped to shape parts of this paper as well as her continuing scientific and material support.
I would further like to express my gratitude to Oleg Lisovyy, Marta Mazzocco, Peter Miller and Yang Shi for their interest and discussions  at various stages of this work.

I am  grateful to Yousuke Ohyama, Jean-Pierre Ramis and Jacques Sauloy for the interesting email exchanges about the present topic of investigation. I would also like to thank Hajime Nagoya, for sharing mathematica code for the tau-function expansions in \cite{jimbonagoyasakai}, and Marco Bertola for introducing me to the notion of Tyurin parameters.

A part of the work was done during the author's residence at the Isaac Newton Institute during the Fall of 2022. I am grateful to the Møller institute and organisers of the program ”Applicable resurgent asymptotics: towards a universal theory” for their hospitality.

\subsection{Notation}
Here, we briefly describe the notation used in this paper. The symbol $\sigma_3$ is the well-known Pauli matrix $\sigma_3=\operatorname{diag}(1,-1)$. We fix a $0<q<1$, set $\tau=-\frac{2\pi i}{\log q}$, so that $\tau\in i\,\mathbb{R}_{>0}$, and define the lattice
\begin{equation*}
    \Lambda_\tau=\mathbb{Z}\cdot 1+\mathbb{Z}\cdot \tau.
\end{equation*}

The $q$-Pochhammer symbol is the (convergent) product
\begin{equation*}
(z;q)_\infty=\prod_{k=0}^{\infty}{(1-q^kz)}\qquad (z\in\mathbb{C}).
\end{equation*}
Note that the entire function $(z;q)_\infty$ satisfies
\begin{equation*}
(qz;q)_\infty=\frac{1}{1-z}(z;q)_\infty,
\end{equation*}
with $(0;q)_\infty=1$ and, moreover, possesses simple zeros at points in $q^{-\mathbb{N}}$. The $\mathit{q}$-theta function 
\begin{equation}\label{eq:thetasym}
\theta_q(z)=(z;q)_\infty(q/z;q)_\infty\quad (z\in \mathbb{C}^*),\quad \mathbb{C}^*:=\mathbb{C}\setminus\{0\},
\end{equation}
is analytic on $\mathbb{C}^*$, with essential singularities at $z=0, \infty$, and has simple zeros on the $q$-line $q^\mathbb{Z}$. It satisfies
\begin{equation}\label{eq:qtheta_identities}
\theta_q(qz)=-\frac{1}{z}\theta_q(z)=\theta_q(1/z).
\end{equation}

We make use of the $q$-gamma function
\begin{equation*}
    \Gamma_q(x)=(1-q)^{1-x}\frac{(q;q)_\infty}{(q^x;q)_\infty},
\end{equation*}
and define
\begin{equation*}
    \vartheta_\tau(x):=\theta_q(q^x)=\frac{(1-q)(q;q)_\infty^2}{\Gamma_q(x)\Gamma_q(1-x)},
\end{equation*}
so that
\begin{equation*}
    \vartheta_\tau(x+\tau)=\vartheta_\tau(x),\quad \vartheta_\tau(x+1)=-e^{\frac{2\pi i}{\tau}x}\vartheta_\tau(x)=\vartheta_\tau(-x).
\end{equation*}
For $n\in\mathbb{N}^*$, we use the common abbreviation for repeated products of these functions,
\begin{align*}
\theta_q(z_1,\ldots,z_n)&=\theta_q(z_1)\cdot \ldots\cdot \theta_q(z_n),\\
\vartheta_\tau(z_1,\ldots,z_n)&=\vartheta_\tau(z_1)\cdot \ldots\cdot \vartheta_\tau(z_n),\\
\Gamma_q(z_1,\ldots,z_n)&=\Gamma_q(z_1)\cdot\ldots\cdot \Gamma_q(z_n),\\
(z_1,\ldots,z_n;q)_\infty&=(z_1;q)_\infty\cdot\ldots\cdot (z_n;q)_\infty.
\end{align*}
We further denote
\begin{equation}\label{eq:defiFq}
 F_q(\alpha, \beta;\gamma
; z) :=\;_{2}\phi_1 \left[\begin{matrix} 
q^\alpha, q^\beta \\ 
q^\gamma \end{matrix} 
; q,z \right],
\end{equation}
where $_2\phi_1$ is the standard Heine $q$-hypergeometric function, for $\alpha,\beta,\gamma\in\mathbb{C}$ with $\gamma\notin \mathbb{Z}_{\leq 0}$.

%% file: results.tex
\section{Set up and main results}\label{sec:results}
In this section, we start by discussing the general set up, introducing the linear problem for $q\Psix$ and defining the corresponding monodromy manifold. Then, in Subsection \ref{subsec:previous}, some of the results in \cite{roffelsenjoshiqpvi} are recalled, that we will be building on in this paper. In particular, the affine Segre surface corresponding to $q\Psix$ is introduced.
In Subsection \ref{subsec:geometry}, we discuss some fundamental properties of this surface, like its smoothness, number of lines and their intersection points.

 We then move on to the asymptotics of solutions of $q\Psix$. In Subsections \ref{subsec:asympzero} and \ref{subsec:asympinfty}, the  respective asymptotics near $t=0$ and $t=\infty$ of generic solutions are described. Then, in Subsection \ref{subsec:asymplines}, it is shown how the generic asymptotic expansions truncate on the different lines of the Segre surface.
This is followed by Subsection \ref{subsection:intersect}, in which  double or doubly truncated asymptotics on intersection points of lines are discussed.

In subsection \ref{subsec:nonlinear_con}, the solution to the nonlinear connection problem for the general solution of $q\Psix$ is given and shown to yield some surprising consequences regarding  solutions on open domains.

Finally, in Subsection \ref{subsec:continuumlimit}, the continuum limit of our connection formulae is described and shown to match exactly with the classical formulas by Jimbo \cite{jimbo1982} that solve the nonlinear connection problem for $\Psix$. In particular, this yields an explicit bi-rational mapping between the Jimbo-Fricke cubic and the continuum limit of the Segre surface. We end the subsection, announcing two forthcoming works.

\subsection{Set up}\label{subsec:setup}
Take parameters $\Theta:=(\theta_0,\theta_t,\theta_1,\theta_\infty)\in\mathbb{C}^4$ and consider a rank two linear system
\begin{align}
    Y(qz)&=A(z,t)Y(z),\label{eq:linear_problem}\\
    A(z,t)&=A_0(t)+z A_1(t)+z^2 A_2,\nonumber\label{eq:matrix_polynomial}
\end{align}
where $A(z,t)$ is a $2\times 2$ matrix polynomial with determinant given by
\begin{equation}\label{eq:detA}
|A(z,t)|=\left(z-q^{+\theta_t}t\right)\left(z-q^{-\theta_t}t\right)\left(z-q^{+\theta_1}\right)\left(z-q^{-\theta_1}\right),
\end{equation}
and
\begin{equation*}
    A_0(t)=H(t)\begin{pmatrix}
    q^{+\theta_0}t & 0\\
    0 & q^{-\theta_0}t \end{pmatrix}H(t)^{-1},\quad  A_2=\begin{pmatrix}
    q^{-\theta_\infty} & 0\\
    0 & q^{+\theta_\infty}\end{pmatrix},\label{eq:diagonal}
\end{equation*}
for an $H(t)\in GL_2(\mathbb{C})$.

We assume that
\begin{equation}\label{eq:param_assumptions_1}
2\theta_0,2\theta_t,2\theta_1,2\theta_\infty\notin \Lambda_\tau,\qquad
\epsilon_0 \theta_0+\epsilon_t \theta_t+\epsilon_1 \theta_1+\epsilon_\infty \theta_\infty\notin \Lambda_\tau,
\end{equation}
for any $\epsilon_j\in\{\pm 1\}$, $j=0,t,1,\infty$, and we let $t$ vary in a $q$-line $q^\mathbb{Z}t_0$, for some fixed $t_0\in\mathbb{C}^*$, with
\begin{equation}\label{eq:param_assumptions_2}
    t_0\notin q^{\mathbb{Z}\pm (\theta_t+\theta_1)},q^{\mathbb{Z}\pm (\theta_t-\theta_1)}, q^{\mathbb{Z}\pm (\theta_0+\theta_\infty)},q^{\mathbb{Z}\pm (\theta_0-\theta_\infty)}.
\end{equation}
Under these assumptions the linear system has no resonance and is guaranteed to be irreducible, see \cite{roffelsenjoshiqpvi}*{\S 1.1}, where we note the following correspondence between the
 $\theta$-parameters and the $\kappa$-parameters used in \cite{roffelsenjoshiqpvi},
\begin{equation*}
    \kappa_0=q^{\theta_0},\quad \kappa_t=q^{\theta_t},\quad \kappa_1=q^{\theta_1},\quad \kappa_\infty=q^{-\theta_\infty}.
\end{equation*}
In particular, $A_{12}(z)$ cannot be identically zero and we parametrise $A$ by introducing the standard coordinates $\{f,g,w\}$, through
\begin{subequations}\label{eq:coordinates_linear}
\begin{align}
    A_{12}(z,t)&=q^{\theta_\infty} w\,(z-f),\\
    A_{22}(f,t)&=q(f-q^{+\theta_1})(f-q^{-\theta_1})g,
\end{align}
\end{subequations}
see Jimbo and Sakai \cite{jimbosakai}. The coefficient matrix is then given in terms of $\{f,g,w\}$ by
\begin{equation*}
    A(z,t)=\begin{pmatrix}
    q^{-\theta_\infty} ((z-f)(z-\alpha)+g_1) & q^{\theta_\infty}\, w(z-f)\\
    q^{-\theta_\infty}\, w^{-1}(\gamma\, z+\delta) & q^{\theta_\infty}((z-f)(z-\beta)+g_2)
    \end{pmatrix},
\end{equation*}
where
\begin{align}
g_1&=q^{-1+\theta_\infty}(f-q^{\theta_t} t)(f-q^{-\theta_t}t)g^{-1},\label{eq:g1}\\ 
g_2&=q^{1-\theta_\infty}(f-q^{\theta_1})(f-q^{-\theta_1})g,\nonumber
\end{align}
and, temporarily using the notation $\mathring{\kappa}_j=q^{\theta_j}+q^{-\theta_j}$, for $j\in\{0,t,1\}$,
\begin{align*}
\alpha&=\frac{1}{(q^{2\theta_\infty}-1)f}
\left(g_1-q^{\theta_\infty}\mathring{\kappa}_0t+q^{2\theta_\infty}g_2+q^{2\theta_\infty}(\mathring{\kappa}_tt+\mathring{\kappa}_1)f-2 q^{2\theta_\infty} f^2\right),\\
\beta&=\frac{1}{(1-q^{2\theta_\infty})f}
\left(g_1-q^{\theta_\infty}\mathring{\kappa}_0t+q^{2\theta_\infty}g_2+(\mathring{\kappa}_tt+\mathring{\kappa}_1)f-2 f^2\right),\\
\gamma&=g_1+g_2+f^2+2(\alpha+\beta)f+\alpha\beta -(t^2+\mathring{\kappa}_t\mathring{\kappa}_1t+1),\\           
\delta&=f^{-1}(t^2-(g_1+\alpha f)(g_2+\beta f)),
\end{align*}

We denote by $Y_0(z,t)$ and $Y_\infty(z,t)$ the canonical solutions of the linear system respectively given by convergent series expansions around $z=0$ and $z=\infty$ of the following form,
\begin{subequations}\label{eq:linear_sys_solutions}
\begin{align}
Y_0(z,t)&=z^{\log_q(t)}\Psi_0(z,t)z^{\theta_0\sigma_3},& \Psi_0(z,t)&=H(t)+\sum_{n=1}^\infty z^n M_n(t),\\
Y_\infty(z,t)&=z^{\log_q(z/q)}\Psi_\infty(z,t) z^{-\theta_\infty\sigma_3},& \Psi_\infty(z,t)&=I+\sum_{n=1}^\infty z^{-n} N_n(t),
\end{align}
\end{subequations}
whose existence follows from the general existence theorems in Carmichael \cite{carmichael1912}. We note that the matrix functions $\Psi_\infty(z,t)$ and $\Psi_0(z,t)^{-1}$ extend to single-valued analytic functions in $z$ on $\mathbb{CP}^1\setminus\{0\}$ and $\mathbb{CP}^1\setminus\{\infty\}$ respectively. 

In Birkhoff's classical treatment \cite{birkhoffgeneralized1913} of the Riemann problem for $q$-difference systems, a crucial role is played by the connection matrix which, in the current context, is defined by
\begin{equation*}
    P(z,t)=Y_0(z,t)^{-1}Y_\infty(z,t),
\end{equation*}
where we remark that $P(z,t)$ is $q$-periodic with respect to $z$, i.e. $P(qz,t)=P(z,t)$.
An analytic characterisation of this matrix is best done via the associated matrix function
\begin{equation*}
    C(z,t):=\Psi_0(z,t)^{-1}\Psi_\infty(z,t)=z^{\log_q(qt/z)}z^{\theta_0\sigma_3}P(z,t)z^{\theta_\infty\sigma_3},
\end{equation*}
as the latter is single-valued in $z$. Indeed, $C(z,t)$ is a $2\times 2$ matrix function which satisfies the following analytic properties.
\begin{enumerate}
    \item It is a single-valued analytic function in $z\in\mathbb{C}^*$.
    \item It satisfies the $q$-difference equation
    \begin{equation}\label{eq:connectionqdif}
        C(qz,t)=\frac{t}{z^2}q^{\theta_0\sigma_3}C(z,t)q^{\theta_\infty\sigma_3}.
    \end{equation}
    \item Its determinant is given by
    \begin{equation}\label{eq:connectiondet}
    |C(z,t)|=c\, \theta_q\left(q^{-\theta_t}\frac{z}{t},q^{+\theta_t}\frac{z}{t},q^{-\theta_1}z,q^{+\theta_1}z\right),
    \end{equation}
    for some $c\in\mathbb{C}^*$.
\end{enumerate}
We accordingly make the following definition.
\begin{definition}\label{def:connection_matrix_space}
We denote by $\mathfrak{C}(\Theta,t)$, for any fixed $t\in\mathbb{C}^*$, the set of all $2\times 2$ matrix functions satisfying properties (1)-(3) above.
\end{definition}
Next we consider deforming the linear system \eqref{eq:linear_problem}, as $t\rightarrow qt$, in such a way that the connection matrix is left invariant. 
As $P(z,t)$ is only rigidly defined up to left-multiplication by invertible diagonal matrices (due to the freedom in choosing $H(t)$), we thus impose
$P(z,qt)=D(t)P(z,t)$ or, equivalently,
\begin{equation}\label{eq:connection_deformation}
    C(z,qt)=z\,D(t)C(z,t),
\end{equation}
for some diagonal matrix $D(t)$.

A principal result in Jimbo and Sakai \cite{jimbonagoyasakai}, is that \eqref{eq:connection_deformation} holds true if and only if $(f,g)$ satisfies $q\Psix$ and $w$ satisfies the auxiliary equation
\begin{equation}\label{eq:auxiliary}
    \frac{\overline{w}}{w}=\frac{q^{1-\theta_\infty}\overline{g}-1}{q^{\theta_\infty}\overline{g}-1}.
\end{equation}
Furthermore, in that case, $H(t)$ may be chosen such that $D(t)=I$.

It follows that we may associate, to a solution $(f,g)$ of $q\Psix(\Theta,t_0)$, the matrix $C_0(z):=C(z,t_0)$, defined via the linear system, up to a natural equivalence. One half of this natural equivalence comes from the fact that $C_0(z)$ is only rigidly defined up to left-multiplication by arbitrary diagonal matrices. The other half comes from the gauge freedom 
\begin{equation*}
    A(z,t)\mapsto DA(z,t)D^{-1},
\end{equation*}
for an arbitrary invertible diagonal matrix $D$,  which rescales $w$, but leaves $(f,g)$ invariant, and has the effect of scaling $C_0(z)$ by $C_0(z)\mapsto C_0(z)D^{-1}$.

For these reasons, the linear system \eqref{eq:linear_problem} yields a mapping
\begin{equation}\label{eq:monodromy_mapping}
    (f,g)\mapsto [C_0(z)],
\end{equation}
where $[C_0(z)]$ denotes the equivalence class of $C_0(z)$ in $\mathfrak{C}(\Theta,t_0)$ with respect to arbitrary left and right-multiplication by invertible diagonal matrices.

\begin{definition}\label{def:monodromy_manifold}
We denote by $\mathcal{M}(\Theta,t_0)$, the space of matrix functions $\mathfrak{C}(\Theta,t_0)$, defined in Definition \ref{def:connection_matrix_space}, quotiented by arbitrary left and right-multiplication by invertible diagonal matrices. We refer to $\mathcal{M}(\Theta,t_0)$ as the monodromy manifold of $q\Psix(\Theta,t_0)$.

Correspondingly, we call the mapping \eqref{eq:monodromy_mapping}, which associates with any solution $(f,g)$ of $q\Psix(\Theta,t_0)$, a corresponding point on the monodromy manifold, the monodromy mapping.
\end{definition}
The space $\mathcal{M}(\Theta,t_0)$ was first introduced by Ohyama et al. \cite{ohyamaramissualoy}. They showed that it naturally comes with an algebraic structure and studied lines on it. They furthermore showed how any element in the monodromy manifold admits a Mano decomposition, which will be a fundamental tool in this paper.

We further emphasise that, contrary to the differential setting, the monodromy manifold depends on the initial point $t_0$. This dependence is non-trivial, since the geometry of the monodromy manifold changes drastically when $t_0$ takes value in some of the $q$-lines in equation \eqref{eq:param_assumptions_2}, see
\cite[Theorem 2.17]{roffelsenjoshiqpvi}.

\subsection{Summary of results in a previous work}\label{subsec:previous}
In this section, we recall some results in \cite{roffelsenjoshiqpvi} that we require in what follows.
\begin{proposition}[Corollary 2.13 in \cite{roffelsenjoshiqpvi}]\label{prop:bijective}
The monodromy mapping, given in Definition \ref{def:monodromy_manifold}, is bijective.
\end{proposition}

We proceed to introduce some coordinates on the monodromy manifold $\mathcal{M}(\Theta,t_0)$. For any $2\times2$ matrix $R$ of rank $1$, let $R_1$ and $R_2$ be respectively its first and second column, so $R=(R_1,R_2)$, then we define $\pi(R)\in\mathbb{CP}^1$ by
\begin{equation*}
R_1=\pi(R)R_2,
\end{equation*}
with $\pi(R)=0$ if and only if $R_1=(0,0)^T$ and $\pi(R)=\infty$ if and only if $R_2=(0,0)^T$.

 Take some monodromy $[C(z)]\in \mathcal{M}(\Theta,t_0)$ and denote
 \begin{equation}\label{eq:intro_xnotation}
    (x_1,x_2,x_3,x_4)=\left(q^{+\theta_t}t_0,q^{-\theta_t}t_0,q^{+\theta_1},q^{-\theta_1}\right).
\end{equation}
Let $1\leq k\leq 4$, then $|C(z)|$ has a simple zero at $z=x_k$ and thus $C(x_k)$ has rank $1$. We define the coordinates
\begin{equation*}
\rho_k=\pi(C(x_k)),\quad (1\leq k\leq 4).
\end{equation*}

Note that $\rho=(\rho_1,\rho_2,\rho_3,\rho_4)$ is invariant under left multiplication of $C(z)$ by diagonal matrices, but right-multiplication translates to scalar multiplication of $\rho$. This means that only the equivalence class $[\rho]\in(\mathbb{CP}^1)^4/\mathbb{C}^*$ is rigidly defined by the monodromy $[C(z)]$. The $\rho_k$'s $1\leq k\leq 4$, are a special case of Tyurin parameters and we will often refer to them by this name.

In \cite{roffelsenjoshiqpvi}, it was shown that $\rho$ satisfies the following quadratic equation (in affine form),
\begin{equation}\label{eq:quadratic}
0=T(\rho):=T_{12}\rho_1\rho_2+T_{13}\rho_1\rho_3+T_{14}\rho_1\rho_4+T_{23}\rho_2\rho_3+T_{24}\rho_2\rho_4+T_{34}\rho_3\rho_4,
\end{equation}
with coefficients given by
\begin{align*}
T_{12}&=+\theta_q\left(q^{\theta_0+\theta_\infty}t_0,q^{-\theta_0+\theta_\infty}t_0\right)\vartheta_\tau(2\theta_t,2\theta_1),\\
T_{34}&=+\theta_q\left(q^{\theta_0-\theta_\infty}t_0,q^{-\theta_0-\theta_\infty}t_0\right)\vartheta_\tau(2\theta_t,2\theta_1)q^{2\theta_\infty},\\
T_{13}&=-\theta_q\left(q^{\theta_t-\theta_1}t_0,q^{-\theta_t+\theta_1}t_0\right)\vartheta_\tau\left(\theta_0+\theta_t+\theta_1+\theta_\infty,-\theta_0+\theta_t+\theta_1+\theta_\infty\right),\\
T_{24}&=-\theta_q\left(q^{\theta_t-\theta_1}t_0,q^{-\theta_t+\theta_1}t_0\right)\vartheta_\tau\left(\theta_0-\theta_t-\theta_1+\theta_\infty,-\theta_0-\theta_t-\theta_1+\theta_\infty\right)q^{2\theta_t+2\theta_1},\\
T_{23}&=+\theta_q\left(q^{\theta_t+\theta_1}t_0,q^{-\theta_t-\theta_1}t_0\right)\vartheta_\tau\left(\theta_0-\theta_t+\theta_1+\theta_\infty,-\theta_0-\theta_t+\theta_1+\theta_\infty\right)q^{2\theta_t},\\
T_{14}&=+\theta_q\left(q^{\theta_t+\theta_1}t_0,q^{-\theta_t-\theta_1}t_0\right)\vartheta_\tau\left(\theta_0+\theta_t-\theta_1+\theta_\infty,-\theta_0+\theta_t-\theta_1+\theta_\infty\right)q^{2\theta_1}.
\end{align*}
To be precise, using homogeneous coordinates $\rho_k=[\rho_k^x: \rho_k^y]\in \mathbb P^1$, $1\le k\le 4$, this equation should be read as
 \begin{equation}\label{eq:Thom}
      0=T_{12}\rho_1^x\rho_2^x\rho_3^y\rho_4^y+T_{13}\rho_1^x\rho_2^y\rho_3^x\rho_4^y+\ldots+T_{34}\rho_1^y\rho_2^y\rho_3^x\rho_4^x.
  \end{equation}
Note that none of the coefficients of $T(\rho)$ vanish, precisely due to parameter assumptions \eqref{eq:param_assumptions_1} and \eqref{eq:param_assumptions_2}.

Denote by
    \begin{equation*}
T'(\rho)=T_{12}'\rho_1\rho_2+T_{13}'\rho_1\rho_3+T_{14}'\rho_1\rho_4+T_{23}'\rho_2\rho_3+T_{24}'\rho_2\rho_4+T_{34}'\rho_3\rho_4,
\end{equation*}
the polynomial $T(\rho)$ with $\theta_0$ set equal to $0$. Then we further recall from \cite{roffelsenjoshiqpvi}, that
\begin{equation}\label{eq:Thom0}
      0\neq T_{12}'\rho_1^x\rho_2^x\rho_3^y\rho_4^y+T_{13}'\rho_1^x\rho_2^y\rho_3^x\rho_4^y+\ldots+T_{34}'\rho_1^y\rho_2^y\rho_3^x\rho_4^x,
  \end{equation}
and equations \eqref{eq:Thom} and \eqref{eq:Thom0} completely describe the space of values of the $\rho$-coordinates on the monodromy manifold, modulo scalar multiplication.

We proceed to define a set of affine coordinates on the monodromy manifold, using inequality \eqref{eq:Thom0}. Take $1\leq i<j\leq 4$ and define the coordinate
\begin{equation}\label{eq:eta_defi}
    \eta_{ij}:=\frac{\vartheta_\tau(+\frac{1}{2},-\frac{1}{2})}{\vartheta_\tau(+\theta_0,-\theta_0)}\frac{T_{ij}\rho_i\rho_j}{T'(\rho)}.
\end{equation}
So, for example, $\eta_{12}$ is given by
\begin{equation*}
\eta_{12}=\frac{\vartheta_\tau(+\frac{1}{2},-\frac{1}{2})}{\vartheta_\tau(+\theta_0,-\theta_0)}\frac{T_{12}\rho_1^x\rho_2^x\rho_3^y\rho_4^y}{T_{12}'\rho_1^x\rho_2^x\rho_3^y\rho_4^y+T_{13}'\rho_1^x\rho_2^y\rho_3^x\rho_4^y+\ldots+T_{34}'\rho_1^y\rho_2^y\rho_3^x\rho_4^x},
\end{equation*}
in homogeneous coordinates. Note that $\eta_{ij}$ is invariant under scalar multiplication $\rho\mapsto c \rho$, $c\in\mathbb{C}^*$, and is furthermore necessarily finite due to inequality \eqref{eq:Thom0}.

Therefore the $\eta$-coordinates form a set of affine coordinates which take value in $\mathbb{C}^6$. It furthermore follows, by construction, that they satisfy the following four equations,
\begin{subequations}\label{eq:eta_equations}
\begin{align}
&\eta_{12}+\eta_{13}+\eta_{14}+\eta_{23}+\eta_{24}+\eta_{34}=0,\label{eq:eta_equationsa}\\
&a_{12}\eta_{12}+a_{13}\eta_{13}+a_{14}\eta_{14}+a_{23}\eta_{23}+a_{24}\eta_{24}+a_{34}\eta_{34}+a_\infty=0,\label{eq:eta_equationsb}\\
    &\eta_{13}\eta_{24}-b_1\eta_{12}\eta_{34}=0,\label{eq:eta_equationsc}\\ 
    &\eta_{14}\eta_{23}-b_2\eta_{12}\eta_{34}=0,\label{eq:eta_equationsd}
\end{align}
\end{subequations}
where
\begin{align*}
    &a_{12}=\prod_{\epsilon=\pm 1}\frac{\theta_q(q^{+\theta_\infty}t_0)}{\theta_q(q^{\epsilon \hspace{0.5mm}\theta_0+\theta_\infty}t_0)}, &
    &a_{34}=\prod_{\epsilon=\pm 1}    \frac{\theta_q(q^{-\theta_\infty}t_0)}{\theta_q(q^{\epsilon \hspace{0.5mm}\theta_0-\theta_\infty}t_0)},\\
    &a_{13}=\prod_{\epsilon=\pm 1} \frac{\vartheta_\tau(\theta_t+\theta_1+\theta_\infty)}{\vartheta_\tau(\epsilon \,\theta_0+\theta_t+\theta_1+\theta_\infty)}, &
    &a_{24}=\prod_{\epsilon=\pm 1}\frac{\vartheta_\tau(-\theta_t-\theta_1+\theta_\infty)}{\vartheta_\tau(\epsilon \,\theta_0-\theta_t-\theta_1+\theta_\infty)},\\
    &a_{23}=\prod_{\epsilon=\pm 1}\frac{\vartheta_\tau(-\theta_t+\theta_1+\theta_\infty)}{\vartheta_\tau(\epsilon \,\theta_0-\theta_t+\theta_1+\theta_\infty)}, &
    &a_{14}=\prod_{\epsilon=\pm 1}\frac{\vartheta_\tau(\theta_t-\theta_1+\theta_\infty)}{\vartheta_\tau(\epsilon \,\theta_0+\theta_t-\theta_1+\theta_\infty)},
\end{align*}
and
\begin{equation*}
    a_\infty=-\frac{\vartheta_\tau(+\frac{1}{2},-\frac{1}{2})}{\vartheta_\tau(+\theta_0,-\theta_0)},\quad b_1=\frac{T_{13}T_{24}}{T_{12}T_{34}},\quad
    b_2=\frac{T_{14}T_{23}}{T_{12}T_{34}}.
\end{equation*}

\begin{definition}\label{def:affine_variety}
We denote by $\mathcal{F}(\Theta,t_0)$ the affine algebraic surface in
\begin{equation*}
    \{(\eta_{12},\eta_{13},\eta_{14},\eta_{23},\eta_{24},\eta_{34})\in\mathbb{C}^6\}
\end{equation*}
defined by equations \eqref{eq:eta_equations}. We correspondingly denote by 
\begin{equation*}
 \Phi:\mathcal{M}(\Theta,t_0)\rightarrow \mathcal{F}(\Theta,t_0), [C(z)]\rightarrow \eta,
\end{equation*}
the mapping defined through the $\eta$-coordinates \eqref{eq:eta_defi}.
\end{definition}
Then we have the following theorem.
\begin{theorem}[Theorem 2.20 in \cite{roffelsenjoshiqpvi}]\label{thm:main_affine}
The mapping $\Phi$, given in Definition \ref{def:affine_variety}, is a bijective correspondence between the
monodromy manifold $\mathcal{M}(\Theta,t_0)$ and the affine algebraic surface $\mathcal{F}(\Theta,t_0)$.
\end{theorem}

Denote the solution space of $q\Psix(\Theta,t_0)$ by
\begin{equation*}
    \mathcal{I}(\Theta,t_0)=\{\text{solutions of $q\Psix(\Theta,t_0)$}\}.
\end{equation*}
Composition of the monodromy mapping, defined in Definition \ref{def:monodromy_manifold}, and the mapping $\Phi$ in Definition \ref{def:affine_variety}, gives a bijective mapping
\begin{equation}\label{eq:rhmapping}
    \textrm{RH}:\mathcal{I}(\Theta,t_0)\rightarrow \mathcal{F}(\Theta,t_0).
\end{equation}
We will consider the $\eta\in \mathcal{F}(\Theta,t_0)$ as a set of global coordinates on the solution space of $q\Psix(\Theta,t_0)$, under this bijection. We note that $\mathcal{I}(\Theta,t_0)$ can naturallly be identified with the initial value space at $t=t_0$ of $q\Psix$, so that it becomes a rational complex surface \cite{s:01}. It can then be shown that $\textrm{RH}$ is a (non-algebraic) biholomorphic mapping. We formulate this as the following proposition.
\begin{proposition}\label{prop:rhdiffeo}
Upon identifying $\mathcal{I}(\Theta,t_0)$ with the initial value space of $q\Psix(\Theta)$ at $t=t_0$, the mapping $\textrm{RH}$, defined in equation \eqref{eq:rhmapping}, is a biholomorphism.
\end{proposition}
The proof of this proposition is technical and can be found in Appendix \ref{appendix:rhdif}.

\subsection{Geometry of an affine Segre surface}\label{subsec:geometry}
In this section, we discuss some classical algebro-geometric aspects of the surface $\mathcal{F}=\mathcal{F}(\Theta,t_0)$.

Firstly, we make the important observation that $\mathcal{F}$ is an affine Segre surface. An easy way to see this, is by using the linear equations \eqref{eq:eta_equationsa} and \eqref{eq:eta_equationsb} to eliminate either the pair $\{\eta_{13},\eta_{23}\}$ or the pair $\{\eta_{14},\eta_{24}\}$. We may always do so, as at least one of the relevant determinants,
\begin{align*}
    &\begin{vmatrix}
    1 & 1\\
    a_{13} & a_{23}
    \end{vmatrix}=-q^{-\theta_t+\theta_1+\theta_\infty}\frac{\vartheta_\tau(\theta_0)\vartheta_\tau(-\theta_0)\vartheta_\tau(2\theta_\infty+2\theta_1,2\theta_t)}{\prod_{\epsilon=\pm 1}\vartheta_\tau(\epsilon\,\theta_0+\theta_t+\theta_1+\theta_\infty,\epsilon\,\theta_0-\theta_t+\theta_1+\theta_\infty)},\\
    &\begin{vmatrix}
    1 & 1\\
    a_{14} & a_{24}
    \end{vmatrix}=-q^{-\theta_t-\theta_1+\theta_\infty}\frac{\vartheta_\tau(\theta_0)\vartheta_\tau(-\theta_0)\vartheta_\tau(2\theta_\infty-2\theta_1,2\theta_t)}{\prod_{\epsilon=\pm 1}\vartheta_\tau(\epsilon\,\theta_0+\theta_t-\theta_1+\theta_\infty,\epsilon\,\theta_0-\theta_t-\theta_1+\theta_\infty)},
\end{align*}
must be nonzero, due to conditions \eqref{eq:param_assumptions_1}. The remaining two equations \eqref{eq:eta_equationsc} and \eqref{eq:eta_equationsd} are then both quadratic in the remaining variables.

Upon introducing projective coordinates
\begin{equation*}
[\eta_{12}:\eta_{13}:\eta_{14}:\eta_{23}:\eta_{24}:\eta_{34}:1]=[N_{12}:N_{13}:N_{14}:N_{23}:N_{24}:N_{34}:N_\infty],
\end{equation*}
equations \eqref{eq:eta_equations} become
\begin{subequations}\label{segre:projective}
\begin{align}
&N_{12}+N_{13}+N_{14}+N_{23}+N_{24}+N_{34}=0,\label{eq:eta_equationsaproj}\\    &a_{12}N_{12}+a_{13}N_{13}+a_{14}N_{14}+a_{23}N_{23}+a_{24}N_{24}+a_{34}N_{34}+a_\infty N_{\infty}=0,\label{eq:eta_equationsbproj}\\
    &N_{13}N_{24}-b_1N_{12}N_{34}=0,\label{eq:eta_equationscproj}\\ 
    &N_{14}N_{23}-b_2N_{12}N_{34}=0.\label{eq:eta_equationsdproj}
\end{align}
\end{subequations}
We denote the corresponding projective completion of $\mathcal{F}$ in $\mathbb{CP}^6$  by $\widehat{\mathcal{F}}$, which is thus a classical Segre (quartic) surface. 
For background information on Segre surfaces, we recommend Dolgachev \cite{dolgachev}*{\S 8.6} and Kunyavski{\i} et al. \cite{kunyavski}. We have the following proposition for $\widehat{\mathcal{F}}$, which is proven in Section \ref{sec:smoothness}.
\begin{proposition}\label{lemma:smooth}
The Segre surface $\widehat{\mathcal{F}}$ is smooth.
\end{proposition}
The hyperplane section at infinity,
\begin{equation}\label{eq:segre_infinity}
    X:=\widehat{\mathcal{F}}\setminus\mathcal{F},
\end{equation}
is isomorphic to the intersection of two quadrics in $\mathbb{CP}^3$, and in this regard we mention an interesting paper by Honda \cite{hondaminitwistor}, where special hyperplane sections of (smooth and singular) Segre surfaces are discussed in relation to minitwistor spaces.

It is classical result, essentially due to Segre himself \cite{Segre}, that smooth Segre surfaces have $16$ lines. To describe these lines on $\widehat{\mathcal{F}}$, we use the Tyurin parameters $\rho$, as well as the dual Tyurin parameters,
\begin{equation*}
    \widetilde{\rho}_k=\pi(C(x_k)^T) \qquad (1\leq k\leq 4).
\end{equation*}
The lines are then indirectly given by
\begin{align*}
    &\mathcal{L}_k^0:\;\; \rho_k=0,& &\mathcal{L}_k^\infty:\;\; \rho_k=\infty, \\
    &\widetilde{\mathcal{L}}_k^0:\;\; \widetilde{\rho}_k=0,& &\widetilde{\mathcal{L}}_k^\infty:\;\; \widetilde{\rho}_k=\infty, 
\end{align*}
in conjunction with equation \eqref{eq:eta_defi}, where $1\leq k\leq 4$.

To make this description more concrete, we introduce the following ratios of Tyurin parameters,
\begin{equation}\label{eq:def_rhoij}
    \rho_{ij}:=\frac{\rho_i}{\rho_j}\in\mathbb{CP}^1,\quad \widetilde{\rho}_{ij}:=\frac{\widetilde{\rho}_i}{\widetilde{\rho}_j}\in\mathbb{CP}^1,\qquad (i\neq j, 1\leq i,j\leq 4).
\end{equation}
These ratios were first introduced by Ohyama et al. \cite{ohyamaramissualoy}, to study the monodromy manifold and construct Mano decompositions, amongst other things.

We call these ratios (dual) Tyurin ratios from here on. Each of them is a rational function on $\widehat{\mathcal{F}}$. For example, $\rho_{12}$ can be written explicitly as
\begin{equation}\label{eq:rho12}
    \rho_{12}=\frac{T_{23}\,\eta_{13}}{T_{13}\,\eta_{23}}=\frac{T_{24}\,\eta_{14}}{T_{14}\,\eta_{24}}.
\end{equation}
In general, for any choice of indices $\{i,j,k,l\}=\{1,2,3,4\}$, 
\begin{equation}\label{eq:rhogen}
    \rho_{ij}=\frac{T_{jk}\,\eta_{ik}}{T_{ik}\,\eta_{jk}}=\frac{T_{jl}\,\eta_{il}}{T_{il}\,\eta_{jl}},
\end{equation}
with the convention that $T_{ij}=T_{ji}$ and $\eta_{ij}=\eta_{ji}$ for $1\leq i,j\leq 4$.

Explicit formulas for the dual Tyurin ratios are given in the following lemma, which is proven in Section \ref{sec:tyurinquotients}.
\begin{lemma}\label{lem:duality}
The Tyurin ratios in equation \eqref{eq:def_rhoij} are related via the transformations
\begin{align}
   \widetilde{\rho}_{ij}&=\frac{x_j}{x_i}M_{(i,j,k,l)}\left( \rho_{ik};\theta_0,\theta_\infty\right)M_{(j,l,k,i)}\left(\rho_{jk},\theta_0,\theta_\infty\right),\label{eq:rhoduality1}\\
    \rho_{ij}&=\frac{x_j}{x_i}M_{(i,j,k,l)}\left( \widetilde{\rho}_{ik};\theta_\infty,\theta_0\right)M_{(j,l,k,i)}\left(\widetilde{\rho}_{jk},\theta_\infty,\theta_0\right),\label{eq:rhoduality2}
\end{align}
for any $\{i,j,k,l\}=\{1,2,3,4\}$, where the $x_i$ are defined as in equation \eqref{eq:intro_xnotation} and $M_{(i,j,k,l)}$ denotes the M\"obius transform
\begin{equation*}
M_{(i,j,k,l)}(Z;\theta_0,\theta_\infty)=\frac{\theta_q\left(\frac{x_i x_l}{t_0}\frac{q^{\theta_\infty}}{q^{\theta_0}} \right)\theta_q\left(  \frac{x_k x_l}{t_0} \frac{q^{-\theta_\infty}}{q^{\theta_0}}\right)Z-\theta_q\left(\frac{x_k x_l}{t_0} \frac{q^{\theta_\infty}}{q^{\theta_0}}\right)\theta_q\left(  \frac{x_i x_l}{t_0} \frac{q^{-\theta_\infty}}{q^{\theta_0}}\right)}{\theta_q\left(\frac{x_i x_j}{t_0} \frac{q^{\theta_\infty}}{q^{\theta_0}}\right)\theta_q\left(  \frac{x_j x_k}{t_0} \frac{q^{-\theta_\infty}}{q^{\theta_0}}\right)Z-\theta_q\left(\frac{x_j x_k}{t_0} \frac{q^{\theta_\infty}}{q^{\theta_0}}\right)\theta_q\left(  \frac{x_i x_j}{t_0} \frac{q^{-\theta_\infty}}{q^{\theta_0}}\right)}.
\end{equation*}
\end{lemma}
\begin{remark}\label{rem:duality}
Fix some labeling $\{i,j,k,l\}=\{1,2,3,4\}$, then the above lemma gives two formulas to compute the dual Tyurin ratio $\widetilde{\rho}_{ij}$, namely, formula \eqref{eq:rhoduality1}, and the same formula with $k$ and $l$ interchanged,
\begin{equation*}
    \widetilde{\rho}_{ij}=\frac{x_j}{x_i}M_{(i,j,l,k)}\left( \rho_{il};\theta_0,\theta_\infty\right)M_{(j,k,l,i)}\left(\rho_{jl},\theta_0,\theta_\infty\right).\label{eq:rhoduality1int}
\end{equation*}
It is possible, for any of the two formulas, that the two M\"obius transforms on the right-hand side evaluate to $0$ and $\infty$, so that their product is ill-defined. There are only two values of the Tyurin parameters $\rho$, up to scalar equivalence, for which this happens in both formulas. One is characterised by
\begin{subequations}\label{eq;mobiussingular}
\begin{align}
    M_{(i,j,k,l)}\left( \rho_{ik};\theta_0,\theta_\infty\right)&=0, & M_{(j,l,k,i)}\left(\rho_{jk},\theta_0,\theta_\infty\right)&=\infty,\\
    M_{(i,j,l,k)}\left( \rho_{il};\theta_0,\theta_\infty\right)&=0, & M_{(j,k,l,i)}\left(\rho_{jl},\theta_0,\theta_\infty\right)&=\infty,
\end{align}
\end{subequations}
and the correct corresponding value of $\widetilde{\rho}_{ij}$ is given by
\begin{equation}\label{eq:singularmobiusdualrho}
    \widetilde{\rho}_{ij}=\frac{\theta_q\left(\frac{x_j}{x_k},\frac{x_l}{x_i},\frac{x_ix_l}{t_0q^{+\theta_0+\theta_\infty}},\frac{x_ix_l}{t_0q^{+\theta_0-\theta_\infty}}\right)}{\theta_q\left(\frac{x_i}{x_k},\frac{x_l}{x_j},\frac{x_jx_l}{t_0q^{+\theta_0+\theta_\infty}},\frac{x_jx_l}{t_0q^{+\theta_0-\theta_\infty}}\right)}.
\end{equation}
The other is characterised by \eqref{eq;mobiussingular}, after swapping $0\leftrightarrow \infty$ on the right-hand side of each equality, and the corresponding value for $\widetilde{\rho}_{ij}$ is then given by equation \eqref{eq:singularmobiusdualrho} after swapping $k\leftrightarrow l$.
\end{remark}

We have the following important fact for the Tyurin ratios, proven in Section \ref{sec:tyurinquotients}.
\begin{lemma}\label{lemma:global_mero}
Each of the Tyurin ratios defined in equation \eqref{eq:def_rhoij} extends to a globally meromorpic function on $\widehat{\mathcal{F}}$.
\end{lemma}

The indirect descriptions of the lines above can now be made precise as follows. For $\diamond\in\{0,\infty\}$ and $1\leq j\leq 4$,
\begin{align}\label{eq:linesdefi}
    \mathcal{L}_j^\diamond&=\{\eta\in\widehat{\mathcal{F}}:\rho_{jk}=\diamond\text{ for }1\leq k\leq 4, k\neq j\},\\
    \widetilde{\mathcal{L}}_j^\diamond&=\{\eta\in\widehat{\mathcal{F}}:\widetilde{\rho}_{jk}=\diamond\text{ for }1\leq k\leq 4, k\neq j\}.\label{eq:linesdualdefi}
\end{align}
In Section \ref{sec:lines}, we prove that these sets are indeed lines and give explicit, linear descriptions for them.

It is a classical fact, due to Segre \cite{Segre}, that any line on a smooth Segre surface, intersects with precisely $5$ other lines, yielding a total of $40$ intersection points.
As to the surface $\widehat{F}$, the intersections of different lines are specified in the following theorem, proven in Section \ref{sec:intersectionpoints}.
\begin{theorem}\label{thm:intersection}
None of the $16$ lines on $\widehat{\mathcal{F}}(\Theta,t_0)$ lie entirely in the hyperplane section at infinity, defined in equation \eqref{eq:segre_infinity}.
The lines intersect as specified in the Clebsch graph in Figure \ref{fig:lines_intersection}. In particular, each line intersects with exactly five other lines, yielding a total of 40 intersection points. Furthermore, under the additional parameter conditions, 
\begin{equation*}
    2\theta_0\pm 2\theta_t,2\theta_0\pm 2\theta_1,2\theta_\infty\pm 2\theta_t, 2\theta_\infty\pm 2\theta_1\notin \Lambda_\tau,
\end{equation*}
and $t_0\notin q^{\mathbb{Z}+U}$, where
\begin{align*}
U=&\{\epsilon_0\theta_0+\epsilon_\infty\theta_\infty+2\epsilon_t\theta_t:\epsilon_{0,t,\infty}\in\{\pm 1\}\}\cup\{\epsilon_t\theta_t+\epsilon_1\theta_1+2\epsilon_\infty\theta_\infty:\epsilon_{t,1,\infty}\in\{\pm 1\}\}\cup\\
&\{\epsilon_0\theta_0+\epsilon_\infty\theta_\infty+2\epsilon_1\theta_1:\epsilon_{0,1,\infty}\in\{\pm 1\}\}\cup\{\epsilon_t\theta_t+\epsilon_1\theta_1+2\epsilon_0\theta_0:\epsilon_{0,t,1}\in\{\pm 1\}\},
\end{align*}
all $40$ intersection points are finite, i.e. they lie in $\mathcal{F}(\Theta,t_0)$.
\end{theorem}

	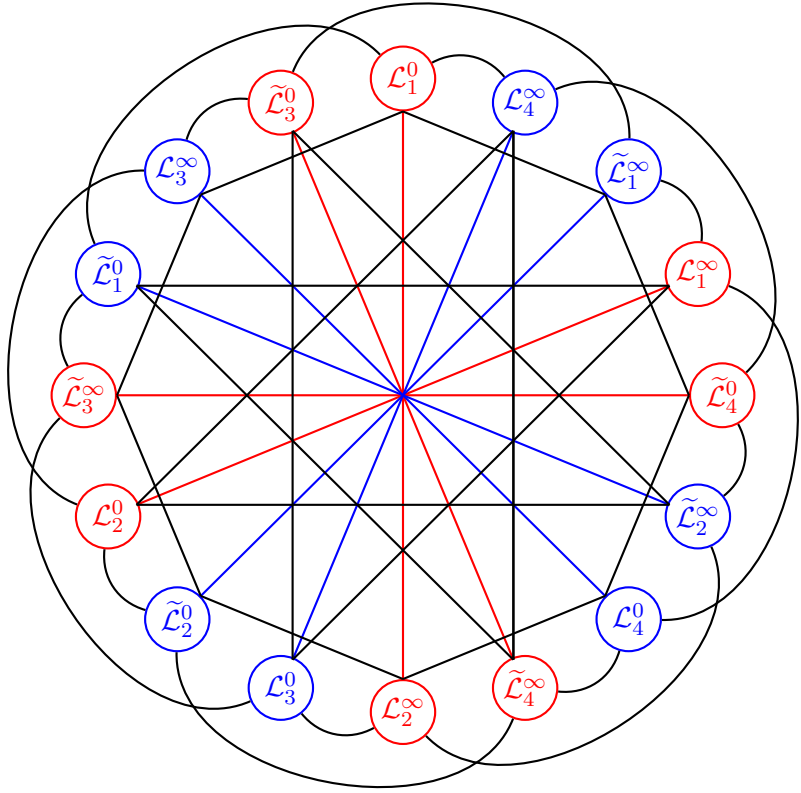
\begin{figure}[t]
		\centering
\begin{tikzpicture}
\tikzset{every state/.style={inner sep=-1pt,minimum size=2.2em}}
\def\ra{4.2}

\node[state,blue,thick] (Li4) at ( {\ra*cos(90-1*22.5)} , {\ra*sin(90-1*22.5)} ) {$\mathcal{L}_4^\infty$};
\node[state,blue,thick] (Lhi1) at ( {\ra*cos(90-2*22.5)} , {\ra*sin(90-2*22.5)} ) {$\widetilde{\mathcal{L}}_1^\infty$};

\node[state,red,thick] (Li1) at ( {\ra*cos(90-3*22.5)} , {\ra*sin(90-3*22.5)} ) {$\mathcal{L}_1^\infty$};
\node[state,red,thick] (Lh04) at ( {\ra*cos(90-4*22.5)} , {\ra*sin(90-4*22.5)} ) {$\widetilde{\mathcal{L}}_4^0$};

\node[state,blue,thick] (Lhi2) at ( {\ra*cos(90-5*22.5)} , {\ra*sin(90-5*22.5)} ) {$\widetilde{\mathcal{L}}_2^\infty$};
\node[state,blue,thick] (L04) at ( {\ra*cos(90-6*22.5)} , {\ra*sin(90-6*22.5)} ) {$\mathcal{L}_4^0$};

\node[state,red,thick] (Lhi4) at ( {\ra*cos(90-7*22.5)} , {\ra*sin(90-7*22.5)} ) {$\widetilde{\mathcal{L}}_4^\infty$};
\node[state,red,thick] (Li2) at ( {\ra*cos(90-8*22.5)} , {\ra*sin(90-8*22.5)} ) {$\mathcal{L}_2^\infty$};

\node[state,blue,thick] (L03) at ( {\ra*cos(90-9*22.5)} , {\ra*sin(90-9*22.5)} ) {$\mathcal{L}_3^0$};
\node[state,blue,thick] (Lh02) at ( {\ra*cos(90-10*22.5)} , {\ra*sin(90-10*22.5)} ) {$\widetilde{\mathcal{L}}_2^0$};

\node[state,red,thick] (L02) at ( {\ra*cos(90-11*22.5)} , {\ra*sin(90-11*22.5)} ) {$\mathcal{L}_2^0$};
\node[state,red,thick] (Lhi3) at ( {\ra*cos(90-12*22.5)} , {\ra*sin(90-12*22.5)} ) {$\widetilde{\mathcal{L}}_3^\infty$};

\node[state,blue,thick] (Lh01) at ( {\ra*cos(90-13*22.5)} , {\ra*sin(90-13*22.5)} ) {$\widetilde{\mathcal{L}}_1^0$};
\node[state,blue,thick] (Li3) at ( {\ra*cos(90-14*22.5)} , {\ra*sin(90-14*22.5)} ) {$\mathcal{L}_3^\infty$};

\node[state,red,thick] (Lh03) at ( {\ra*cos(90-15*22.5)} , {\ra*sin(90-15*22.5)} ) {$\widetilde{\mathcal{L}}_3^0$};
\node[state,red,thick] (L01) at ( {\ra*cos(90-0*22.5)} , {\ra*sin(90-0*22.5)} ) {$\mathcal{L}_1^0$};

\draw[red,thick] (L01) -- (Li2);
\draw[red,thick] (Lh03) -- (Lhi4);
\draw[red,thick] (Lhi3) -- (Lh04);
\draw[red,thick] (L02) -- (Li1);

\draw[blue,thick] (Li3) -- (L04);
\draw[blue,thick] (Lh01) -- (Lhi2);
\draw[blue,thick] (Li4) -- (L03);
\draw[blue,thick] (Lhi1) -- (Lh02);

\draw[black,thick] (Li3.south east) -- (L01.south);
\draw[black,thick] (L01.south) -- (Lhi1.south west);
\draw[black,thick] (Lhi1.south west) -- (Lh04.west);
\draw[black,thick] (Lh04.west) -- (L04.north west);
\draw[black,thick] (L04.north west) -- (Li2.north);
\draw[black,thick] (Li2.north) -- (Lh02.north east);
\draw[black,thick] (Lh02.north east) -- (Lhi3.east);
\draw[black,thick] (Lhi3.east) -- (Li3.south east);

\draw[black,thick] ($(Lh03.south)!0.5!(Lh03.south east)$) -- ($(Lhi2.west)!0.5!(Lhi2.north west)$);
\draw[black,thick] ($(Lhi2.west)!0.5!(Lhi2.north west)$) -- ($(L02.east)!0.5!(L02.north east)$);
\draw[black,thick] ($(L02.east)!0.5!(L02.north east)$) -- ($(Li4.south)!0.5!(Li4.south west)$) ;
\draw[black,thick] ($(Li4.south)!0.5!(Li4.south west)$) -- ($(Lhi4.north west)!0.5!(Lhi4.north)$);
\draw[black,thick] ($(Li4.south)!0.5!(Li4.south west)$) -- ($(Lhi4.north west)!0.5!(Lhi4.north)$);
\draw[black,thick] ($(Lhi4.north west)!0.5!(Lhi4.north)$) -- ($(Lh01.south east)!0.5!(Lh01.east)$);
\draw[black,thick] ($(Lh01.south east)!0.5!(Lh01.east)$) -- ($(Li1.south west)!0.5!(Li1.west)$);
\draw[black,thick] ($(Li1.south west)!0.5!(Li1.west)$) -- ($(L03.north east)!0.5!(L03.north)$);
\draw[black,thick] ($(L03.north east)!0.5!(L03.north)$) -- ($(Lh03.south)!0.5!(Lh03.south east)$);

\path[-]     (Lh03) edge   [bend left=80,thick]   (Lhi1);
\path[-]     (Lhi1) edge   [bend left=40,thick]   (Li1);
\path[-]     (Li1) edge   [bend left=80,thick]   (L04);
\path[-]     (L04) edge   [bend left=40,thick]   (Lhi4);
\path[-]     (Lhi4) edge   [bend left=80,thick]   (Lh02);
\path[-]     (Lh02) edge   [bend left=40,thick]   (L02);
\path[-]     (L02) edge   [bend left=80,thick]   (Li3);
\path[-]     (Li3) edge   [bend left=40,thick]   (Lh03);

\path[-]     (Lh01) edge   [bend left=80,thick]   (L01);
\path[-]     (L01) edge   [bend left=40,thick]   (Li4);
\path[-]     (Li4) edge   [bend left=80,thick]   (Lh04);
\path[-]     (Lh04) edge   [bend left=40,thick]   (Lhi2);
\path[-]     (Lhi2) edge   [bend left=80,thick]   (Li2);
\path[-]     (Li2) edge   [bend left=40,thick]   (L03);
\path[-]     (L03) edge   [bend left=80,thick]   (Lhi3);
\path[-]     (Lhi3) edge   [bend left=40,thick]   (Lh01);
\end{tikzpicture}
\caption{Clebsch graph encoding the configuration of lines and their points of intersection on the Segre surface $\widehat{\mathcal{F}}(\Theta,t_0)$. Each vertex represents a line and each edge represents an intersection point of the two lines corresponding to its endpoints. Blue (resp. red) vertices correspond to lines on which the asymptotics at $t=0$ (resp. $t=\infty$) are truncated, under suitable parameter assumptions. Blue (resp. red) edges correspond to intersection points at which the asymptotics at $t=0$ (resp. $t=\infty$) are doubly truncated. Black edges correspond to intersection points where the asymptotics at $t=0$ and $t=\infty$ are simultaneously truncated, under certain parameter conditions.}
\label{fig:lines_intersection}
\end{figure}

\begin{remark}
  As this work was being finalised, an interesting preprint was placed on the arXiv by Ramis and Sauloy \cite{ramis2023geometry}, regarding the geometry of the monodromy manifold, the affine Segre surface and related topics.  
  They independently show that the lines on the affine Segre surface are given by equations \eqref{eq:linesdefi}, see \cite{ramis2023geometry}*{Proposition 5.8}, after this result was communicated privately by the author.
\end{remark}

\subsection{Generic asymptotics near $\boldsymbol{t=0}$}\label{subsec:asympzero}
In this section we give the asymptotics at $t=0$ of generic solutions of $q\Psix(\Theta,t_0)$.
 Let $\mathcal{E}_0$ denote the elliptic function
\begin{equation*}
    \mathcal{E}_0(\sigma)=\frac{\vartheta_\tau(\sigma-\theta_1+\theta_\infty,\sigma+\theta_1-\theta_\infty)}{\vartheta_\tau(\sigma+\theta_1+\theta_\infty,\sigma-\theta_1-\theta_\infty)}.
\end{equation*}
This is an elliptic function of degree $2$ with periods $1$ and $\tau$, that plays an important geometric role in the Mano decomposition discussed in Section \ref{sec:dec_algebraic}, and was first derived by Ohyama et al. \cite{ohyamaramissualoy}*{\S 5.1.3} in that context. It has the reflection symmetry $\mathcal{E}_0(-\sigma)=\mathcal{E}_0(\sigma)$ and, therefore, the equation 
\begin{equation}\label{eq:elliptic_equation0}
    \mathcal{E}_0(\sigma)=\rho_{34},
\end{equation}
has two solutions (counting multiplicity) in the fundamental domain
\begin{equation*}
    \{\sigma\in\mathbb{C}:-\tfrac{1}{2}\leq \Re\sigma<\tfrac{1}{2}\text{ and } \tfrac{1}{2}i\tau \leq \Im\sigma <-\tfrac{1}{2}i\tau\}.
\end{equation*}
It then follows from the reflection symmetry that \eqref{eq:elliptic_equation0} has at least one solution, which we denote by $\sigma_{0t}$, in the domain
\begin{equation}\label{eq:domain0}
    \{\sigma\in\mathbb{C}:0\leq \Re\sigma\leq \tfrac{1}{2}\text{ and } \tfrac{1}{2}i\tau \leq \Im\sigma <-\tfrac{1}{2}i\tau\}.
\end{equation}
 If $0<\Re \sigma_{0t}<\tfrac{1}{2}$, then the solution is unique, otherwise a second solution is given by its complex conjugate $\overline{\sigma}_{0t}$, or $\overline{\sigma}_{0t}-\tau$, in the domain \eqref{eq:domain0}, which may coincide with $\sigma_{0t}$.

\begin{definition}\label{def:generic0}
We call a point $\eta$, in the affine Segre surface $\mathcal{F}(\Theta,t_0)$, $0$-generic if $\sigma_{0t}$ satisfies
\begin{equation*}
    0<\Re \sigma_{0t}<\tfrac{1}{2}, 
\end{equation*}
and
\begin{equation}\label{eq:line_conditions}
    \sigma_{0t}\not\equiv\pm (\theta_0+\theta_t),\pm (\theta_0-\theta_t),\pm (\theta_\infty+\theta_1),\pm (\theta_\infty-\theta_1) \mod{\Lambda_\tau}.
\end{equation}
\end{definition}
It follows from conditions \eqref{eq:line_conditions} that a $0$-generic point in $\mathcal{F}$ does not lie on any of the eight lines denoted by blue vertices in Figure \ref{fig:lines_intersection}.
In the following theorem, proven in Section \ref{sec:extract_asymp}, we give the asymptotics of generic solutions near $t=0$.
\begin{theorem}\label{thm:generic_asymp_zero}
Take any $0$-generic point $\eta$ in the affine Segre surface $\mathcal{F}(\Theta,t_0)$ and let $(f,g)$ denote the corresponding solution of $q\Psix(\Theta,t_0)$. Denote $t_m=q^mt_0$, $m\in\mathbb{Z}$, then $f$ and $g$ have complete asymptotic expansions as $t_m\rightarrow 0$, of the form,
\begin{align*}
    f(t_m)&=\sum_{n=1}^\infty\sum_{k=-n}^n F_{n,k}r_{0t}^k(- t_m)^{n+2k\sigma_{0t}},\\
    g(t_m)&=\sum_{n=1}^\infty\sum_{k=-n}^n G_{n,k}r_{0t}^k(- t_m)^{n+2k\sigma_{0t}},
\end{align*}
which are absolutely convergent for large enough $m\geq 0$, where
\begin{align*}
    F_{1,\pm 1}&=q^{-\theta_t}\frac{\bigl(q^{\theta_t+\theta_0\mp\sigma_{0t}}-1\bigr)\bigl(q^{\theta_t-\theta_0\mp\sigma_{0t}}-1\bigr)\bigl(q^{\theta_1+\theta_\infty\mp\sigma_{0t}}-1\bigr)}{\bigl(q^{\theta_1+\theta_\infty\pm \sigma_{0t}}-1\bigr)\bigl(q^{\sigma_{0t}}-q^{-\sigma_{0t}}\bigr)^2},\\
    F_{1,0}&=\frac{2\bigl(q^{\theta_t}+q^{-\theta_t}\bigr)-\bigl(q^{\theta_0}+q^{-\theta_0}\bigr)\bigl(q^{ \sigma_{0t}}+q^{- \sigma_{0t}}\bigr)}{\bigl(q^{ \sigma_{0t}}-q^{- \sigma_{0t}}\bigr)^2},\\
    G_{1,0}&=\frac{2\bigl(q^{ \theta_0}+q^{- \theta_0}\bigr)-\bigl(q^{ \theta_t}+q^{- \theta_t}\bigr)\bigl(q^{ \sigma_{0t}}+q^{- \sigma_{0t}}\bigr)}{\bigl(q^{ \sigma_{0t}}-q^{- \sigma_{0t}}\bigr)^2}q^{-1},\\
    G_{1,\pm 1}&=-q^{-1\mp\sigma_{0t}}F_{1,\pm 1},
\end{align*}
and the higher order coefficients may be computed recursively via the $q\Psix$ equation, where we note that each coefficient
\begin{equation}\label{eq:coefficients_zero}
    F_{n,k}=F_{n,k}(\Theta,\sigma_{0t}),\quad
    G_{n,k}=G_{n,k}(\Theta,\sigma_{0t})\qquad (-n\leq k \leq n, n\geq 1),
\end{equation}
only depends on the parameters $\Theta$ and $\sigma_{0t}$.
The branches of the complex powers are principal and the pair of integration constants $\{\sigma_{0t},r_{0t}\}$ is related to $\eta$, via the pair of Tyurin ratios $\{\rho_{24},\rho_{34}\}$, as follows.  The exponent $\sigma_{0t}$ is defined as the unique solution of equation \eqref{eq:elliptic_equation0} in the domain \eqref{eq:domain0}, and
\begin{equation*}
    r_{0t}=c_{0t}\times s_{0t},
\end{equation*}
with
\begin{align*}
c_{0t}&=\frac{\Gamma_q(1-2\sigma_{0t})^2}{\Gamma_q(1+2\sigma_{0t})^2}\prod_{\epsilon=\pm1}\frac{ \Gamma_q(1+\theta_t+\epsilon\,\theta_0+\sigma_{0t})  \Gamma_q(1+\theta_1+\epsilon\,\theta_\infty+\sigma_{0t})}{ \Gamma_q(1+\theta_t+\epsilon\,\theta_0-\sigma_{0t})  \Gamma_q(1+\theta_1+\epsilon\,\theta_\infty-\sigma_{0t})},\\
s_{0t}&=-(-t_0)^{-2\sigma_{0t}}M_{0t}(\rho_{24}),
\end{align*}
where $M_{0t}(\cdot)$ is the M\"obius transformation
\begin{equation*}
    M_{0t}(Z)=\frac{\vartheta_\tau(\theta_1-\theta_\infty+\sigma_{0t})\theta_q(q^{\theta_t+\theta_\infty+\sigma_{0t}}t_0^{-1})-Z
\vartheta_\tau(\theta_1+\theta_\infty+\sigma_{0t})\theta_q(q^{\theta_t-\theta_\infty+\sigma_{0t}}t_0^{-1})}{\vartheta_\tau(\theta_1-\theta_\infty-\sigma_{0t})\theta_q(q^{\theta_t+\theta_\infty-\sigma_{0t}}t_0^{-1})-Z
\vartheta_\tau(\theta_1+\theta_\infty-\sigma_{0t})\theta_q(q^{\theta_t-\theta_\infty-\sigma_{0t}}t_0^{-1})}.
\end{equation*}
\end{theorem}

%\begin{align*}
%s&=-q^{\sigma_{0t}}(\shortminus t_0)^{-\sigma_{0t}} \frac{\vartheta_\tau(\frac{1}{2}(\theta_1-\theta_\infty+\sigma_{0t}))\theta_q(q^{-\frac{1}{2}(\theta_t+\theta_\infty+\sigma_{0t})}t_0)-q^{-\theta_\infty}\rho_{24}
%\vartheta_\tau(\frac{1}{2}(\theta_1+\theta_\infty+\sigma_{0t}))\theta_q(q^{-\frac{1}{2}(\theta_t-\theta_\infty+\sigma_{0t})}t_0)}{\vartheta_\tau(\frac{1}{2}(\theta_1-\theta_\infty-\sigma_{0t}))\theta_q(q^{-\frac{1}{2}(\theta_t+\theta_\infty-\sigma_{0t})}t_0)-q^{-\theta_\infty}\rho_{24}
%\vartheta_\tau(\frac{1}{2}(\theta_1+\theta_\infty-\sigma_{0t}))\theta_q(q^{-\frac{1}{2}(\theta_t-\theta_\infty-\sigma_{0t})}t_0)},\\
%\end{align*}

\begin{remark}
The integration constants $\{\sigma_{0t},s_{0t}\}$ form a pair of local coordinates on $\mathcal{F}(\Theta,t_0)$. Geometrically, $\sigma_{0t}$ plays the role of the intermediate exponent and $s_{0t}$ plays the role of a twist parameter in a Mano decomposition of the monodromy, as detailed in Section \ref{sec:dec_algebraic}. We further note that, in the above theorem, we have related these integration constants to $\eta$ using the intermediate pair of Tyurin ratios $\{\rho_{34},\rho_{24}\}$, which also define local coordinates on the monodromy manifold.
\end{remark}
\begin{remark}\label{remark:weakeningzero}
The asymptotic expansions in Theorem \ref{thm:generic_asymp_zero} remain valid when $\Re \sigma_{0t}=0$, as long as $\sigma_{0t}\not\equiv 0,\tfrac{\tau}{2}$ modulo $\Lambda_\tau$. In this case the leading order behaviour of the solution is damped oscillatory as $t_m\rightarrow 0$. This is proven at the end of Section \ref{sec:extract_asymp}.
In the exceptional cases, when $\sigma_{0t}\equiv 0$ or $\sigma_{0t}\equiv \tfrac{\tau}{2}$ modulo $\Lambda_\tau$,  the corresponding Mano decomposition is logarithmic, see Ohyama et al. \cite{ohyamaramissualoy}*{Theorem 5.13(2)}. Drawing on analogy with $\Psix$, see Guzzetti \cite{guzzettiloga}, one would expect that the solution admits asymptotic expansions as $t_m\rightarrow 0$ which are power series in $(-t_m)$ and $\log_q(-t_m)=m\log_q(-t_0)$. This case will not be studied in this paper. We further note that formal asymptotic behaviours expressed in terms of $\log_q(t)$ were found for solutions of the $q$-difference Painlev\'e equation $qP(A_1)$, in
\cite{joshi_roffelsenqpa1}*{\S 3.3}.
\end{remark}

\begin{remark}
For generic $\sigma_{0t}$, the value of $s_{0t}$ cannot be $0$ or $\infty$ on $\mathcal{F}(\Theta,t_0)$. This can for example be derived from inequality \eqref{eq:Thom0}. Geometrically, upon fixing the value of the exponent $\sigma_{0t}$, $\eta=\eta(s_{0t})$ traces out a conic in $\widehat{\mathcal{F}}$. The values $s_{0t}=0,\infty$ correspond to the two intersections of this conic with the hyperplane section at infinity, $X=\widehat{\mathcal{F}}\setminus \mathcal{F}$.
\end{remark}
\begin{remark}
Note that the asymptotic formulas in the above theorem are independent of the choice of branch of the complex powers, as
\begin{equation*}
    r_{0t}^k(-t_m)^{2k\sigma_{0t}}=-c_{0t}^kM_0(\rho_{24})^k q^{2k\sigma_{0t}}.
\end{equation*}
When considering the continuum limit, the choice of the branch cut along $[0,\infty)$ becomes very convenient, since, as $q\rightarrow 1$, the four $q$-lines on the right-hand side of equation \eqref{eq:param_assumptions_2} condensate on the positive real line $[0,\infty)$. In other words, all the problematic points in the $t$-plane for our approach asymptotically lie on the branch cut $[0,\infty)$.
\end{remark}
\begin{remark}\label{remark:JNS}
We compared the terms with $n=1$ in the asymptotic expansion of $f(t_m)$ against the tau-function expansions in \cite{jimbonagoyasakai}. They match under the identification
\begin{equation*}
\Theta=\Theta^{\text{JNS}},\qquad f(t)=q^{\theta_1+1}y^{\text{JNS}}(q^{\theta_t+\theta_1}t),\qquad \sigma_{0t}=\sigma^{\text{JNS}},
\end{equation*}
and
\begin{align*}
\frac{r_{0t}}{s^{\text{JNS}}}=&(-1)^{1-2\sigma_{0t}}q^{2(\theta_t+\theta_1)\sigma_{0t}}\frac{1-q^{\theta_1+\theta_\infty+\sigma_{0t}}}{1-q^{\theta_1+\theta_\infty-\sigma_{0t}}}\frac{\Gamma_q(-2\sigma_{0t})^2}{\Gamma_q(+2\sigma_{0t})^2}\times\\
&\prod_{\epsilon=\pm 1}\frac{\Gamma_q(\frac{1}{2}-\theta_1+\epsilon\theta_\infty+\sigma_{0t})\Gamma_q(\theta_t+\epsilon \theta_0+\sigma_{0t})}{\Gamma_q(\frac{1}{2}-\theta_1+\epsilon\theta_\infty-\sigma_{0t})\Gamma_q(\theta_t+\epsilon \theta_0-\sigma_{0t})}.
\end{align*}
Comparison with the asymptotics by Mano \cite{manoqpvi} is given in Remark \ref{rem:mano_compare}.
\end{remark}

As explained in Remark \ref{remark:weakeningzero}, the asymptotic expansions in Theorem \ref{thm:generic_asymp_zero} remain valid when $\Re \sigma_{0t}=0$.
On the other hand, when $\Re \sigma_{0t}=\tfrac{1}{2}$, the series representations of $f$ and $g$ break down as they are no longer asymptotic expansions (they e.g. contain infinitely many terms of order a constant).
To cover this case, we have the dual asymptotic expansions, detailed in the following theorem and proven in Section \ref{sec:extract_asymp_zero_dual}.
\begin{theorem}\label{thm:generic_asymp_zero_dual}
Take any $0$-generic point $\eta$ in the affine Segre surface $\mathcal{F}(\Theta,t_0)$ and let $(f,g)$ denote the corresponding solution of $q\Psix(\Theta,t_0)$. Recall the notation in equation \eqref{eq:coefficients_zero} for the coefficients of the series in Theorem \ref{thm:generic_asymp_zero}.
Denote $t_m=q^mt_0$, $m\in\mathbb{Z}$, then $t/f$ and $t/g$ have complete asymptotic expansions as $t_m\rightarrow 0$, of the form,
\begin{align*}
    \frac{t_m}{f(t_m)}&=\sum_{n=1}^\infty\sum_{k=-n}^n F_{n,k}(\widehat{\Theta},\widehat{\sigma}_{0t})\widehat{r}_{0t}^k(- t_m)^{n+2k\widehat{\sigma}_{0t}},\\
    \frac{t_m}{q^{\frac{3}{2}}g(t_m)}&=\sum_{n=1}^\infty\sum_{k=-n}^n G_{n,k}(\widehat{\Theta},\widehat{\sigma}_{0t})\widehat{r}_{0t}^k(- t_m)^{n+2k\widehat{\sigma}_{0t}},
\end{align*}
which are absolutely convergent for large enough $m\geq 0$,
where
\begin{equation*}
\widehat{\theta}_0=\theta_\infty-\tfrac{1}{2},\quad
\widehat{\theta}_t=\theta_1,\quad
\widehat{\theta}_1=\theta_t,\quad
\widehat{\theta}_\infty=\theta_0+\tfrac{1}{2},
\end{equation*}
the branches of the complex powers are principal and the pair of integration constants $\{\widehat{\sigma}_{0t},\widehat{r}_{0t}\}$ is related to the pair of integration constants $\{\sigma_{0t},r_{0t}\}$ in Theorem \ref{thm:generic_asymp_zero}, by
\begin{equation*}
    \widehat{\sigma}_{0t}=\tfrac{1}{2}-\sigma_{0t},\quad \widehat{r}_{0t}=-\frac{F_{1,-1}(\Theta,\sigma_{0t})F_{1,-1}(\widehat{\Theta},\widehat{\sigma}_{0t})}{r_{0t}}=-\frac{F_{1,-1}(\Theta,\sigma_{0t})F_{1,-1}(\widehat{\Theta},\widehat{\sigma}_{0t})}{c_{0t}s_{0t}}.
\end{equation*}
The latter pair is in turn related to $\eta$ via
 the pair of dual Tyurin ratios $\{\widetilde{\rho}_{12},\widetilde{\rho}_{13}\}$, by
\begin{equation*}
   \frac{\vartheta_\tau(\sigma_{0t}-\theta_t+\theta_0,\sigma_{0t}+\theta_t-\theta_0)}{\vartheta_\tau(\sigma_{0t}-\theta_t-\theta_0,\sigma_{0t}+\theta_t+\theta_0)}= \widetilde{\rho}_{12},
\end{equation*}
and
\begin{equation}\label{eq:twist0talt}
s_{0t}=-(-t_0)^{-2\sigma_{0t}}\widehat{M}_{0t}(\widetilde{\rho}_{13}),
\end{equation}
where $\widehat{M}_{0t}(\cdot)$ is the M\"obius transformation
\begin{equation*}
    \widehat{M}_{0t}(Z)=\frac{\vartheta_\tau(\theta_t-\theta_0+\sigma_{0t})\theta_q(q^{\theta_1+\theta_0+\sigma_{0t}}t_0^{-1})-Z
\vartheta_\tau(\theta_t+\theta_0+\sigma_{0t})\theta_q(q^{\theta_1-\theta_0+\sigma_{0t}}t_0^{-1})}{\vartheta_\tau(\theta_t-\theta_0-\sigma_{0t})\theta_q(q^{\theta_1+\theta_0-\sigma_{0t}}t_0^{-1})-Z
\vartheta_\tau(\theta_t+\theta_0-\sigma_{0t})\theta_q(q^{\theta_1-\theta_0-\sigma_{0t}}t_0^{-1})}.
\end{equation*}
\end{theorem}
\begin{remark}\label{remark:weakeningzerodual}
The asymptotic expansions in Theorem \ref{thm:generic_asymp_zero_dual} remain valid when $\Re \sigma_{0t}=\frac{1}{2}$, as long as $\sigma_{0t}\not\equiv \tfrac{1}{2}, \tfrac{1}{2}+\tfrac{\tau}{2}$ modulo $\Lambda_\tau$. In this case, the leading order behaviours of $f(t_m)^{-1}$ and $g(t_m)^{-1}$ are oscillatory as $t_m\rightarrow 0$. 
\end{remark}

Note that the two different expressions for the twist parameter $s_{0t}$, in theorems
 \ref{thm:generic_asymp_zero} and \ref{thm:generic_asymp_zero_dual}, imply the remarkably identity,
 \begin{equation*}
     M_{0t}(\rho_{24})=\widehat{M}_{0t}(\widetilde{\rho}_{13}).
 \end{equation*}

\subsection{Generic asymptotics near $\boldsymbol{t=\infty}$}\label{subsec:asympinfty}
In this section, we discuss the asymptotics of solutions of $q\Psix$ around $t=\infty$. Let $\mathcal{E}_\infty$ denote the elliptic function
\begin{equation*}
    \mathcal{E}_\infty(\sigma)=\frac{\vartheta_\tau(\sigma-\theta_t+\theta_\infty,\sigma+\theta_t-\theta_\infty)}{\vartheta_\tau(\sigma+\theta_t+\theta_\infty,\sigma-\theta_t-\theta_\infty)}.
\end{equation*}
This is an elliptic function of degree $2$ with periods $1$ and $\tau$, that plays an important geometric role in the Mano decomposition discussed in Section \ref{sec:dec_algebraicII}, and again was first derived by Ohyama et al. \cite{ohyamaramissualoy}*{\S 5.1.3} in that context.
It has the reflection symmetry $\mathcal{E}_\infty(-\sigma)=\mathcal{E}_\infty(\sigma)$ and, therefore, the equation 
\begin{equation}\label{eq:elliptic_equationinf}
    \mathcal{E}_\infty(\sigma)=\rho_{12},
\end{equation}
has at least one solution, which we denote by $\sigma_{01}$, in the domain
\eqref{eq:domain0}. If $0<\Re \sigma_{01}<\tfrac{1}{2}$, then the solution is unique, otherwise a second solution is given by its complex conjugate $\overline{\sigma}_{01}$, or $\overline{\sigma}_{01}-\tau$, in the domain \eqref{eq:domain0}, which may coincide with $\sigma_{01}$.

\begin{definition}\label{def:genericinfty}
We call a point $\eta$, in the affine Segre surface $\mathcal{F}(\Theta,t_0)$, $\infty$-generic if $\sigma_{01}$ satisfies
\begin{equation*}
    0<\Re \sigma_{01}<\tfrac{1}{2}, 
\end{equation*}
and
\begin{equation*}
    \sigma_{01}\not\equiv\pm (\theta_0+\theta_1),\pm (\theta_0-\theta_1),\pm (\theta_\infty+\theta_t),\pm (\theta_\infty-\theta_t) \mod{\Lambda_\tau}.
\end{equation*}
\end{definition}
The following theorem, proven in Section \ref{sec:extract_asymp_infty}, describes the asymptotics of generic solutions near $t=\infty$.
\begin{theorem}\label{thm:generic_asymp_infty}
Take any $\infty$-generic point $\eta$ in the affine Segre surface $\mathcal{F}(\Theta,t_0)$ and let $(f,g)$ denote the corresponding solution of $q\Psix(\Theta,t_0)$. Denote $t_m=q^mt_0$, $m\in\mathbb{Z}$, then $t^{-1}f$ and $g^{-1}$ have complete asymptotic expansions as $t_m\rightarrow \infty$ of the form,
\begin{align*}
    \frac{f(t_m)}{t_m}&=\sum_{n=1}^\infty\sum_{k=-n}^n \dot{F}_{n,k}r_{01}^k(- t_m)^{-(n+2k\sigma_{01})},\\
    \frac{1}{g(t_m)}&=\sum_{n=1}^\infty\sum_{k=-n}^n \dot{G}_{n,k}r_{01}^k(- t_m)^{-(n+2k\sigma_{01})},
\end{align*}
which are absolutely convergent for large enough $m\leq 0$,
where
\begin{align*}
    \dot{F}_{1,\pm 1}&=q^{-\theta_1}\frac{\bigl(q^{\theta_1+\theta_0\mp\sigma_{01}}-1\bigr)\bigl(q^{\theta_1-\theta_0\mp\sigma_{01}}-1\bigr)\bigl(q^{\theta_t+\theta_\infty\mp\sigma_{01}}-1\bigr)}{\bigl(q^{\theta_t+\theta_\infty\pm \sigma_{01}}-1\bigr)\bigl(q^{\sigma_{01}}-q^{-\sigma_{01}}\bigr)^2},\\
    \dot{F}_{1,0}&=\frac{2\bigl(q^{\theta_1}+q^{-\theta_1}\bigr)-\bigl(q^{\theta_0}+q^{-\theta_0}\bigr)\bigl(q^{\sigma_{01}}+q^{-\sigma_{01}}\bigr)}{\bigl(q^{\sigma_{01}}-q^{-\sigma_{01}}\bigr)^2},\\
    \dot{G}_{1,0}&=\frac{2\bigl(q^{\theta_0}+q^{-\theta_0}\bigr)-\bigl(q^{\theta_1}+q^{-\theta_1}\bigr)\bigl(q^{\sigma_{01}}+q^{-\sigma_{01}}\bigr)}{\bigl(q^{\sigma_{01}}-q^{-\sigma_{01}}\bigr)^2}q,\\
    \dot{G}_{1,\pm 1}&=-q^{1\pm \sigma_{01}}\dot{F}_{1,\pm 1},
\end{align*}
and the higher order coefficients may be computed recursively via the $q\Psix$ equation,  where we note that each coefficient
\begin{equation}\label{eq:coefficients_infty}
    \dot{F}_{n,k}=\dot{F}_{n,k}(\Theta,\sigma_{01}),\quad
    \dot{G}_{n,k}=\dot{G}_{n,k}(\Theta,\sigma_{01})\qquad (-n\leq k \leq n, n\geq 1),
\end{equation}
only depends on the parameters $\Theta$ and $\sigma_{01}$.
Here the branches of the complex powers are principal and the pair of integration constants $\{\sigma_{01},r_{01}\}$ is related to $\eta$, via the pair of Tyurin ratios $\{\rho_{12},\rho_{24}\}$, as follows.  The exponent $\sigma_{01}$ is defined as the unique solution of equation \eqref{eq:elliptic_equationinf} in the domain \eqref{eq:domain0}, and
\begin{equation*}
    r_{01}=c_{01}\times s_{01},
\end{equation*}
with
\begin{align*}
c_{01}&=\frac{\Gamma_q(1-2\sigma_{01})^2}{\Gamma_q(1+2\sigma_{01})^2}\prod_{\epsilon=\pm1}\frac{ \Gamma_q(1+\theta_1+\epsilon\,\theta_0+\sigma_{01})  \Gamma_q(1+\theta_t+\epsilon\,\theta_\infty+\sigma_{01})}{ \Gamma_q(1+\theta_1+\epsilon\,\theta_0-\sigma_{01})  \Gamma_q(1+\theta_t+\epsilon\,\theta_\infty-\sigma_{01})},\\
s_{01}&=-(-t_0)^{2\sigma_{01}}M_{01}(\rho_{24}),
\end{align*}
where $M_{01}(\cdot)$ is the M\"obius transformation
\begin{equation*}
    M_{01}(Z)=\frac{\vartheta_\tau(\theta_t+\theta_\infty+\sigma_{01})\theta_q(q^{\theta_1-\theta_\infty+\sigma_{01}}t_0)-Z
\vartheta_\tau(\theta_t-\theta_\infty+\sigma_{01})\theta_q(q^{\theta_1+\theta_\infty+\sigma_{01}}t_0)}{\vartheta_\tau(\theta_t+\theta_\infty-\sigma_{01})\theta_q(q^{\theta_1-\theta_\infty-\sigma_{01}}t_0)-Z
\vartheta_\tau(\theta_t-\theta_\infty-\sigma_{01})\theta_q(q^{\theta_1+\theta_\infty-\sigma_{01}}t_0)}.
\end{equation*}
\end{theorem}
\begin{remark}
The integration constants $\{\sigma_{01},s_{01}\}$ form a pair of local coordinates on $\mathcal{F}(\Theta,t_0)$. Geometrically, $\sigma_{01}$ plays the role of the intermediate exponent and $s_{01}$ plays the role of a twist parameter in a Mano decomposition of the monodromy, as detailed in Section \ref{sec:dec_algebraicII}. We further note that, in the above theorem, we have related these integration constants to $\eta$ using the intermediate pair of Tyurin ratios $\{\rho_{12},\rho_{24}\}$, which also define local coordinates on the monodromy manifold.
\end{remark}

\begin{remark}\label{remark:coefficients}
The coefficients in the series expansions around $t=\infty$, see equation \eqref{eq:coefficients_infty}, and the coefficients in the series expansions around $t=0$, see equation \eqref{eq:coefficients_zero}, are interesting combinatorical objects in themselves, related by
\begin{equation*}
   \dot{F}_{n,k}=F_{n,k}|_{\theta_t\leftrightarrow \theta_1,\sigma_{0t}\mapsto \sigma_{01}},\quad
   \dot{G}_{n,k}=q^{1+n+2k\sigma_{01}}G_{n,k}|_{\theta_t\leftrightarrow \theta_1,\sigma_{0t}\mapsto \sigma_{01}},
\end{equation*}
for $-n\leq k\leq n$ and $n\geq 1$.
\end{remark}
\begin{remark}
The series representations in Theorem \ref{thm:generic_asymp_infty} remain valid when $\Re \sigma_{01}=0$, as long as $\sigma_{01}\not \equiv 0,\tfrac{1}{2}\tau$ modulo $\Lambda_\tau$. In this case, the leading order behaviours of $f(t_m)$ and $t_m/g(t_m)$ are oscillatory as $t_m\rightarrow \infty$.
\end{remark}
\begin{remark}\label{rem:mano_compare}
We compared the leading order terms in the series expansions in Theorems \ref{thm:generic_asymp_zero} and \ref{thm:generic_asymp_infty} with the corresponding expansions in Mano \cite{manoqpvi}*{Theorem 5.1}. After identifying the parameters by
\begin{align*}
\theta_1^M&=q^{+\theta_0}, & \theta_2^M&=q^{-\theta_0}, &a_1^M&=q^{+\theta_t}, & a_2^M=q^{-\theta_t},  \\
\kappa_1^M&=q^{-\theta_\infty}, & \kappa_2^M&=q^{+\theta_\infty}, &a_3^M&=q^{+\theta_1}, & a_4^M=q^{-\theta_1},  
\end{align*}
they match, with the integration constants related by
\begin{align*}
    \sigma_1&=-q^{+\sigma_{0t}}, & \sigma_2&=-q^{-\sigma_{0t}}, & \frac{\widetilde{w}_0}{w_0}&=(-1)^{1+2\sigma_{0t}}(q^{2\sigma_{0t}}-1)F_{1,1} r_{0t}, \\
   \overline{\sigma}_1&=-q^{+\sigma_{01}}   & \overline{\sigma}_2&=-q^{-\sigma_{01}}, & \frac{\widetilde{\overline{w}}_0}{\overline{w}_0}&=(-1)^{1+2\sigma_{01}}(q^{2\sigma_{01}}-1)\dot{F}_{1,1}r_{01}.
\end{align*}
\end{remark}

When $\Re \sigma_{01}=\tfrac{1}{2}$, the expansions in Theorem \ref{thm:generic_asymp_infty} are no longer asymptotic expansions. To cover this case, we have the dual series representations given in the following theorem, proven in Section \ref{sec:extract_asymp_infty}.
\begin{theorem}\label{thm:generic_asymp_infty_dual}
Take any $\infty$-generic point $\eta$ in the affine Segre surface $\mathcal{F}(\Theta,t_0)$ and let $(f,g)$ denote the corresponding solution of $q\Psix(\Theta,t_0)$. Recall the notation in equation \eqref{eq:coefficients_infty} for the coefficients of the series in Theorem \ref{thm:generic_asymp_infty}.
Denote $t_m=q^mt_0$, $m\in\mathbb{Z}$, then $1/f$ and $g/t$ have complete asymptotic expansions as $t_m\rightarrow \infty$, of the form,
\begin{align*}
    \frac{1}{f(t_m)}&=\sum_{n=1}^\infty\sum_{k=-n}^n \dot{F}_{n,k}(\widehat{\Theta},\widehat{\sigma}_{01})\widehat{r}_{01}^k(- t_m)^{-(n+2k\widehat{\sigma}_{01})},\\
    \frac{q^{\frac{3}{2}}g(t_m)}{t_m}&=\sum_{n=1}^\infty\sum_{k=-n}^n \dot{G}_{n,k}(\widehat{\Theta},\widehat{\sigma}_{01})\widehat{r}_{01}^k(- t_m)^{-(n+2k\widehat{\sigma}_{01})},
\end{align*}
which are absolutely convergent for large enough $m\geq 0$,
where
\begin{equation*}
\widehat{\theta}_0=\theta_\infty-\tfrac{1}{2},\quad
\widehat{\theta}_t=\theta_1,\quad
\widehat{\theta}_1=\theta_t,\quad
\widehat{\theta}_\infty=\theta_0+\tfrac{1}{2},
\end{equation*}
the branches of the complex powers are principal and the pair of integration constants $\{\widehat{\sigma}_{01},\widehat{r}_{01}\}$ is related to the pair of integration constants $\{\sigma_{01},r_{01}\}$ in Theorem \ref{thm:generic_asymp_infty}, by
\begin{equation*}
    \widehat{\sigma}_{01}=\tfrac{1}{2}-\sigma_{01},\quad \widehat{r}_{01}=-\frac{\dot{F}_{1,-1}(\Theta,\sigma_{01})\dot{F}_{1,-1}(\widehat{\Theta},\widehat{\sigma}_{01})}{r_{01}}=-\frac{\dot{F}_{1,-1}(\Theta,\sigma_{01})\dot{F}_{1,-1}(\widehat{\Theta},\widehat{\sigma}_{01})}{c_{0t}s_{0t}}.
\end{equation*}
The latter pair is in turn related to $\eta$ via
 the pair of dual Tyurin ratios $\{\widetilde{\rho}_{34},\widetilde{\rho}_{13}\}$, by
\begin{equation*}
   \frac{\vartheta_\tau(\sigma_{01}-\theta_1+\theta_0,\sigma_{01}+\theta_1-\theta_0)}{\vartheta_\tau(\sigma_{01}-\theta_1-\theta_0,\sigma_{01}+\theta_1+\theta_0)}= \widetilde{\rho}_{34},
\end{equation*}
and
\begin{equation}\label{eq:twist01alt}
s_{01}=-(-t_0)^{2\sigma_{01}}\widehat{M}_{01}(\widetilde{\rho}_{13}),
\end{equation}
where $\widehat{M}_{0t}(\cdot)$ is the M\"obius transformation
\begin{equation*}
    \widehat{M}_{01}(Z)=\frac{\vartheta_\tau(\theta_1+\theta_0+\sigma_{01})\theta_q(q^{\theta_t-\theta_0+\sigma_{01}}t_0)-Z
\vartheta_\tau(\theta_1-\theta_0+\sigma_{01})\theta_q(q^{\theta_t+\theta_0+\sigma_{01}}t_0)}{\vartheta_\tau(\theta_1+\theta_0-\sigma_{01})\theta_q(q^{\theta_t-\theta_0-\sigma_{01}}t_0)-Z
\vartheta_\tau(\theta_1-\theta_0-\sigma_{01})\theta_q(q^{\theta_t+\theta_0-\sigma_{01}}t_0)}.
\end{equation*}
\end{theorem}
\begin{remark}\label{remark:weakeninginftydual}
The asymptotic expansions in Theorem \ref{thm:generic_asymp_infty_dual} remain valid when $\Re \sigma_{01}=\frac{1}{2}$, as long as $\sigma_{01}\not\equiv \tfrac{1}{2}, \tfrac{1}{2}+\tfrac{\tau}{2}$ modulo $\Lambda_\tau$. In this case, the leading order behaviours of $t_m/f(t_m)$ and $g(t_m)$ are oscillatory as $t_m\rightarrow \infty$. 
\end{remark}

% \begin{remark}
% One may develop the reciprocal of the asymptotic expansion of $g^{-1}$ as an asymptotic series, giving a convergent expansion of the form,
% \begin{equation*}
%     g(t_m)=t_m^2\sum_{n=1}^\infty\sum_{k=-n}^\infty \widetilde{G}_{n,k}r_{01}^k(- t_m)^{-(n+2k\sigma_{01})},
% \end{equation*}
% where the inner summation index $k$ is no longer restricted to $k\leq n$ and some of the obvious symmetries in the original coefficients are lost.
% \end{remark}

\subsection{Asymptotics on lines}\label{subsec:asymplines}
In the following proposition, proven in Section \ref{sec:extract_asymp}, the asymptotics of solutions on the line $\widetilde{\mathcal{L}}_2^\infty$ are given, under the assumption $\Re (\theta_t-\theta_0)<\tfrac{1}{2}$.
\begin{proposition}[Asymptotics on $\widetilde{\mathcal{L}}_2^\infty$] \label{prop:asymptotics_lineLd2i}
Assume $\theta_t-\theta_0\notin \frac{1}{2}\Lambda_\tau$ and $\Re (\theta_t-\theta_0)<\tfrac{1}{2}$.
Take a point $\eta$ in the affine Segre surface $\mathcal{F}(\Theta,t_0)$ which lies on the line $\widetilde{\mathcal{L}}_2^\infty$ and let $(f,g)$ denote the corresponding solution of $q\Psix(\Theta,t_0)$. Denote $t_m=q^mt_0$, $m\in\mathbb{Z}$, then $f$ and $g$ have complete and convergent asymptotic expansions as $t_m\rightarrow 0$ of the form,
\begin{align*}
    f(t_m)&=\sum_{n=1}^\infty\sum_{k=-n}^0 F_{n,k}r_{0t}^k(- t_m)^{n+2k(\theta_t-\theta_0)},\\
    g(t_m)&=\sum_{n=1}^\infty\sum_{k=-n}^0 G_{n,k}r_{0t}^k(- t_m)^{n+2k(\theta_t-\theta_0)},
\end{align*}
where
\begin{align*}
    F_{1,-1}&=-\frac{\bigl(q^{\theta_t}-q^{-\theta_t}\bigr)\bigl(q^{\theta_1+\theta_\infty-\theta_0+\theta_t}-1\bigr)}{\bigl(q^{\theta_1+\theta_\infty+\theta_0-\theta_t}-1\bigr)\bigl(q^{2\theta_0-2\theta_t}-1\bigr)},  & G_{1,-1}&=-q^{-1+\theta_t-\theta_0}F_{1,-1},\\
    F_{1,0}&=\frac{q^{\theta_0}-q^{-\theta_0}}{q^{\theta_t-\theta_0}-q^{\theta_0-\theta_t}}, &
    G_{1,0}&=q^{-1}\frac{q^{\theta_t}-q^{-\theta_t}}{q^{\theta_0-\theta_t}-q^{\theta_t-\theta_0}},
\end{align*}
and the higher order coefficients may be computed recursively via the $q\Psix$ equation.
Here the branches of the complex powers are principal and the integration constant $r_{0t}$ is related to  $\rho_{24}$ via the same formulas as in Theorem \ref{thm:generic_asymp_zero} with $\sigma_{0t}=\theta_t-\theta_0$.
\end{proposition}
\begin{remark}
If we consider $r_{0t}$ as a free variable, then $\eta=\eta(r_{0t})$ traces out the line $\widetilde{\mathcal{L}}_2^\infty$ in the Segre surface $\widehat{\mathcal{F}}$. The value $r_{0t}=\infty$ corresponds to the intersection point of this line with the line $\widetilde{\mathcal{L}}_1^0$, see Corollary \ref{coro:intersectionanalytic}. The value $r_{0t}=0$ corresponds to the intersection point of the line with the hyperplane section at infinity, $X=\widehat{\mathcal{F}}\setminus \mathcal{F}$, and is thus excluded.
\end{remark}
\begin{remark}
Upon letting $\sigma_{0t}\rightarrow\theta_t-\theta_0$ in the generic asymptotic expansion around $t=0$, given in Theorem \ref{thm:generic_asymp_zero}, the coefficients $F_{n,k},G_{n,k}$, with $k>0$, vanish and the expansion coincides with the series in the above proposition. In other words, the generic asymptotic expansion truncates on the line $\widetilde{\mathcal{L}}_2^\infty$ when $\Re (\theta_t-\theta_0)<\tfrac{1}{2}$.
\end{remark}
\begin{remark}
When $\Re (\theta_t-\theta_0)>\tfrac{1}{2}$, the truncation on the line $\widetilde{\mathcal{L}}_2^\infty$, described by Proposition \ref{prop:asymptotics_lineLd2i}, does not occur. Namely, let $l\in\mathbb{N}^*$ be such that
\begin{equation*}
    -\tfrac{1}{2}<\Re(\theta_t-\theta_0-l)<\tfrac{1}{2},
\end{equation*}
then the generic asymptotic expansion in Theorem \ref{thm:generic_asymp_zero} continues to hold
 with $\sigma_{0t}=\theta_t-\theta_0-l$ and the formula for $r_{0t}$ unchanged, noting that $c_{0t}$ remains a well-defined nonzero scalar. Using Mathematica, we computed higher-order coefficients in the expansion, with $-n\leq k\leq n$ and $n$ going up to $7$. None of them vanish when $\sigma_{0t}$ is set equal $\theta_t-\theta_0-l$, with other parameters generic. Similarly, we did not note any truncation in the dual asymptotic expansion around $t=0$ in Theorem \ref{thm:generic_asymp_zero_dual}, nor in the asymptotic expansions around $t=\infty$ in Theorems \ref{thm:generic_asymp_infty} and \ref{thm:generic_asymp_infty_dual}.
\end{remark}
The following corollary is proven in Section \ref{sec:extract_asymp}.
\begin{corollary}\label{coro:intersectionanalytic}
Suppose $\theta_0-\theta_t\notin \frac{1}{2}\Lambda_\tau$. Then the intersection point of the lines $\widetilde{\mathcal{L}}_1^0$ and $\widetilde{\mathcal{L}}_2^\infty$ is finite, and the solution $(f,g)$ corresponding to
this intersection point is given by a convergent power series expansion around $t_m=0$,
\begin{equation*}
    f(t_m)=\sum_{n=1}^\infty F_{n,0}(- t_m)^{n},\quad
    g(t_m)=\sum_{n=1}^\infty G_{n,0}(- t_m)^{n},
\end{equation*}
where
\begin{equation*}
     F_{1,0}=\frac{q^{\theta_0}-q^{-\theta_0}}{q^{\theta_t-\theta_0}-q^{\theta_0-\theta_t}}, \quad
    G_{1,0}=q^{-1}\frac{q^{\theta_t}-q^{-\theta_t}}{q^{\theta_0-\theta_t}-q^{\theta_t-\theta_0}},
\end{equation*}
and the higher order terms may be computed recursively via the $q\Psix$ equation.
\end{corollary}
Laurent series solutions of $q\Psix$ like the one in the above corollary were first derived by Ohyama \cite{ohyamamero}. We discuss them in greater detail in the next subsection.

In Section \ref{sec:symmetries}, we derive the action of six symmetries of $q\Psix$ on monodromy data, which we denote by  $r_0,r_t,r_1,r_\infty, v_{t1}$ and $v_{0\infty}$. Their definition and the results are given in Table \ref{table:symmetries}. We note that each of these symmetries is an involution, and
\begin{equation}\label{eq:conjugation_sym}
   r_1=v_{t1} r_t v_{t1},\quad r_\infty=v_{0\infty}r_0v_{0\infty}.
\end{equation}
These symmetries act transitively on the sixteen lines of the Segre surface, as graphically displayed in Figure \ref{figure:action_on_lines}. This means that we can deduce truncations of asymptotics on each of the sixteen lines, by application of these symmetries to the results in Proposition \ref{prop:asymptotics_lineLd2i}. This yields the following theorem, proven in Section \ref{sec:lines_asymp}.

\begin{theorem}\label{thm:line_truncations}
Let $\eta\in\mathcal{F}(\Theta,t_0)$ and $(f,g)$ denote the corresponding solution of $q\Psix(\Theta,t_0)$.
  The asymptotic expansions around $t=0$, in Theorem \ref{thm:generic_asymp_zero}, take the truncated form
  \begin{equation*}
    f(t_m)=\sum_{n=1}^\infty\sum_{k=-n}^0 F_{n,k}r_{0t}^k(- t_m)^{n+2k\sigma_{0t}},\quad
    g(t_m)=\sum_{n=1}^\infty\sum_{k=-n}^0 G_{n,k}r_{0t}^k(- t_m)^{n+2k\sigma_{0t}},
\end{equation*}
when
\begin{align*}
    \eta\in \widetilde{\mathcal{L}}_1^0\text{ with } \sigma_{0t}&=+\theta_0-\theta_t, &\text{ or}  & &
    \eta\in \widetilde{\mathcal{L}}_1^\infty\text{ with } \sigma_{0t}&=-\theta_0-\theta_t,\text{ or}\\
    \eta\in \widetilde{\mathcal{L}}_2^0\text{ with } \sigma_{0t}&=+\theta_0+\theta_t, &\text{ or} & &
    \eta\in \widetilde{\mathcal{L}}_2^\infty\text{ with } \sigma_{0t}&=-\theta_0+\theta_t,
\end{align*}
under the parameter assumptions $\Re \sigma_{0t}<\tfrac{1}{2}$, $\sigma_{0t}\notin \tfrac{1}{2}\Lambda_\tau$, in each case.\\
  The dual asymptotic expansions around $t=0$, in Theorem \ref{thm:generic_asymp_zero_dual}, take the truncated form
\begin{equation*}
    \frac{t_m}{f(t_m)}=\sum_{n=1}^\infty\sum_{k=-n}^0 \widehat{F}_{n,k}\widehat{r}_{0t}^k(- t_m)^{n+2k\widehat{\sigma}_{0t}},\quad
    \frac{t_m}{q^{\frac{3}{2}}g(t_m)}=\sum_{n=1}^\infty\sum_{k=-n}^0 \widehat{G}_{n,k}\widehat{r}_{0t}^k(- t_m)^{n+2k\widehat{\sigma}_{0t}},
\end{equation*}
when
\begin{align*}
    \eta\in \mathcal{L}_3^0\text{ with } \widehat{\sigma}_{0t}&=+\tfrac{1}{2}-\theta_\infty+\theta_1, &\text{ or}  & &
    \eta\in \mathcal{L}_3^\infty\text{ with } \widehat{\sigma}_{0t}&=-\tfrac{1}{2}+\theta_\infty+\theta_1,\text{ or}\\
    \eta\in \mathcal{L}_4^0\text{ with } \widehat{\sigma}_{0t}&=+\tfrac{1}{2}-\theta_\infty-\theta_1, &\text{ or} & &
    \eta\in \mathcal{L}_4^\infty\text{ with } \widehat{\sigma}_{0t}&=-\tfrac{1}{2}+\theta_\infty-\theta_1,
\end{align*}
under the parameter assumptions $\Re \widehat{\sigma}_{0t}<\tfrac{1}{2}$, $\widehat{\sigma}_{0t}\notin \tfrac{1}{2}\Lambda_\tau$, in each case.\\
  The asymptotic expansions around $t=\infty$, in Theorem \ref{thm:generic_asymp_infty}, take the truncated form
\begin{equation*}
    \frac{f(t_m)}{t_m}=\sum_{n=1}^\infty\sum_{k=-n}^0\dot{F}_{n,k}r_{01}^k(- t_m)^{-(n+2k\sigma_{01})},\quad
    \frac{1}{g(t_m)}=\sum_{n=1}^\infty\sum_{k=-n}^0\dot{G}_{n,k}r_{01}^k(- t_m)^{-(n+2k\sigma_{01})},
\end{equation*}
when
\begin{align*}
    \eta\in \widetilde{\mathcal{L}}_3^0\text{ with } \sigma_{01}&=+\theta_0-\theta_1, &\text{ or}  & &
    \eta\in \widetilde{\mathcal{L}}_3^\infty\text{ with } \sigma_{01}&=-\theta_0-\theta_1,\text{ or}\\
    \eta\in \widetilde{\mathcal{L}}_4^0\text{ with } \sigma_{01}&=+\theta_0+\theta_1, &\text{ or} & &
    \eta\in \widetilde{\mathcal{L}}_4^\infty\text{ with } \sigma_{01}&=-\theta_0-\theta_1,
\end{align*}
under the parameter assumptions $\Re \sigma_{01}<\tfrac{1}{2}$, $\sigma_{01}\notin \tfrac{1}{2}\Lambda_\tau$, in each case.\\
  The dual asymptotic expansions around $t=\infty$, in Theorem \ref{thm:generic_asymp_infty_dual}, take the truncated form
\begin{align*}
    \frac{1}{f(t_m)}&=\sum_{n=1}^\infty\sum_{k=-n}^0 \widehat{\dot{F}}_{n,k}\widehat{r}_{01}^k(- t_m)^{-(n+2k\widehat{\sigma}_{01})},\\
    \frac{q^{\frac{3}{2}}g(t_m)}{t_m}&=\sum_{n=1}^\infty\sum_{k=-n}^0 \widehat{\dot{G}}_{n,k}\widehat{r}_{01}^k(- t_m)^{-(n+2k\widehat{\sigma}_{01})},
\end{align*}
when
\begin{align*}
    \eta\in \mathcal{L}_1^0\text{ with } \widehat{\sigma}_{01}&=-\tfrac{1}{2}+\theta_\infty-\theta_t, &\text{ or}  & &
    \eta\in \mathcal{L}_1^\infty\text{ with } \widehat{\sigma}_{01}&=+\tfrac{1}{2}-\theta_\infty-\theta_t,\text{ or}\\
    \eta\in \mathcal{L}_2^0\text{ with } \widehat{\sigma}_{01}&=-\tfrac{1}{2}+\theta_\infty+\theta_t, &\text{ or} & &
    \eta\in \mathcal{L}_2^\infty\text{ with } \widehat{\sigma}_{01}&=\tfrac{1}{2}-\theta_\infty+\theta_t,
\end{align*}
under the parameter assumptions $\Re \widehat{\sigma}_{01}<\tfrac{1}{2}$, $\widehat{\sigma}_{01}\notin \tfrac{1}{2}\Lambda_\tau$, in each case.\\

\end{theorem}

\begingroup
\renewcommand{\arraystretch}{1.8}
\begin{table}[t]
\centering
\begin{tabular}{c || c | c | c | c || c | c || c || c | c || c}

& $\theta_0$ & $\theta_t$ & $\theta_1$ & $\theta_\infty$ & $f$ & $g$ & $t$ & $\rho_i$ & $\widetilde{\rho}_i$ &  $\eta_{ij}$\\
 \hline\hline
$r_0$ & $-\theta_0$ & $\theta_t$ & $\theta_1$ & $\theta_\infty$ & $f$ & $g$ & $t$ &  $\rho_i$ & $1/\widetilde{\rho}_i$ &  $\eta_{ij}$\\
 \hline
$r_t$ & $\theta_0$ & $-\theta_t$ & $\theta_1$ & $\theta_\infty$ & $f$ & $g$ & $t$ &  $\rho_{\alpha_t(i)}$ & $\widetilde{\rho}_{\alpha_t(i)}$ &  $\eta_{\alpha_t(i)\alpha_t(j)}$\\
 \hline
 $r_1$ & $\theta_0$ & $\theta_t$ & $-\theta_1$ & $\theta_\infty$ & $f$ & $g$ & $t$ & $\rho_{\alpha_1(i)}$ & $\widetilde{\rho}_{\alpha_1(i)}$ &  $\eta_{{\alpha_1(i)}{\alpha_1(j)}}$\\
  \hline
 $r_{\infty}$ & $\theta_0$ & $\theta_t$ & $\theta_1$ & $1-\theta_\infty$ & $f$ & $g$ & $t$ & $\frac{x_i^2}{t_0}\rho_i^{-1}$ & $\widetilde{\rho}_i$ &  $\eta_{kl}$\\
  \hline\hline
 $v_{t1}$ & $\theta_0$ & $\theta_1$ & $\theta_t$ & $\theta_\infty$ & $\frac{f}{t}$ & $\frac{1}{q\overline{g}}$ & $t^{-1}$ & $\rho_{\alpha_{t1}(i)}$ & $\widetilde{\rho}_{\alpha_{t1}(i)}$ &  $\eta_{{\alpha_{t1}(i)}{\alpha_{t1}(j)}}$\\
   \hline
 $v_{0\infty}$ & $\theta_\infty-\frac{1}{2}$ & $\theta_1$ & $\theta_t$ & $\theta_0+\frac{1}{2}$ & $\frac{t}{f}$ & $q^{-\frac{3}{2}}\frac{t}{g}$ & $t$ & $t_0^{\frac{1}{2}}\frac{x_{\beta(i)}^{-1}}{\widetilde{\rho}_{\beta(i)}}$ & $t_0^{-\frac{1}{2}}\frac{x_{\beta(i)}}{\rho_{\beta(i)}}$ &  \text{-}\\
\end{tabular}
\caption{Some symmetries of $q\Psix$ and their actions on (dual) Tyurin parameters and the $\eta$-coordinates, for any labelling $\{i,j,k,l\}=\{1,2,3,4\}$, where the $\alpha$'s and $\beta$ denote different permutations in $S_4$:\hspace{2mm}
 $\alpha_t=(1\;2),\hspace{2mm}\alpha_1=(3\;4),\hspace{2mm} \alpha_{t1}=(1\;3)\;(2\;4),\hspace{2mm}\beta=(1\;4)\;(2\;3).$ The bottom-right entry is missing as we did not compute it.}
\label{table:symmetries}
\end{table}
\endgroup

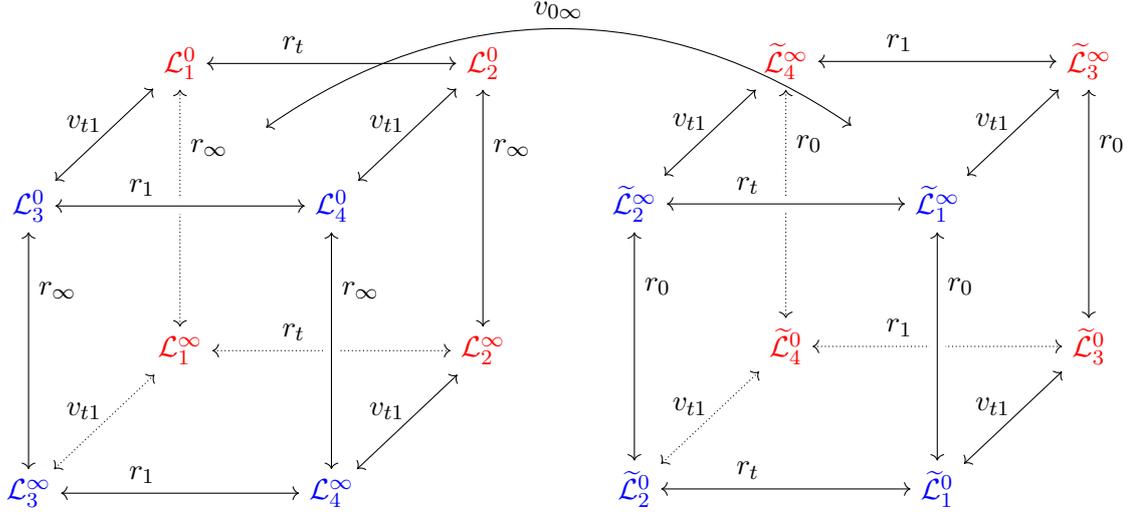
\begin{figure}[t]
\begin{tikzpicture}
  \matrix (m) [matrix of math nodes, row sep=3em,
    column sep=3em]{
    & {\color{red}\mathcal{L}_1^0}& & {\color{red}\mathcal{L}_2^0}  & & \color{red}\widetilde{\mathcal{L}}_4^\infty &  & \color{red}\widetilde{\mathcal{L}}_3^\infty\\
    {\color{blue}\mathcal{L}_3^0} & & {\color{blue}\mathcal{L}_4^0}  &  & \color{blue}\widetilde{\mathcal{L}}_2^\infty &  & \color{blue}\widetilde{\mathcal{L}}_1^\infty &\\
    & {\color{red}\mathcal{L}_1^\infty} & & {\color{red}\mathcal{L}_2^\infty}  &  & \color{red}\widetilde{\mathcal{L}}_4^0 &  & \color{red}\widetilde{\mathcal{L}}_3^0\\
    {\color{blue}\mathcal{L}_3^\infty} & & {\color{blue}\mathcal{L}_4^\infty} &  & {\color{blue}\widetilde{\mathcal{L}}_2^0} &  & \color{blue}\widetilde{\mathcal{L}}_1^0 & \\};
  \path[-stealth]
    (m-1-2) edge[<->] node[above,xshift=-5mm] {$r_t$} (m-1-4) edge[<->] node[above left,xshift=1mm,yshift=-1mm] {$v_{t1}$} (m-2-1)
            edge[<->]   [densely dotted] node[right,yshift=8mm] {$r_\infty$} (m-3-2)
    (m-1-4) edge[<->] node[right,yshift=8mm] {$r_\infty$} (m-3-4) edge[<->] node[above left,xshift=1mm,yshift=-1mm] {$v_{t1}$} (m-2-3)
    (m-2-1) edge[<->] [-,line width=6pt,draw=white] (m-2-3)
            edge[<->] node[above,xshift=-5mm] {$r_1$} (m-2-3) edge[<->] node[right,yshift=8mm] {$r_\infty$} (m-4-1)
    (m-3-2) edge[<->] [densely dotted] node[above,xshift=-5mm] {$r_t$} (m-3-4)
            edge[<->] [densely dotted] node[above left,xshift=1mm,yshift=-1mm] {$v_{t1}$} (m-4-1)
    (m-4-1) edge[<->] node[above,xshift=-5mm] {$r_1$} (m-4-3)
    (m-3-4) edge[<->] node[above left,xshift=1mm,yshift=-1mm] {$v_{t1}$} (m-4-3)
    (m-2-3) edge[<->]  [-,line width=6pt,draw=white] (m-4-3)
            edge[<->] node[right,yshift=8mm] {$r_{\infty}$} (m-4-3);
    \path[-stealth]
    (m-1-6) edge[<->] node[above,xshift=-5mm] {$r_1$} (m-1-8) edge[<->] node[above left,xshift=1mm,yshift=-1mm] {$v_{t1}$} (m-2-5)
            edge[<->]   [densely dotted] node[right,yshift=8mm] {$r_0$} (m-3-6)
    (m-1-8) edge[<->] node[right,yshift=8mm] {$r_0$} (m-3-8) edge[<->] node[above left,xshift=1mm,yshift=-1mm] {$v_{t1}$} (m-2-7)
    (m-2-5) edge[<->] [-,line width=6pt,draw=white] (m-2-7)
            edge[<->] node[above,xshift=-5mm] {$r_t$} (m-2-7) edge[<->] node[right,yshift=8mm] {$r_0$} (m-4-5)
    (m-3-6) edge[<->] [densely dotted] node[above,xshift=-5mm] {$r_1$} (m-3-8)
            edge[<->] [densely dotted] node[above left,xshift=1mm,yshift=-1mm] {$v_{t1}$} (m-4-5)
    (m-4-5) edge[<->] node[above,xshift=-5mm] {$r_t$} (m-4-7)
    (m-3-8) edge[<->] node[above left,xshift=1mm,yshift=-1mm] {$v_{t1}$} (m-4-7)
    (m-2-7) edge[<->]  [-,line width=6pt,draw=white] (m-4-7)
            edge[<->] node[right,yshift=8mm] {$r_{0}$} (m-4-7);

\node (c1) at (m-1-2) {};
\node (c2) at (m-2-3) {};
\node (c12) at ($(c1)!0.5!(c2)$) {};

\node (d1) at (m-1-6) {};
\node (d2) at (m-2-7) {};
\node (d12) at ($(d1)!0.5!(d2)$) {};

\path[<->]  (c12)  edge   [bend left=35]   node[above] {$v_{0\infty}$} (d12);
\end{tikzpicture}
\caption{Illustration of the action of the symmetries $r_0,r_t,r_1,v_{t1}$ and $v_{0\infty}$ on the sixteen lines. Here, all vertical and diagonal edges correspond respectively to application of $r_{0}$ and $v_{t1}$. All horizontal edges on the frontal faces of the cubes correspond to application of $r_1$, all other horizontal edges to application of $r_t$. Finally, the action of $v_{0\infty}$ interchanges the vertices of the two cubes, as realised by simply swapping them.}
\label{figure:action_on_lines}
\end{figure}

\subsection{Asymptotics on intersection points}\label{subsection:intersect}
In corollary \ref{coro:intersectionanalytic}, it is shown that the unique intersection point of the lines $\widetilde{\mathcal{L}}_1^0$ and $\widetilde{\mathcal{L}}_2^\infty$ corresponds to a solution which admits a convergent power series expansion around $t=0$.

We refer to solutions which admit a convergent Laurent series expansion in $t$, around $t=0$ or $t=\infty$, as Kaneko-Ohyama solutions, since Kaneko \cite{kanekomero} classified such solutions for $\Psix$ and Ohyama \cite{ohyamamero} classified such solutions for $q\Psix$.
For generic parameters, there exist precisely eight Kaneko-Ohyama solutions. 
We correspondingly call the eight intersection points in $\widehat{\mathcal{F}}$, coloured in red or blue in Figure \ref{fig:lines_intersection}, Kaneko-Ohyama points.

In the following theorem, proven in Section \ref{sec:asymptotics_intersection}, it is shown that the Kaneko-Ohyama solutions correspond exactly to the Kaneko-Ohyama points under the Riemann-Hilbert correspondence.

\begin{theorem}\label{thm:kanekoohyama}
Under the Riemann-Hilbert correspondence, Kaneko-Ohyama solutions of $q\Psix$ are in one-to-one correspondence with Kaneko-Ohyama points on the affine Segre surface $\mathcal{F}(\Theta,t_0)$, as detailed below, where we use the notation $\diamond_+=\infty$ and $\diamond_-=0$.
\begin{enumerate}
\item If $\theta_0\pm\theta_t\notin \frac{1}{2}\Lambda_\tau$, then the intersection point of $\widetilde{\mathcal{L}}_1^{\diamond_{\pm}}$ and $\widetilde{\mathcal{L}}_2^{\diamond_{\mp}}$ is finite and the corresponding solution $(f,g)$ admits a convergent power series expansion around $t=0$, of the form
\begin{align*}
    f(t_m)&=\sum_{n=1}^\infty F_{n}(- t_m)^{n}, &
    g(t_m)&=\sum_{n=1}^\infty G_{n}(- t_m)^{n},\\
    F_{1}&=-\frac{q^{\theta_0}-q^{-\theta_0}}{q^{\pm\theta_t+\theta_0}-q^{\mp\theta_t-\theta_0}}, &
    G_{1}&=-q^{-1}\frac{q^{\pm\theta_t}-q^{\mp\theta_t}}{q^{\pm\theta_t+\theta_0}-q^{\mp\theta_t-\theta_0}}.
\end{align*}
%     \item If $\theta_0-\theta_t\notin \frac{1}{2}\Lambda_\tau$, then the intersection point of $\widetilde{\mathcal{L}}_1^0$ and $\widetilde{\mathcal{L}}_2^\infty$ is finite and the corresponding solution $(f,g)$ admits a convergent power series expansion around $t=0$, of the form
% \begin{align*}
%     f(t_m)&=\sum_{n=1}^\infty F_{n}(- t_m)^{n}, &
%     g(t_m)&=\sum_{n=1}^\infty G_{n}(- t_m)^{n},\\
%     F_{1}&=\frac{q^{\theta_0}-q^{-\theta_0}}{q^{\theta_t-\theta_0}-q^{\theta_0-\theta_t}}, &
%     G_{1}&=q^{-1}\frac{q^{\theta_t}-q^{-\theta_t}}{q^{\theta_0-\theta_t}-q^{\theta_t-\theta_0}}.
% \end{align*}
% \item If $\theta_0+\theta_t\notin \frac{1}{2}\Lambda_\tau$, then the intersection point of $\widetilde{\mathcal{L}}_1^\infty$ and $\widetilde{\mathcal{L}}_2^0$ is finite and the corresponding solution $(f,g)$ admits a convergent power series expansion around $t=0$, of the form
% \begin{align*}
%     f(t_m)&=\sum_{n=1}^\infty F_{n}(- t_m)^{n}, &
%     g(t_m)&=\sum_{n=1}^\infty G_{n}(- t_m)^{n},\\
%     F_{1}&=\frac{q^{-\theta_0}-q^{\theta_0}}{q^{\theta_t+\theta_0}-q^{-\theta_0-\theta_t}}, &
%     G_{1}&=q^{-1}\frac{q^{\theta_t}-q^{-\theta_t}}{q^{-\theta_0-\theta_t}-q^{\theta_t+\theta_0}}.
% \end{align*}
\item If $\theta_\infty\pm\theta_1\notin \frac{1}{2}\Lambda_\tau$, then the intersection point of $\mathcal{L}_3^{\diamond_{\pm}}$ and $\mathcal{L}_4^{\diamond_{\mp}}$ is finite and the corresponding solution $(f,g)$ admits a convergent power series expansion around $t=0$, of the form
\begin{align*}
    f(t_m)&=\sum_{n=0}^\infty F_{n}(- t_m)^{n}, &
    g(t_m)&=\sum_{n=0}^\infty G_{n}(- t_m)^{n},\\
    F_{0}&=\frac{q^{\theta_\infty\pm\theta_1}-q^{1-\theta_\infty\mp\theta_1}}{q^{\theta_\infty}-q^{1-\theta_\infty}}, &
    G_{0}&=q^{-1}\frac{q^{\theta_\infty\pm\theta_1}-q^{1-\theta_\infty\mp\theta_1}}{q^{\pm\theta_1}-q^{\mp\theta_1}}.
\end{align*}
% \item If $\theta_\infty-\theta_1\notin \frac{1}{2}\Lambda_\tau$, then the intersection point of $\mathcal{L}_3^0$ and $\mathcal{L}_4^\infty$ is finite and the corresponding solution $(f,g)$ admits a convergent power series expansion around $t=0$, of the form
% \begin{align*}
%     f(t_m)&=\sum_{n=0}^\infty F_{n}(- t_m)^{n}, &
%     g(t_m)&=\sum_{n=0}^\infty G_{n}(- t_m)^{n},\\
%     F_{0}&=\frac{q^{\theta_\infty-\theta_1}-q^{1-\theta_\infty+\theta_1}}{q^{\theta_\infty}-q^{1-\theta_\infty}}, &
%     G_{0}&=q^{-1}\frac{q^{1-\theta_\infty+\theta_1}-q^{\theta_\infty-\theta_1}}{q^{\theta_1}-q^{-\theta_1}}.
% \end{align*}
% \item If $\theta_\infty+\theta_1\notin \frac{1}{2}\Lambda_\tau$, then the intersection point of $\mathcal{L}_3^\infty$ and $\mathcal{L}_4^0$ is finite and the corresponding solution $(f,g)$ admits a convergent power series expansion around $t=0$, of the form
% \begin{align*}
%     f(t_m)&=\sum_{n=0}^\infty F_{n}(- t_m)^{n}, &
%     g(t_m)&=\sum_{n=0}^\infty G_{n}(- t_m)^{n},\\
%     F_{0}&=\frac{q^{\theta_\infty+\theta_1}-q^{1-\theta_\infty-\theta_1}}{q^{\theta_\infty}-q^{1-\theta_\infty}}, &
%     G_{0}&=q^{-1}\frac{q^{1-\theta_\infty-\theta_1}-q^{\theta_\infty+\theta_1}}{q^{-\theta_1}-q^{\theta_1}}.
% \end{align*}
\item If $\theta_\infty\pm\theta_t\notin \frac{1}{2}\Lambda_\tau$, then the intersection point of $\mathcal{L}_1^{\diamond_{\pm}}$ and $\mathcal{L}_2^{\diamond_{\mp}}$ is finite and the corresponding solution $(f,g)$ admits a convergent Laurent series expansion around $t=\infty$, of the form
\begin{align*}
    f(t_m)&=\sum_{n=-1}^\infty F_{n}(- t_m)^{-n}, &
    g(t_m)&=\sum_{n=0}^\infty G_{n}(- t_m)^{-n},\\
    F_{-1}&=-\frac{q^{\theta_\infty\pm\theta_t}-q^{1-\theta_\infty\mp\theta_t}}{q^{\theta_\infty}-q^{1-\theta_\infty}}, &
    G_{0}&=\frac{q^{\pm\theta_t}-q^{\mp\theta_t}}{q^{\theta_\infty\pm\theta_t}-q^{1-\theta_\infty\mp\theta_t}}.
\end{align*}
% \item If $\theta_\infty-\theta_t\notin \frac{1}{2}\Lambda_\tau$, then the intersection point of $\mathcal{L}_1^0$ and $\mathcal{L}_2^\infty$ is finite and the corresponding solution $(f,g)$ admits a convergent Laurent series expansion around $t=\infty$, of the form
% \begin{align*}
%     f(t_m)&=\sum_{n=-1}^\infty F_{n}(- t_m)^{-n}, &
%     g(t_m)&=\sum_{n=0}^\infty G_{n}(- t_m)^{-n},\\
%     F_{-1}&=\frac{q^{1-\theta_\infty+\theta_t}-q^{\theta_\infty-\theta_t}}{q^{\theta_\infty}-q^{1-\theta_\infty}}, &
%     G_{0}&=\frac{q^{\theta_t}-q^{-\theta_t}}{q^{1-\theta_\infty+\theta_t}-q^{\theta_\infty-\theta_t}}.
% \end{align*}
% \item If $\theta_\infty+\theta_t\notin \frac{1}{2}\Lambda_\tau$, then the intersection point of $\mathcal{L}_1^\infty$ and $\mathcal{L}_2^0$ is finite and the corresponding solution $(f,g)$ admits a convergent Laurent series expansion around $t=\infty$, of the form
% \begin{align*}
%     f(t_m)&=\sum_{n=-1}^\infty F_{n}(- t_m)^{-n}, &
%     g(t_m)&=\sum_{n=0}^\infty G_{n}(- t_m)^{-n},\\
%     F_{-1}&=\frac{q^{1-\theta_\infty-\theta_t}-q^{\theta_\infty+\theta_t}}{q^{\theta_\infty}-q^{1-\theta_\infty}}, &
%     G_{0}&=\frac{q^{-\theta_t}-q^{\theta_t}}{q^{1-\theta_\infty-\theta_t}-q^{\theta_\infty+\theta_t}}.
% \end{align*}
\item If $\theta_0\pm\theta_1\notin \frac{1}{2}\Lambda_\tau$, then the intersection point of $\widetilde{\mathcal{L}}_3^{\diamond_{\pm}}$ and $\widetilde{\mathcal{L}}_4^{\diamond_{\mp}}$ is finite and the corresponding solution $(f,g)$ admits a convergent Laurent series expansion around $t=\infty$, of the form
\begin{align*}
    f(t_m)&=\sum_{n=0}^\infty F_{n}(- t_m)^{-n}, &
    g(t_m)&=\sum_{n=-1}^\infty G_{n}(- t_m)^{-n},\\
    F_{0}&=\frac{q^{\theta_0}-q^{-\theta_0}}{q^{\pm\theta_1+\theta_0}-q^{\mp\theta_1-\theta_0}}, &
    G_{-1}&=-q^{-1}\frac{q^{\pm\theta_1+\theta_0}-q^{\mp\theta_1-\theta_0}}{q^{\pm\theta_1}-q^{\mp\theta_1}}.
\end{align*}
% \item If $\theta_0-\theta_1\notin \frac{1}{2}\Lambda_\tau$, then the intersection point of $\widetilde{\mathcal{L}}_3^0$ and $\widetilde{\mathcal{L}}_4^\infty$ is finite and the corresponding solution $(f,g)$ admits a convergent Laurent series expansion around $t=\infty$, of the form
% \begin{align*}
%     f(t_m)&=\sum_{n=0}^\infty F_{n}(- t_m)^{-n}, &
%     g(t_m)&=\sum_{n=-1}^\infty G_{n}(- t_m)^{-n},\\
%     F_{0}&=\frac{q^{\theta_0}-q^{-\theta_0}}{q^{\theta_0-\theta_1}-q^{\theta_1-\theta_0}}, &
%     G_{-1}&=q^{-1}\frac{q^{\theta_0-\theta_1}-q^{\theta_1-\theta_0}}{q^{\theta_1}-q^{-\theta_1}}.
% \end{align*}
% \item If $\theta_0+\theta_1\notin \frac{1}{2}\Lambda_\tau$, then the intersection point of $\widetilde{\mathcal{L}}_3^\infty$ and $\widetilde{\mathcal{L}}_4^0$ is finite and the corresponding solution $(f,g)$ admits a convergent Laurent series expansion around $t=\infty$, of the form
% \begin{align*}
%     f(t_m)&=\sum_{n=0}^\infty F_{n}(- t_m)^{-n}, &
%     g(t_m)&=\sum_{n=-1}^\infty G_{n}(- t_m)^{-n},\\
%     F_{0}&=\frac{q^{\theta_0}-q^{-\theta_0}}{q^{\theta_1+\theta_0}-q^{-\theta_0-\theta_1}}, &
%     G_{-1}&=q^{-1}\frac{q^{-\theta_0-\theta_1}-q^{\theta_1+\theta_0}}{q^{\theta_1}-q^{-\theta_1}}.
% \end{align*}
\end{enumerate}
Furthermore, in each case, the parameter assumption is not only sufficient but also necessary for the existence of the solution and finiteness of the corresponding Kaneko-Ohyama point.
\end{theorem}
\begin{remark}
In a recent talk, Ohyama \cite{ohyama_connection_talk} presented a characterisation of the Tyurin parameters associated to one of the solutions of $q\Psix$ analytic at $t=0$, and a corresponding paper is in preparation \cite{ohyama_connection}.
\end{remark}

Note that any of the Kaneko-Ohyama solutions of $q\Psix(\Theta,t_0)$ admits an obvious uniformisation in $t_0$, realised by simply replacing $t_m$ by $t$ in the corresponding Laurent series in Theorem \ref{thm:kanekoohyama}. The corresponding solutions $f=f(t)$ and $g=g(t)$ then form meromorphic functions on $\mathbb{C}^*$.
For example, the solution $(f,g)$ corresponding to the intersection point $\{N_*\}=\widetilde{\mathcal{L}}_1^{\infty}\cap\widetilde{\mathcal{L}}_2^{0}$, is given by the convergent power series
\begin{align*}
    f(t)&=\sum_{n=1}^\infty F_{n}(- t)^{n}, &
    g(t)&=\sum_{n=1}^\infty G_{n}(- t)^{n},\\
    F_{1}&=-\frac{q^{\theta_0}-q^{-\theta_0}}{q^{\theta_t+\theta_0}-q^{-\theta_t-\theta_0}}, &
    G_{1}&=-q^{-1}\frac{q^{+\theta_t}-q^{-\theta_t}}{q^{+\theta_t+\theta_0}-q^{-\theta_t-\theta_0}},
\end{align*}
which defines a meromorphic solution of $q\Psix$ on $\mathbb{C}$. (This is also the original definition in \cite{ohyamamero}.)
Geometrically, this uniformisation is possible since the intersection point $N_*=N_*(t_0)\in\mathbb{C}^6$ is an analytic and $q$-periodic function in $t_0\in\mathbb{C}^*$.

The solutions corresponding to the remaining $40-8=32$ intersection points also admit such a uniformisation, on $\mathbb{C}^*$ minus one $q$-line of points.
Before we get there, we note that the action of the symmetries in Table \ref{table:symmetries}, induces an action on the intersection points of lines on the Segre surface. The intersection points make up three orbits, one orbit consisting of the eight Kaneko-Ohyama points, see Figure \ref{figure:kaneko_ohyama_orbit}. A second given by the orbit of $\widetilde{\mathcal{L}}_1^0\cap\mathcal{L}_1^0$, containing $16$ elements, see Figure \ref{figure:orbit_I}. The third given by the orbit of $\widetilde{\mathcal{L}}_1^0\cap\widetilde{\mathcal{L}}_3^\infty$, also containing $16$ elements, see Figure \ref{figure:orbit_II}.
We will only discuss one solution, corresponding to the representative $\widetilde{\mathcal{L}}_1^0\cap\widetilde{\mathcal{L}}_3^\infty$ of the third orbit. Its properties, under suitable parameter conditions, are described in the following theorem, which is proven in Section \ref{sec:asymptotics_intersection}, 

\begin{theorem}\label{thm:orbitII}
Let $\Theta\in\mathbb{C}^4$ be such that the parameter conditions \eqref{eq:param_assumptions_1} are satisfied and further assume that
\begin{equation*}
    \Re(\theta_0-\theta_t)<\tfrac{1}{2},\quad
    \Re(-\theta_0-\theta_1)<\tfrac{1}{2}.
\end{equation*}
Define
\begin{equation*}
    U:=\mathbb{CP}^1\setminus q^{\mathbb{Z}-2\theta_0+\theta_t-\theta_1},
\end{equation*}
then the solution $(f,g)$ of $q\Psix$, corresponding to the intersection point
\begin{equation*}
\{N_*(t_0)\}=\widetilde{\mathcal{L}}_1^0\cap\widetilde{\mathcal{L}}_3^\infty,
\end{equation*}
admits an asymptotic expansion as $t\rightarrow 0$, of the form
\begin{align*}
    f(t)&=\frac{q^{\theta_0}-q^{-\theta_0}}{q^{\theta_0-\theta_t}-q^{\theta_t-\theta_0}}\;t+t E_0(t)+\sum_{n=2}^\infty\sum_{k=0}^n f_{n,k}t^n E_0(t)^k,\\
    g(t)&=-q^{-1}\frac{q^{\theta_t}-q^{-\theta_t}}{q^{\theta_0-\theta_t}-q^{\theta_t-\theta_0}}\;t-q^{-1+\theta_0-\theta_t}t E_0(t)+\sum_{n=2}^\infty\sum_{k=0}^n g_{n,k}t^n E_0(t)^k,
\end{align*}
uniformly convergent for small enough $t\in K\setminus\{0\}$, for any compact subset $ K\subseteq U$ with $q K=K$, and simultaneously an asymptotic expansion as $t\rightarrow \infty$, of the form
\begin{align*}
    f(t)&=\frac{q^{\theta_0}-q^{-\theta_0}}{q^{\theta_0+\theta_1}-q^{-\theta_0-\theta_1}}+ E_\infty(t)+t\sum_{n=2}^\infty\sum_{k=0}^{n} \dot{f}_{n,k}t^{-n} E_\infty(t)^{k},\\
    \frac{1}{g(t)}&=\left(q^{-1}\frac{q^{\theta_0+\theta_1}-q^{-\theta_0-\theta_1}}{q^{\theta_1}-q^{\theta_1}}\right)^{-1}t^{-1}-q^{1+\theta_0+\theta_1}t^{-1}E_\infty(t)+\sum_{n=2}^\infty\sum_{k=0}^n \dot{g}_{n,k}t^{-n} E_\infty(t)^k,
\end{align*}
uniformly convergent for large enough $t\in K\setminus\{\infty\}$, for any compact subset $ K\subseteq U$ with $q K=K$, where $E_0(t)$ and $E_\infty(t)$ are the meromorphic functions
\begin{equation}\label{eq:e0infdefi}
E_0(t)=c_0\frac{\theta_q(q^{-\theta_t-\theta_1}t)}{\theta_q(q^{-2\theta_0+\theta_t-\theta_1}t)},\qquad
E_\infty(t)=c_\infty\frac{\theta_q(q^{-\theta_t-\theta_1}t^{-1})}{\theta_q(q^{+2\theta_0-\theta_t+\theta_1}t^{-1})},
\end{equation}
with
\begin{align*}
c_0&=-q^{-2\theta_0+\theta_t}\frac{\Gamma_q(1+2\theta_t,+2\theta_0-2\theta_t,1-\theta_0+\theta_t+\theta_1-\theta_\infty,-\theta_0+\theta_t+\theta_1+\theta_\infty)}{\Gamma_q(1-2\theta_0+2\theta_t)^2\Gamma_q(+2\theta_0,+\theta_0-\theta_t+\theta_1+\theta_\infty,1+\theta_0-\theta_t+\theta_1-\theta_\infty)},\\
c_\infty&=-q^{+2\theta_0+\theta_1}\frac{\Gamma_q(1+2\theta_1,-2\theta_0-2\theta_1,1+\theta_0+\theta_t+\theta_1-\theta_\infty,-\theta_0+\theta_t+\theta_1+\theta_\infty)}{\Gamma_q(1+2\theta_0+2\theta_1)^2\Gamma_q(-2\theta_0,-\theta_0-\theta_t-\theta_1+\theta_\infty,1-\theta_0+\theta_t-\theta_1-\theta_\infty)}.
\end{align*}
Here the non-leading coefficients can be determined recursively in each of the two expansions by direct substitution into the $q\Psix$ equation. In particular, $f$ and $g$ are meromorphic functions on $\mathbb{C}^*\setminus q^{\mathbb{Z}-2\theta_0+\theta_t-\theta_1}$.
Furthermore, when $t\in q^{\mathbb{Z}-2\theta_0+\theta_t-\theta_1}$, the intersection point $N_*(t)$ lies in the hyperplane section at infinity of the Segre surface $\widehat{\mathcal{F}}(\Theta,t)$ and the solution is thus ill-defined there.
\end{theorem}
It is unclear to the author whether the isolated singularities of the solution in Theorem \ref{thm:orbitII}, in the $q$-line $q^{\mathbb{Z}-2\theta_0+\theta_t-\theta_1}$, are essential singularities or poles. They cannot be apparent, as that would contradict Proposition \ref{prop:bijective}.

In Figure \ref{fig:numericsorbitII}, the function $f$ in Theorem \ref{thm:orbitII} is plotted on the negative real line for some values of the parameters.

\begin{figure}[t]
\centering
\includegraphics[width=0.8\textwidth]{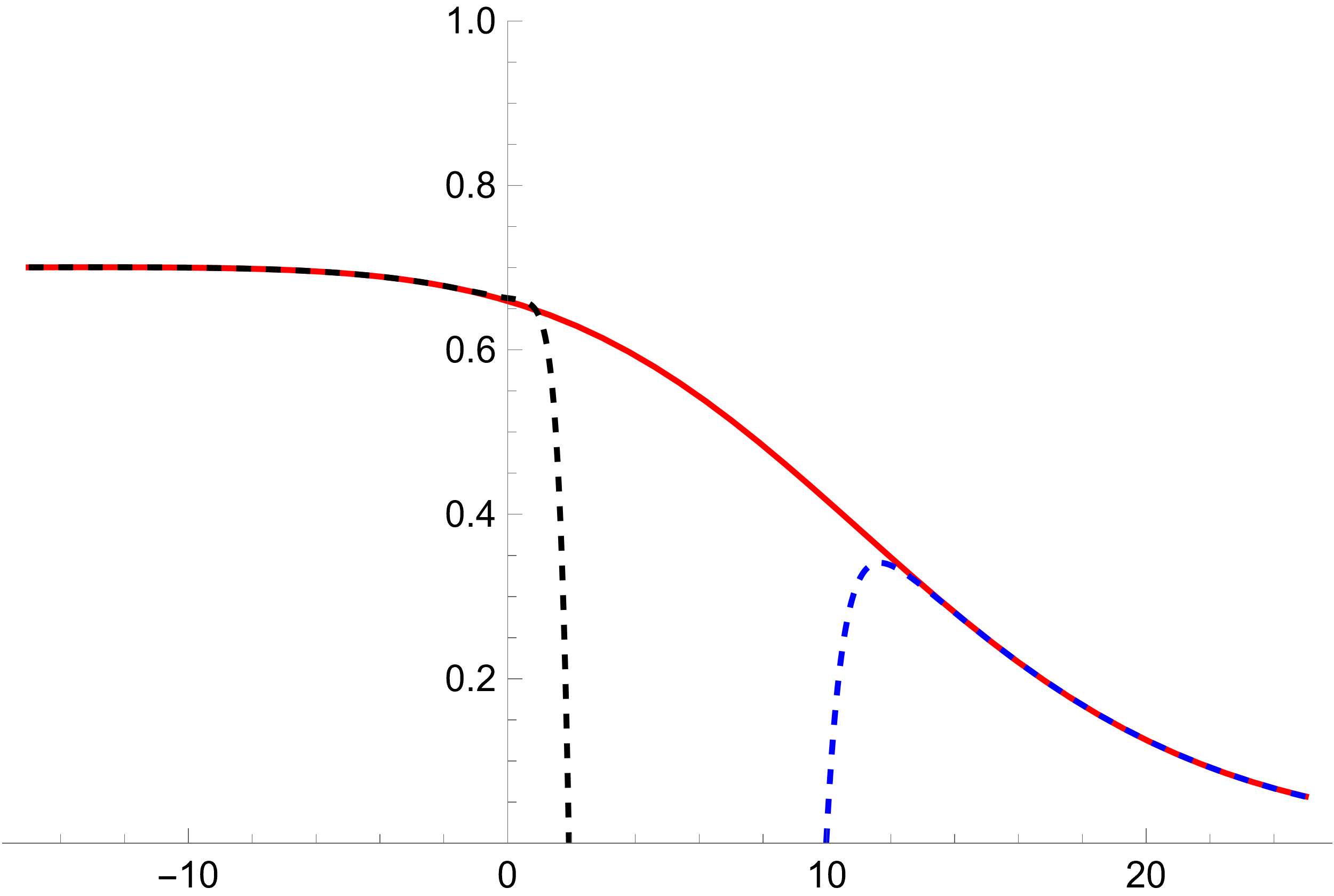}\\
    \caption{\label{fig:numericsorbitII} A plot of the function $f(-q^r)$ in red, with $f=f(t)$ defined in Theorem \ref{thm:orbitII}, where $r$ ranges in the interval $(-15,25)$, with parameter values $\theta_0=\tfrac{1}{3}$, $\theta_t=\tfrac{1}{5}$, $\theta_1=\tfrac{1}{7}$, $\theta_\infty=\tfrac{1}{11}$ and $q=\exp(-\tfrac{1}{4})$. Furthermore, in respectively dashed blue and black the series expansions around $t=0$ and $t=\infty$, in Theorem \ref{thm:orbitII}, are plotted, showing their limited range of convergence. The numerics of the true solution were obtained by analytic continuation via the $q\Psix$ equation of the convergent series solution around $t=\infty$. }
\end{figure}

\subsection{The nonlinear connection problem and uniformisation in $\boldsymbol{t_0}$}\label{subsec:nonlinear_con}
Theorem \ref{thm:generic_asymp_zero} gives the asymptotics of the general (i.e. generic) solution of $q\Psix(\Theta,t_0)$ as $t_m\rightarrow 0$, expressed in terms of the pair of integration constants $\{\sigma_{0t},r_{0t}\}$, which are related to monodromy via the pair of Tyurin ratios $\{\rho_{34},\rho_{24}\}$. On the other hand, Theorem \ref{thm:generic_asymp_infty} gives the asymptotics of the general solution as $t_m\rightarrow \infty$,
expressed in terms of the pair of integration constants $\{\sigma_{01},r_{01}\}$, which are related to monodromy via the pair of Tyurin ratios $\{\rho_{12},\rho_{24}\}$.

The nonlinear connection problem for the general solution of $q\Psix(\Theta,t_0)$ is now reduced to relating the pairs of Tyurin ratios $\{\rho_{34},\rho_{24}\}$ and  $\{\rho_{12},\rho_{24}\}$. The solution to the latter problem is given by the quadratic equation \eqref{eq:quadratic}, which, after division by $\rho_2\rho_4$, becomes
\begin{equation}\label{eq:nonlinear_connection}
T_{13}\rho_{12}\rho_{34}+T_{14}\rho_{12}+T_{23}\rho_{34}+T_{12}\rho_{12}\rho_{24}+T_{34}\rho_{34}/\rho_{24}+T_{24}=0.
\end{equation}

With the nonlinear connection problem solved on $q$-lines, we can study the uniformisation of the asymptotic expansions around $t=0$ and $t=\infty$ in Theorems \ref{thm:generic_asymp_zero} and \ref{thm:generic_asymp_infty} to open domains. Here, `uniformisation' signifies uniform in $t_0$. 
Mano \cite{manoqpvi} constructed two complex parameter families of solutions on open connected domains, invariant on multiplication by $q$, whose asymptotics are described (after a translation) by the series expansions in Theorem \ref{thm:generic_asymp_zero} near $t=0$, with $\sigma_{0t}$ and $r_{0t}$ \textit{true constants}, that is, $\sigma_{0t}$ and $r_{0t}$ are independent of $t_0$ rather than merely $q$-periodic in $t_0$. Here we note that Jimbo et al. \cite{jimbonagoyasakai} derive analogous expansions, with all coefficients given explicitly, near $t=0$ for the corresponding tau-function. Mano \cite{manoqpvi} also constructed two complex parameter families of solutions on open domains as above, whose asymptotics are described by the series expansions in Theorem \ref{thm:generic_asymp_infty} near $t=\infty$ with $\sigma_{01}$ and $r_{01}$ true constants. It is further suggested in \cite{manoqpvi}*{Theorem 5.1},
to consider solutions which admit expansions around $t=0$ and $t=\infty$ with both pairs of integration constants given by true constants simultaneously. 
Somewhat surprisingly, there exist at most two such solutions (counting branches), when the parameters satisfy assumptions \eqref{eq:param_assumptions_1}, as shown in the following corollary.
\begin{corollary}
Assume parameter conditions \eqref{eq:param_assumptions_1} and let $U\subseteq\mathbb{C}^*$ be a (non-empty) open, connected domain, invariant under multiplication by $q$. Suppose a solution admits, simultaneously on $U$, the asymptotic expansion in Theorem \ref{thm:generic_asymp_zero} as $t\rightarrow 0$, with $\{\sigma_{0t},r_{0t}\}$ true constants (i.e. independent of $t_0$), as well as the asymptotic expansion in Theorem \ref{thm:generic_asymp_infty} as $t\rightarrow \infty$, with $\{\sigma_{01},r_{01}\}$ true constants. Then
\begin{equation*}
    \sigma_{0t}=\sigma_{01}=\tfrac{1}{4},
\end{equation*}
and
\begin{equation*}
    s_{0t}=\pm iq^{\frac{1}{2}(1-\theta_t-\theta_1)},\quad s_{01}=\mp iq^{\frac{1}{2}(1-\theta_t-\theta_1)}.
\end{equation*}
\end{corollary}
\begin{proof}
Since $r_{0t}$ and $r_{01}$ are constants, the twist parameters $s_{0t}$ and $s_{01}$ must be constants. Then, from Theorems \ref{thm:generic_asymp_zero} and \ref{thm:generic_asymp_infty}, we obtain two different formulas for the Tyurin ratio $\rho_{24}=\rho_{24}(t_0)$, namely
\begin{equation*}
    M_{0t}^{-1}[-s_{0t}(-t_0)^{+2\sigma_{0t}s_{0t}}]=\rho_{24}(t_0)=
    M_{01}^{-1}[-s_{01}(-t_0)^{-2\sigma_{01}s_{01}}],
\end{equation*}
where the M\"obius transforms $M_{0t}[\cdot]$ and  $M_{01}[\cdot]$ are given explicitly in the theorems.
Equating the formulas on the left and right and some simplification yields
\begin{equation}\label{eq:uniformisationnot}
a(t_0)s_{0t}(-t_0)^{2\sigma_{0t}+2\sigma_{01}}+b(t_0)s_{0t}s_{01}(-t_0)^{2\sigma_{0t}}+c(t_0)(-t_0)^{2\sigma_{01}}+d(t_0)s_{01}=0,
\end{equation}
where
\begin{align*}
    a(t_0)=&q^{-2\sigma_{0t}}\vartheta_\tau(\theta_t+\theta_1-\sigma_{0t}+\sigma_{01})\theta_q(q^{\sigma_{0t}+\sigma_{01}}t_0),\\
    b(t_0)=&q^{-2\sigma_{0t}} \vartheta_\tau(\theta_t+\theta_1-\sigma_{0t}-\sigma_{01})\theta_q(q^{\sigma_{0t}-\sigma_{01}}t_0),\\    c(t_0)=&\vartheta_\tau(\theta_t+\theta_1+\sigma_{0t}+\sigma_{01})\theta_q(q^{-\sigma_{0t}+\sigma_{01}}t_0),\\
     d(t_0)=&\vartheta_\tau(\theta_t+\theta_1+\sigma_{0t}-\sigma_{01})\theta_q(q^{-\sigma_{0t}-\sigma_{01}}t_0),
\end{align*}
which holds identically on $U$, and thus identically on the universal covering of $\mathbb{C}^*$. Since none of the coefficients in \eqref{eq:uniformisationnot} have monodromy as $t_0\mapsto e^{2\pi i}t_0$, taking the difference of the left-hand side of \eqref{eq:uniformisationnot} with itself after shifting $t_0\mapsto e^{2\pi i}t_0$, yields
\begin{gather*}
    a(t_0)(e^{4\pi i(\sigma_{0t}+\sigma_{01})}-1)s_{0t}+
    b(t_0)(e^{4\pi i\sigma_{0t}}-1)s_{0t}s_{01}(-t_0)^{-2\sigma_{01}}\\+
    c(t_0)(e^{4\pi i\sigma_{01}}-1)(-t_0)^{-2\sigma_{0t}}=0.
\end{gather*}
Now, $2\sigma_{0t}\notin\mathbb{Z}$ and $2\sigma_{01}\notin\mathbb{Z}$, as else the asymptotic expansions around $t=0$ or $t=\infty$ respectively are ill-defined. The above equation can thus only be satisfied identically if 
\begin{equation*}
    \sigma_{0t}+\sigma_{01},\sigma_{0t}-\sigma_{01}\in \tfrac{1}{2}\mathbb{Z}.
\end{equation*}
Recalling further that $0\leq \Re\sigma_{0t},\Re\sigma_{0t}<\tfrac{1}{2}$, it follows that $\sigma_{0t}=\sigma_{01}=\tfrac{1}{4}$. This in turn simplifies equation \eqref{eq:uniformisationnot} substantially, as now $d(t_0)=-t_0a(t_0)$ and $c(t_0)=-q^{1-\theta_t-\theta_1}b(t_0)$, so that equation \eqref{eq:uniformisationnot} reduces to the following two algebraic equations among the twist parameters,
\begin{equation*}
    s_{0t}+s_{01}=0,\quad s_{0t}s_{01}=q^{1-\theta_t-\theta_1},
\end{equation*}
from which the corollary follows.
\end{proof}
As the corollary shows, forcing all four integration constants in the asymptotic expansions to be true constants, heavily restricts the solution. It might, however, be interesting to try and construct solutions such that only the exponents $\{\sigma_{0t},\sigma_{01}\}$ are true constants. In other words, $r_{0t}$ and $r_{01}$ are allowed to vary $q$-periodically with $t_0$. This way, the values of the Tyurin ratios $\rho_{34}$ and $\rho_{12}$ are fixed by equations \eqref{eq:elliptic_equation0} and \eqref{eq:elliptic_equationinf}, and $\rho_{24}$ then is required to solve the quadratic \eqref{eq:nonlinear_connection}, uniformly in $t_0$. Starting with some initial point $\rho_{24}(t_0^*)$  that solves the quadratic at $t_0=t_0^*$, this problem has a locally unique solution as long as one stays away from points where the discriminant of \eqref{eq:nonlinear_connection}, with respect to $\rho_{24}$, has a zero or pole of odd multiplicity. If the solution $\rho_{24}$ of the quadratic thus obtained by analytic continuation, satisfies $\rho_{24}(qt_0)=\rho_{24}(t_0)$, then this gives a corresponding solution of $q\Psix$, with generally a finite number of $q$-lines  where the solution has square root singularities, as well as a finite number of $q$-lines on which the solution is ill-defined, as the twist parameters becomes zero or infinite at such points\footnote{Technically the Riemann-Hilbert correspondence is ill-defined at these points as $\eta=\eta(t_0)$ lies in the hyperplane at infinity there.}, and meromorpic elsewhere. The solution in Theorem \ref{thm:orbitII} is an example of this.

\subsection{The continuum limit and announcement of results}\label{subsec:continuumlimit}
By letting $q\rightarrow 1$ in $q\Psix$, under the assumption that $f\rightarrow u$ and $g\rightarrow (u-t)/(u-1)$, a formal computation shows that $q\Psix$ reduces to $\Psix$:
\begin{align*}
u_{tt}=&\left(\frac{1}{u}+\frac{1}{u-1}+\frac{1}{u-t}\right)\frac{u_t^2}{2}-\left(\frac{1}{t}+\frac{1}{t-1}+\frac{1}{u-t}\right)u_t\\
&+\frac{u(u-1)(u-t)}{t^2(t-1)^2}\left(\alpha+\frac{\beta t}{u^2}+\frac{\gamma(t-1)}{(u-1)^2}+\frac{\delta t(t-1)}{(u-t)^2}\right),
\end{align*}
where
\begin{equation*}
    \alpha=\tfrac{1}{2}(2\theta_\infty-1)^2,\quad \beta=-2\theta_0^2,\quad \gamma =2\theta_1^2,\quad \delta=\tfrac{1}{2}-2\theta_t^2.
\end{equation*}

By applying this limit to the leading order behaviours near $t=0$ and $t=\infty$ in Theorems \ref{thm:generic_asymp_zero} and \ref{thm:generic_asymp_infty}, and using that
\begin{equation*}
  \Gamma_q(\beta)\rightarrow \Gamma(\beta),\qquad   \frac{\theta_q(q^\alpha)}{\theta_q(q^\beta)}\rightarrow \frac{\sin \pi \alpha}{\sin \pi \beta},\qquad \frac{\theta_q(q^\alpha t_0)}{\theta_q(q^\beta t_0)}\rightarrow (-t_0)^{\beta-\alpha}\qquad (\arg(-t_0)<\pi),
\end{equation*}
as $q\uparrow 1$, for $\beta\notin\mathbb{Z}$, we formally obtain
\begin{align*}
u(t)&\sim +\frac{(\theta_0+\theta_t+\sigma_{0t}\hspace{0.5mm})(-\theta_0+\theta_t+\sigma_{0t}\hspace{0.5mm})(\theta_\infty+\theta_1+\sigma_{0t}\hspace{0.5mm})}{4(\theta_\infty+\theta_1-\sigma_{0t}\hspace{0.5mm})\sigma_{0t}^2}\frac{1}{r_{0t}}(-t)^{1-2\sigma_{0t}} &(t\rightarrow 0),\\
&\sim -\frac{(\theta_0+\theta_1+\sigma_{01})(-\theta_0+\theta_1+\sigma_{01})(\theta_\infty+\theta_t+\sigma_{01})}{4(\theta_\infty+\theta_t-\sigma_{01})\sigma_{01}^2}\frac{1}{r_{01}}(-1/t)^{-2\sigma_{01}}
&(t\rightarrow \infty),
\end{align*}
where
\begin{equation*}
    r_{0t}=c_{0t}\times s_{0t},\quad r_{01}=c_{01}\times s_{01},
\end{equation*}
with
\begin{align*}
 c_{0t}&=\frac{\Gamma(1-2\sigma_{0t}\hspace{0.5mm})^2}{\Gamma(1+2\sigma_{0t}\hspace{0.5mm})^2}\prod_{\epsilon=\pm1}\frac{ \Gamma(1+\theta_t+\epsilon\,\theta_0+\sigma_{0t}\hspace{0.5mm})  \Gamma(1+\theta_1+\epsilon\,\theta_\infty+\sigma_{0t}\hspace{0.5mm})}{ \Gamma(1+\theta_t+\epsilon\,\theta_0-\sigma_{0t}\hspace{0.5mm})  \Gamma(1+\theta_1+\epsilon\,\theta_\infty-\sigma_{0t}\hspace{0.5mm})},\\  
 c_{01}&=\frac{\Gamma(1-2\sigma_{01})^2}{\Gamma(1+2\sigma_{01})^2}\prod_{\epsilon=\pm1}\frac{ \Gamma(1+\theta_1+\epsilon\,\theta_0+\sigma_{01})  \Gamma(1+\theta_t+\epsilon\,\theta_\infty+\sigma_{01})}{ \Gamma(1+\theta_1+\epsilon\,\theta_0-\sigma_{01})  \Gamma(1+\theta_t+\epsilon\,\theta_\infty-\sigma_{01})},
\end{align*}
and
\begin{align*}
s_{0t}&=\frac{\sin(\pi(-\theta_\infty+\theta_1+\sigma_{0t}))-\dot{\rho}_{24}\sin(\pi(+\theta_\infty+\theta_1+\sigma_{0t}))}{\sin(\pi(+\theta_\infty-\theta_1+\sigma_{0t}))-\dot{\rho}_{24}\sin(\pi(-\theta_\infty-\theta_1+\sigma_{0t}))},\\
    s_{01}&=\frac{\sin(\pi(+\theta_\infty+\theta_t+\sigma_{01}))-\dot{\rho}_{24}\sin(\pi(-\theta_\infty+\theta_t+\sigma_{01}))}{\sin(\pi(-\theta_\infty-\theta_t+\sigma_{01}))-\dot{\rho}_{24}\sin(\pi(+\theta_\infty-\theta_t+\sigma_{01}))}.
\end{align*}
Here $\dot{\rho}_{24}:=(-t_0)^{-2\theta_\infty}\rho_{24}$, equations \eqref{eq:elliptic_equation0} and \eqref{eq:elliptic_equationinf} become
\begin{align*}
    \rho_{34}&=\frac{\sin(\pi(+\theta_\infty-\theta_1+\sigma_{0t}))\sin(\pi(-\theta_\infty+\theta_1+\sigma_{0t}))}{\sin(\pi(+\theta_\infty+\theta_1+\sigma_{0t}))\sin(\pi(-\theta_\infty-\theta_1+\sigma_{0t}))},\\
    \rho_{12}&=\frac{\sin(\pi(+\theta_\infty-\theta_t+\sigma_{01}))\sin(\pi(-\theta_\infty+\theta_t+\sigma_{01}))}{\sin(\pi(+\theta_\infty+\theta_t+\sigma_{01}))\sin(\pi(-\theta_\infty-\theta_t+\sigma_{01}))},
\end{align*}
and equation \eqref{eq:nonlinear_connection} reduces to the following algebraic relation among them,
\begin{equation}\label{eq:nonlinear_connectionlim}
\mathcal{T}_{13}\rho_{12}\rho_{34}+\mathcal{T}_{14}\rho_{12}+\mathcal{T}_{23}\rho_{34}+\mathcal{T}_{0}\left(\rho_{12}\dot{\rho}_{24}+\rho_{34}/\dot{\rho}_{24}\right)+\mathcal{T}_{24}=0,
\end{equation}
where
\begin{align*}
\mathcal{T}_{0}&=+\sin(2\pi \theta_t)\sin(2\pi \theta_1),\\
\mathcal{T}_{13}&=-\sin(\pi(\theta_0+\theta_t+\theta_1+\theta_\infty))\sin(\pi(-\theta_0+\theta_t+\theta_1+\theta_\infty)),\\
\mathcal{T}_{24}&=-\sin(\pi(\theta_0-\theta_t-\theta_1+\theta_\infty))\sin(\pi(-\theta_0-\theta_t-\theta_1+\theta_\infty)),\\
\mathcal{T}_{23}&=+\sin(\pi(\theta_0-\theta_t+\theta_1+\theta_\infty))\sin(\pi(-\theta_0-\theta_t+\theta_1+\theta_\infty)),\\
\mathcal{T}_{14}&=+\sin(\pi(\theta_0+\theta_t-\theta_1+\theta_\infty))\sin(\pi(-\theta_0+\theta_t-\theta_1+\theta_\infty)).
\end{align*}

These formulae are in exact agreement with the solution by Jimbo \cite{jimbo1982}\footnote{We note the following typo in \cite{jimbo1982}: the last `$\mp$' at the bottom of page 1141 should have been a `$\pm$', as pointed out in \cite{guzzettitabulation}, see also \cite{boalchklein}.},  to the nonlinear connection problem between $t=0$ and $t=\infty$ on the domain $\arg(-t)<\pi$ for the general solution of $\Psix$. 
Namely, according to Jimbo's formulae in \cite{jimbo1982}, modified appropriately for the domain $\arg(-t)<\pi$ by application of an element of the braid group \cites{dubmazzocco,lisovyyalgebraic}, the triple of trace coordinates $\{X_{0t},X_{t1},X_{01}\}$,
\begin{equation*}
  X_{0t}:=2\cos(2\pi \sigma_{0t}),\quad X_{01}:=2\cos(2\pi \sigma_{01}),
\end{equation*}
and $X_{t1}$ defined via the twist parameter $s_{0t}$ by
\begin{align*}
    s_{0t}&=\frac{a+b\; e^{2\pi i\sigma_{0t}}}{d},\\
    a&=\tfrac{1}{2}i \sin(2\pi \sigma_{0t})X_{t1}+\cos(2\pi\theta_0)\cos(2\pi\theta_1)+\cos(2\pi\theta_\infty)\cos(2\pi\theta_t),\\
    b&=\tfrac{1}{2}i \sin(2\pi \sigma_{0t})X_{01}+\cos(2\pi\theta_t)\cos(2\pi\theta_1)+\cos(2\pi\theta_\infty)\cos(2\pi\theta_0),\\
    d&=4\prod_{\epsilon=\pm 1}\sin[\pi(\theta_0+\epsilon(\theta_t-\sigma_{0t}))]\sin[\pi(\theta_\infty+\epsilon(\theta_1-\sigma_{0t}))],
\end{align*}
satisfy the cubic equation
\begin{equation}\label{eq:pvicubic}
X_{0t}X_{01}X_{t1}
+X_{0t}^2+X_{01}^2+X_{t1}^2
-\omega_1 X_{0t}-\omega_2X_{01}-\omega_3 X_{t1}+\omega_4=0,
\end{equation}
where
\begin{gather*}
    \omega_1=v_0v_t+v_1v_\infty,\quad
    \omega_2=v_0v_1+v_tv_\infty,\quad
    \omega_3=v_0v_\infty+v_t v_1,\\
    \omega_4=v_0^2+v_t^2+v_1^2+v_\infty^2+v_0v_tv_1v_\infty-4,    
\end{gather*}
with $v_k:=2\cos(2\pi \theta_k)$ for $k=0,t,1,\infty$.

The explicit relations between the triples $\{X_{0t},X_{t1},X_{01}\}$ and $\{\sigma_{0t},s_{0t},\sigma_{01}\}$ on the one hand, and the explicit relations between $\{\sigma_{0t},s_{0t},\sigma_{01}\}$ and $\{\rho_{34},\dot{\rho}_{24},\rho_{12}\}$ on the other hand, give
\begin{equation*}
\rho_{34}=\frac{X_{0t}-2\cos(2\pi(\theta_\infty-\theta_1))}{X_{0t}-2\cos(2\pi(\theta_\infty+\theta_1))},\qquad
\rho_{12}=\frac{X_{01}-2\cos(2\pi(\theta_\infty-\theta_t))}{X_{01}-2\cos(2\pi(\theta_\infty+\theta_t))},
\end{equation*}
and
\begin{equation*}
 \dot{\rho}_{24}=\frac{2X_{t1}\cos(\pi(\theta_\infty-\theta_1))-R_-(X_{0t},X_{01})}{2X_{t1}\cos(\pi(\theta_\infty+\theta_1))-R_+(X_{0t},X_{01})}, 
\end{equation*}
with $R_{\pm}=R_{\pm}(X_{0t},X_{01})$ given by
\begin{equation*}
R_{\pm}=r_{\pm}-e^{-\pi i(\theta_1\pm\theta_\infty)}X_{0t}X_{01}
-2i X_{0t}\sin(\pi(\theta_1+2\theta_t\mp\theta_\infty))+2iX_{01}\sin(\pi(\theta_1\pm\theta_\infty)),
\end{equation*}
where the constant term reads
\begin{gather*}
    r_{\pm}=2e^{-\pi i(\theta_1\pm \theta_\infty)}\cos(2\pi(\theta_t\mp\theta_\infty))+2i e^{2\pi i(\theta_1\pm \theta_\infty)}
\sin(\pi(\theta_1+2\theta_t\mp\theta_\infty))\\
+2\cos(\pi(\theta_1\pm\theta_\infty))\left(e^{2\pi i(\theta_t-\theta_1)}+2e^{\mp 2\pi i\theta_\infty}\cos(2\pi\theta_0))\right).
\end{gather*}

A direct computation shows that this bi-rational relation between the triple of Tyurin parameters $\{\rho_{34},\dot{\rho}_{24},\rho_{12}\}$, and the triple of trace coordinates $\{X_{0t},X_{t1},X_{01}\}$, maps a generic point on the Jimbo-Fricke cubic \eqref{eq:pvicubic} to a point on \eqref{eq:nonlinear_connectionlim} and vice versa, showing the exact match between Jimbo's nonlinear connection formulae and ours for the general solution.

Furthermore, we can similarly take the formal limit of the Segre surface equations \eqref{eq:eta_equations}, yielding
\begin{align*}
&\eta_{12}+\eta_{13}+\eta_{14}+\eta_{23}+\eta_{24}+\eta_{34}=0,\\
&\eta_{12}+\alpha_{13}\eta_{13}+\alpha_{14}\eta_{14}+\alpha_{23}\eta_{23}+\alpha_{24}\eta_{24}+\eta_{34}+\alpha_\infty=0,\\
    &\eta_{13}\eta_{24}-\beta_1\eta_{12}\eta_{34}=0,\\ 
    &\eta_{14}\eta_{23}-\beta_2\eta_{12}\eta_{34}=0,
\end{align*}
where $\alpha_\infty=-\sin(\pi\theta_0)^{-2}$ and the other coefficients are given by
\begin{align*}
    &\alpha_{13}=\prod_{\epsilon=\pm 1} \frac{\sin(\pi\left(\theta_t+\theta_1+\theta_\infty\right))}{\sin(\pi\left(\epsilon \,\theta_0+\theta_t+\theta_1+\theta_\infty\right))}, &
    &\alpha_{24}=\prod_{\epsilon=\pm 1}\frac{\sin(\pi\left(-\theta_t-\theta_1+\theta_\infty\right))}{\sin(\pi\left(\epsilon \,\theta_0-\theta_t-\theta_1+\theta_\infty\right))},\\
    &\alpha_{14}=\prod_{\epsilon=\pm 1}\frac{\sin(\pi\left(\theta_t-\theta_1+\theta_\infty\right))}{\sin(\pi\left(\epsilon \,\theta_0+\theta_t-\theta_1+\theta_\infty\right))}, &
    &\alpha_{23}=\prod_{\epsilon=\pm 1}\frac{\sin(\pi\left(-\theta_t+\theta_1+\theta_\infty\right))}{\sin(\pi\left(\epsilon \,\theta_0-\theta_t+\theta_1+\theta_\infty\right))},
\end{align*}
and
\begin{align*}
    \beta_1&=\prod_{\epsilon=\pm 1} \frac{\sin(\pi\left(\epsilon\, \theta_0-\theta_t-\theta_1+\theta_\infty\right))\sin(\pi\left(\epsilon\, \theta_0+\theta_t+\theta_1+\theta_\infty\right))}{\sin(2\pi\theta_t)\sin(2\pi\theta_1)},\\
    \beta_2&=
    \prod_{\epsilon=\pm 1} \frac{\sin(\pi\left(\epsilon \, \theta_0-\theta_t+\theta_1+\theta_\infty\right))\sin(\pi\left(\epsilon \, \theta_0+\theta_t-\theta_1+\theta_\infty\right))}{\sin(2\pi\theta_t)\sin(2\pi\theta_1)}.
\end{align*}
The defining equations of the $\eta$-coordinates, equations \eqref{eq:eta_defi}, also have well-defined formal limits as $q\rightarrow 1$, so that the bi-rational equivalence between the $\eta$-coordinates and the triple $\{\rho_{34},\dot{\rho}_{24},\rho_{12}\}$, gives an explicit bi-rational mapping between the limiting Segre surface and the cubic surface \eqref{eq:pvicubic}. In a forthcoming work \cite{joshimazzoccoroffelsen}, this mapping is characterised algebraically and shown to map the affine part of the limiting Segre surface biholomorphically onto the affine part of the cubic surface. 

Furthermore, the author \cite{roffelsencontinuum} is currently preparing a manuscript in which the formal limit described above is made rigorous on the level of solutions by proving convergence of the solutions of associated Riemann-Hilbert problems as $q\uparrow 1$.

% \begin{equation*}
% \dot{\rho}_{24}=\frac{e^{-\pi i(\theta_1-\theta_\infty)}X_{0t}X_{01}
% +2i X_{0t}\sin(\pi(\theta_1+\theta_\infty+2\theta_t))-2iX_{01}\sin(\pi(\theta_1-\theta_\infty))+2X_{t1}\cos(\pi(\theta_1-\theta_\infty))-c(-\theta_\infty)}{e^{-\pi i(\theta_1+\theta_\infty)}X_{0t}X_{01}
% +2i X_{0t}\sin(\pi(\theta_1-\theta_\infty+2\theta_t))-2iX_{01}\sin(\pi(\theta_1+\theta_\infty))+2X_{t1}\cos(\pi(\theta_1+\theta_\infty))-c(+\theta_\infty)},
% \end{equation*}

%% file: segre_proofs.tex
\section{Geometry of the Segre surface}\label{sec:geometrysegre}
In this section we prove all the statements made on the Segre surface $\widehat{\mathcal{F}}$ in Section \ref{subsec:geometry}. In Subsection \ref{sec:smoothness} we show that the Segre surface smooth, proving Proposition \ref{lemma:smooth}. Then in Subsection \ref{sec:tyurinquotients} we study the (dual) Tyurin ratios, in particular proving Lemmas \ref{lem:duality} and \ref{lemma:global_mero}. In Subsection \ref{sec:lines} we derive explicit expressions for the sixteen lines on the Segre surface and in Subsection \ref{sec:intersectionpoints} we study their points of intersection, proving in particular Theorem \ref{thm:intersection}.

\subsection{Smoothness of the Segre surface}\label{sec:smoothness}
In this section, we show that the Segre surface $\widehat{\mathcal{F}}$ is smooth. It follows from \cite{roffelsenjoshiqpvi}*{Theorem 2.17}, that there can be no singularities in the affine part $\mathcal{F}$ under the parameter assumptions \eqref{eq:param_assumptions_1} and \eqref{eq:param_assumptions_2}, but this leaves open the possibility for a singularity in the hyperplane section at infinity. In what follows, we  give a direct proof of the smoothness of $\widehat{\mathcal{F}}$ by elementary means.

\begin{proof}[Proof of Proposition \ref{lemma:smooth}]
Consider the Jacobian matrix corresponding to equations \eqref{segre:projective},
\begin{equation*}
    J=\begin{pmatrix}
        1 & 1 & 1 & 1 & 1 & 1 & 0\\
        a_{12} & a_{13} & a_{14} & a_{23} & a_{24} & a_{34} & a_\infty\\
        -b_1 N_{34} & N_{24} & 0 & 0 & N_{13} & -b_1 N_{12} & 0\\
        -b_2 N_{34} & 0 & N_{23} & N_{14} & 0 & -b_2 N_{12} & 0\\
    \end{pmatrix}.
\end{equation*}
To prove the lemma, it is enough to show that there exists no point $N\in \widehat{\mathcal{F}}$ such that $J$ has row-rank less than 4. By subtracting $a_{13}$ times the first row of $J$ from the second row and similarly subtracting $N_{24}$ times the first row from the third, we find that $J$ is row-equivalent to
\begin{equation*}
\begin{pmatrix}
        1 & 1 & 1 & 1 & 1 & 1 & 0\\
        a_{12}-a_{13} & 0 & a_{14}-a_{13} & a_{23}-a_{13} & a_{24}-a_{13} & a_{34}-a_{13} & a_\infty\\
        -b_1 N_{34}-N_{24} & 0 & -N_{24} & -N_{24} & N_{13}-N_{24} & -b_1 N_{12}-N_{24} & 0\\
        -b_2 N_{34} & 0 & N_{23} & N_{14} & 0 & -b_2 N_{12} & 0\\
    \end{pmatrix}.
\end{equation*}
Since $a_\infty\neq 0$, it follows that $J$ has rank less than four if and only if the following two row vectors are linearly dependent,
\begin{equation*}
\begin{array}{lcccccr}
    v_1=( & -b_1 N_{34}-N_{24}   & -N_{24} & -N_{24} & N_{13}-N_{24} & -b_1 N_{12}-N_{24} & ), \\
    v_2=( & -b_2 N_{34}  & N_{23} & N_{14} & 0 & -b_2 N_{12} & ).
\end{array}
\end{equation*}

Suppose that $N\in \widehat{\mathcal{F}}$ is such that these two vectors are linearly dependent. We proceed to derive a contradiction, yielding the lemma.

Firstly, we note that neither $v_1$ nor $v_2$ can be zero. For, if $v_1=0$, then
\begin{equation*}
    N_{12}=N_{13}=N_{24}=N_{34}=0.
\end{equation*}
Equations \eqref{eq:eta_equationsaproj}, \eqref{eq:eta_equationsbproj} and  \eqref{eq:eta_equationsdproj} then respectively reduce to
\begin{align*}
 &N_{14}+N_{23}=0,\\
&a_{14}N_{14}+a_{23}N_{23}+a_{\infty}N_\infty=0,\\
&N_{14}N_{23}=0,   
\end{align*}
showing that also $N_{14}=N_{23}=N_\infty=0$ and we have arrived at a contradiction. Similarly, it follows that $v_2$ cannot equal zero.

Therefore, for $v_1$ and $v_2$ to be linearly dependent, there must exist a $\lambda\in\mathbb{C}^*$ such that $v_2=\lambda v_1$. We may now solve for some of the variables in terms of $N_{34}$, giving
\begin{align*}
N_{13}&=+(b_2-b_1\lambda) N_{34}\lambda^{-1},\\
N_{14}&=-(b_2-b_1\lambda) N_{34},\\
N_{23}&=-(b_2-b_1\lambda) N_{34},\\
N_{24}&=+(b_2-b_1\lambda) N_{34}\lambda^{-1},
\end{align*}
and the equation $v_2-\lambda v_1=0$ then reduces to
\begin{equation*}
    (b_2-b_1\lambda)(N_{12}-N_{34})=0.
\end{equation*}
If $b_2-b_1\lambda=0$, then we deduce that all the entries of $N$ are zero, as before, yielding a contradiction. Therefore, we must have $N_{12}=N_{34}$.

If $N_{34}=0$, then again we see that all the entries of $N$ are zero, yielding a contradiction. So $N_{34}\neq 0$ and Equations \eqref{eq:eta_equationsaproj}, \eqref{eq:eta_equationscproj} and  \eqref{eq:eta_equationsdproj} now reduce to the following quadratic equations respectively,
\begin{align*}
    &b_1\lambda^2-(b_1+b_2-1)\lambda+b_2=0,\\
    &b_1(b_1-1)\lambda^2-2b_1b_2\lambda+b_2^2=0,\\
    &b_1\lambda^2-2b_1b_2\lambda+b_2(b_2-1)=0.
\end{align*}
These three equations admit a common solution $\lambda$ if and only if the following common factor, in their mutual resultants, vanishes,
\begin{equation}\label{eq:b1b2vanishing}
 (b_1-b_2)^2-2(b_1+b_2)+1=0.
\end{equation}

By substituting the defining equations for $b_1$ and $b_2$, we see that the left-hand side equals
\begin{equation*}
 (b_1-b_2)^2-2(b_1+b_2)+1=\frac{|H|}{T_{12}^2T_{34}^2},
\end{equation*}
where $H$ is the Hessian matrix of $T(\rho)$,
\begin{equation}\label{eq:hessian}
    H=\begin{pmatrix}
    0 & T_{12} & T_{13} & T_{14}\\
    T_{12} &0  & T_{23} & T_{24}\\
    T_{13} & T_{23} & 0 & T_{34}\\
    T_{14} & T_{24} & T_{34} & 0\\
    \end{pmatrix}.
\end{equation}
The determinant of this Hessian was computed in \cite{roffelsenjoshiqpvi}*{Eq. (5.27)} and reads
\begin{equation*}
    |H|=[q^{-\theta_0+\theta_t+\theta_1+\theta_\infty}\vartheta_\tau(2\theta_0,2\theta_t,2\theta_1,2\theta_\infty)\theta_q(q^{\theta_t+\theta_1}t_0,q^{\theta_t-\theta_1}t_0,q^{-\theta_t+\theta_1}t_0,q^{-\theta_t-\theta_1}t_0)]^2.
\end{equation*}
In particular, it cannot vanish due to the parameter assumptions \eqref{eq:param_assumptions_1} and \eqref{eq:param_assumptions_2}. Thus equation \eqref{eq:b1b2vanishing} cannot hold true and we have arrived at a contradiction. The proposition follows.
\end{proof}

\subsection{Ratios of Tyurin parameters}\label{sec:tyurinquotients}
In this section, we study the (dual) Tyurin ratios defined in equation \eqref{eq:def_rhoij}, and prove Lemma \ref{lem:duality}, Lemma \ref{lemma:global_mero} and Remark \ref{rem:duality}. In fact, we will prove them in a slightly more symmetric setting.
Namely, we consider any $2\times 2$ matrix function $C(z)$, which is analytic on $\mathbb{C}^*$,  satisfies a $q$-difference equation
\begin{equation*}
        C(qz)=\frac{t_0}{z^2}q^{\theta_0\sigma_3}C(z)q^{\theta_\infty\sigma_3},
    \end{equation*}
and whose determinant equals
    \begin{equation*}
    |C(z)|=c\, \theta_q\left(z/x_1,z/x_2,z/x_3,z/x_4\right),
    \end{equation*}
for some nonzero scalar $c\in\mathbb{C}^*$, where $x_k$, $1\leq k \leq 4$, are nonzero complex numbers that satisfy $x_1x_2x_3x_4=t_0^2$, and
\begin{equation}\label{eq:parameterassumptionsx}
     \frac{x_j}{x_k}\notin q^{\mathbb{Z}},\qquad \frac{q^{\epsilon_0 \theta_0+\epsilon_\infty \theta_\infty}t_0}{x_j x_k}\notin q^{\mathbb{Z}}\qquad (1\leq j<k\leq 4,\epsilon_{0,\infty}\in\{\pm 1\}).
\end{equation}
The first conditions in the above equation, plus $2\theta_0,2\theta_\infty\notin \Lambda_\tau$, are the non-resonance conditions, and the second conditions are the non-splitting conditions, in \cites{ohyamaramissualoy,roffelsenjoshiqpvi}. Under the parameter identification
\begin{equation}\label{eq:identification}
    x_1=q^{+\theta_t}t_0,\quad x_2=q^{-\theta_t}t_0,\quad x_3=q^{+\theta_1},\quad x_4=q^{-\theta_1},
\end{equation}
these conditions coincide with the union of the conditions in equations \eqref{eq:param_assumptions_1} and \eqref{eq:param_assumptions_2}.

We define Tyurin and dual Tyurin parameters,
\begin{equation*}
    \rho_j=\pi[C(x_j)],\quad \widetilde{\rho}_j=\pi[C(x_j)^T]\qquad (1\leq j\leq 4),
\end{equation*}
as well as their ratios,
\begin{equation*}
\rho_{jk}=\frac{\rho_j}{\rho_k},\quad \widetilde{\rho}_{jk}=\frac{\widetilde{\rho}_j}{\widetilde{\rho}_k}\qquad (1\leq j,k\leq 4,j\neq k),
\end{equation*}
which are always well-defined elements of $\mathbb{CP}^1$, due to the non-splitting conditions.

The Tyurin and dual Tyurin parameters satisfy the homogeneous quadratic equations,
\begin{equation}\label{eq:tyurinquadratics}
    \sum_{1\leq i<j\leq 4}T_{ij}(\theta_0,\theta_\infty) \rho_i\rho_j=0,\quad
    \sum_{1\leq i<j\leq 4}T_{ij}(\theta_\infty,\theta_0) \widetilde{\rho}_i \widetilde{\rho}_j=0,
\end{equation}
where, for any labeling $\{i,j,k,l\}=\{1,2,3,4\}$,
\begin{equation*}
    T_{ij}(\theta_0,\theta_\infty)=\frac{x_2x_4}{t_0q^{\theta_0+\theta_\infty}}x_i x_l\theta_q\left(\frac{x_i}{x_j},\frac{x_k}{x_l},\frac{x_ix_j}{t_0q^{+\theta_0-\theta_\infty}},\frac{x_kx_l}{t_0q^{+\theta_0+\theta_\infty}}\right).
\end{equation*}
We will now prove Lemma \ref{lem:duality}, which relates the Tyurin and dual Tyurin ratios bi-rationally, in this symmetric setting.
\begin{proof}[Proof of Lemma \ref{lem:duality}]
We derive an explicit formula for $\widetilde{\rho}_{12}$ in terms of $\rho_{13}$ and $\rho_{23}$, in the symmetric setting described in the first part of this subsection. This will yield formula \eqref{eq:rhoduality1} in Lemma \ref{lem:duality}, with $(i,j,k,l)=(1,2,3,4)$, and the others then follow by symmetry.

We start with the fact that the space of scalar analytic functions $c(z)$, satisfying
\begin{equation*}
    c(qz)=\frac{t_0}{z^2}q^{\epsilon_0 \theta_0+\epsilon_\infty \theta_\infty}c(z),
\end{equation*}
is a vector space of dimension two, with a basis given by
\begin{equation*}
    c_1(z)=\theta_q(z/x_1)\theta_q\left(\frac{x_1 z}{t_0 q^{\epsilon_0 \theta_0+\epsilon_\infty \theta_\infty}}\right),\quad
    c_2(z)=\theta_q(z/x_1)\theta_q\left(\frac{x_2 z}{t_0 q^{\epsilon_0 \theta_0+\epsilon_\infty \theta_\infty}}\right),
\end{equation*}
for any $\epsilon_{0,\infty}\in\{\pm 1\}$.
It follows that we may write
\begin{align*}
C(z)=&\theta_q(z/x_1)\begin{pmatrix} u_{11} \theta_q\big(\frac{x_1 z}{t_0 q^{+\theta_0+ \theta_\infty}}\big) & u_{12} \theta_q\big(\frac{x_1 z}{t_0 q^{+\theta_0- \theta_\infty}}\big)   \\
u_{21} \theta_q\big(\frac{x_1 z}{t_0 q^{-\theta_0+ \theta_\infty}}\big) & u_{22} \theta_q\big(\frac{x_1 z}{t_0 q^{-\theta_0- \theta_\infty}}\big) 
\end{pmatrix}+\\
&\theta_q(z/x_2)\begin{pmatrix} v_{11} \theta_q\big(\frac{x_2 z}{t_0 q^{+\theta_0+ \theta_\infty}}\big) & v_{12} \theta_q\big(\frac{x_2 z}{t_0 q^{+\theta_0- \theta_\infty}}\big)   \\
v_{21} \theta_q\big(\frac{x_2 z}{t_0 q^{-\theta_0+ \theta_\infty}}\big) & v_{22} \theta_q\big(\frac{x_2 z}{t_0 q^{-\theta_0- \theta_\infty}}\big) 
\end{pmatrix},
\end{align*}
for some coefficients $u_{ij}, v_{ij}$, $1\leq i,j\leq 2$, which cannot all be zero.

Next, we are going to derive homogeneous linear relations for the $u_{ij}, v_{ij}$, $1\leq i,j\leq 2$, with coefficients involving the Tyurin parameters and dual Tyurin parameters.
Firstly, the equations
\begin{equation*}
    \pi[C(x_2)]=\rho_2,\quad \pi[C(x_2)^T]=\widetilde{\rho}_2,
\end{equation*}
yield four linear homogeneous equations for the $u_{ij}$, $1\leq i,j\leq 2$. Using homogeneous coordinates,
\begin{equation*}
    \rho_k=(\rho_k^x:\rho_k^y)\in\mathbb{CP}^1,\quad \widetilde{\rho}_k=(\widetilde{\rho}_k^x:\widetilde{\rho}_k^y)\in\mathbb{CP}^1\qquad (1\leq k\leq 4),
\end{equation*}  
they can be written as
\begin{equation}\label{eq:Mueq}
    M_u \cdot (u_{11},u_{12},u_{21},u_{22})^T=0,
\end{equation}
where
\begin{equation*}
M_u=\begin{pmatrix}
-\rho_2^y\theta_q\big(\frac{x_1x_2}{q^{+\theta_0+\theta_\infty}t_0}\big) & \rho_2^x\theta_q\big(\frac{x_1x_2}{q^{+\theta_0-\theta_\infty}t_0}\big) & 0 & 0\\
-\widetilde{\rho}_2^y\theta_q\big(\frac{x_1x_2}{q^{+\theta_0+\theta_\infty}t_0}\big) & 0 &
\widetilde{\rho}_2^x\theta_q\big(\frac{x_1x_2}{q^{-\theta_0+\theta_\infty}t_0}\big) & 0\\
0 & -\widetilde{\rho}_2^y\theta_q\big(\frac{x_1x_2}{q^{+\theta_0-\theta_\infty}t_0}\big) & 0 & \widetilde{\rho}_2^x\theta_q\big(\frac{x_1x_2}{q^{-\theta_0-\theta_\infty}t_0}\big)\\
0 & 0 & -\rho_2^y\theta_q\big(\frac{x_1x_2}{q^{-\theta_0+\theta_\infty}t_0}\big) & \widetilde{\rho}_2^x\theta_q\big(\frac{x_1x_2}{q^{-\theta_0-\theta_\infty}t_0}\big)
\end{pmatrix}.
\end{equation*}
Note, furthermore, that this matrix is not invertible, so that \eqref{eq:Mueq} indeed admits non-trivial solutions.

Similarly, the equations
\begin{equation*}
    \pi[C(x_1)]=\rho_1, \quad \pi[C(x_1)^T]=\widetilde{\rho}_1,
\end{equation*}
yield four linear homogeneous equations for the $v_{ij}$, $1\leq i,j\leq 2$, given by
\begin{equation}\label{eq:Mveq}
    M_v \cdot (v_{11},v_{12},v_{21},v_{22})^T=0,
\end{equation}
where
\begin{equation*}
M_v=\begin{pmatrix}
-\rho_1^y\theta_q\big(\frac{x_1x_2}{q^{+\theta_0+\theta_\infty}t_0}\big) & \rho_1^x\theta_q\big(\frac{x_1x_2}{q^{+\theta_0-\theta_\infty}t_0}\big) & 0 & 0\\
-\widetilde{\rho}_1^y\theta_q\big(\frac{x_1x_2}{q^{+\theta_0+\theta_\infty}t_0}\big) & 0 &
\widetilde{\rho}_1^x\theta_q\big(\frac{x_1x_2}{q^{-\theta_0+\theta_\infty}t_0}\big) & 0\\
0 & -\widetilde{\rho}_1^y\theta_q\big(\frac{x_1x_2}{q^{+\theta_0-\theta_\infty}t_0}\big) & 0 & \widetilde{\rho}_1^x\theta_q\big(\frac{x_1x_2}{q^{-\theta_0-\theta_\infty}t_0}\big)\\
0 & 0 & -\rho_1^y\theta_q\big(\frac{x_1x_2}{q^{-\theta_0+\theta_\infty}t_0}\big) & \widetilde{\rho}_1^x\theta_q\big(\frac{x_1x_2}{q^{-\theta_0-\theta_\infty}t_0}\big)
\end{pmatrix}.
\end{equation*}

Finally, the equation
\begin{equation*}
    \pi[C(x_3)]=\rho_3,
\end{equation*}
gives two more homogeneous equations for the coefficients,
\begin{align*}
&\rho_3^x\left(u_{12}\theta_q(x_3/x_1)\theta_q\left(\frac{x_1x_3}{q^{+\theta_0-\theta_\infty}t_0}\right)+v_{12}\theta_q(x_3/x_2)\theta_q\left(\frac{x_2x_3}{q^{+\theta_0-\theta_\infty}t_0}\right)\right)-\\
&\rho_3^y\left(u_{11}\theta_q(x_3/x_1)\theta_q\left(\frac{x_1x_3}{q^{+\theta_0+\theta_\infty}t_0}\right)+v_{11}\theta_q(x_3/x_2)\theta_q\left(\frac{x_2x_3}{q^{+\theta_0+\theta_\infty}t_0}\right)\right)=0,\\
&\rho_3^x\left(u_{22}\theta_q(x_3/x_1)\theta_q\left(\frac{x_1x_3}{q^{-\theta_0-\theta_\infty}t_0}\right)+v_{22}\theta_q(x_3/x_2)\theta_q\left(\frac{x_2x_3}{q^{-\theta_0-\theta_\infty}t_0}\right)\right)-\\
&
\rho_3^y\left(u_{21}\theta_q(x_3/x_1)\theta_q\left(\frac{x_1x_3}{q^{-\theta_0+\theta_\infty}t_0}\right)+v_{21}\theta_q(x_3/x_2)\theta_q\left(\frac{x_2x_3}{q^{-\theta_0+\theta_\infty}t_0}\right)\right)=0.
\end{align*}

In summary, we have a system of four homogeneous linear equations for the $(u_{ij})$,  equation \eqref{eq:Mueq}, whose rank is less than four. We have a system of four homogeneous linear equations for the $(v_{ij})$,  equation \eqref{eq:Mveq}, whose rank is less than four. And we have two homogeneous linear equations for both sets of variables, $(u_{ij})$ and $(v_{ij})$. By taking the two homogeneous linear equations for both sets of variables and adding the first three equations of \eqref{eq:Mueq} and the last three equations of \eqref{eq:Mveq}, we arrive at a homogeneous linear system of eight equations for the eight variables $(u_{ij})$ and $(v_{ij})$. Since $C(z)\not\equiv 0$, the corresponding determinant of the homogeneous linear system must be zero. By computing this determinant and setting it equal to zero, we obtain
\begin{align*}
&\widetilde{\rho}_1^x\widetilde{\rho}_2^y\bigg(\rho_2^x\rho_3^y\theta_q\left(\frac{x_1x_2}{q^{-\theta_0-\theta_\infty}t_0}\right)\theta_q\left(\frac{x_1x_3}{q^{-\theta_0+\theta_\infty}t_0}\right)-
\rho_2^y\rho_3^x\theta_q\left(\frac{x_1x_2}{q^{-\theta_0+\theta_\infty}t_0}\right)\theta_q\left(\frac{x_1x_3}{q^{-\theta_0-\theta_\infty}t_0}\right)\bigg)\\
&\bigg(\rho_1^x\rho_3^y\theta_q\left(\frac{x_1x_2}{q^{+\theta_0-\theta_\infty}t_0}\right)\theta_q\left(\frac{x_2x_3}{q^{+\theta_0+\theta_\infty}t_0}\right)-
\rho_1^y\rho_3^x\theta_q\left(\frac{x_1x_2}{q^{+\theta_0+\theta_\infty}t_0}\right)\theta_q\left(\frac{x_2x_3}{q^{+\theta_0-\theta_\infty}t_0}\right)\bigg)=\\
&\widetilde{\rho}_1^y\widetilde{\rho}_2^x
\bigg(\rho_2^x\rho_3^y\theta_q\left(\frac{x_1x_2}{q^{+\theta_0-\theta_\infty}t_0}\right)\theta_q\left(\frac{x_1x_3}{q^{+\theta_0+\theta_\infty}t_0}\right)-
\rho_2^y\rho_3^x\theta_q\left(\frac{x_1x_2}{q^{+\theta_0+\theta_\infty}t_0}\right)\theta_q\left(\frac{x_1x_3}{q^{+\theta_0-\theta_\infty}t_0}\right)\bigg)\\
&\bigg(\rho_1^x\rho_3^y\theta_q\left(\frac{x_1x_2}{q^{-\theta_0-\theta_\infty}t_0}\right)\theta_q\left(\frac{x_2x_3}{q^{-\theta_0+\theta_\infty}t_0}\right)-
\rho_1^y\rho_3^x\theta_q\left(\frac{x_1x_2}{q^{-\theta_0+\theta_\infty}t_0}\right)\theta_q\left(\frac{x_2x_3}{q^{-\theta_0-\theta_\infty}t_0}\right)\bigg),
\end{align*}
This gives equation \eqref{eq:rhoduality1}, with $(i,j,k,l)=(1,2,3,4)$,  after several applications of the formula
\begin{equation*}
    \theta_q\left(\frac{x_ix_j}{q^{-\theta_0\pm\theta_\infty}t_0}\right)=-\frac{x_ix_j}{q^{-\theta_0\pm\theta_\infty}t_0}\theta_q\left(\frac{x_kx_l}{q^{+\theta_0\mp\theta_\infty}t_0}\right)\qquad
    (\{i,j,k,l\}=\{1,2,3,4\}).
\end{equation*}

By permuting indices, we consequently obtain equation \eqref{eq:rhoduality1} in general. Furthermore, by replacing $C(z)\mapsto C(z)^T$ in the above, we interchange $\theta_0\leftrightarrow \theta_\infty$ and $\rho_j\leftrightarrow \widetilde{\rho}_j$, $1\leq j\leq 4$. thus yielding the dual equations \eqref{eq:rhoduality2}. The lemma follows.
\end{proof}

Next, we prove Lemma \ref{lemma:global_mero} and Remark \ref{rem:duality}.

\begin{proof}[Proof of Lemma \ref{lemma:global_mero} and Remark \ref{rem:duality}]
We start by proving that the Tyurin ratios $\rho_{ij}$, $1\leq i,j\leq 4, i\neq j$, define analytic functions from $\widehat{\mathcal{F}}$ to $\mathbb{CP}^1$. 

Consider the ratio $\rho_{12}$, which admits the following rational expressions on $\widehat{\mathcal{F}}$, see equation \eqref{eq:rho12},
\begin{equation*}
    \rho_{12}=\frac{T_{23}\,N_{13}}{T_{13}\,N_{23}}=\frac{T_{24}\,N_{14}}{T_{14}\,N_{24}}.
\end{equation*}
The only potential obstacle to $\rho_{12}$ defining an analytic function from $\widehat{\mathcal{F}}$ to $\mathbb{CP}^1$, would be a point $N\in \widehat{\mathcal{F}}$ with
\begin{equation*}
    N_{13}=N_{23}=N_{14}=N_{24}=0.
\end{equation*}
But such points do not exist, as, given these four equalities, equations \eqref{eq:eta_equationsaproj}, \eqref{eq:eta_equationsbproj} and \eqref{eq:eta_equationscproj} reduce to
\begin{align*}
    &N_{12}+N_{34}=0,\\
    &a_{12}N_{12}+a_{34}N_{34}+a_\infty N_\infty=0,\\
    &N_{12}N_{34}=0.
\end{align*}
which yield $N_{12}=N_{34}=N_{\infty}=0$. Similarly, we prove that the other Tyurin ratios $\rho_{ij}$, $1\leq i,j\leq 4, i\neq j$, define analytic functions from $\widehat{\mathcal{F}}$ to $\mathbb{CP}^1$.

Now, to prove that the dual Tyurin ratios define analytic functions from $\widehat{\mathcal{F}}$ to $\mathbb{CP}^1$, and thus Lemma \ref{lem:duality}, we start by proving Remark \ref{rem:duality}. Without loss of generality, we will restrict ourselves to proving all the statements for $\widetilde{\rho}_{12}$. We start by writing down the two formulas that Lemma \ref{lem:duality} provides for $\widetilde{\rho}_{12}$,
\begin{align}
   \widetilde{\rho}_{12}&=\frac{x_2}{x_1}M_{(1,2,3,4)}\left( \rho_{13};\theta_0,\theta_\infty\right)M_{(2,4,3,1)}\left(\rho_{23},\theta_0,\theta_\infty\right)\label{eq:rhoduality12a}\\
   &=\frac{x_2}{x_1}M_{(1,2,4,3)}\left( \rho_{14};\theta_0,\theta_\infty\right)M_{(2,3,4,1)}\left(\rho_{24},\theta_0,\theta_\infty\right).\label{eq:rhoduality12b}
\end{align}
These formulas show that $\widetilde{\rho}_{12}$ is a rational function on $\widehat{\mathcal{F}}$, since the Tyurin ratios on the right-hand sides are.

The only points in $\widehat{\mathcal{F}}$ where $\widetilde{\rho}_{12}$ might be singular, are points where, for each of the two right-hand sides, one of the M\"obius transforms evaluates to $0$, and the other to $\infty$. In other words, at such points, either
\begin{itemize}
    \item[(1)] $M_{(1,2,3,4)}(\rho_{13})=0$ and $M_{(2,4,3,1)}(\rho_{23})=\infty$, or
    \item[(2)] $M_{(1,2,3,4)}(\rho_{13})=\infty$ and $M_{(2,4,3,1)}(\rho_{23})=0$,
\end{itemize}
and either
\begin{itemize}
    \item[(i)] $M_{(1,2,4,3)}(\rho_{14})=\infty$ and $M_{(2,3,4,1)}(\rho_{24})=0$, or
    \item[(ii)] $M_{(1,2,4,3)}(\rho_{14})=0$ and $M_{(2,3,4,1)}(\rho_{24})=\infty$.
\end{itemize}
This gives us four cases to consider.

Let us first show the combinations (1) plus (ii) and (2) plus (i) can never occur, as they contradict the equality
\begin{equation}\label{eq:rhotrivial}
    \rho_{13} \rho_{24}=\rho_{23}\rho_{14}.
\end{equation}
Indeed, if (1) plus (ii) holds, then (1) specifies the values of $\rho_{13}$ and $\rho_{23}$ and (ii) specifies the values of $\rho_{14}$ and $\rho_{24}$, giving
\begin{align*}
    0=&(\rho_{13} \rho_{24}-\rho_{23}\rho_{14})\;\theta_q\left(\frac{x_3x_4}{t_0q^{+\theta_0-\theta_\infty}}\right)^{-2}\\
    =&
    \theta_q\bigg(\frac{x_1x_3}{t_0q^{+\theta_0-\theta_\infty}},\frac{x_1x_4}{t_0q^{+\theta_0+\theta_\infty}},\frac{x_2x_3}{t_0q^{+\theta_0+\theta_\infty}},\frac{x_2x_4}{t_0q^{+\theta_0-\theta_\infty}}\bigg)-\\
    &\theta_q\bigg(\frac{x_1x_3}{t_0q^{+\theta_0+\theta_\infty}},\frac{x_1x_4}{t_0q^{+\theta_0-\theta_\infty}},\frac{x_2x_3}{t_0q^{+\theta_0-\theta_\infty}},\frac{x_2x_4}{t_0q^{+\theta_0+\theta_\infty}}\bigg)\\
    &=\frac{x_2x_4}{t_0q^{+\theta_0+\theta_\infty}}\theta_q\bigg(q^{-2\theta_0},q^{2\theta_\infty},\frac{x_1}{x_4},\frac{x_3}{x_2}\bigg),
\end{align*}
where the last equality follows from standard manipulations of $q$-theta functions. But the expression on the last line cannot equal zero, due to the non-resonance conditions, yielding a contradiction. Similarly, the combination (2) plus (i) cannot occur, as it similarly implies
\begin{equation*}
    0=\frac{x_2x_4}{t_0q^{+\theta_0+\theta_\infty}}\theta_q\bigg(q^{-2\theta_0},q^{2\theta_\infty},\frac{x_2}{x_4},\frac{x_3}{x_4}\bigg).
\end{equation*}
We are left to consider the two combinations (1) plus (i) and (2) plus (ii).

Considering combination (1) plus (i), (1) specifies the values of $\rho_{13}$ and $\rho_{23}$, (ii) specifies the values of $\rho_{14}$ and $\rho_{24}$, and
equation \eqref{eq:rhotrivial} is trivially satisfied. Furthermore, by application of equation \eqref{eq:rhoduality1}, with $(i,j,k,l)=(2,3,4,1)$, we find that $\widetilde{\rho}_{23}=0$. Similarly, application of equation \eqref{eq:rhoduality1}, with $(i,j,k,l)=(2,4,3,1)$, gives $\widetilde{\rho}_{24}=\infty$, so $\widetilde{\rho}_{42}=0$. In order to determine the value of $\widetilde{\rho}_{12}$, we note that the second quadratic equation in \eqref{eq:tyurinquadratics} gives
\begin{align}
0=&T_{12}(\theta_\infty,\theta_0)\widetilde{\rho}_{12}\widetilde{\rho}_{23}+
T_{13}(\theta_\infty,\theta_0)\widetilde{\rho}_{12}+
T_{14}(\theta_\infty,\theta_0)\widetilde{\rho}_{12}\widetilde{\rho}_{23}\widetilde{\rho}_{42}+\label{eq:quadraticrhotilde}\\
&T_{23}(\theta_\infty,\theta_0)+
T_{24}(\theta_\infty,\theta_0)\widetilde{\rho}_{23}\widetilde{\rho}_{42}+
T_{34}(\theta_\infty,\theta_0)\widetilde{\rho}_{42},\nonumber
\end{align}
and, using $\widetilde{\rho}_{23}=\widetilde{\rho}_{42}=0$, we then find
\begin{equation*}
    \widetilde{\rho}_{12}=-\frac{T_{23}(\theta_\infty,\theta_0)}{T_{13}(\theta_\infty,\theta_0)}.
\end{equation*}

Similarly, combination (2) plus (ii) yields $\widetilde{\rho}_{24}=\widetilde{\rho}_{32}=0$ and consequently
\begin{equation*}
    \widetilde{\rho}_{12}=-\frac{T_{24}(\theta_\infty,\theta_0)}{T_{14}(\theta_\infty,\theta_0)}.
\end{equation*}
This proves the statement of Remark \ref{rem:duality} for $\widetilde{\rho}_{12}$. By permuting $\{1,2,3,4\}$, we obtain Remark \ref{rem:duality} in general.

We now return to the issue of analyticity of the dual Tyurin ratio $\widetilde{\rho}_{12}$. As we had already shown, $\widetilde{\rho}_{12}$ is analytic on all of $\widehat{\mathcal{F}}$ where combination (1) plus (i) and combination (2) plus (ii) do not occur.

Now, let us take an $N^*\in \widehat{\mathcal{F}}$ and suppose that the corresponding Tyurin ratios satisfy (1) plus (i). We are going to use equation \eqref{eq:quadraticrhotilde} to show that $\widetilde{\rho}_{12}$ is locally analytic around $N=N^*$. Equation \eqref{eq:quadraticrhotilde} gives  the following expression for $\widetilde{\rho}_{12}$,
\begin{equation}\label{eq:rho24alt}
    \widetilde{\rho}_{12}=-\frac{T_{23}(\theta_\infty,\theta_0)+
T_{24}(\theta_\infty,\theta_0)\widetilde{\rho}_{23}\widetilde{\rho}_{42}+
T_{34}(\theta_\infty,\theta_0)\widetilde{\rho}_{42}}{T_{12}(\theta_\infty,\theta_0)\widetilde{\rho}_{23}+T_{13}(\theta_\infty,\theta_0)+T_{14}(\theta_\infty,\theta_0)\widetilde{\rho}_{23}\widetilde{\rho}_{42}}.
\end{equation}
Furthermore, equation \eqref{eq:rhoduality1}, with
$(i,j,k,l)=(2,3,4,1)$ and $(i,j,k,l)=(2,4,3,1)$ respectively, show that $\widetilde{\rho}_{23}=\widetilde{\rho}_{23}(N)$ and $\widetilde{\rho}_{42}=\widetilde{\rho}_{42}(N)$ are locally analytic around $N=N^*$ with
\begin{equation*}
    \widetilde{\rho}_{23}(N^*)=\widetilde{\rho}_{42}(N^*)=0.
\end{equation*}
It now follows from equation \eqref{eq:rho24alt} that $\widetilde{\rho}_{12}$ is locally analytic around $N=N^*$.

It is similarly shown that $\widetilde{\rho}_{12}$ is analytic near any point $N^*\in \widehat{\mathcal{F}}$ where combination (2) plus (ii) occur. It follows that $\widetilde{\rho}_{12}$ is an analytic function to $\mathbb{CP}^1$ on the whole of $\widehat{\mathcal{F}}$. By permuting indices, Lemma \ref{lemma:global_mero} now follows.
\end{proof}

\subsection{Lines on the Segre surface}\label{sec:lines}
In this section, we give explicit expressions for all the lines on the Segre surface $\widehat{\mathcal{F}}$. Firstly, let us consider $\mathcal{L}_1^0$, defined in equation \eqref{eq:linesdefi}. Using the defining equation for the $\eta$-coordinates \eqref{eq:eta_defi}, it follows directly that this set can equally be described by
\begin{align*}
\mathcal{L}_1^0=&\{N\in\widehat{\mathcal{F}}:N_{12}=N_{13}=N_{14}=0\}\\
=&\{N\in\mathbb{CP}^6:N_{12}=N_{13}=N_{14}=0,N_{23}+N_{24}+N_{34}=0,\\
    &\hspace{2.6cm}a_{23}N_{23}+a_{24}N_{24}+a_{34}N_{34}+a_{\infty}N_\infty=0\}.
\end{align*}
Note, in particular, that this is clearly a line.
Similarly, we have
\begin{align*}
    \mathcal{L}_2^0=\{N\in\mathbb{CP}^6:N_{12}=N_{23}=N_{24}=0,N_{13}+N_{14}+N_{34}=0&,\\
    a_{13}N_{13}+a_{14}N_{14}+a_{34}N_{34}+a_{\infty}N_\infty=0&\},\\
    \mathcal{L}_3^0=\{N\in\mathbb{CP}^6:N_{13}=N_{23}=N_{34}=0,N_{12}+N_{14}+N_{24}=0&,\\
    a_{12}N_{12}+a_{14}N_{14}+a_{24}N_{24}+a_{\infty}N_\infty=0&\},\\
    \mathcal{L}_4^0=\{N\in\mathbb{CP}^6:N_{14}=N_{24}=N_{34}=0,N_{12}+N_{13}+N_{23}=0&,\\
    a_{12}N_{12}+a_{13}N_{13}+a_{23}N_{23}+a_{\infty}N_\infty=0&\},
    \end{align*}
    and
    \begin{align*}
    \mathcal{L}_1^\infty=\{N\in\mathbb{CP}^6:N_{23}=N_{24}=N_{34}=0,N_{12}+N_{13}+N_{14}=0&,\\
    a_{12}N_{12}+a_{13}N_{13}+a_{14}N_{14}+a_{\infty}N_\infty=0&\},\\
    \mathcal{L}_2^\infty=\{N\in\mathbb{CP}^6:N_{13}=N_{14}=N_{34}=0,N_{12}+N_{23}+N_{24}=0&,\\
    a_{12}N_{12}+a_{23}N_{23}+a_{24}N_{24}+a_{\infty}N_\infty=0&\},\\
    \mathcal{L}_3^\infty=\{N\in\mathbb{CP}^6:N_{12}=N_{14}=N_{24}=0,N_{13}+N_{23}+N_{34}=0&,\\
    a_{13}N_{13}+a_{23}N_{23}+a_{34}N_{34}+a_{\infty}N_\infty=0&\},\\
    \mathcal{L}_4^\infty=\{N\in\mathbb{CP}^6:N_{12}=N_{13}=N_{23}=0,N_{14}+N_{24}+N_{34}=0&,\\
    a_{14}N_{14}+a_{24}N_{24}+a_{34}N_{34}+a_{\infty}N_\infty=0&\}.
\end{align*}

We proceed to derive explicit expressions for the remaining eight lines. Let $1\leq k\leq 4$ and consider an $N\in \widetilde{\mathcal{L}}_k^0$. Take any $i,j,l$ such that $\{i,j,k,l\}=\{1,2,3,4\}$, then $\widetilde{\rho}_{ki}=\widetilde{\rho}_{kj}=\infty$ and thus equation \eqref{eq:rhoduality2} gives
\begin{equation*}
\rho_{ij}=\frac{x_j}{x_i}M_{(i,j,k,l)}\left( \infty;\theta_\infty,\theta_0\right)M_{(j,l,k,i)}\left(\infty,\theta_\infty,\theta_0\right).
\end{equation*}
In other words,
\begin{equation*}
    \rho_{ij}=P_{ij}^{(k)},\quad P_{ij}^{(k)}:=\frac{\theta_q\left(\frac{x_i x_k}{t_0}\frac{q^{-\theta_\infty}}{q^{\theta_0}} \right)\theta_q\left(\frac{x_j x_k}{t_0}\frac{q^{\theta_\infty}}{q^{\theta_0}} \right)}{\theta_q\left(\frac{x_j x_k}{t_0}\frac{q^{-\theta_\infty}}{q^{\theta_0}} \right)\theta_q\left(\frac{x_i x_k}{t_0}\frac{q^{\theta_\infty}}{q^{\theta_0}} \right)},
\end{equation*}
for any $i,j,l$ such that $\{i,j,k,l\}=\{1,2,3,4\}$.

Let us consider specifically the case $k=1$. Then, we obtain explicit values for
\begin{equation*}
    \rho_{23}=P_{23}^{(1)},\quad \rho_{24}=P_{24}^{(1)},\quad \rho_{34}=P_{34}^{(1)}.
\end{equation*}
These three equations yield a system of six linear relations among the $N$-variables, given by
\begin{align*}
    T_{13} N_{12}-P_{23}^{(1)}T_{12}N_{13}&=0, & 
    T_{14} N_{12}-P_{24}^{(1)}T_{12}N_{14}&=0, &
    T_{14} N_{13}-P_{34}^{(1)}T_{13}N_{14}&=0,\\
    T_{34} N_{24}-P_{23}^{(1)}T_{24}N_{34}&=0, &
    T_{34} N_{23}-P_{24}^{(1)}T_{23}N_{34}&=0, &
    T_{24} N_{23}-P_{34}^{(1)}T_{23}N_{24}&=0,
\end{align*}
whose rank is $4$. It thus defines a plane in the ambient space $\{N\in\mathbb{CP}^6\}$. Next, we consider the intersection of this plane with $\widehat{\mathcal{F}}(\Theta,t_0)$. It is easy to see that \eqref{eq:eta_equationsa}, \eqref{eq:eta_equationsc} and \eqref{eq:eta_equationsd} are identically satisfied on this plane. Equation \eqref{eq:eta_equationsb}, however,
cannot also be identically satisfied on this plane, else the whole plane would lie in $\widehat{\mathcal{F}}(\Theta,t_0)$, so that $\widehat{\mathcal{F}}(\Theta,t_0)$ would be reducible and thus not smooth, contradiction Proposition \ref{lemma:smooth}.
Therefore, equation \eqref{eq:eta_equationsb}
specifies a unique line in the plane, and thus in $\widehat{\mathcal{F}}(\Theta,t_0)$. 
Conversely, by inverting the computations, it is easy to see that any point on this line lies in $\widetilde{\mathcal{L}}_1^0$. We conclude that
\begin{align*}
    \widetilde{\mathcal{L}}_1^0=\{N\in\mathbb{CP}^6:\,&T_{13}N_{12}-P_{23}^{(1)}T_{12}N_{13}=0,T_{14}N_{13}-P_{34}^{(1)}T_{13}N_{14}=0,\\
    &T_{34} N_{24}-P_{23}^{(1)}T_{24}N_{34}=0,T_{34} N_{23}-P_{24}^{(1)}T_{23}N_{34}=0,\\
    &a_{12}N_{12}+a_{13}N_{13}+a_{14}N_{14}+a_{23}N_{23}+a_{24}N_{24}+a_{34}N_{34}+a_{\infty}N_\infty=0\}.
    \end{align*}

Similarly, we obtain the following explicit descriptions of the remaining lines,
\begin{align*}
    \widetilde{\mathcal{L}}_2^0=\{N\in\mathbb{CP}^6:\,&T_{23}N_{12}-P_{13}^{(2)}T_{12}N_{23}=0,T_{24}N_{23}-P_{34}^{(2)}T_{23}N_{24}=0,\\
    &T_{34} N_{14}-P_{13}^{(2)}T_{14}N_{34}=0,T_{34} N_{13}-P_{14}^{(2)}T_{13}N_{34}=0,\\
    &a_{12}N_{12}+a_{13}N_{13}+a_{14}N_{14}+a_{23}N_{23}+a_{24}N_{24}+a_{34}N_{34}+a_{\infty}N_\infty=0\},\\
    \widetilde{\mathcal{L}}_3^0=\{N\in\mathbb{CP}^6:\,&T_{13}N_{23}-P_{21}^{(3)}T_{23}N_{13}=0,T_{34}N_{13}-P_{14}^{(3)}T_{13}N_{34}=0,\\
    &T_{14} N_{24}-P_{21}^{(3)}T_{24}N_{14}=0,T_{14} N_{12}-P_{24}^{(3)}T_{12}N_{14}=0,\\
    &a_{12}N_{12}+a_{13}N_{13}+a_{14}N_{14}+a_{23}N_{23}+a_{24}N_{24}+a_{34}N_{34}+a_{\infty}N_\infty=0\},\\
    \widetilde{\mathcal{L}}_4^0=\{N\in\mathbb{CP}^6:\,&T_{34}N_{24}-P_{23}^{(4)}T_{24}N_{34}=0,T_{14}N_{34}-P_{31}^{(4)}T_{34}N_{14}=0,\\
    &T_{13} N_{12}-P_{23}^{(4)}T_{12}N_{13}=0,T_{13} N_{23}-P_{21}^{(4)}T_{23}N_{13}=0,\\
    &a_{12}N_{12}+a_{13}N_{13}+a_{14}N_{14}+a_{23}N_{23}+a_{24}N_{24}+a_{34}N_{34}+a_{\infty}N_\infty=0\},
    \end{align*}
and
\begingroup
\allowdisplaybreaks
\begin{align*}
    \widetilde{\mathcal{L}}_1^\infty=\{N\in\mathbb{CP}^6:\,&T_{13}N_{12}-P_{23}^{(4)}T_{12}N_{13}=0,T_{14}N_{13}-P_{34}^{(2)}T_{13}N_{14}=0,\\
    &T_{34} N_{24}-P_{23}^{(4)}T_{24}N_{34}=0,T_{34} N_{23}-P_{24}^{(3)}T_{23}N_{34}=0,\\
    &a_{12}N_{12}+a_{13}N_{13}+a_{14}N_{14}+a_{23}N_{23}+a_{24}N_{24}+a_{34}N_{34}+a_{\infty}N_\infty=0\},\\
    \widetilde{\mathcal{L}}_2^\infty=\{N\in\mathbb{CP}^6:\,&T_{23}N_{12}-P_{13}^{(4)}T_{12}N_{23}=0,T_{24}N_{23}-P_{34}^{(1)}T_{23}N_{24}=0,\\
    &T_{34} N_{14}-P_{13}^{(4)}T_{14}N_{34}=0,T_{34} N_{13}-P_{14}^{(3)}T_{13}N_{34}=0,\\
    &a_{12}N_{12}+a_{13}N_{13}+a_{14}N_{14}+a_{23}N_{23}+a_{24}N_{24}+a_{34}N_{34}+a_{\infty}N_\infty=0\},\\
    \widetilde{\mathcal{L}}_3^\infty=\{N\in\mathbb{CP}^6:\,&T_{13}N_{23}-P_{21}^{(4)}T_{23}N_{13}=0,T_{34}N_{13}-P_{14}^{(2)}T_{13}N_{34}=0,\\
    &T_{14} N_{24}-P_{21}^{(4)}T_{24}N_{14}=0,T_{14} N_{12}-P_{24}^{(1)}T_{12}N_{14}=0,\\
    &a_{12}N_{12}+a_{13}N_{13}+a_{14}N_{14}+a_{23}N_{23}+a_{24}N_{24}+a_{34}N_{34}+a_{\infty}N_\infty=0\},\\
    \widetilde{\mathcal{L}}_4^\infty=\{N\in\mathbb{CP}^6:\,&T_{34}N_{24}-P_{23}^{(1)}T_{24}N_{34}=0,T_{14}N_{34}-P_{31}^{(2)}T_{34}N_{14}=0,\\
    &T_{13} N_{12}-P_{23}^{(1)}T_{12}N_{13}=0,T_{13} N_{23}-P_{21}^{(3)}T_{23}N_{13}=0,\\
    &a_{12}N_{12}+a_{13}N_{13}+a_{14}N_{14}+a_{23}N_{23}+a_{24}N_{24}+a_{34}N_{34}+a_{\infty}N_\infty=0\}.
    \end{align*}
\endgroup

\subsection{Intersection points}\label{sec:intersectionpoints}
In this section we consider intersection points among the different lines on the Segre surface $\widehat{\mathcal{F}}$ and prove Theorem \ref{thm:intersection}.
\begin{proof}[Proof of Theorem \ref{thm:intersection}]
It follows from the explicit formulas for the $16$ lines on $\widehat{\mathcal{F}}$, given in Section \ref{sec:lines}, that none of the lines lie entirely in the hyperplane section at infinity.

We know from the general theory of (smooth) Segre surfaces that any line intersect exactly five others. In Section \ref{sec:lines}, each of the $16$ lines in $\mathbb{CP}^6$ is described by $5$ linearly independent equations. Using these descriptions, the problem of determining whether two lines intersect has been reduced to linear algebra.

Let us consider which lines $\mathcal{L}_1^0$ intersects with. Firstly, note that $\mathcal{L}_1^0$ and $\mathcal{L}_1^\infty$ cannot possible intersect, as, if they would at some point $N\in \mathbb{CP}^6$, then we immediately find
\begin{equation}\label{eq:contradiction}
    N_{12}=N_{13}=N_{14}=N_{23}=N_{24}=N_{34}=N_{\infty}=0.
\end{equation}
Similarly, if $\mathcal{L}_1^0$ and $\mathcal{L}_2^0$ intersect at a point $N\in \mathbb{CP}^6$, then we must have
\begin{equation*}
    N_{12}=N_{13}=N_{14}=N_{23}=N_{24}=0,
\end{equation*}
and, by either employing equations \eqref{eq:eta_equationsaproj} and \eqref{eq:eta_equationsbproj}, or by using the further linear equations describing $\mathcal{L}_1^0$ and $\mathcal{L}_2^0$, we find $N_{34}=N_{\infty}=0$, yielding a  contradiction. Analogously, we find that $\mathcal{L}_1^0$ does not intersect with $\mathcal{L}_3^0$ and $\mathcal{L}_4^0$.

On the other hand, a direct calculation shows that $\mathcal{L}_1^0$ and $\mathcal{L}_2^\infty$ do intersect and the intersection point is given by
\begin{equation*}
  N_{12}=N_{13}=N_{14}=N_{34}=0,\quad N_{23}=1,\quad N_{24}=-1,\quad N_\infty=(a_{24}-a_{23})/a_\infty.
\end{equation*}
Furthermore, this intersection point is finite (i.e. lies in $\mathcal{F}$), if and only if $a_{23}\neq a_{24}$, which holds if and only if
\begin{equation*}
    q^{-2\theta_\infty} \frac{x_1}{x_2}\notin q^{\mathbb{Z}}.
\end{equation*}

In general, for any $1\leq i,j\leq 4$, with $i\neq j$, analogous computations show that
\begin{enumerate}
    \item $\mathcal{L}_i^0$ does not intersect with $\mathcal{L}_i^\infty$,
    \item $\mathcal{L}_i^0$ does not intersect with $\mathcal{L}_j^0$,
    \item $\mathcal{L}_i^\infty$ does not intersect with $\mathcal{L}_j^\infty$,
    \item $\mathcal{L}_i^0$ and $\mathcal{L}_j^\infty$ intersect and the intersection point is finite if and only if $a_{jk}\neq a_{jl}$, where $k,l$ are such that $\{i,j,k,l\}=\{1,2,3,4\}$, which holds if and only if
    \begin{equation*}
    q^{-2\theta_\infty} \frac{x_i}{x_j}\notin q^{\mathbb{Z}}.
\end{equation*}
\end{enumerate}
Then, by duality, i.e. by considering the dual Segre surface $\widehat{F}(\widetilde{\Theta},t_0)$, where $\widetilde{\Theta}=\Theta|_{\theta_0\leftrightarrow \theta_\infty}$, related to the dual Tyurin parameters $\widetilde{\rho}$, so that e.g. the line $\mathcal{L}_3^0(\widetilde{\Theta})$ on $\widehat{F}(\widetilde{\Theta},t_0)$ corresponds to the line $\widetilde{\mathcal{L}}_3^0(\Theta)$ on $\widehat{F}(\Theta,t_0)$ , we immediately find that, for any $1\leq i,j\leq 4$, with $i\neq j$,
\begin{enumerate}
    \item[(1)'] $\widetilde{\mathcal{L}}_i^0$ does not intersect with $\widetilde{\mathcal{L}}_i^\infty$,
    \item[(2)'] $\widetilde{\mathcal{L}}_i^0$ does not intersect with $\widetilde{\mathcal{L}}_j^0$,
    \item[(3)'] $\widetilde{\mathcal{L}}_i^\infty$ does not intersect with $\widetilde{\mathcal{L}}_j^\infty$,
    \item[(4)'] $\widetilde{\mathcal{L}}_i^0$ and $\widetilde{\mathcal{L}}_j^\infty$ intersect and the intersection point is finite if and only if $\widetilde{a}_{jk}\neq \widetilde{a}_{jl}$, where $\widetilde{a}_{jk}=a_{jk}|_{\theta_0\leftrightarrow \theta_\infty}$, $\widetilde{a}_{jl}=a_{jl}|_{\theta_0\leftrightarrow \theta_\infty}$ and
    $k,l$ are such that $\{i,j,k,l\}=\{1,2,3,4\}$, which holds if and only if
    \begin{equation*}
    q^{-2\theta_0} \frac{x_i}{x_j}\notin q^{\mathbb{Z}}.
\end{equation*}
\end{enumerate}
Alternatively, one can also verify these by direct calculation.

Let us return now to the line $\mathcal{L}_1^0$. We already know that it intersects $\mathcal{L}_2^\infty$, $\mathcal{L}_3^\infty$, $\mathcal{L}_4^\infty$, and does not intersect $\mathcal{L}_1^\infty$, $\mathcal{L}_2^0$, $\mathcal{L}_3^0$ and $\mathcal{L}_4^0$. Suppose now that $\mathcal{L}_1^0$ would intersect with $\widetilde{\mathcal{L}}_2^0$ at some point $N\in \mathbb{CP}^6$. As $N\in \mathcal{L}_1^0$, we know that $N_{12}=N_{13}=N_{14}$ and by simply setting these three variables equal to zero in the linear system describing $\widetilde{\mathcal{L}}_2^0$, we find $N_{23}=N_{24}=N_{34}=N_{\infty}=0$ and we have arrived at a contradiction. By similar simple arguments, we find that $\mathcal{L}_1^0$ does not intersect with $\widetilde{\mathcal{L}}_2^0$, $\widetilde{\mathcal{L}}_3^0$, $\widetilde{\mathcal{L}}_4^0$, nor $\widetilde{\mathcal{L}}_2^\infty$, $\widetilde{\mathcal{L}}_3^\infty$, $\widetilde{\mathcal{L}}_4^\infty$. More generally, for any $1\leq i,j\leq 4$, $i\neq j$,
\begin{enumerate}
    \item[(a)] $\mathcal{L}_i^0$ and $\widetilde{\mathcal{L}}_j^0$ do not intersect,
    \item[(b)] $\mathcal{L}_i^0$ and $\widetilde{\mathcal{L}}_j^\infty$ do not intersect,
    \item[(c)] $\mathcal{L}_i^\infty$ and $\widetilde{\mathcal{L}}_j^0$ do not intersect,
    \item[(d)] $\mathcal{L}_i^\infty$ and $\widetilde{\mathcal{L}}_j^\infty$ do not intersect.
\end{enumerate}

We return again to $\mathcal{L}_1^0$. We know that $\mathcal{L}_1^0$ must intersect with $5$ lines in $\widetilde{\mathcal{F}}$ and it thus follows that $\mathcal{L}_1^0$ must intersect with $\widetilde{\mathcal{L}}_1^0$ and $\widetilde{\mathcal{L}}_1^\infty$. We proceed to show this explicitly. Firstly, by combining the equations defining the lines $\mathcal{L}_1^0$ and $\widetilde{\mathcal{L}}_1^0$, we find
\begin{align*}
 \mathcal{L}_1^0\cap\widetilde{\mathcal{L}}_1^0=\{N\in\mathbb{CP}^6:\,&N_{12}=N_{13}=N_{14}=0,N_{23}+N_{24}+N_{34}=0,\\ 
    &T_{34} N_{24}-P_{23}^{(1)}T_{24}N_{34}=0,T_{34} N_{23}-P_{24}^{(1)}T_{23}N_{34}=0,\\
&a_{23}N_{23}+a_{24}N_{24}+a_{34}N_{34}+a_{\infty}N_\infty=0\}.
    \end{align*}
This intersection is non-empty if and only if
\begin{equation*}
    \begin{vmatrix}
    1 & 1 & 1 & 0\\
    T_{34} & 0 & -P_{24}^{(1)}T_{23} &0\\
    0 & T_{34} & -P_{23}^{(1)}T_{24} &0\\
    a_{23} & a_{24} & a_{34} & a_\infty
    \end{vmatrix}= 0,
\end{equation*}
which holds if and only if
\begin{equation*}
\begin{vmatrix}
    1 & 1 & 1 \\
    T_{34} & 0 & -P_{24}^{(1)}T_{23} \\
    0 & T_{34} & -P_{23}^{(1)}T_{24} \\
    \end{vmatrix}= 0.
\end{equation*}
The determinant on the left-hand side is given by
\begin{equation*}
    T_{34}\;\theta_q\left(\frac{x_1x_2}{t_0q^{+\theta_0+\theta_\infty}}\right)\theta_q\left(\frac{x_1x_2}{t_0q^{+\theta_0-\theta_\infty}}\right)^{-1}\Delta,
\end{equation*}
where
\begin{align*}
    \Delta=&x_1\;\theta_q\left(\frac{x_3}{x_1},\frac{x_2}{x_4},\frac{x_1x_3}{t_0q^{+\theta_0-\theta_\infty}},\frac{x_2x_4}{t_0q^{+\theta_0-\theta_\infty}}\right)+\\
     &x_2\;\theta_q\left(\frac{x_1}{x_2},\frac{x_3}{x_4},\frac{x_1x_2}{t_0q^{+\theta_0-\theta_\infty}},\frac{x_3x_4}{t_0q^{+\theta_0-\theta_\infty}}\right)+\\
    &x_3\;\theta_q\left(\frac{x_2}{x_3},\frac{x_1}{x_4},\frac{x_2x_3}{t_0q^{+\theta_0-\theta_\infty}},\frac{x_1x_4}{t_0q^{+\theta_0-\theta_\infty}}\right),
\end{align*}
and we are left to show that $\Delta=0$. An easy way to prove this, is by considering $\Delta=\Delta(x_1)$ as a function in $x_1$, and forgetting about the constraint $x_1x_2x_3x_4=t_0^2$. Then $\Delta(x_1)$ is an analytic in $x_1\in\mathbb{C}^*$ and satisfies
\begin{equation*}
    \Delta(qx_1)=q^{\theta_0-\theta_\infty}t_0 x_1^{-2}\Delta(x_1).
\end{equation*}
So it is either a degree two theta function or identically equal to zero. In particular, either $\Delta(x_1)$ has precisely two zeros (counting multiplicity) in $\mathbb{C}^*$ modulo multiplication by $q$, or it is identically equal to zero. By direct calculation, setting $x_1=x_2$ yields
\begin{equation*}
\Delta(x_2)=\theta_q\left(\frac{x_2}{x_4},\frac{x_2x_3}{t_0q^{+\theta_0-\theta_\infty}},\frac{x_2x_4}{t_0q^{+\theta_0-\theta_\infty}}\right)\left(x_2\theta_q\left(\frac{x_3}{x_2}\right)+x_3\theta_q\left(\frac{x_2}{x_3}\right)\right)=0,
\end{equation*}
where the second equality follows from the identity $\theta_q(1/z)=-z^{-1}\theta_q(z)$. Similarly, we find $\Delta(x_3)=\Delta(x_4)=0$, and thus $\Delta(x_1)$ has at least three inequivalent zeros modulo multiplication by $q$, $x_1=x_2,x_3,x_4$. It follows that $\Delta(x_1)\equiv 0$. Thus $\mathcal{L}_1^0$ and $\widetilde{\mathcal{L}}_1^0$ indeed intersect at a (unique) point $N^*\in\mathbb{CP}^6$.

This point is infinite if and only if $N_\infty^*= 0$, which is equivalent to
\begin{align*}
0=&a_{23}N_{23}^*+a_{24}N_{24}^*+a_{34}N_{34}^*\\
=&a_{23}\frac{T_{23}}{T_{34}}P_{24}^{(1)}N_{34}^*+a_{24}\frac{T_{24}}{T_{34}}P_{23}^{(1)}N_{34}^*+a_{34}N_{34}^*\\
=&\frac{N_{34}^*}{T_{34}}\left(a_{23}T_{23}P_{24}^{(1)}+a_{24}T_{24}P_{23}^{(1)}+a_{34}T_{34}\right).
\end{align*}
A lengthy calculation using standard manipulations of $q$-theta functions, shows that the factor in brackets in the last line is equal to
\begin{gather*}
    a_{23}T_{23}P_{24}^{(1)}+a_{24}T_{24}P_{23}^{(1)}+a_{34}T_{34}=\\
q^{-2\theta_0-\theta_\infty}t_0\frac{\theta_q(q^{\theta_0})^2\theta_q(q^{2\theta_\infty})\theta_q\left(\frac{x_2}{x_3},\frac{x_4}{x_2},\frac{x_3}{x_4}\right)}{\theta_q\left(\frac{x_1x_2}{t_0q^{+\theta_0-\theta_\infty}},\frac{x_1x_3}{t_0q^{+\theta_0+\theta_\infty}},\frac{x_1x_4}{t_0q^{+\theta_0+\theta_\infty}}\right)}\theta_q\left(\frac{x_1^2}{t_0q^{+\theta_0+\theta_\infty}}\right).
\end{gather*}
Under the parameter conditions \eqref{eq:parameterassumptionsx}, only the last factor in the above expression can possible vanish. Thus, the intersection point of $\mathcal{L}_1^0$ and $\widetilde{\mathcal{L}}_1^0$ is finite if and only if
\begin{equation*}
    q^{+\theta_0+\theta_\infty} \frac{t_0}{x_1^2}\notin q^{\mathbb{Z}}.
\end{equation*}

A similar computation shows that $\mathcal{L}_1^0$ and $\widetilde{\mathcal{L}}_1^\infty$ intersect at a unique point, and this point is finite if and only if 
\begin{equation*}
    q^{-\theta_0+\theta_\infty} \frac{t_0}{x_1^2}\notin q^{\mathbb{Z}}.
\end{equation*}
By similar computations and permutations of the indices, we find the following statements. For any $1\leq i\leq 4$,
\begin{enumerate}
    \item[(i)] $\mathcal{L}_i^0$ and $\widetilde{\mathcal{L}}_i^0$ intersect and the intersection point is finite if and only if
    \begin{equation*}
    q^{+\theta_0+\theta_\infty} \frac{t_0}{x_i^2}\notin q^{\mathbb{Z}}.
\end{equation*}
    \item[(ii)] $\mathcal{L}_i^0$ and $\widetilde{\mathcal{L}}_i^\infty$ intersect and the intersection point is finite if and only if
    \begin{equation*}
    q^{-\theta_0+\theta_\infty} \frac{t_0}{x_i^2}\notin q^{\mathbb{Z}}.
\end{equation*}
    \item[(iii)] $\mathcal{L}_i^\infty$ and $\widetilde{\mathcal{L}}_i^0$ intersect and the intersection point is finite if and only if
    \begin{equation*}
    q^{+\theta_0-\theta_\infty} \frac{t_0}{x_i^2}\notin q^{\mathbb{Z}}.
\end{equation*}
    \item[(iv)] $\mathcal{L}_i^\infty$ and $\widetilde{\mathcal{L}}_i^\infty$ intersect and the intersection point is finite if and only if
    \begin{equation*}
    q^{-\theta_0-\theta_\infty} \frac{t_0}{x_i^2}\notin q^{\mathbb{Z}}.
\end{equation*}
\end{enumerate}
The intersection diagram in Figure \ref{fig:lines_intersection} now follows from enumerations (1)-(4),(1)'-(4)',(a)-(d) and (i)-(iv) above. Furthermore, we know that all intersection points are finite if and only if, for all $1\leq i,j\leq 4$, $i\neq j$ and all $\epsilon_0,\epsilon_\infty\in \{\pm 1\}$,
\begin{equation*}
       q^{-2\theta_\infty} \frac{x_i}{x_j}\notin q^{\mathbb{Z}},\quad
        q^{-2\theta_0} \frac{x_i}{x_j}\notin q^{\mathbb{Z}},\quad
        q^{\epsilon_0\theta_0+\epsilon\theta_\infty} \frac{t_0}{x_i^2}\notin q^{\mathbb{Z}}.
\end{equation*}
The theorem now follows through the identification of parameters in \eqref{eq:identification}.
\end{proof}

We finish this section with a heuristic dimension count. Consider pairs $(S,H)$, where $S$ a smooth (complex) Segre surface, embedded in $\mathbb{CP}^4$, and $H$ a corresponding hyperplane, so that $S\setminus H$ forms an affine Segre surface. The natural automorphism group $PGL_5(\mathbb{C})$ of the ambient space $\mathbb{CP}^4$ acts on such pairs in the usual way. The space of smooth Segre surfaces modulo $PGL_5(\mathbb{C})$ is two dimensional, as can for example be derived by putting them into diagonal form, see Mabuchi and Mukai \cite{mabuchi}. A heuristic dimension count then shows that the space of affine Segre surfaces, moduli the action by $PGL_5(\mathbb{C})$, has complex dimension six. This fits nicely with the fact that the affine Segre surface $\mathcal{F}$ involves six parameters, $\theta_0,\theta_t,\theta_1,\theta_\infty, t_0$ and $q$. It would interesting to see whether inequivalent parameter values (under the extended affine Weyl group symmetry of $q\Psix$) yield inequivalent affine Segre surfaces, and whether the space of affine Segre surfaces $\{\mathcal{F}\}$, defined by varying all the parameters such that conditions \eqref{eq:param_assumptions_1} and \eqref{eq:param_assumptions_2} are satisfied, is open and/or dense.

%% file: pairs_pants.tex
\section{Mano decompositions}\label{sec:qpairspants}
Decompositions of monodromy of isomonodromic systems along pairs of pants decompositions of underlying Riemann surfaces are very useful in the exact and asymptotic study of their solutions. For example, they were used by Gavrylenko and Lisovyy \cite{gavlis2018} to derive an explicit Fredholm determinant representation for the general solution of $\Psix$. In that context, one considers pairs of pants decompositions of $\mathbb{CP}^1\setminus\{0,t,1,\infty\}$ and correspondingly decomposes global monodromy of a $4$-point Fuchsian system into the monodromies of two local hypergeometric systems. By correspondingly decomposing the Riemann-Hilbert problem (RHP) and constructing local parametrices out of solutions of the two hypergeometric systems, one can recast the original RHP into one with a single jump on a closed contour. The latter RHP is very suitable to then extract asymptotics of the underlying solution of $\Psix$ \cites{itslisovyy,gavlis2018}.
We will apply a similar method to $q\Psix$, close to the approach by Its et al. \cite{itslisovyy} in methodology.

Evidence that a $q$-analog of decomposition of monodromy, along pairs-of-pants decompositions,
 might also be feasible for $q$-difference linear systems first appeared in Mano \cite{manoqpvi}, where it was shown that, generically, the global connection matrix of the Jimbo-Sakai linear problem factorises into the connection matrices of two local $q$-hypergeometric systems. A similar factorisation was also obtained by the author \cite{phdroffelsen}, for the linear problem related to $q\text{P}(A_1)$, devised by Yamada \cite{yamadalax}.

Ohyama et al. \cite{ohyamaramissualoy} established the general factorisation of the connection matrix of the Jimbo-Sakai linear problem into the connection matrices of local hypergeometric systems, which they call Mano decompositions.

Heuristically, it is tempting to think of such factorisations along decompositions of the complex cylinder with four marked points,
\begin{equation*}
    \left(\mathbb{CP}^1\setminus\{0,\infty\},\{q^{\theta_t}t,q^{-\theta_t}t,q^{\theta_1},q^{-\theta_1}\}\right),
\end{equation*}
into two copies of the complex cylinder, each with two marked points, for example
\begin{equation*}
    \left(\mathbb{CP}^1\setminus\{0',\infty\},\{q^{\theta_1},q^{-\theta_1}\}\right)\quad \text{and}\quad 
    \left(\mathbb{CP}^1\setminus\{0,\infty'\},\{q^{\theta_t}t,q^{-\theta_t}t\}\right),
\end{equation*}
or
\begin{equation*}
    \left(\mathbb{CP}^1\setminus\{0'',\infty\}, \{q^{\theta_t}t,q^{-\theta_t}t\}\right)\quad \text{and}\quad 
    \left(\mathbb{CP}^1\setminus \{0,\infty''\},\{q^{\theta_1},q^{-\theta_1}\}\right),
\end{equation*}
see Figure \ref{fig:pairspants}.
 There are ${{4}\choose{2}}=6$ inequivalent such decompositions, but we will only need two of them, decompositions I and II, as pictorially displayed in Figure \ref{fig:pairspants}.

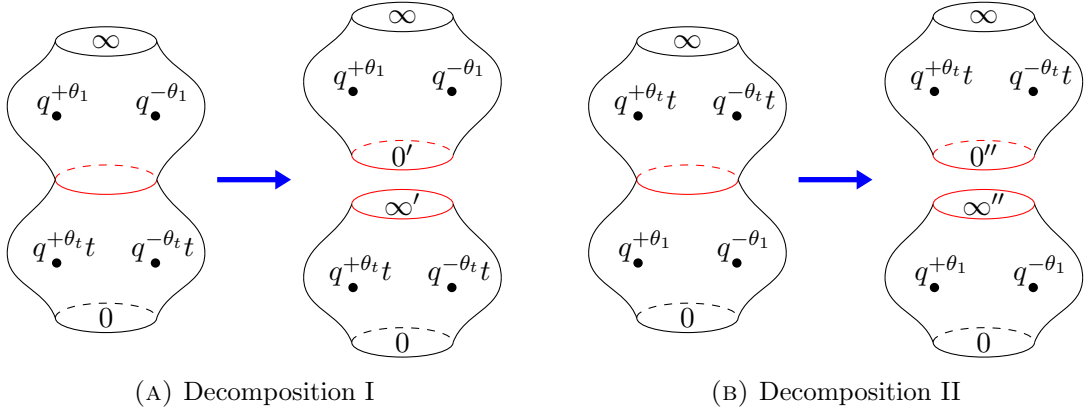
\begin{figure}[t]
     \centering
     \begin{minipage}{15cm}
     \begin{subfigure}[b]{0.49\textwidth}
         \centering
        \begin{tikzpicture}[scale=0.65]
	\tikzstyle{star}  = [circle, minimum width=3.5pt, fill, inner sep=0pt];
	\tikzstyle{starsmall}  = [circle, minimum width=3.5pt, fill, inner sep=0pt];

\draw (0,{3*cos(20)}) ellipse ({3*sin(20)} and 0.3);
%\draw (0,{-3*cos(20)}) ellipse ({3*sin(20)} and 0.3);

\draw ({-3*sin(20)},{-3*cos(20)}) arc (180:360:{3*sin(20)} and 0.3);
\draw [dashed] ({-3*sin(20)},{-3*cos(20)}) arc (180:360:{3*sin(20)} and -0.3);

\draw [red] ({-3*sin(20)},0) arc (180:360:{3*sin(20)} and 0.3);
\draw [dashed,red] ({-3*sin(20)},0) arc (180:360:{3*sin(20)} and -0.3);

\draw ({3*sin(20)},{3*cos(20)}) to[out=-70,in=90] (2,1.5);
\draw ({3*sin(20)},{-3*cos(20)}) to[out=70,in=-90] (2,-1.5);

\draw ({3*sin(20)},{0}) to[out=70,in=-90] (2,1.5);
\draw ({3*sin(20)},{0}) to[out=-70,in=90] (2,-1.5);
 
\draw ({-3*sin(20)},{3*cos(20)}) to[out=-110,in=90] (-2,1.5);
\draw ({-3*sin(20)},{-3*cos(20)}) to[out=110,in=-90] (-2,-1.5);

\draw ({-3*sin(20)},{0}) to[out=110,in=-90] (-2,1.5);
\draw ({-3*sin(20)},{0}) to[out=-110,in=90] (-2,-1.5);

\node at (0,{3*cos(20)}) {$\infty$};
\node at (0,{-3*cos(20)}) {$0$};

\node[starsmall] (x1) at (-1,-1.7) {};
\node at ($(x1)+(0.15,0.35)$) {$q^{+\theta_t}t$};

\node[starsmall] (x2) at (1,-1.7) {};
\node at ($(x2)+(0.15,0.35)$) {$q^{-\theta_t}t$};

\node[starsmall] (x3) at (-1,1.3) {};
\node at ($(x3)+(0.15,0.35)$) {$q^{+\theta_1}$};

\node[starsmall] (x4) at (1,1.3) {};
\node at ($(x4)+(0.15,0.35)$) {$q^{-\theta_1}$};

\draw (6,{0.5+3*cos(20)}) ellipse ({3*sin(20)} and 0.3);
%\draw (0,{-3*cos(20)}) ellipse ({3*sin(20)} and 0.3);

\draw ({6-3*sin(20)},{-0.5-3*cos(20)}) arc (180:360:{3*sin(20)} and 0.3);
\draw [dashed] ({6-3*sin(20)},{-0.5-3*cos(20)}) arc (180:360:{3*sin(20)} and -0.3);

\draw [red] ({6-3*sin(20)},0.5) arc (180:360:{3*sin(20)} and 0.3);
\draw [dashed,red] ({6-3*sin(20)},0.5) arc (180:360:{3*sin(20)} and -0.3);

\draw [red] ({6-3*sin(20)},-0.5) arc (180:360:{3*sin(20)} and 0.3);
\draw [red] ({6-3*sin(20)},-0.5) arc (180:360:{3*sin(20)} and -0.3);

\draw ({6+3*sin(20)},{3*cos(20)+0.5}) to[out=-70,in=90] (6+2,1.5+0.5);
\draw ({6+3*sin(20)},{-3*cos(20)-0.5}) to[out=70,in=-90] (6+2,-1.5-0.5);

\draw ({6+3*sin(20)},{0.5}) to[out=70,in=-90] (6+2,1.5+0.5);
\draw ({6+3*sin(20)},{-0.5}) to[out=-70,in=90] (6+2,-1.5-0.5);
 
\draw ({6-3*sin(20)},{3*cos(20)+0.5}) to[out=-110,in=90] (6-2,1.5+0.5);
\draw ({6-3*sin(20)},{-3*cos(20)-0.5}) to[out=110,in=-90] (6-2,-1.5-0.5);

\draw ({6-3*sin(20)},{0.5}) to[out=110,in=-90] (6-2,1.5+0.5);
\draw ({6-3*sin(20)},{-0.5}) to[out=-110,in=90] (6-2,-1.5-0.5);

\draw[-{Triangle[blue,width=7pt,length=6pt]},blue, line width=2pt](2.25,0) -- (3.75, 0);

\node at (6,{3*cos(20)+0.5}) {$\infty$};
\node at (6,{-3*cos(20)-0.5}) {$0$};

\node at (6,{0.5}) {$0'$};
\node at (6,{-0.5}) {$\infty'$};

\node[starsmall] (x1sh) at (6-1,-1.7-0.5) {};
\node at ($(x1sh)+(0.15,0.35)$) {$q^{+\theta_t}t$};

\node[starsmall] (x2sh) at (6+1,-1.7-0.5) {};
\node at ($(x2sh)+(0.15,0.35)$) {$q^{-\theta_t}t$};

\node[starsmall] (x3sh) at (6-1,1.3+0.5) {};
\node at ($(x3sh)+(0.15,0.35)$) {$q^{+\theta_1}$};

\node[starsmall] (x4sh) at (6+1,1.3+0.5) {};
\node at ($(x4sh)+(0.15,0.35)$) {$q^{-\theta_1}$};

\end{tikzpicture}
         \caption{Decomposition I}
         \label{fig:pairspantsI}
     \end{subfigure}
\hfill
     \begin{subfigure}[b]{0.49\textwidth}
         \centering
         
         \begin{tikzpicture}[scale=0.65]
	\tikzstyle{star}  = [circle, minimum width=3.5pt, fill, inner sep=0pt];
	\tikzstyle{starsmall}  = [circle, minimum width=3.5pt, fill, inner sep=0pt];

\draw (0,{3*cos(20)}) ellipse ({3*sin(20)} and 0.3);
%\draw (0,{-3*cos(20)}) ellipse ({3*sin(20)} and 0.3);

\draw ({-3*sin(20)},{-3*cos(20)}) arc (180:360:{3*sin(20)} and 0.3);
\draw [dashed] ({-3*sin(20)},{-3*cos(20)}) arc (180:360:{3*sin(20)} and -0.3);

\draw [red] ({-3*sin(20)},0) arc (180:360:{3*sin(20)} and 0.3);
\draw [dashed,red] ({-3*sin(20)},0) arc (180:360:{3*sin(20)} and -0.3);

\draw ({3*sin(20)},{3*cos(20)}) to[out=-70,in=90] (2,1.5);
\draw ({3*sin(20)},{-3*cos(20)}) to[out=70,in=-90] (2,-1.5);

\draw ({3*sin(20)},{0}) to[out=70,in=-90] (2,1.5);
\draw ({3*sin(20)},{0}) to[out=-70,in=90] (2,-1.5);
 
\draw ({-3*sin(20)},{3*cos(20)}) to[out=-110,in=90] (-2,1.5);
\draw ({-3*sin(20)},{-3*cos(20)}) to[out=110,in=-90] (-2,-1.5);

\draw ({-3*sin(20)},{0}) to[out=110,in=-90] (-2,1.5);
\draw ({-3*sin(20)},{0}) to[out=-110,in=90] (-2,-1.5);

\node at (0,{3*cos(20)}) {$\infty$};
\node at (0,{-3*cos(20)}) {$0$};

\node[starsmall] (x1) at (-1,1.3) {};
\node at ($(x1)+(0.15,0.35)$) {$q^{+\theta_t}t$};

\node[starsmall] (x2) at (1,1.3) {};
\node at ($(x2)+(0.15,0.35)$) {$q^{-\theta_t}t$};

\node[starsmall] (x3) at (-1,-1.7) {};
\node at ($(x3)+(0.15,0.35)$) {$q^{+\theta_1}$};

\node[starsmall] (x4) at (1,-1.7) {};
\node at ($(x4)+(0.15,0.35)$) {$q^{-\theta_1}$};

\draw (6,{0.5+3*cos(20)}) ellipse ({3*sin(20)} and 0.3);
%\draw (0,{-3*cos(20)}) ellipse ({3*sin(20)} and 0.3);

\draw ({6-3*sin(20)},{-0.5-3*cos(20)}) arc (180:360:{3*sin(20)} and 0.3);
\draw [dashed] ({6-3*sin(20)},{-0.5-3*cos(20)}) arc (180:360:{3*sin(20)} and -0.3);

\draw [red] ({6-3*sin(20)},0.5) arc (180:360:{3*sin(20)} and 0.3);
\draw [dashed,red] ({6-3*sin(20)},0.5) arc (180:360:{3*sin(20)} and -0.3);

\draw [red] ({6-3*sin(20)},-0.5) arc (180:360:{3*sin(20)} and 0.3);
\draw [red] ({6-3*sin(20)},-0.5) arc (180:360:{3*sin(20)} and -0.3);

\draw ({6+3*sin(20)},{3*cos(20)+0.5}) to[out=-70,in=90] (6+2,1.5+0.5);
\draw ({6+3*sin(20)},{-3*cos(20)-0.5}) to[out=70,in=-90] (6+2,-1.5-0.5);

\draw ({6+3*sin(20)},{0.5}) to[out=70,in=-90] (6+2,1.5+0.5);
\draw ({6+3*sin(20)},{-0.5}) to[out=-70,in=90] (6+2,-1.5-0.5);
 
\draw ({6-3*sin(20)},{3*cos(20)+0.5}) to[out=-110,in=90] (6-2,1.5+0.5);
\draw ({6-3*sin(20)},{-3*cos(20)-0.5}) to[out=110,in=-90] (6-2,-1.5-0.5);

\draw ({6-3*sin(20)},{0.5}) to[out=110,in=-90] (6-2,1.5+0.5);
\draw ({6-3*sin(20)},{-0.5}) to[out=-110,in=90] (6-2,-1.5-0.5);

\draw[-{Triangle[blue,width=7pt,length=6pt]},blue, line width=2pt](2.25,0) -- (3.75, 0);

\node at (6,{3*cos(20)+0.5}) {$\infty$};
\node at (6,{-3*cos(20)-0.5}) {$0$};

\node at (6,{0.5}) {$0''$};
\node at (6,{-0.5}) {$\infty''$};

\node[starsmall] (x1sh) at (6-1,+1.3+0.5) {};
\node at ($(x1sh)+(0.15,0.35)$) {$q^{+\theta_t}t$};

\node[starsmall] (x2sh) at (6+1,+1.3+0.5) {};
\node at ($(x2sh)+(0.15,0.35)$) {$q^{-\theta_t}t$};

\node[starsmall] (x3sh) at (6-1,-1.7-0.5) {};
\node at ($(x3sh)+(0.15,0.35)$) {$q^{+\theta_1}$};

\node[starsmall] (x4sh) at (6+1,-1.7-0.5) {};
\node at ($(x4sh)+(0.15,0.35)$) {$q^{-\theta_1}$};

\end{tikzpicture}

         \caption{Decomposition II}
         \label{fig:pairspantsII}
     \end{subfigure}
\end{minipage}

        \caption{Graphical illustration of two decompositions relevant for the asymptotic analysis of $q\Psix$.}
        \label{fig:pairspants}
\end{figure}

In Section \ref{sec:rhp} we recall the definition of the general Riemann-Hilbert problem of $q\Psix$. In Section \ref{sec:hyp_model}, we discuss the Heine hypergeometric system, which forms the local model for the different factors in the Mano decompositions.

Then, in Section \ref{sec:dec_algebraic} we discuss the Mano decomposition I of the monodromy, from an algebraic point of view, following Ohyama et al. \cite{ohyamaramissualoy}. This is followed by Section \ref{sec:dec_analytic}, in which we describe the associated factorisation of the Riemann-Hilbert problem. Then, in Sections \ref{sec:dec_algebraicII} and \ref{sec:dec_analyticII}, we discuss the analogous algebraic and analytic realisations of Mano decomposition II.

\subsection{The $q\Psix$ Riemann-Hilbert Problem}
\label{sec:rhp}

Take a point $\eta$ on the affine Segre surface $\mathcal{F}(\Theta,t_0)$ and let $C(z)\in\mathfrak{C}(\Theta,t_0)$ be a corresponding connection matrix. For $m\in\mathbb{Z}$, let $\gamma^{\B{m}}$ be a Jordan curve in $\mathbb{C}$ such that, denoting by $D_-^{\B{m}}$ and $D_+^{\B{m}}$ its inside and outside within $\mathbb{C}$ respectively, we have $0\in D_-^{\B{m}}$ and
\begin{align*}
q^{\mathbb{Z}_{>0}}\cdot \{q^{\theta_t} t_m,q^{-\theta_t} t_m,q^{\theta_1},q^{-\theta_1}\}&\subseteq D_-^{\B{m}},\\
q^{\mathbb{Z}_{\leq 0}}\cdot \{q^{\theta_t} t_m,q^{-\theta_t} t_m,q^{\theta_1},q^{-\theta_1}\}&\subseteq D_+^{\B{m}}, 
\end{align*}
where we use the notation $U\cdot V=\{uv:u\in U,v\in V\}$ for compatible sets $U$ and $V$, and
    \begin{equation*}
     D_-^{\B{m+1}}\subseteq D_-^{\B{m}},
    \end{equation*}
 see Figure \ref{fig:contours}.

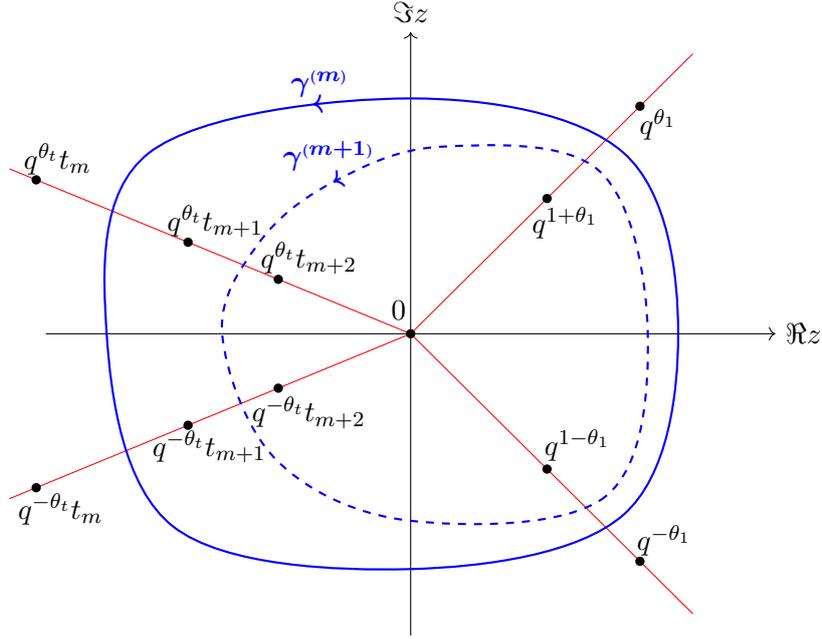
\begin{figure}[ht]
	\centering
	\begin{tikzpicture}[scale=0.8]
	\draw[->] (-6,0)--(6,0) node[right]{$\Re{z}$};
	\draw[->] (0,-5)--(0,5) node[above]{$\Im{z}$};
	\tikzstyle{star}  = [circle, minimum width=3.5pt, fill, inner sep=0pt];
	\tikzstyle{starsmall}  = [circle, minimum width=3.5pt, fill, inner sep=0pt];

	\draw[domain=-1.1:6,smooth,variable=\x,red] plot ({exp(-\x*ln(2))*2*5/3*cos((-3*pi/8+5*pi/4+\x*0) r)},{-exp(-\x*ln(2))*2*5/3*sin((-3*pi/8+5*pi/4+\x*0) r)});	
    \draw[domain=-1.1:6,smooth,variable=\x,red] plot ({exp(-\x*ln(2))*2*5/3*cos((-3*pi/8+5*pi/4+\x*0) r)},{exp(-\x*ln(2))*2*5/3*sin((-3*pi/8+5*pi/4+\x*0) r)});

	\draw[domain=-1.3:6,smooth,variable=\x,red] plot ({exp(-\x*ln(2))*2*4/3*cos((pi/4) r)},{exp(-\x*ln(2))*2*4/3*sin((pi/4) r)});	
	\draw[domain=-1.3:6,smooth,variable=\x,red] plot ({exp(-\x*ln(2))*2*4/3*cos((pi/4) r)},{-exp(-\x*ln(2))*2*4/3*sin((pi/4) r)});
	
	\draw [blue,thick,decoration={markings, mark=at position 0.6 with {\arrow{<}}},
	postaction={decorate}] plot [smooth cycle,tension=0.6] coordinates {(4.4,0) (3.5,-3) (0,-3.9) (-3.9,-3.2) (-5,0) (-4.3,3.0) (0,3.9) (3.5,3) };
	
	\draw [blue,dashed,thick,decoration={markings, mark=at position 0.53 with {\arrow{<}}},
	postaction={decorate}] plot [smooth cycle,tension=0.6] coordinates {(3.9,0) (3.2,-2.7) (0,-3.1) (-2.2,-2.1) (-3.1,0) (-2.5,1.5) (-1.3,2.5) (0.5,3.1) (3.1,2.7) };
	
	\node[starsmall]     (or) at ({0},{0} ) {};
	\node     at ($(or)+(-0.2,0.4)$) {$0$};
	
	\node[starsmall]     (qk1) at ({sqrt(sqrt(2))*2*4/3*cos((-pi/4) r)},{-sqrt(sqrt(2))*2*4/3*sin((-pi/4) r)} ) {};
	\node     at ($(qk1)+(0.3,-0.4)$) {$q^{1+\theta_1}$};
	\node[star]     (k1) at ({2*2*4/3*cos((-pi/4) r)},{-2*2*4/3*sin((-pi/4) r)} ) {};
	\node     at ($(k1)+(0.3,-0.3)$) {$q^{\theta_1}$};
	
	\node[starsmall]     (qkm1) at ({sqrt(sqrt(2))*2*4/3*cos((-pi/4) r)},{sqrt(sqrt(2))*2*4/3*sin((-pi/4) r)} ) {};
	\node     at ($(qkm1)+(0.5,0.4)$) {$q^{1-\theta_1}$};
	\node[star]     (km1) at ({2*2*4/3*cos((-pi/4) r)},{2*2*4/3*sin((-pi/4) r)} ) {};
	\node     at ($(km1)+(0.4,0.35)$) {$q^{-\theta_1}$};
	
	\node[star]     (qqkt) at ({sqrt(2)*5/3*cos((-3*pi/8+5*pi/4+2/4*0) r)},{sqrt(2)*5/3*sin((-3*pi/8+5*pi/4+2/4*0) r)} ) {};
	\node     at ($(qqkt)+(0.5,0.35)$) {$q^{\theta_t}t_{m+2}$};
	\node[star]     (qkt) at ({sqrt(sqrt(2))*2*5/3*cos((-3*pi/8+5*pi/4-1/4*0) r)},{sqrt(sqrt(2))*2*5/3*sin((-3*pi/8+5*pi/4-1/4*0) r)} ) {};
	\node     at ($(qkt)+(0.45,0.35)$) {$q^{\theta_t}t_{m+1}$};	
	\node[star]     (kt) at ({2*2*5/3*cos((-3*pi/8+5*pi/4-0) r)},{2*2*5/3*sin((-3*pi/8+5*pi/4-0) r)} ) {};
	\node     at ($(kt)+(0.35,0.35)$) {$q^{\theta_t}t_{m}$};

	\node[star]     (qqktm) at ({sqrt(2)*5/3*cos((-3*pi/8+5*pi/4) r)},{-sqrt(2)*5/3*sin((-3*pi/8+5*pi/4) r)} ) {};
	\node     at ($(qqktm)+(+0.5,-0.4)$) {$q^{-\theta_t}t_{m+2}$};
	\node[star]     (qktm) at ({sqrt(sqrt(2))*2*5/3*cos((-3*pi/8+5*pi/4) r)},{-sqrt(sqrt(2))*2*5/3*sin((-3*pi/8+5*pi/4) r)} ) {};
	\node     at ($(qktm)+(0.35,-0.35)$) {$q^{-\theta_t}t_{m+1}$};
	\node[star]     (ktm) at ({2*2*5/3*cos((-3*pi/8+5*pi/4) r)},{-2*2*5/3*sin((-3*pi/8+5*pi/4) r)} ) {};
	\node     at ($(ktm)+(0.4,-0.35)$) {$q^{-\theta_t}t_{m}$};

	\node[blue]     at (-1.5,4.15) {$\boldsymbol{\gamma^{\B{m}}}$};
	
	\node[blue]     at (-1.35,2.95) {$\boldsymbol{\gamma^{\B{m+1}}}$};

%	\node[blue]     at (1.90,5.08) {$\boldsymbol{D_+}$};
	
%	\node[blue]     at (1.28,2.75) {$\boldsymbol{D_-}$};
	
	\end{tikzpicture}
	\caption{An example of contours $\gamma^{\B{m}}$ and $\gamma^{\B{m+1}}$ , where $t_m=q^mt_0$ and the red lines denote the four half-lines $q^{\mathbb{R}}\cdot x$, $x\in\{q^{\pm \theta_t}t_0,q^{\pm \theta_1}\}$.}
	\label{fig:contours}
\end{figure}

Then the main Riemann-Hilbert problem for $q\Psix$ reads as follows.

\begin{rhprob}\label{rhp:main}
Find a $2\times 2$ matrix-valued function $\Psi(z)=\Psi^{\B{m}}(z)$, which satisfies the following conditions.
  \begin{enumerate}[label={{\rm (\roman *)}}]
  \item $\Psi^{\B{m}}(z)$ is analytic on $\mathbb{C}\setminus\gamma^{\B{m}}$.
    \item $\Psi^{\B{m}}(z)$ has continuous boundary values $\Psi_+^{\B{m}}(z)$ and $\Psi_-^{\B{m}}(z)$ as $z$ approaches $\gamma^{\B{m}}$ from its respective out and inside, related by
		\begin{equation*}
		\Psi_+^{\B{m}}(z)=\Psi_-^{\B{m}}(z)z^m C(z),\quad z\in \gamma^{\B{m}}.
              \end{equation*}
            \item $\Psi^{\B{m}}(z)$ satisfies
              \begin{equation*}
		\Psi^{\B{m}}(z)=I+\mathcal{O}\left(z^{-1}\right)\quad z\rightarrow \infty.
              \end{equation*}
              \end{enumerate}
\end{rhprob}

In \cite{roffelsenjoshiqpvi}, it is shown that this Riemann-Hilbert problem is solvable. Namely, let $\mathfrak{M}$ denote the set of integers $m\in\mathbb{Z}$ such that the solution $\Psi^{\B{m}}(z)$ of the above Riemann-Hilbert problem exists. Then $\mathfrak{M}\neq \emptyset$ and, in fact, for any $m\in\mathbb{Z}$, if $m\notin \mathfrak{M}$, then $m\pm 1\in \mathfrak{M}$.

Take an $m\in \mathfrak{M}$ and denote by $\Psi_\infty^{\B{m}}(z)$ the analytic continuation of $\Psi^{\B{m}}(z)$ from the outside $D_+^{\B{m}}$ to $\mathbb{C}^*$. Similarly, let $\Psi_0^{\B{m}}(z)$  denote the analytic continuation of $\Psi^{\B{m}}(z)$ from the inside  $D_-^{\B{m}}$ to $\mathbb{C}$. Then
\begin{equation}
   \Psi_\infty^{\B{m}}(z)=\Psi_0^{\B{m}}(z)z^mC(z) \label{eq:globaljump},
\end{equation}
and
\begin{subequations}\label{eq:true_sol}
    \begin{align}
Y_0(z,t_m)&:=z^{\log_q(t_m)}\Psi_0^{\B{m}}(z)z^{\theta_0\sigma_3},\\
Y_\infty(z,t_m)&:=z^{\log_q(z/q)}\Psi_\infty^{\B{m}}(z) z^{-\theta_\infty\sigma_3},
\end{align}
\end{subequations}
define solutions of the same linear system
\begin{equation}\label{eq:linear_system}
    Y(qz)=A(z,t_m)Y(z),
\end{equation}
which is of the form \eqref{eq:linear_problem}. By construction, this defines an isomonodromic family in $t_m$ and therefore yields an associated solution $(f,g)$ of $q\Psix(\Theta,t_0)$. In this way, solving the Riemann-Hilbert problem is equivalent to the point-wise inversion of the monodromy mapping defined in Definition \ref{def:monodromy_manifold}.

We further note that the values of $m$ for which the Riemann-Hilbert problem is not solvable, i.e. $m\in\mathbb{Z}\setminus \mathfrak{M}$, correspond precisely to the times $t_m$ for which $(f(t_m),g(t_m))=(\infty,q^{-\theta_\infty})$.

In Section \ref{sec:dec_analytic} we recast this Riemann-Hilbert problem in a factorised form of local problems solvable using Heine's hypergeometric functions. In the next section, we discuss the class of such local problems.

\subsection{The Heine Hypergeometric System}\label{sec:hyp_model}
Consider a linear system
\begin{align}
    Y(qz)&=A(z)Y(z),\label{eq:hypsystem}\\
    A(z)&=A_0+z A_1,\nonumber
\end{align}
with $A_0$ and $A_1$ constant $2\times 2$ invertible matrices.  In this section, we discuss how such a system can be solved in terms of Heine's basic hypergeometric functions and furthermore recall the explicit description of the corresponding connection matrix.

In the following we denote
\begin{itemize}
    \item the eigenvalues of $A_0$ by $\{\sigma_{1},\sigma_{2}\}$,
    \item the eigenvalues of $A_1$ by $\{\mu_{1}^{-1},\mu_2^{-1}\}$,
    \item the roots of the determinant $|A(z)|$ by $\{x_1,x_2\}$,
\end{itemize}
so that
\begin{equation*}
    \sigma_1\sigma_2\mu_1\mu_2=x_1x_2.
\end{equation*}
These parameters form the critical data of the linear system and we will impose that there is no resonance, i.e.
\begin{equation*}
\frac{\sigma_1}{\sigma_2},\frac{\mu_1}{\mu_2},\frac{x_1}{x_2}\notin q^\mathbb{Z}.
\end{equation*}
We further assume that the linear system is irreducible, which amounts to
\begin{equation}\label{eq:irred_assumption}
    -\frac{\sigma_1\mu_1}{x_1},-\frac{\sigma_1\mu_2}{x_1},-\frac{\sigma_2\mu_1}{x_1},-\frac{\sigma_2\mu_2}{x_1}\notin q^\mathbb{Z}.
\end{equation}
Here, irreducible means that the linear system cannot be brought into triangular form by a rational gauge transform, or, equivalently, that the monodromy of the linear system is not triangular or anti-triangular. We comment on this in more detail further below.

Through application of a constant gauge transform $Y\mapsto G Y$, we may diagonalise $A_1$ so that the coefficients of $A$ are in the standard form
\begin{equation*}
    A_0=\begin{pmatrix}
    \alpha & \beta\, w\\
    \gamma\, w^{-1} & \delta\\
    \end{pmatrix},\quad
      A_1=\begin{pmatrix}
    1/\mu_1 & 0\\
    0 & 1/\mu_2\\
    \end{pmatrix},
\end{equation*}
where
\begin{align*}
    \alpha=&\frac{x_1+x_2+\mu_2(\sigma_1+\sigma_2)}{\mu_2-\mu_1}, &    \beta=&\frac{(x_1+\sigma_1\mu_2)(x_2+\sigma_1\mu_1)}{\sigma_1\mu_2(\mu_2-\mu_1)},\\
    \delta=&\frac{x_1+x_2+\mu_1(\sigma_1+\sigma_2)}{\mu_1-\mu_2}, &
    \gamma=&\frac{(x_1+\sigma_1\mu_1)(x_2+\sigma_1\mu_2)}{\sigma_1\mu_1(\mu_1-\mu_2)},
\end{align*}
and $w$ is a remaining unknown parameter related to the freedom of gauging the linear system by a constant diagonal matrix. We refer to $w$ as the gauge parameter and note that the normalised linear system is completely fixed by the critical data and gauge parameter. 

Further normalisation of the linear system is possible by rescaling
\begin{equation*}
    Y(z)\mapsto z^{\log_q(s)}Y(cz),\quad (c,s\in\mathbb{C}^*),
\end{equation*}
so that $A(z)\mapsto s A(cz)$, which allows each of the products $\sigma_1\sigma_2$, $\mu_1\mu_2$ and $x_1x_2$ to be scaled to equal $1$. This will however not be necessary for our purposes in this section.

We proceed to introduce canonical solutions of the linear system near $z=\infty$ and $z=0$ and recall the classical connection formulas between them.

Near $z=\infty$, a canonical solution is given by
\begin{align}
    Y_\infty(z)&=z^{\frac{1}{2}\log_q(z/q)}\Psi_\infty(z) \begin{pmatrix} z^{-\log_q(\mu_1)} & 0\\
    0 & z^{-\log_q(\mu_2)}
    \end{pmatrix},\label{eq:hypersolinfbranch}\\
    \Psi_\infty(z)&=\widehat{\Psi}_\infty(z) \begin{pmatrix} \left(\frac{q x_1}{z};q\right)_\infty & 0\\
    0 & \left(\frac{q x_2}{z};q\right)_\infty
    \end{pmatrix},\label{eq:hypersolinf}
\end{align}
where
\begin{equation*}
 \widehat{\Psi}_\infty(z)=\begin{pmatrix}
\hspace{2.1mm}\;_{2}\phi_1 \left[\begin{matrix} 
-\frac{\sigma_1\mu_1}{x_1}, -\frac{\sigma_2\mu_1}{x_1} \\ 
\frac{\mu_2}{\mu_1} \end{matrix} 
; q,\frac{qx_1}{z} \right] & \frac{w\, r_1}{z}  \;_{2}\phi_1 \left[\begin{matrix} 
-\frac{q\sigma_1\mu_2}{x_2}, -\frac{q\sigma_2\mu_2}{x_2} \\ 
q^2\frac{\mu_1}{\mu_2} \end{matrix} 
; q,\frac{qx_2}{z} \right]\\
\frac{r_2}{w\,z} \;_{2}\phi_1 \left[\begin{matrix} 
-\frac{q\sigma_1\mu_1}{x_1}, -\frac{q\sigma_2\mu_1}{x_1} \\ 
q^2\frac{\mu_2}{\mu_1} \end{matrix} 
; q,\frac{qx_1}{z} \right] & \hspace{3.35mm} \;_{2}\phi_1 \left[\begin{matrix} 
-\frac{\sigma_1\mu_2}{x_2}, -\frac{\sigma_2\mu_2}{x_2} \\ 
\frac{\mu_1}{\mu_2} \end{matrix} 
; q,\frac{qx_2}{z} \right]\\
 \end{pmatrix},
\end{equation*}
with
\begin{equation*}
    r_1=\frac{q \mu_1\mu_2\beta}{\mu_1-q\mu_2},\qquad r_2=\frac{q \mu_1\mu_2\gamma}{\mu_2-q\mu_1}.
\end{equation*}

Near $z=0$, we have a canonical solution
\begin{align}
    Y_0(z)&=\Psi_0(z) \begin{pmatrix} z^{\log_q(\sigma_1)} & 0\\
    0 & z^{\log_q(\sigma_2)}
    \end{pmatrix},\label{eq:hypersolzerobranch}\\
    \Psi_0(z)&=\widehat{\Psi}_0(z)\begin{pmatrix} \big(\frac{z}{x_1};q\big)_\infty^{-1} & 0\\
    0 & \big(\frac{z}{x_2};q\big)_\infty^{-1}
    \end{pmatrix}, \label{eq:hypersolzero}
\end{align}
where
\begin{equation*}
 \widehat{\Psi}_0(z)=\begin{pmatrix}
 h_{11}\;_{2}\phi_1 \left[\begin{matrix} 
-\frac{\sigma_1\mu_1}{x_1}, -\frac{q\sigma_1\mu_2}{x_1} \\ 
\frac{q\sigma_1}{\sigma_2} \end{matrix} 
; q,\frac{z}{x_2} \right] & h_{12}  \;_{2}\phi_1 \left[\begin{matrix} 
-\frac{\sigma_2\mu_1}{x_2}, -\frac{q\sigma_2\mu_2}{x_2} \\ 
\frac{q\sigma_2}{\sigma_1} \end{matrix} 
; q,\frac{z}{x_1} \right]\\
h_{21} \;_{2}\phi_1 \left[\begin{matrix} 
-\frac{\sigma_1\mu_2}{x_1}, -\frac{q\sigma_1\mu_1}{x_1} \\ 
\frac{q\sigma_1}{\sigma_2} \end{matrix} 
; q,\frac{z}{x_2} \right] & h_{22} \;_{2}\phi_1 \left[\begin{matrix} 
-\frac{\sigma_2\mu_2}{x_2}, -\frac{q\sigma_2\mu_1}{x_2} \\ 
\frac{q\sigma_2}{\sigma_1} \end{matrix} 
; q,\frac{z}{x_1} \right]\\
 \end{pmatrix},
\end{equation*}
with
\begin{equation*}
h=\begin{pmatrix}
1+\frac{x_1}{\sigma_1\mu_2} & w(1+\frac{x_1}{\sigma_2\mu_2})\\
w^{-1}(1+\frac{x_1}{\sigma_1\mu_1}) & 1+\frac{x_1}{\sigma_2\mu_1}
\end{pmatrix}\begin{pmatrix}s_1^{-1} & 0\\
0 & s_2^{-1}\\
\end{pmatrix}.
\end{equation*}
Here $s_{1},s_2\in\mathbb{C}^*$ are two free scaling parameters and we note that $h$ diagonalises $A_0$, that is, $h^{-1}A_0 h=\operatorname{diag}(\sigma_1,\sigma_2)$.

It is a classical result, essentially due to Watson \cite{watson}, see also Le Caine \cite{lecaine}, that the corresponding connection matrix,
\begin{equation*}
    C_H(z):=\Psi_0(z)^{-1}\Psi_\infty(z),
\end{equation*}
is explicitly given by
\begin{equation}\label{eq:explicit_connection}
    C(z)=\begin{pmatrix} s_1 & 0\\
    0 & s_2\end{pmatrix}
    \begin{pmatrix}
   {\color{white} w^{-1}} c_{11}\,\theta_q(-\frac{z}{\sigma_1\mu_1}) & w\, c_{12}\,\theta_q(-\frac{z}{\sigma_1\mu_2})\\
    w^{-1} c_{21}\,\theta_q(-\frac{z}{\sigma_2\mu_1}) & {\color{white} w}\, c_{22}\,\theta_q(-\frac{z}{\sigma_2\mu_2})\\
    \end{pmatrix},
\end{equation}
where
\begin{align*}
    c_{11}&=\displaystyle
\frac{\left(-\frac{\sigma_2\mu_1}{x_1},-\frac{q\sigma_2\mu_1}{x_2};q\right)_\infty}{\left(\frac{\mu_1}{\mu_2},\frac{\sigma_2}{\sigma_1};q\right)_\infty}, &
\displaystyle
c_{12}&=\frac{\left(-\frac{\sigma_2\mu_2}{x_1},-\frac{q\sigma_2\mu_2}{x_2};q\right)_\infty}{\left(\frac{\mu_2}{\mu_1},\frac{\sigma_2}{\sigma_1};q\right)_\infty},\\
c_{21}&=\displaystyle\frac{\left(-\frac{\sigma_1\mu_1}{x_1},-\frac{q\sigma_1\mu_1}{x_2};q\right)_\infty}{\left(\frac{\mu_1}{\mu_2},\frac{\sigma_1}{\sigma_2};q\right)_\infty}, &
c_{22}&=\displaystyle
\frac{\left(-\frac{\sigma_1\mu_2}{x_1},-\frac{q\sigma_1\mu_2}{x_2};q\right)_\infty}{\left(\frac{\mu_2}{\mu_1},\frac{\sigma_1}{\sigma_2};q\right)_\infty}.
\end{align*}

Note that the connection matrix is determined by the critical data
\begin{equation*}
    \Xi:=\{\sigma_1,\sigma_2,x_1,x_2,\mu_1,\mu_2\},
\end{equation*}
the gauge parameter $w$ and the two free scaling parameters $s_{1,2}$.

The matrix functions $\Psi_\infty(z)$ and $\Psi_0(z)^{-1}$ are analytic on $\mathbb{CP}^1\setminus \{0\}$ and $\mathbb{CP}^1\setminus \{\infty\}$ respectively, and their determinants are given by
\begin{align*}
|\Psi_\infty(z)|&=\left(q x_1/z,qx_2/z;q\right)_\infty,\\
|\Psi_0(z)|&=\left(z/x_1,z/x_2;q\right)_\infty.
\end{align*}
Correspondingly, the connection matrix is analytically characterised by the following properties.
\begin{enumerate}[label=\textbf{S.\arabic*}]
    \item\label{item:s1} $C(z)$ is analytic on $\mathbb{C}^*$.
    \item\label{item:s2} $C(z)$ satisfies the $q$-difference equation
    \begin{equation*}
        C(qz)=z^{-1}\begin{pmatrix}
        \sigma_1 & 0\\
        0 & \sigma_2
        \end{pmatrix} C(z) \begin{pmatrix}
        \mu_1 & 0\\
        0 & \mu_2
        \end{pmatrix}.
    \end{equation*}
    \item\label{item:s3} The determinant of $C(z)$ is given by
    \begin{equation*}
        |C(z)|=c\, \theta_q\left(z/x_1,z/x_2\right),
    \end{equation*}
    for some nonzero multiplier $c\in\mathbb{C}^*$.
\end{enumerate}

Now, take any Jordan curve $\gamma_0\subset \mathbb{C}^*$, such that, denoting by $D_-$ and $D_+$ its inside and outside within $\mathbb{C}$,
\begin{align*}
    q^nx_j\in \begin{cases}D_+ & n\leq 0,\\
    D_- & n>0,\\
    \end{cases}
\end{align*}
for $n\in\mathbb{Z}$ and $j=1,2$. Then
\begin{equation*}
    \Psi(z)=\begin{cases}\Psi_\infty(z) & z\in D_+,\\
    \Psi_0(z) & z\in D_-,\\
    \end{cases}
\end{equation*}
defines the unique solution of the following Riemann-Hilbert problem.
\begin{rhprob}[The model Riemann-Hilbert problem]\label{rhp:model}
    Find a $2\times 2$ matrix-valued function $\Psi(z)$ which satisfies the following conditions.
\begin{enumerate}
    \item $\Psi(z)$ is analytic on $\mathbb{C}\setminus \gamma_0$.
    \item The matrix function $\Psi(z)$ has continuous boundary values $\Psi_+(z)$ and $\Psi_-(z)$ as $z$ approaches $\gamma_0$ from its respective out and inside respectively, related by the jump condition
    \begin{equation*}
        \Psi_+(z)=\Psi_-(z)C(z)\quad (z\in \gamma_0).
    \end{equation*}
    \item As $z\rightarrow \infty$,
    \begin{equation*}
        \Psi(z)=I+\mathcal{O}(z^{-1}).
    \end{equation*}
\end{enumerate}
\end{rhprob}

From the explicit formula for the connection matrix, it follows that its Tyurin data are given by
\begin{equation}\label{eq:tyurin}
     \rho_j=\pi(C(x_j))=\displaystyle\frac{1}{w}\frac{\big(\frac{\mu_2}{\mu_1},-\frac{\sigma_2\mu_1}{x_1},-\frac{\sigma_2\mu_1}{x_2};q\big)_\infty}{\big(\frac{\mu_1}{\mu_2},-\frac{\sigma_2\mu_2}{x_1},-\frac{\sigma_2\mu_2}{x_2};q\big)_\infty}\times  \frac{\theta_q\big(-\frac{x_j}{\sigma_1\mu_1}\big)}{\theta_q\big(-\frac{x_j}{\sigma_1\mu_2}\big)},
\end{equation}
for $j=1,2$. This immediately yields the following important formula,
\begin{equation}\label{eq:tyurinqoutient}
    \frac{\rho_1}{\rho_2}=\displaystyle\frac{\theta_q\big(-\frac{x_1}{\sigma_1\mu_1},-\frac{x_2}{\sigma_1\mu_2}\big)}{\theta_q\big(-\frac{x_2}{\sigma_1\mu_1},-\frac{x_1}{\sigma_1\mu_2}\big)},
\end{equation}
which relates the critical data of the linear system to the Tyurin data of the connection matrix.  Similarly, we have the following formula for the dual Tyurin data $\widetilde{\rho}_{1,2}:=\pi(C(x_{1,2})^T)$,
\begin{equation}\label{eq:tyurinqoutientdual}
    \frac{\widetilde{\rho}_1}{\widetilde{\rho}_2}=\displaystyle\frac{\theta_q\big(-\frac{x_1}{\sigma_1\mu_1},-\frac{x_2}{\sigma_2\mu_1}\big)}{\theta_q\big(-\frac{x_2}{\sigma_1\mu_1},-\frac{x_1}{\sigma_2\mu_1}\big)}.
\end{equation}
The right-hand sides of equations \eqref{eq:tyurin} and \eqref{eq:tyurinqoutientdual}, considered respectively as functions of $\sigma_1$ and $\mu_1$, are $q$-periodic functions which will play an important role in the Mano decompositions. They are derived in \cite{ohyamaramissualoy}*{\S 5.1.3 and \S 5.3.2} by geometric means.

Let us now consider a reducible case. All of the above formula remain valid when we set
\begin{equation*}
\mu_1=-q^{-k}\frac{x_1}{\sigma_1},
\end{equation*}
for some $k\in\mathbb{Z}$. When $k<0$, the connection matrix is lower-triangular and $\rho_2=\infty$. On the other hand, when $k\geq 0$, the connection matrix is upper-triangular and $\rho_1=0$. 

If we further specialise to $k=0$, then the coefficient matrix $A(z)$ and the matrix functions $\Psi_\infty(z)$ and $\Psi_0(z)$ are also upper-triangular. Explicitly, when $k=0$,
\begin{equation*}
    A(z)=\begin{pmatrix}
    \sigma_1 & \frac{w(x_1-x_2)}{\mu_2}\\
    0 & \sigma_2
    \end{pmatrix}+z\begin{pmatrix}
    1/\mu_1 & 0\\
    0 & 1/\mu_2\\
    \end{pmatrix},
\end{equation*}
and
\begin{equation*}
     \widehat{\Psi}_\infty(z)=\begin{pmatrix}
     1 & \frac{w\,r_1}{z} g_\infty(z)\\
     0 & 1
     \end{pmatrix},\qquad \widehat{\Psi}_0(z)=\begin{pmatrix}
     h_{11} & h_{12} \,g_0(z)\\
     0 & h_{22}
     \end{pmatrix},
\end{equation*}
where
\begin{equation*}
    g_\infty(z)=\;_{2}\phi_1 \left[\begin{matrix} 
q,\frac{q\sigma_1}{\sigma_2} \\ 
q^2\frac{\mu_1}{\mu_2} \end{matrix} 
; q,\frac{qx_2}{z} \right],\quad
    g_0(z)=\;_{2}\phi_1 \left[\begin{matrix} 
q,\frac{\mu_1}{\mu_1} \\ 
\frac{q\sigma_2}{\sigma_1} \end{matrix} 
; q,\frac{z}{x_1} \right].
\end{equation*}

Correspondingly, the connection matrix reads
\begin{equation*}
    C(z)=\begin{pmatrix} s_1 & 0\\
    0 & s_2\end{pmatrix}
    \begin{pmatrix}
     c_{11}\,\theta_q(\frac{z}{x_1}) & w\, c_{12}\,\theta_q(-\frac{z}{\sigma_1\mu_2})\\
   0 & {\color{white} w}\, c_{22}\,\theta_q(\frac{z}{x_2})\\
    \end{pmatrix},
\end{equation*}
and
\begin{equation*}
   \rho_1=0,\quad  \rho_2=-\frac{\mu_2}{\mu_1 w} \frac{\big(\frac{q x_1}{x_2},\frac{q\mu_2}{\mu_1};q\big)_\infty}{\big(q,q\frac{\sigma_1}{\sigma_2};q\big)_\infty}.
\end{equation*}
Note that, by setting $w=0$, we further obtain the diagonal case where $A$, $\Psi_{0}$, $\Psi_\infty$ and $C$ are all diagonal, with $\rho_1=0$ and $\rho_2=\infty$.

We finish this section with studying connection matrices of Heine hypergeometric type abstractly.

\begin{definition}
We call any $2\times 2$ matrix $C(z)$, satisfying properties \ref{item:s1}-\ref{item:s3}, for some choice of non-resonant critical data $\Xi$, a connection matrix of Heine hypergeometric type. We say $C(z)$ is reducible if and only if it is triangular or anti-triangular.
\end{definition}
\begin{definition}
We denote by $C_H(z;\Xi,w,s_{1},s_2)$ the connection matrix given in equation \eqref{eq:explicit_connection}, for any choice of non-resonant critical data $\Xi$, gauge parameter $w$ and scalars $s_{1,2}$.
\end{definition}
\begin{lemma}\label{lem:abstract_to_explicit}
Let $C(z)$ be any irreducible connection matrix of Heine hypergeometric type, for some critical data $\Xi$, then
\begin{equation*}
    C(z)=C_H(z;\Xi,w,s_{1},s_2),
\end{equation*}
is the connection matrix of a Heine hypergeometric system with critical data $\Xi$, for some values of the gauge parameter $w$ and scalars $s_1,s_2$. 
\end{lemma}
\begin{proof}
Let the critical data be given by $\Xi=\{\sigma_1,\sigma_2,x_1,x_2,\mu_1,\mu_2\}$ and denote by $(\rho_1,\rho_2)$ the Tyurin data of $C(z)$ at $(x_1,x_2)$. We start by proving equation \eqref{eq:tyurinqoutient} for this `abstract' connection matrix. To do this, note that
\begin{equation}
    C(z)=\begin{pmatrix}
   \gamma_{11}\,\theta_q(-\frac{z}{\sigma_1\mu_1}) &  \gamma_{12}\,\theta_q(-\frac{z}{\sigma_1\mu_2})\\
     \gamma_{21}\,\theta_q(-\frac{z}{\sigma_2\mu_1}) &  \gamma_{22}\,\theta_q(-\frac{z}{\sigma_2\mu_2})\\
    \end{pmatrix},
\end{equation}
for some constants $\gamma_{ij}, 1\leq i,j\leq 2$, all nonzero. It follows that
\begin{equation*}
    \rho_j=\frac{C_{11}(x_j)}{C_{12}(x_j)}=\frac{\gamma_{11}\,\theta_q(-\frac{x_j}{\sigma_1\mu_1})}{\gamma_{12}\,\theta_q(-\frac{x_j}{\sigma_1\mu_2})},
\end{equation*}
for $j=1,2$, from which equation \eqref{eq:tyurinqoutient} follows directly.

Next we consider the connection matrix 
\begin{equation*}
    \widehat{C}(z)=C_H(z;\Xi,w,s_{1},s_2).
\end{equation*}
Let $\widehat{\rho}_{1,2}$ denote its Tyurin data at $x_{1,2}$. Due to  equation \eqref{eq:tyurinqoutient}, we must have
\begin{equation*}
    \frac{\widehat{\rho}_1}{\widehat{\rho}_2}=\frac{\rho_1}{\rho_2},
\end{equation*}
and, using equation \eqref{eq:tyurin}, we can further scale the gauge parameter $w$ such that $\widehat{\rho}_{1,2}=\rho_{1,2}$. Now, consider the quotient
\begin{equation*}
    D(z)=C(z)\widehat{C}(z)^{-1}.
\end{equation*}
As the Tyurin data of $C(z)$ and $\widehat{C}(z)$ are the same, $D(z)$ is analytic at $z=x_{1,2}$ and thus analytic on $\mathbb{C}^*$. It, furthermore, satisfies
\begin{equation*}
    D(qz)=\begin{pmatrix}
     \sigma_1 & 0\\
     0 & \sigma_2
    \end{pmatrix}D(z)\begin{pmatrix}
     \sigma_1^{-1} & 0\\
     0 & \sigma_2^{-1}
    \end{pmatrix}.
\end{equation*}
So $D_{11}(z)$ and $D_{22}(z)$ are analytic functions which are $q$-periodic. By application of Liouville's theorem to $D_{11}(q^x)$ and $D_{22}(q^x)$, it follows that they must be constant. Similarly, $D_{12}(z)$ is analytic on $\mathbb{C}^*$ and it satisfies $D_{12}(qz)=r\,D_{12}(z)$, where $r=\frac{\sigma_1}{\sigma_2}\notin q^\mathbb{Z}$. It is easy to see that there is only one function satisfying these conditions, $D_{12}(z)\equiv 0$. Similarly, $D_{21}(z)\equiv 0$ and it follows that $D(z)=D$ is a constant matrix. So
\begin{equation*}
    C(z)=D\widehat{C}(z),
\end{equation*}
and by rescaling $s_1\mapsto s_1 D_{11}^{-1}$ and $s_2\mapsto s_2 D_{22}^{-1}$, we obtain $C(z)=C_H(z;\Xi,w,s_{1},s_2)$ and the lemma follows.
\end{proof}

\begin{lemma}\label{lem:reducible}
Let $C(z)$ be any connection matrix of Heine hypergeometric type, for some non-resonant critical data $\Xi$, and denote by
\begin{equation*}
    \rho_{1,2}=\pi(C(x_{1,2})),\quad \widetilde{\rho}_{1,2}=\pi(C(x_{1,2})^T),
\end{equation*}
its Tyurin data and dual Tyurin data respectively. Then, the following are equivalent,
\begin{enumerate}
    \item $C(z)$ is reducible,
    \item $\rho_1\in \{0,\infty\}$ or $\rho_2\in \{0,\infty\}$,
    \item $\widetilde{\rho}_1\in \{0,\infty\}$ or $\widetilde{\rho}_2\in \{0,\infty\}$,
    \item one of the parameter conditions in \eqref{eq:irred_assumption} is violated.
\end{enumerate}
Suppose now that $C(z)$ is reducible, so that one of the parameter conditions in \eqref{eq:irred_assumption} is violated.\\
If $-\frac{\sigma_1\mu_1}{x_1}=q^k\in q^{\mathbb{Z}}$, then one or both of the following hold true.
\begin{itemize}
    \item $(\rho_1,\widetilde{\rho}_2)=(0,\infty)$ and
\begin{equation*}
    C(z)=\begin{pmatrix} s_1 & 0\\
    0 & s_2\end{pmatrix}
    \begin{pmatrix}
     z^k\theta_q(\frac{z}{x_1}) & w\, \theta_q(-\frac{z}{\sigma_1\mu_2})\\
   0 & z^{-k} \theta_q(\frac{z}{x_2})\\
    \end{pmatrix},
\end{equation*}
    for some $s_{1,2}\in \mathbb{C}^*$ and $w\in\mathbb{C}$.
\item $(\widetilde{\rho}_1,\rho_2)=(0,\infty)$ and
\begin{equation*}
    C(z)=\begin{pmatrix} s_1 & 0\\
    0 & s_2\end{pmatrix}
    \begin{pmatrix}
     z^k\theta_q(\frac{z}{x_1}) & 0\\
   w^{-1}\, \theta_q(-\frac{z}{\sigma_2\mu_1}) & z^{-k} \theta_q(\frac{z}{x_2})\\
    \end{pmatrix},
\end{equation*}
    for some $s_{1,2}\in \mathbb{C}^*$ and $w^{-1}\in\mathbb{C}$.
\end{itemize}
If $-\frac{\sigma_2\mu_2}{x_1}=q^k\in q^{\mathbb{Z}}$, then one or both of the following hold true.
\begin{itemize}
    \item $(\widetilde{\rho}_1,\rho_2)=(\infty,0)$ and
\begin{equation*}
    C(z)=\begin{pmatrix} s_1 & 0\\
    0 & s_2\end{pmatrix}
    \begin{pmatrix}
     z^{-k}\theta_q(\frac{z}{x_2}) & w\, \theta_q(-\frac{z}{\sigma_1\mu_2})\\
   0 & z^{k} \theta_q(\frac{z}{x_1})\\
    \end{pmatrix},
\end{equation*}
    for some $s_{1,2}\in \mathbb{C}^*$ and $w\in\mathbb{C}$.
\item $(\rho_1, \widetilde{\rho}_2)=(\infty,0)$ and
\begin{equation*}
    C(z)=\begin{pmatrix} s_1 & 0\\
    0 & s_2\end{pmatrix}
    \begin{pmatrix}
     z^{-k}\theta_q(\frac{z}{x_2}) & 0\\
   w^{-1}\, \theta_q(-\frac{z}{\sigma_2\mu_1}) & z^{k} \theta_q(\frac{z}{x_1})\\
    \end{pmatrix},
\end{equation*}
    for some $s_{1,2}\in \mathbb{C}^*$ and $w^{-1}\in\mathbb{C}$.
\end{itemize}
If $-\frac{\sigma_2\mu_1}{x_1}=q^k\in q^{\mathbb{Z}}$, then one or both of the following hold true.
\begin{itemize}
    \item $(\widetilde{\rho}_1,\rho_2)=(\infty,\infty)$ and
\begin{equation*}
    C(z)=\begin{pmatrix} s_1 & 0\\
    0 & s_2\end{pmatrix}
    \begin{pmatrix}
     w\, \theta_q(-\frac{z}{\sigma_1\mu_1}) & z^{-k} \theta_q(\frac{z}{x_2}) \\
   z^{k}\theta_q(\frac{z}{x_1}) &  0\\
    \end{pmatrix},
\end{equation*}
    for some $s_{1,2}\in \mathbb{C}^*$ and $w\in\mathbb{C}$.
\item $(\rho_1, \widetilde{\rho}_2)=(0,0)$ and
\begin{equation*}
    C(z)=\begin{pmatrix} s_1 & 0\\
    0 & s_2\end{pmatrix}
    \begin{pmatrix}
     0 & z^{-k} \theta_q(\frac{z}{x_2}) \\
   z^{k}\theta_q(\frac{z}{x_1}) &  w^{-1} \theta_q(-\frac{z}{\sigma_2\mu_2})\\
    \end{pmatrix},
\end{equation*}
    for some $s_{1,2}\in \mathbb{C}^*$ and $w^{-1}\in\mathbb{C}$.
\end{itemize}
If $-\frac{\sigma_1\mu_2}{x_1}=q^k\in q^{\mathbb{Z}}$, then one or both of the following hold true.
\begin{itemize}
    \item $(\rho_1,\widetilde{\rho}_2)=(\infty,\infty)$ and
\begin{equation*}
    C(z)=\begin{pmatrix} s_1 & 0\\
    0 & s_2\end{pmatrix}
    \begin{pmatrix}
     w\, \theta_q(-\frac{z}{\sigma_1\mu_1}) & z^{k} \theta_q(\frac{z}{x_1}) \\
   z^{-k}\theta_q(\frac{z}{x_2}) &  0\\
    \end{pmatrix},
\end{equation*}
    for some $s_{1,2}\in \mathbb{C}^*$ and $w\in\mathbb{C}$.
\item $(\widetilde{\rho}_1,\rho_2)=(0,0)$ and
\begin{equation*}
    C(z)=\begin{pmatrix} s_1 & 0\\
    0 & s_2\end{pmatrix}
    \begin{pmatrix}
     0 & z^{k} \theta_q(\frac{z}{x_1}) \\
   z^{-k}\theta_q(\frac{z}{x_2}) &  w^{-1} \theta_q(-\frac{z}{\sigma_2\mu_2})\\
    \end{pmatrix},
\end{equation*}
    for some $s_{1,2}\in \mathbb{C}^*$ and $w^{-1}\in\mathbb{C}$.
\end{itemize}
\end{lemma}
\begin{proof}
The proof is straightforward but laborious. 
Let us start by assuming that $C(z)$ is reducible. Then, either $C(z)$ is triangular or anti-triangular. If $C(z)$ is triangular, then
\begin{equation*}
C_{11}(z)C_{22}(z)=|C(z)|=c\, \theta_q(z/x_1,z/x_2),
\end{equation*}
for some $c\in\mathbb{C}^*$ and hence
\begin{equation*}
    C_{11}(z)=c_1\,z^k\theta_q(z/x_i),\quad C_{22}(z)=c_2\,z^k\theta_q(z/x_j),
\end{equation*}
for some scalars $c_{1,2}$, $k\in\mathbb{Z}$ and $\{i,j\}=\{1,2\}$. This yields four separate cases, defined by
 $C(z)$ being upper or lower triangular (or both) and $(i,j)=(1,2)$ or $(i,j)=(2,1)$.\\
Similarly, if $C(z)$ is anti-triangular, then
\begin{equation*}
C_{12}(z)C_{21}(z)=-|C(z)|=c\, \theta_q(z/x_1,z/x_2),
\end{equation*}
for some $c\in\mathbb{C}^*$ and thus
\begin{equation*}
    C_{12}(z)=c_1\,z^k\theta_q(z/x_i),\quad C_{21}(z)=c_2\,z^k\theta_q(z/x_j),
\end{equation*}
for some scalars $c_{1,2}$, $k\in\mathbb{Z}$ and $\{i,j\}=\{1,2\}$. This yields four separate cases, defined by
 $C(z)$ being upper or lower anti-triangular (or both) and $(i,j)=(1,2)$ or $(i,j)=(2,1)$.
 
 This gives a total of $4+4=8$ different possible forms of $C(z)$, described explicitly in the lemma. The values of $\rho_{1,2}$ and $\widetilde{\rho}_{1,2}$ can be read directly from the formulas and, in particular, in each case (2) and (3) hold true. Finally, (4) follows from comparing the explicit formulas for $C(z)$ with the $q$-difference equation in \ref{item:s2}.
 
 Note also that at most one of the parameter conditions in equation \eqref{eq:irred_assumption} can be violated, as otherwise the critical data $\Xi$ are reducible. We have thus established the second part of the lemma and shown that (1) implies (2),(3) and (4).
 
 For the remainder of the proof, note that
 \begin{equation}\label{eq:cexlicit}
    C(z)=\begin{pmatrix}
   \gamma_{11}\,\theta_q(-\frac{z}{\sigma_1\mu_1}) &  \gamma_{12}\,\theta_q(-\frac{z}{\sigma_1\mu_2})\\
     \gamma_{21}\,\theta_q(-\frac{z}{\sigma_2\mu_1}) &  \gamma_{22}\,\theta_q(-\frac{z}{\sigma_2\mu_2})\\
    \end{pmatrix},
\end{equation}
 for some complex coefficients $\gamma_{ij}\in\mathbb{C}$, $1\leq i,j\leq 2$, by the $q$-difference equation in \ref{item:s2}.
 
 Next, assume (2) and, without loss of generality, that $\rho_1=0$. Then,  $C_{11}(x_1)=C_{21}(x_1)=0$.  It follows that either $\gamma_{11}=0$ or $-\frac{\sigma_1\mu_1}{x_1}\in q^{\mathbb{Z}}$. Similarly, either $\gamma_{21}(z)=0$ or $-\frac{\sigma_2\mu_1}{x_1}\in q^{\mathbb{Z}}$. Since both parameter conditions on the critical data cannot hold simultaneously, either $\gamma_{11}=0$ or $\gamma_{21}=0$ and thus $C(z)$ must be reducible. This proves that (1) and (2) are equivalent.
 
 Analogously, it is shown that (1) and (3) are equivalent. Finally, assume (4) and without loss of generality, consider $-\frac{\sigma_1\mu_1}{x_1}\in q^{\mathbb{Z}}$. If $C_{11}(z)\equiv 0$, then (1) holds and we are done. So suppose $C_{11}(z)\not\equiv 0$. Since $C_{11}(x_1)=0$ and $|C(x_1)|=0$, it follows that $C_{12}(x_1)=0$ or $C_{21}(x_1)=0$. By \eqref{eq:cexlicit}, and the fact that the critical data are non-resonant, it follows that $\gamma_{12}=0$ or $\gamma_{21}=0$. This finishes the proof of the lemma.
\end{proof}

\subsection{Algebraic Decomposition I of the monodromy}
\label{sec:dec_algebraic}
In this section we construct the first of two factorisations of the global connection matrix of the Jimbo-Sakai linear system into two local connection matrices coming from Heine hypergeometric systems, relevant in this paper.

The existence of such factorisations was proven in full generality by Ohyama et al. \cite{ohyamaramissualoy}[Theorem 5.13]. Here we only discuss the generic setting, excluding the so called logarithmic cases. On the other hand, we provide additional explicit formulas for all the different factors. 

In this section, we construct the factorisation, along the decomposition illustrated in Figure \ref{fig:pairspants}(A), and call it Mano decomposition I.  We start with the following lemma, which will be helpful for our purposes.
\begin{lemma}\label{lem:pi}
    Let $R$ and $N=(n_{ij})$ be $2\times 2$ matrices of rank $1$ and $2$ respectively, then
\begin{equation*}
  \pi(NR)=\pi(R),\quad  \pi(RN)=M_{N}(\pi(R)),\quad M_{N}(Z):=\frac{n_{11}Z+n_{21}}{n_{12}Z+n_{22}}.
\end{equation*}
\end{lemma}
\begin{proof}
These identities follow directly from the definition of $\pi(\cdot)$.
\end{proof}

The describe decomposition I, recall the definition of the elliptic function
\begin{equation*}
\mathcal{E}_0(\sigma)=\frac{\vartheta_\tau(\sigma-\theta_1+\theta_\infty,\sigma+\theta_1-\theta_\infty)}{\vartheta_\tau(\sigma+\theta_1+\theta_\infty,\sigma-\theta_1-\theta_\infty)}.
\end{equation*}
This is an elliptic function of degree $2$ with periods $1$ and $\tau$. It furthermore has the reflection symmetry $\mathcal{E}_0(-\sigma)=\mathcal{E}_0(\sigma)$. We have the following proposition, which corresponds to the generic case (1) in \cite{ohyamaramissualoy}[Theorem 5.13].

\begin{proposition}\label{prop:algebraicdecomzero}
Let $C(z)\in\mathfrak{C}(\Theta,t_0)$ be a connection matrix with corresponding Tyurin data $\rho$ and recall the notation $\rho_{34}:=\rho_3/\rho_4\in\mathbb{CP}^1$. If
\begin{equation}\label{eq:nonlogassumption}
    \rho_{34}\neq \mathcal{E}_0(\sigma), \text{ for } \sigma=0,\tfrac{1}{2},\tfrac{\tau}{2},\tfrac{1}{2}+\tfrac{\tau}{2},
\end{equation}
then, for any solution $\sigma_{0t}$ of $\mathcal{E}_0(\sigma)=\rho_{34}$, we can construct a factorisation of the form,
\begin{equation}\label{eq:factorisation}
    z^mC(z)=D_mC^i\left(z/t_m\right)(-t_m)^{\sigma_{0t} \sigma_3}C^e(z),
\end{equation}
valid for all $m\in\mathbb{Z}$, where 
\begin{itemize}
    \item $C^e(z)$, independent of $m$, is the connection matrix of a Heine hypergeometric system with critical data $\Xi^e$ given by
    \begin{align*}
    \sigma_1^e&=-q^{-\sigma_{0t}}, &  x_1^e&=q^{+\theta_1}, &   \mu_1^e&=q^{+\theta_\infty},\\
    \sigma_2^e&=-q^{+\sigma_{0t}}, &  x_2^e&=q^{-\theta_1}, & \mu_2^e&=q^{-\theta_\infty}.
    \end{align*}
    \item $C^i(z)$, independent of $m$, is the connection matrix of a Heine hypergeometric system with critical data $\Xi^i$ given by
    \begin{align*}
    \sigma_1^i&=q^{+\theta_0}, &  x_1^i&=q^{+\theta_t}, &   \mu_1^i&=-q^{+\sigma_{0t}},\\
    \sigma_2^i&=q^{-\theta_0}, &  x_2^i&=q^{-\theta_t}, & \mu_2^i&=-q^{-\sigma_{0t}}.
    \end{align*}
    \item $D_m$ is a diagonal matrix, explicitly given by
    \begin{equation*}
        D_m=(-t_0)^m q^{\frac{1}{2}m(m+1)}q^{m\theta_0\sigma_3}.
    \end{equation*}
    \end{itemize}
Let $\widetilde{\rho}$ denote the dual Tyurin data of $C(z)$, then we have the identifications
\begin{equation}\label{eq:tyurin_identification}
    (\rho_3,\rho_4)=(\rho_1^e,\rho_2^e),\qquad
    (\widetilde{\rho}_1,\widetilde{\rho}_2)=(\widetilde{\rho}_1^i,\widetilde{\rho}_2^i),
\end{equation}
where $(\rho_1^e,\rho_2^e)$ are the Tyurin data of $C^e(z)$ and $(\widetilde{\rho}_1^i,\widetilde{\rho}_2^i)$ are the dual Tyurin data of $C^i(z)$. In particular, the following hold true.
\begin{itemize}
    \item If $\rho_{3},\rho_4\notin\{0,\infty\}$, then $$\sigma_{0t}\not\equiv\pm (\theta_\infty+\theta_1),\pm (\theta_\infty-\theta_1)\mod{\Lambda_\tau},$$ and $C^e(z)$ is irreducible and given by
    $C^e(z)=C_H(z;\Xi^e,w^e,s_1^e,s_2^e)$, for some gauge parameter $w^e$ and scaling parameters $s_{1,2}^e$.\\ Else, $C^e(z)$ is reducible.
    \item If $\widetilde{\rho}_1,\widetilde{\rho}_2\notin\{0,\infty\}$, then $$\sigma_{0t}\not\equiv\pm (\theta_0+\theta_t),\pm (\theta_0-\theta_t)\mod{\Lambda_\tau},$$ and $C^i(z)$ is irreducible and given by
    $C^i(z)=C_H(z;\Xi^i,w^i,s_1^i,s_2^i)$,
    for some gauge parameter $w^i$ and scaling parameters $s_{1,2}^i$.\\ Else, $C^i(z)$ is reducible.
\end{itemize}
\end{proposition}
\begin{proof} 
Suppose a factorisation as described in the proposition exists, then equation \eqref{eq:tyurin_identification} follows directly from Lemma \ref{lem:pi}. In particular, it follows from equation \eqref{eq:tyurinqoutient} that necessarily $\mathcal{E}_0(\sigma_{0t})=\rho_1^e/\rho_2^e=\rho_{34}$.

To prove the proposition, we start with the construction of a connection matrix $C^e(z)$  of Heine hypergeometric type whose Tyurin data are equal to those of $C(z)$ at $x_{1,2}^e$. If $\rho_3\in\{0,\infty\}$ or $\rho_4\in\{0,\infty\}$, then  $\sigma_{0t}\equiv \pm(\theta_\infty+\theta_1)$ or $\sigma_{0t}\equiv \sigma_{0t}\pm (\theta_\infty-\theta_1) \mod{\Lambda_\tau}$, for some choice of signs, and $C^e(z)$ has to be chosen reducible, by Lemma \ref{lem:reducible}. In any of the four cases, we can use Lemma \ref{lem:reducible} to construct a connection matrix $C^e(z)$, with critical data $\Xi^e$, that satisfies
\begin{equation}\label{eq:tyurinmatch}
    (\rho_1^e,\rho_2^e)=(\rho_3,\rho_4).
\end{equation}
On the other hand, if $\rho_3,\rho_4\notin\{0,\infty\}$ then
$\sigma_{0t}\not\equiv \pm(\theta_\infty+\theta_1),\pm (\theta_\infty-\theta_1) \mod{\Lambda_\tau}$, and we set $C^e(z)=C_H(z;\Xi^e,w^e,s_1^e,s_2^e)$, with gauge parameter
\begin{equation} \label{eq:we}
    w^e=-\frac{1}{q^{2\theta_\infty}\rho_4} \frac{\Gamma_q(1+2\theta_\infty)\Gamma_q(1+\theta_1-\theta_\infty-\sigma_{0t})\Gamma_q(1+\theta_1-\theta_\infty+ \sigma_{0t})}{\Gamma_q(1-2\theta_\infty)\Gamma_q(1+\theta_1+\theta_\infty-\sigma_{0t})\Gamma_q(1+\theta_1+\theta_\infty+ \sigma_{0t})}.
\end{equation}
The connection matrix $C^e(z)$ is irreducible and equation \eqref{eq:tyurinmatch} holds due to equations \eqref{eq:tyurin} and \eqref{eq:tyurinqoutient}.

We now simply define
\begin{equation*}
    C^i(\zeta)=C(t_0\zeta)C^e(t_0\zeta)^{-1}(-t_0)^{-\sigma_{0t} \sigma_3},
\end{equation*}
so that
\begin{equation}\label{eq:factorisationproof}
    C(z)=C^i\left(z/t_0\right)(-t_0)^{\sigma_{0t} \sigma_3}C^e(z),
\end{equation}
which is exactly equation \eqref{eq:factorisation} for $m=0$ with $D_0=I$.

We proceed to check that $C^i(z)$ is the connection matrix of a hypergeometric system with critical data $\Xi^i$. 
Firstly, note that the determinant of $C^e(z)$ vanishes only on the $q$-lines $q^{\pm\theta_1+\mathbb{Z}}$, so that $C^i(z)$ may only have poles on the $q$-lines $t_0^{-1}q^{\pm\theta_1+\mathbb{Z}}$. Furthermore, as $C(z)$ and $C^e(z)$ have the same Tyurin data at $q^{\pm \theta_1}$, the quotient $C(z)C^e(z)^{-1}$ is analytic at $z=q^{\pm \theta_1}$. It thus follows that $C^i(z)$ is analytic on $\mathbb{C}^*$.

Secondly, it follows by direct computation that
  \begin{equation}\label{eq:qdifCi}
        C^i(q\zeta)=\zeta^{-1}\begin{pmatrix}
        \sigma_1^i & 0\\
        0 & \sigma_2^i
        \end{pmatrix} C^i(\zeta) \begin{pmatrix}
        \mu_1^i & 0\\
        0 & \mu_2^i
        \end{pmatrix}.
    \end{equation}
Finally, since
\begin{align*}
    |C(z)|&=c\, \theta_q\left(q^{-\theta_t}\frac{z}{t_0},q^{+\theta_t}\frac{z}{t_0},q^{-\theta_1}z,q^{+\theta_1}z\right),\\
    |C^e(z)|&=c^e\, \theta_q\left(q^{-\theta_1}z,q^{+\theta_1}z\right),
    \end{align*}
for some multipliers $c,c^e\in\mathbb{C}^*$, we have
\begin{equation*}
       |C^i(z)|=c^i\, \theta_q\left(q^{-\theta_t}z,q^{+\theta_t}z\right),\quad c^i:=\frac{c}{c^e}.
\end{equation*}
It follows that $C^i(z)$ is the connection matrix of a Heine hypergeometric system with critical data $\Xi^i$. Furthermore, by Lemma \ref{lem:pi}, we must have 
\begin{equation*}
    (\widetilde{\rho}_1,\widetilde{\rho}_2)=(\widetilde{\rho}_1^i,\widetilde{\rho}_2^i),
\end{equation*}
where $(\widetilde{\rho}_1^i,\widetilde{\rho}_2^i)$ are the dual Tyurin data of $C^i(z)$. The final item in the proposition follows directly from this identification and Lemma \ref{lem:reducible}.

We have now proven the statement of the proposition with $m=0$. By $m$-fold application of equation \eqref{eq:qdifCi}, we obtain equation \eqref{eq:factorisation} for general $m$ and the proposition follows.
\end{proof}

Next, we introduce the twist parameter $s_{0t}$ corresponding to Mano decomposition I. Take a $0$-generic point $\eta\in\mathcal{F}(\Theta,t_0)$, according to Definition \ref{def:generic0}. Then, by Proposition \ref{prop:algebraicdecomzero}, we may write the corresponding connection matrix $C(z)$ as
\begin{equation}\label{eq:factorisationm0}
    C(z)=C^i\left(z/t_0\right)(-t_0)^{\sigma_{0t} \sigma_3}C^e(z),
\end{equation}
where
\begin{equation*}
   C^e(z)=C_H(z;\Xi^e,w^e,s_1^e,s_2^e), \quad C^i(z)=C_H(z;\Xi^i,w^i,s_1^i,s_2^i),
\end{equation*}
for some gauge and scaling parameters $w^e,w^i,s_{1,2}^e,s_{1,2}^i\in\mathbb{C}^*$, with external and internal critical data $\Xi^e$ and $\Xi^i$ as in the proposition.

The affine Segre surface is two dimensional, but so far we have only given one relation between the $\eta$-variables and the parameters in the decomposition, through the identity $\mathcal{E}_0(\sigma_{0t})=\rho_{34}$.

We proceed to derive a second identity which involves the twist parameter. To this end, we first consider how the different parameters of the decomposition are affected by scaling freedoms.
Firstly, the connection matrix $C(z)$ is only determined by the coordinates $\eta$ up to arbitrary left and right-multiplication by diagonal matrices. In terms of the above parameters, left-multiplication,
\begin{equation*}
    C(z)\mapsto B\,C(z),\quad B=\operatorname{diag}(b_1,b_2),
\end{equation*}
simply rescales $s_1^i\mapsto b_1s_1^i$ and $s_2^i\mapsto b_2s_2^i$. Similarly, right-multiplication,
\begin{equation*}
    C(z)\mapsto C(z)B,\quad B=\operatorname{diag}(b_1,b_2),
\end{equation*}
rescales
\begin{equation*}
    s_1^e\mapsto b_1 s_1^e,\quad s_2^e\mapsto b_2 s_2^e,\quad w^e\mapsto \frac{b_2}{b_1}w^e.
\end{equation*}
Then, there is the further freedom of internal scaling,
\begin{equation*}
    C^e(z)\mapsto B^{-1} C^e(z),\quad C^i(z)\mapsto C^i(z)B,\quad B=\operatorname{diag}(b_1,b_2),
\end{equation*}
which rescales
\begin{equation*}
  s_1^e\mapsto b_1^{-1} s_1^e,\quad s_2^e\mapsto b_2^{-1} s_2^e,\quad  s_1^i\mapsto b_1 s_1^i,\quad s_2^i\mapsto b_2 s_2^i,\quad w^i\mapsto \frac{b_2}{b_1}w^i.
\end{equation*}

Invariant under each of these three scalings, are the quantities
\begin{equation}\label{eq:twistdefi}
    r_{0t}:=\frac{s_1^ew^e}{s_2^ew^i},\quad s_{0t}:=r_{0t}/c_{0t},
\end{equation}
where $c_{0t}$ is defined by
\begin{equation*}
c_{0t}=\frac{\Gamma_q(1-2\sigma_{0t})^2}{\Gamma_q(1+2\sigma_{0t})^2}\prod_{\epsilon=\pm1}\frac{ \Gamma_q(1+\theta_t+\epsilon\,\theta_0+\sigma_{0t})  \Gamma_q(1+\theta_1+\epsilon\,\theta_\infty+\sigma_{0t})}{ \Gamma_q(1+\theta_t+\epsilon\,\theta_0-\sigma_{0t})  \Gamma_q(1+\theta_1+\epsilon\,\theta_\infty-\sigma_{0t})}.
\end{equation*}
We call $s_{0t}$ the twist parameter corresponding to Mano decomposition I.

\begin{lemma}\label{lem:twist_s0t}
The twist parameter is explicitly related to the $\eta$-coordinates via the formula
\begin{equation}\label{eq:twist0t}
 s_{0t}=-(-t_0)^{-2\sigma_{0t}}M_{0t}(\rho_{24}),
\end{equation}
where $M_{0t}(\cdot)$ is the M\"obius transformation
\begin{equation*}
    M_{0t}(Z)=\frac{\vartheta_\tau(\theta_1-\theta_\infty+\sigma_{0t})\theta_q(q^{\theta_t+\theta_\infty+\sigma_{0t}}t_0^{-1})-Z
\vartheta_\tau(\theta_1+\theta_\infty+\sigma_{0t})\theta_q(q^{\theta_t-\theta_\infty+\sigma_{0t}}t_0^{-1})}{\vartheta_\tau(\theta_1-\theta_\infty-\sigma_{0t})\theta_q(q^{\theta_t+\theta_\infty-\sigma_{0t}}t_0^{-1})-Z
\vartheta_\tau(\theta_1+\theta_\infty-\sigma_{0t})\theta_q(q^{\theta_t-\theta_\infty-\sigma_{0t}}t_0^{-1})},
\end{equation*}
and $\rho_{24}=\rho_{24}(\eta)$ is defined in equation \eqref{eq:def_rhoij}.
\end{lemma}
\begin{proof}
To prove the lemma, we need explicit formulas for the gauge parameters $w^e$ and $w^i$. We already have a formula for $w^e$, equation \eqref{eq:we}. To compute $w^i$, we first apply equation \eqref{eq:tyurin}, with critical data $\Xi^i$ and Tyurin parameter $\rho_2^i:=\pi(C^i(x_2^i))$, to obtain
\begin{equation*}
    w^i=-\frac{1}{q^{2\sigma_{0t}}\rho_2^i} \frac{\Gamma_q(1+2\sigma_{0t})\Gamma_q(1+\theta_t+\theta_0-\sigma_{0t})\Gamma_q(1+\theta_t-\theta_0-\sigma_{0t})}{\Gamma_q(1-2\sigma_{0t})\Gamma_q(1+\theta_t+\theta_0+\sigma_{0t})\Gamma_q(1+\theta_t-\theta_0+\sigma_{0t})}.
\end{equation*}
By substituting this equation and equation \eqref{eq:we} into the defining equation of the twist parameter, equation \eqref{eq:twistdefi}, we get
\begin{align}\label{eq:twistintermediate}
    s_{0t}=&c\,q^{2(\sigma_{0t}-\theta_\infty)}\frac{s_1^e \rho_2^i}{s_2^e \rho_4},\\
c:=&\frac{\Gamma_q(1+2\sigma_{0t})\Gamma_q(1+2\theta_\infty)\Gamma_q(1+\theta_1-\theta_\infty-\sigma_{0t})^2}{\Gamma_q(1-2\sigma_{0t})\Gamma_q(1-2\theta_\infty)\Gamma_q(1+\theta_1+\theta_\infty+\sigma_{0t})^2}.\label{eq:c}
\end{align}

Next, we apply Lemma \ref{lem:pi} to equation \eqref{eq:factorisationm0}, with $z=q^{-\theta_t}t_0$, to obtain
\begin{equation*}
    \rho_2=\frac{C_{11}^e(q^{-\theta_1})(-t_0)^{2\sigma_{0t}}\rho_2^i+C_{21}^e(q^{-\theta_1})}{C_{12}^e(q^{-\theta_1})(-t_0)^{2\sigma_{0t}}\rho_2^i+C_{22}^e(q^{-\theta_1})}.
\end{equation*}
 By solving this equation for $\rho_2^i$ and substuting the explicit expressions for the coefficients of $C^e(z)$, we obtain
 \begin{equation*}
   \rho_2^i=-(-t_0)^{-2\sigma_{0t}}M_{0t}(\rho_{24})
   q^{2(\theta_\infty-\sigma_{0t})}\frac{s_2^e }{s_1^e} \rho_4 c^{-1},
 \end{equation*}
where $M_{0t}(\cdot)$ is the M\"obius transform in the lemma and the factor $c$ is given in equation \eqref{eq:c}. Combining this equation with equation \eqref{eq:twistintermediate} gives the explicit formula for $s_{0t}$ in the lemma.
\end{proof}

It is important to note that, given any $s_{0t}\in\mathbb{C}^*$ and $\sigma_{0t}\in\mathbb{C}$ satisfying
\begin{equation}\label{eq:dense_open}
    \sigma_{0t}\not\equiv 0,\tfrac{1}{2},\tfrac{\tau}{2},\tfrac{1}{2}+\tfrac{\tau}{2},\pm (\theta_0+\theta_t),\pm (\theta_0-\theta_t),\pm (\theta_\infty+\theta_1),\pm (\theta_\infty-\theta_1) \mod{\Lambda_\tau},
\end{equation}
Proposition \ref{prop:algebraicdecomzero} gives an explicit construction of a corresponding connection matrix, whose Tyurin data $\rho$ are entirely determined by $\mathcal{E}_0(\sigma_{0t})=\rho_{34}$, equation \eqref{eq:twist0t} and equation \eqref{eq:we}. They thus define a unique corresponding point $\eta=\eta(\sigma_{0t},s_{0t})$ on the affine Segre surface $\mathcal{F}(\Theta,t_0)$.

In particular, the pair $\{\sigma_{0t},s_{0t}\}$ form a set of local coordinates on the dense open subset of $\mathcal{F}(\Theta,t_0)$ defined by equation \eqref{eq:dense_open}.
This dense open subset contains all $0$-generic points, see Definition \ref{def:generic0}.

\subsection{Analytic Decomposition I of the main RHP}
\label{sec:dec_analytic}
Take a point $\eta\in\mathcal{F}(\Theta,t_0)$ and a corresponding connection matrix $C(z)\in\mathfrak{C}(\Theta,t_0)$.
Recall the main Riemann-Hilbert problem, RHP \ref{rhp:main}, for the corresponding solution $(f,g)$ of $q\Psix$.

Assuming conditions \eqref{eq:nonlogassumption}, Proposition \ref{prop:algebraicdecomzero} gives an algebraic decomposition of the connection matrix $C(z)$ into interior and exterior connection matrices $C^i(z)$ and $C^e(z)$, as detailed in the proposition. In this section, we describe how RHP \ref{rhp:main} can be decomposed correspondingly.

Let $\Psi(z)=\Psi^{\B{m}}(z)$ denote the solution of RHP \ref{rhp:main} and let $\Psi_\infty^{\B{m}}(z)$ and $\Psi_0^{\B{m}}(z)$ denote respectively the analytic continuation of $\Psi^{\B{m}}(z)$ from the outside $D_+^{\B{m}}$ and inside $D_-^{\B{m}}$ of the curve $\gamma^{\B{m}}$. Then, due to equations \eqref{eq:globaljump} and \eqref{eq:factorisation},
\begin{align*}
\Psi_\infty^{\B{m}}(z)&=\Psi_0^{\B{m}}(z)z^mC(z)\\
    &=\Psi_0^{\B{m}}(z)D_mC^i\left(z/t_m\right)(-t_m)^{\sigma_{0t} \sigma_3}C^e(z).\nonumber
\end{align*}

We now introduce an intermediate matrix function $\Psi_{0t}^{\B{m}}(z)$ by
\begin{align*}
   \Psi_\infty^{\B{m}}(z)&=\Psi_{0t}^{\B{m}}(z)C^e(z),\\
   \Psi_{0t}^{\B{m}}(z)&=\Psi_{0}^{\B{m}}D_m C^i\left(z/t_m\right)(-t_m)^{\sigma_{0t}\sigma_3}.
\end{align*}
Then, in addition to $Y_\infty(z)$ and $Y_0(z)$ defined in equations \eqref{eq:true_sol},
\begin{equation}\label{eq:explicitsolyot}
Y_{0t}(z,t_m):=(-1)^{\log_q(z)}z^{\frac{1}{2}\log_q(z/q)}\Psi_{0t}^{\B{m}}(z) z^{-\sigma_{0t}\sigma_3},
\end{equation}
also defines a solution to the linear system \eqref{eq:linear_system}.

The matrix function $\Psi_\infty^{\B{m}}(z)$ is analytic and invertible on $\mathbb{C}^*$, away from points in the discrete set
\begin{equation*}
    q^{\mathbb{Z}_{>0}}\cdot \{q^{\theta_t} t_m,q^{-\theta_t} t_m, q^{\theta_1},q^{-\theta_1}\},
\end{equation*}
the matrix function $\Psi_0^{\B{m}}(z)$ is analytic and invertible on $\mathbb{C}$, away from points in the discrete set
\begin{equation*}
q^{\mathbb{Z}_{\leq 0}}\cdot \{q^{\theta_t} t_m,q^{-\theta_t} t_m, q^{\theta_1},q^{-\theta_1}\},
\end{equation*}
and the matrix function $\Psi_{0t}^{\B{m}}(z)$ is analytic and invertible on $\mathbb{C}^*$, away from points in the discrete set
\begin{equation*}
\left(q^{\mathbb{Z}_{>0}}\cdot \{q^{\theta_t} t_m,q^{-\theta_t} t_m\}\right)\cup \left(q^{\mathbb{Z}_{\leq 0}}\cdot \{ q^{\theta_1},q^{-\theta_1}\}\right).
\end{equation*}

% \begin{figure}[ht]
% \centering
% \includegraphics[width=12cm]{jump_contour.jpg}
%     \caption{babla}
%     \label{fig:contours}
% \end{figure}

\begin{figure}[ht]
	\centering
	\begin{tikzpicture}[scale=0.8]
	\draw[->] (-6,0)--(6,0) node[right]{$\Re{z}$};
	\draw[->] (0,-6)--(0,6) node[above]{$\Im{z}$};
	\tikzstyle{star}  = [circle, minimum width=3.5pt, fill, inner sep=0pt];
	\tikzstyle{starsmall}  = [circle, minimum width=3.5pt, fill, inner sep=0pt];

	\draw[domain=-1.3:6,smooth,variable=\x,red] plot ({exp(-\x*ln(2))*2*4/3*cos((pi/8) r)},{exp(-\x*ln(2))*2*4/3*sin((pi/8) r)});	
	\draw[domain=-1.3:6,smooth,variable=\x,red] plot ({exp(-\x*ln(2))*2*4/3*cos((pi/8) r)},{-exp(-\x*ln(2))*2*4/3*sin((pi/8) r)});

 \draw[blue,thick,decoration={markings, mark=at position 0.21 with {\arrow{>}}},
	postaction={decorate}] (0,0) ellipse (5.4cm and 5.4cm);

 \draw[blue,thick,decoration={markings, mark=at position 0.21 with {\arrow{>}}},
	postaction={decorate}] (0,0) ellipse (2.75cm and 2.75cm);

 \draw[black,thick,,dashed,decoration={markings, mark=at position 0.21 with {\arrow{>}}},
	postaction={decorate}] (0,0) ellipse (4.5cm and 4.5cm);

    \node[starsmall]     (or) at ({0},{0} ) {};
	\node     at ($(or)+(0.2,0.4)$) {$0$};

    \node[starsmall]     (qk1) at ({sqrt(sqrt(2))*2*3/2*cos((-pi/8) r)},{-sqrt(sqrt(2))*2*3/2*sin((-pi/8) r)} ) {};
	\node     at ($(qk1)+(0.3,-0.4)$) {$q^{1+\theta_1}$};
	\node[star]     (k1) at ({2*2*3/2*cos((-pi/8) r)},{-2*2*3/2*sin((-pi/8) r)} ) {};
	\node     at ($(k1)+(0.3,-0.3)$) {$q^{\theta_1}$};
	
	\node[starsmall]     (qkm1) at ({sqrt(sqrt(2))*2*3/2*cos((-pi/8) r)},{sqrt(sqrt(2))*2*3/2*sin((-pi/8) r)} ) {};
	\node     at ($(qkm1)+(0.3,0.4)$) {$q^{1-\theta_1}$};
	\node[star]     (km1) at ({2*2*3/2*cos((-pi/8) r)},{2*2*3/2*sin((-pi/8) r)} ) {};
	\node     at ($(km1)+(0.4,0.35)$) {$q^{-\theta_1}$};

    \draw[domain=-0.95:6,smooth,variable=\x,red] plot ({exp(-\x*ln(2))*2*5/3*cos((-3.8*pi/8+5*pi/4+\x*0) r)},{-exp(-\x*ln(2))*2*5/3*sin((-3.8*pi/8+5*pi/4+\x*0) r)});	
    \draw[domain=-0.95:6,smooth,variable=\x,red] plot ({exp(-\x*ln(2))*2*5/3*cos((-3.8*pi/8+5*pi/4+\x*0) r)},{exp(-\x*ln(2))*2*5/3*sin((-3.8*pi/8+5*pi/4+\x*0) r)});

	\node[star]     (qkt) at ({sqrt(sqrt(2))*2*4/5*cos((-3.8*pi/8+5*pi/4-1/4*0) r)},{sqrt(sqrt(2))*2*4/5*sin((-3.8*pi/8+5*pi/4-1/4*0) r)} ) {};
	\node     at ($(qkt)+(0.53,0.4)$) {$q^{\theta_t}t_{m+1}$};	
	\node[star]     (kt) at ({2*2*0.83*cos((-3.8*pi/8+5*pi/4-0) r)},{2*2*0.83*sin((-3.8*pi/8+5*pi/4-0) r)} ) {};
	\node     at ($(kt)+(0.4,0.4)$) {$q^{\theta_t}t_{m}$};

	\node[star]     (qktm) at ({sqrt(sqrt(2))*2*4/5*cos((-3.8*pi/8+5*pi/4) r)},{-sqrt(sqrt(2))*2*4/5*sin((-3.8*pi/8+5*pi/4) r)} ) {};
	\node     at ($(qktm)+(0.48,-0.38)$) {$q^{-\theta_t}t_{m+1}$};
	\node[star]     (ktm) at ({2*2*0.83*cos((-3.8*pi/8+5*pi/4) r)},{-2*2*0.83*sin((-3.8*pi/8+5*pi/4) r)} ) {};
	\node     at ($(ktm)+(0.4,-0.35)$) {$q^{-\theta_t}t_{m}$};

 	\node[blue]     at ({5.8*cos((0.21*2*pi) r)},{5.8*sin((0.21*2*pi) r)}) {$\boldsymbol{\gamma_e}$};

 	\node[black]     at ({4.83*cos((0.21*2*pi) r)},{4.83*sin((0.21*2*pi) r)}) {$\boldsymbol{\gamma_{0t}}$};

 	\node[blue]     at ({3.2*cos((0.21*2*pi) r)+0.15},{3.2*sin((0.21*2*pi) r)})  {$\boldsymbol{\gamma_i^{\B{m}}}$};

 	\node[blue]     at ({5.70*cos((0.3*2*pi) r)},{5.70*sin((0.3*2*pi) r)}) {${\scriptstyle\boldsymbol{+}}$};

 	\node[black]     at ({4.77*cos((0.3*2*pi) r)},{4.77*sin((0.3*2*pi) r)}) {${\scriptstyle\boldsymbol{+}}$};

 	\node[blue]     at ({3.05*cos((0.3*2*pi) r)},{3.05*sin((0.3*2*pi) r)})  {${\scriptstyle\boldsymbol{+}}$};

 	\node[blue]     at ({5.14*cos((0.3*2*pi) r)},{5.14*sin((0.3*2*pi) r)}) {${\scriptstyle\boldsymbol{-}}$};

 	\node[black]     at ({4.3*cos((0.3*2*pi) r)},{4.3*sin((0.3*2*pi) r)}) {${\scriptstyle\boldsymbol{-}}$};

 	\node[blue]     at ({2.52*cos((0.3*2*pi) r)},{2.52*sin((0.3*2*pi) r)})  {${\scriptstyle\boldsymbol{-}}$};

	\end{tikzpicture}
	\caption{Topological representation of Jordan curves $\gamma_e$, $\gamma_{0t}$ and $\gamma_i^{\B{m}}$ relative to each other, the origin and points in the $q$-lines $q^{\mathbb{Z}}\cdot x$, $x\in\{q^{\pm \theta_t}t_0,q^{\pm \theta_1}\}$. For the sake of simplicity, the contours are displayed as circles, but we emphasise that they generally are not.}
	\label{fig:analytic_decomp_I}
\end{figure}

To formulate the factorised RHP corresponding to decomposition I, we take three analytic\footnote{From here on, we will impose that curves are analytic so that the Cauchy operator acts as a bounded operator on corresponding $L^2$ spaces.} Jordan curves, $\gamma_e$, $\gamma_{0t}$ and an $m$-dependent curve $\gamma_i^{\B{m}}=q^m\cdot \gamma_i^{(0)}$, as described pictorially in Figure \ref{fig:contours}, for $m\geq 0$.
Namely, $\gamma_i^{\B{m}}$ lies entirely in the inside of $\gamma_{0t}$, $\gamma_{0t}$ lies entirely in the inside of $\gamma_e$, and letting $D_0^{\B{m}}=q^m D_0^{\B{0}}$ denote the inside of $\gamma_i^{\B{m}}$ and $D_\infty$ denote the outside of $\gamma_e$,  and $D_{0t}^{\B{m}}$ denote the annulus-like region in between $\gamma_i^{(m)}$ and $\gamma_e$, then
\begin{equation*}
    0\in D_0^{\B{m}}\subseteq \mathbb{C}\setminus D_\infty,\qquad D_0^{\B{m}}\sqcup\gamma_i^{\B{m}}\sqcup D_{0t}^{\B{m}}\sqcup \gamma_e\sqcup D_\infty=\mathbb{C},\quad \gamma_{0t}\subseteq D_{0t}^{\B{m}},
\end{equation*}
and
\begin{align*}
q^{\mathbb{Z}_{>0}}\cdot \{q^{\theta_t} t_m,q^{-\theta_t} t_m\}&\subseteq D_0^{\B{m}},\\
q^{\mathbb{Z}_{\leq 0}}\cdot \{q^{\theta_t} t_m,q^{-\theta_t} t_m\}&\subseteq \mathbb{C}\setminus D_0^{\B{m}}, \\
q^{\mathbb{Z}_{>0}}\cdot \{q^{\theta_1},q^{-\theta_1}\}&\subseteq \mathbb{C}\setminus D_\infty,\\
q^{\mathbb{Z}_{\leq 0}}\cdot \{q^{\theta_1},q^{-\theta_1}\}&\subseteq D_\infty,
\end{align*}
for all $m\geq 0$.

We now define
\begin{equation*}
    \Phi^{\B{m}}(z)=\begin{cases}
    \Psi_\infty^{\B{m}}(z) & \text{if }z\in D_\infty,\\
    \Psi_{0t}^{\B{m}}(z) & \text{if }z\in D_{0t}^{\B{m}},\\
    \Psi_0^{\B{m}}(z) & \text{if }z\in D_0^{\B{m}},\\
    \end{cases}
\end{equation*}
so that $\Phi^{\B{m}}(z)$ solves the following RHP.

\begin{rhprob}\label{rhp:decomI}
For $m\geq 0$, find a $2\times 2$ matrix-valued function $\Phi(z)=\Phi^{\B{m}}(z)$, which satisfies the following conditions.
  \begin{enumerate}[label={{\rm (\roman *)}}]
  \item $\Phi^{\B{m}}(z)$ is analytic on $\mathbb{C}$ away from the curves $\gamma_e$ and $\gamma_i^{\B{m}}$.
    \item $\Phi^{\B{m}}(z)$ has continuous boundary values $\Phi_+^{\B{m}}(z)$ and $\Phi_-^{\B{m}}(z)$ as $z$ approaches $\gamma_e$ or $\gamma_i^{\B{m}}$ from their respective out and insides, related by
\begin{align*}
\Phi_+^{\B{m}}(z)&=\Phi_-^{\B{m}}(z)C^e(z) & &(z\in \gamma_e),\\
\Phi_+^{\B{m}}(z)&=\Phi_-^{\B{m}}(z)D_m C^i\left(z/t_m\right)(-t_m)^{\sigma_{0t}\sigma_3} &   &(z\in \gamma_i^{\B{m}}).
  \end{align*}  
 \item $\Phi^{\B{m}}(z)$ satisfies
              \begin{equation*}
		\Phi^{\B{m}}(z)=I+\mathcal{O}\left(z^{-1}\right)\quad z\rightarrow \infty.
              \end{equation*}
              \end{enumerate}
\end{rhprob}

\begin{remark}
We note that a RHP, similarly factorised to RHP \ref{rhp:decomI}, is the main RHP studied in Jimbo et al. \cite{jimbonagoyasakai}*{\S 3.1}.
\end{remark}

\subsection{Algebraic decomposition II of the monodromy}\label{sec:dec_algebraicII}
Analogous to Proposition \ref{prop:algebraicdecominfty}, we have the following factorisation of the monodromy, which we call Mano decomposition II, along the decomposition illustrated in Figure \ref{fig:pairspants}(B).

\begin{proposition}\label{prop:algebraicdecominfty}
Let $C(z)\in\mathfrak{C}(\Theta,t_0)$ be a connection matrix with corresponding Tyurin data $\rho$ and recall the notation $\rho_{12}:=\rho_1/\rho_2\in\mathbb{CP}^1$.
If
\begin{equation}\label{eq:nonlogassumptioninfty}
    \rho_{12}\neq \mathcal{E}_\infty(\sigma), \text{ for } \sigma=0,\tfrac{1}{2},\tfrac{\tau}{2},\tfrac{1}{2}+\tfrac{\tau}{2},
\end{equation}
then, for any solution $\sigma_{01}$ of $\mathcal{E}_\infty(\sigma)=\rho_{12}$, we can construct a factorisation of the form,
\begin{equation}\label{eq:factorisationinfty}
    z^m{C}(z)=\widehat{C}^i\left(z\right)(-t_m)^{-\sigma_{01} \sigma_3}\widehat{C}^e(z/t_m)\widehat{D}_m,
\end{equation}
valid for all $m\in\mathbb{Z}$, where 
\begin{itemize}
    \item $\widehat{C}^e(z)$, independent of $m$, is the connection matrix of a Heine hypergeometric system with critical data $\widehat{\Xi}^e$ given by
   \begin{align*}
    \widehat{\sigma}_1^e&=-q^{-\sigma_{01}}, &  \widehat{x}_1^e&=q^{+\theta_t}, &   \widehat{\mu}_1^e&=q^{+\theta_\infty},\\
    \widehat{\sigma}_2^e&=-q^{+\sigma_{01}}, &  \widehat{x}_2^e&=q^{-\theta_t}, & \widehat{\mu}_2^e&=q^{-\theta_\infty}.
    \end{align*}
    \item $\widehat{C}^i(z)$, independent of $m$, is the connection matrix of a Heine hypergeometric system with critical data $\widehat{\Xi}^i$ given by
    \begin{align*}
    \widehat{\sigma}_1^i&=q^{+\theta_0}, &  \widehat{x}_1^i&=q^{+\theta_1}, &   \widehat{\mu}_1^i&=-q^{+\sigma_{01}},\\
    \widehat{\sigma}_2^i&=q^{-\theta_0}, &  \widehat{x}_2^i&=q^{-\theta_1}, & \widehat{\mu}_2^i&=-q^{-\sigma_{01}},
    \end{align*}
    \item $\widehat{D}_m$ is a diagonal matrix, explicitly given by
    \begin{equation*}
        \widehat{D}_m=(-t_0)^m q^{\frac{1}{2}m(m+1)}q^{m\theta_\infty\sigma_3}.
    \end{equation*}
    \end{itemize}
Let $\widetilde{\rho}$ denote the dual Tyurin data of $C(z)$, then we have the identifications
\begin{equation*}
    (\rho_1,\rho_2)=(\rho_1^e,\rho_2^e),\qquad
    (\widetilde{\rho}_3,\widetilde{\rho}_4)=(\widetilde{\rho}_1^i,\widetilde{\rho}_2^i),
\end{equation*}
where $(\rho_1^e,\rho_2^e)$ are the Tyurin data of $C^e(z)$ and $(\widetilde{\rho}_1^i,\widetilde{\rho}_2^i)$ are the dual Tyurin data of $C^i(z)$. In particular, the following hold true.
\begin{itemize}
    \item If $\rho_{1},\rho_2\notin\{0,\infty\}$, then $$\sigma_{01}\not\equiv\pm (\theta_\infty+\theta_t),\pm (\theta_\infty-\theta_t)\mod{\Lambda_\tau},$$ and $\widehat{C}^e(z)$ is irreducible and given by
    $\widehat{C}^e(z)=C_H(z;\widehat{\Xi}^e,\widehat{w}^e,\widehat{s}_1^e,\widehat{s}_2^e)$, for some gauge parameter $\widehat{w}^e$ and scaling parameters $\widehat{s}_{1,2}^e$.\\ Else, $\widehat{C}^e(z)$ is reducible.
    \item If $\widetilde{\rho}_3,\widetilde{\rho}_4\notin\{0,\infty\}$, then $$\sigma_{01}\not\equiv\pm (\theta_0+\theta_1),\pm (\theta_0-\theta_1)\mod{\Lambda_\tau},$$ and $\widehat{C}^i(z)$ is irreducible and given by
    $\widehat{C}^i(z)=C_H(z;\widehat{\Xi}^i,\widehat{w}^i,\widehat{s}_1^i,\widehat{s}_2^i)$,
    for some gauge parameter $\widehat{w}^i$ and scaling parameters $\widehat{s}_{1,2}^i$.\\ Else, $\widehat{C}^i(z)$ is reducible.
\end{itemize}
\end{proposition}
\begin{proof}
   The proposition is proven analogously to Proposition \ref{prop:algebraicdecomzero}.
\end{proof}

Next, we introduce the twist parameter $s_{01}$ corresponding to Mano decomposition II. Take a $\infty$-generic point $\eta\in\mathcal{F}(\Theta,t_0)$, according to Definition \ref{def:genericinfty}. Then, by Proposition \ref{prop:algebraicdecominfty}, we may write the corresponding connection matrix $C(z)$ as
\begin{equation}\label{eq:factorisationinftym0}
    C(z)=\widehat{C}^i\left(z\right)(-t_0)^{-\sigma_{01} \sigma_3}\widehat{C}^e(z/t_0),
\end{equation}
where
\begin{equation*}
\widehat{C}^e(z)=C_H(z;\widehat{\Xi}^e,\widehat{w}^e,\widehat{s}_1^e,\widehat{s}_2^e), \quad \widehat{C}^i(z)=C_H(z;\Xi^i,\widehat{w}^i,\widehat{s}_1^i,\widehat{s}_2^i),
\end{equation*}
for some gauge and scaling parameters $\widehat{w}^e,\widehat{w}^i,\widehat{s}_{1,2}^e,\widehat{s}_{1,2}^i\in\mathbb{C}^*$, with external and internal critical data $\widehat{\Xi}^e$ and $\widehat{\Xi}^i$ as in the proposition. Analogous to equation \eqref{eq:we}, we have the following explicit equation for the gauge parameter $\widehat{w}^e$,
\begin{equation} \label{eq:whe}
    \widehat{w}^e=-\frac{1}{q^{2\theta_\infty}\rho_2} \frac{\Gamma_q(1+2\theta_\infty)\Gamma_q(1+\theta_t-\theta_\infty-\sigma_{01})\Gamma_q(1+\theta_t-\theta_\infty+ \sigma_{01})}{\Gamma_q(1-2\theta_\infty)\Gamma_q(1+\theta_t+\theta_\infty-\sigma_{01})\Gamma_q(1+\theta_t+\theta_\infty+ \sigma_{01})}.
\end{equation}

We define the twist parameter $s_{01}$, corresponding to Mano decomposition II, by the formulas,
\begin{equation}\label{eq:twistdefi01}
    r_{01}:=\frac{\widehat{s}_1^e\widehat{w}^e}{\widehat{s}_2^e\widehat{w}^i},\quad s_{01}:=r_{01}/c_{01},
\end{equation}
where $c_{01}$ is defined by
\begin{equation*}
c_{01}=\frac{\Gamma_q(1-2\sigma_{01})^2}{\Gamma_q(1+2\sigma_{01})^2}\prod_{\epsilon=\pm1}\frac{ \Gamma_q(1+\theta_1+\epsilon\,\theta_0+\sigma_{01})  \Gamma_q(1+\theta_t+\epsilon\,\theta_\infty+\sigma_{01})}{ \Gamma_q(1+\theta_1+\epsilon\,\theta_0-\sigma_{01})  \Gamma_q(1+\theta_t+\epsilon\,\theta_\infty-\sigma_{01})}.
\end{equation*}

The twist parameter is explicitly related to the $\eta$-coordinates in the following lemma.

\begin{lemma}\label{lem:twist_s01}
The twist parameter $s_{01}$ is explicitly related to the $\eta$-coordinates via the formula
\begin{equation}\label{eq:twist01}
s_{01}=-(-t_0)^{2\sigma_{01}}M_{01}(\rho_{24}),
\end{equation}
where $M_{01}(\cdot)$ is the M\"obius transformation
\begin{equation*}
M_{01}(Z)=\frac{\vartheta_\tau(\theta_t+\theta_\infty+\sigma_{01})\theta_q(q^{\theta_1-\theta_\infty+\sigma_{01}}t_0)-Z
\vartheta_\tau(\theta_t-\theta_\infty+\sigma_{01})\theta_q(q^{\theta_1+\theta_\infty+\sigma_{01}}t_0)}{\vartheta_\tau(\theta_t+\theta_\infty-\sigma_{01})\theta_q(q^{\theta_1-\theta_\infty-\sigma_{01}}t_0)-Z
\vartheta_\tau(\theta_t-\theta_\infty-\sigma_{01})\theta_q(q^{\theta_1+\theta_\infty-\sigma_{01}}t_0)},
\end{equation*}
and $\rho_{24}=\rho_{24}(\eta)$ is defined in equation \eqref{eq:def_rhoij}.
\end{lemma}
\begin{proof}
The lemma is proven analogously to Lemma \ref{lem:twist_s0t}.
\end{proof}

We note that, given any $s_{01}\in\mathbb{C}^*$ and $\sigma_{01}\in\mathbb{C}$ satisfying
\begin{equation}\label{eq:dense_open_infty}
    \sigma_{01}\not\equiv 0,\tfrac{1}{2},\tfrac{\tau}{2},\tfrac{1}{2}+\tfrac{\tau}{2},\pm (\theta_0+\theta_1),\pm (\theta_0-\theta_1),\pm (\theta_\infty+\theta_t),\pm (\theta_\infty-\theta_t) \mod{\Lambda_\tau},
\end{equation}
Proposition \ref{prop:algebraicdecominfty} gives an explicit construction of a corresponding connection matrix, whose Tyurin data $\rho$ are entirely determined by $\mathcal{E}_\infty(\sigma_{01})=\rho_{12}$, equation \eqref{eq:twist01} and equation \eqref{eq:whe}. They thus define a unique corresponding point $\eta=\eta(\sigma_{01},s_{01})$ on the affine Segre surface $\mathcal{F}(\Theta,t_0)$.

In particular, the pair $\{\sigma_{01},s_{01}\}$ form a set of local coordinates on the dense open subset of $\mathcal{F}(\Theta,t_0)$ defined by equation \eqref{eq:dense_open_infty}.
This dense open subset contains all $\infty$-generic points, see Definition \ref{def:genericinfty}.

\subsection{Analytic decomposition II of the main RHP}\label{sec:dec_analyticII} In this section, we describe how the main RHP can be decomposed corresponding to Mano decomposition II, analogously to the decomposition described in Section \ref{sec:dec_analytic}.

Take a point $\eta\in\mathcal{F}(\Theta,t_0)$ and a corresponding connection matrix $C(z)\in\mathfrak{C}(\Theta,t_0)$.
Recall the main Riemann-Hilbert problem, RHP \ref{rhp:main}, for the corresponding solution $(f,g)$ of $q\Psix$.

Assuming conditions \eqref{eq:nonlogassumptioninfty}, Proposition \ref{prop:algebraicdecomzero} gives an algebraic decomposition of the connection matrix $C(z)$ into interior and exterior connection matrices $\widehat{C}^i(z)$ and $\widehat{C}^e(z)$, as detailed in the proposition.

Let $\Psi(z)=\Psi^{\B{m}}(z)$ denote the solution of RHP \ref{rhp:main} and let $\Psi_\infty^{\B{m}}(z)$ and $\Psi_0^{\B{m}}(z)$ denote respectively the analytic continuation of $\Psi^{\B{m}}(z)$ from the outside $D_+^{\B{m}}$ and inside $D_-^{\B{m}}$ of the curve $\gamma^{\B{m}}$. Then, due to equations \eqref{eq:globaljump} and \eqref{eq:factorisationinfty},
\begin{equation*}
    \Psi_\infty^{\B{m}}(z)=\Psi_0^{\B{m}}(z)\widehat{C}^i\left(z\right)(-t_m)^{-\sigma_{01} \sigma_3}\widehat{C}^e(z/t_m)\widehat{D}_m.
\end{equation*}

We now introduce an intermediate matrix function $\Psi_{01}^{\B{m}}(z)$ by
\begin{align*}
   \Psi_\infty^{\B{m}}(z)&=\Psi_{01}^{\B{m}}(z)(-t_m)^{-\sigma_{01} \sigma_3}\widehat{C}^e(z/t_m)\widehat{D}_m,\\
   \Psi_{01}^{\B{m}}(z)&=\Psi_{0}^{\B{m}}(z)\widehat{C}^i\left(z\right).
\end{align*}
Then, in addition to $Y_\infty(z)$ and $Y_0(z)$ defined in equations \eqref{eq:true_sol},
\begin{equation*}
Y_{01}(z,t_m):=(-1)^{\log_q(z)}z^{\frac{1}{2}\log_q(z/q)+\log_q(t_m)}\Psi_{01}^{\B{m}}(z) z^{-\sigma_{01}\sigma_3},
\end{equation*}
also defines a solution to the linear system \eqref{eq:linear_system}.

The matrix function $\Psi_{01}^{\B{m}}(z)$ is analytic and invertible on $\mathbb{C}^*$, away from points in the discrete set
\begin{equation*}
\left(q^{\mathbb{Z}_{>0}}\cdot \{ q^{\theta_1},q^{-\theta_1}\} \right)\cup \left(q^{\mathbb{Z}_{\leq 0}}\cdot \{q^{\theta_t} t_m,q^{-\theta_t} t_m\} \right).
\end{equation*}

\begin{figure}[ht]
	\centering
	\begin{tikzpicture}[scale=0.8]
	\draw[->] (-6,0)--(6,0) node[right]{$\Re{z}$};
	\draw[->] (0,-6)--(0,6) node[above]{$\Im{z}$};
	\tikzstyle{star}  = [circle, minimum width=3.5pt, fill, inner sep=0pt];
	\tikzstyle{starsmall}  = [circle, minimum width=3.5pt, fill, inner sep=0pt];

	\draw[domain=-1.3:6,smooth,variable=\x,red] plot ({exp(-\x*ln(2))*2*4/3*cos((pi/8) r)},{exp(-\x*ln(2))*2*4/3*sin((pi/8) r)});	
	\draw[domain=-1.3:6,smooth,variable=\x,red] plot ({exp(-\x*ln(2))*2*4/3*cos((pi/8) r)},{-exp(-\x*ln(2))*2*4/3*sin((pi/8) r)});

 \draw[blue,thick,decoration={markings, mark=at position 0.21 with {\arrow{>}}},
	postaction={decorate}] (0,0) ellipse (5.4cm and 5.4cm);

 \draw[blue,thick,decoration={markings, mark=at position 0.21 with {\arrow{>}}},
	postaction={decorate}] (0,0) ellipse (2.6cm and 2.6cm);

 \draw[black,thick,,dashed,decoration={markings, mark=at position 0.21 with {\arrow{>}}},
	postaction={decorate}] (0,0) ellipse (3.54cm and 3.54cm);

    \node[starsmall]     (or) at ({0},{0} ) {};
	\node     at ($(or)+(0.2,0.4)$) {$0$};

    \node[starsmall]     (qk1) at ({sqrt(sqrt(2))*2*5/6*cos((-pi/8) r)},{-sqrt(sqrt(2))*2*5/6*sin((-pi/8) r)} ) {};
	\node     at ($(qk1)+(0.17,-0.35)$) {$q^{1+\theta_1}$};
	\node[star]     (k1) at ({2*2*6/6*cos((-pi/8) r)},{-2*2*6/6*sin((-pi/8) r)} ) {};
	\node     at ($(k1)+(0.3,-0.3)$) {$q^{\theta_1}$};
	
	\node[starsmall]     (qkm1) at ({sqrt(sqrt(2))*2*5/6*cos((-pi/8) r)},{sqrt(sqrt(2))*2*5/6*sin((-pi/8) r)} ) {};
	\node     at ($(qkm1)+(0.17,0.4)$) {$q^{1-\theta_1}$};
	\node[star]     (km1) at ({2*2*6/6*cos((-pi/8) r)},{2*2*6/6*sin((-pi/8) r)} ) {};
	\node     at ($(km1)+(0.4,0.35)$) {$q^{-\theta_1}$};

    \draw[domain=-0.95:6,smooth,variable=\x,red] plot ({exp(-\x*ln(2))*2*5/3*cos((-3.8*pi/8+5*pi/4+\x*0) r)},{-exp(-\x*ln(2))*2*5/3*sin((-3.8*pi/8+5*pi/4+\x*0) r)});	
    \draw[domain=-0.95:6,smooth,variable=\x,red] plot ({exp(-\x*ln(2))*2*5/3*cos((-3.8*pi/8+5*pi/4+\x*0) r)},{exp(-\x*ln(2))*2*5/3*sin((-3.8*pi/8+5*pi/4+\x*0) r)});

	\node[star]     (qkt) at ({sqrt(sqrt(2))*2*3.5/2*cos((-3.8*pi/8+5*pi/4-1/4*0) r)},{sqrt(sqrt(2))*2*3.5/2*sin((-3.8*pi/8+5*pi/4-1/4*0) r)} ) {};
	\node     at ($(qkt)+(0.53,0.5)$) {$q^{\theta_t}t_{m+1}$};	
	\node[star]     (kt) at ({2*2*3/2*cos((-3.8*pi/8+5*pi/4-0) r)},{2*2*3/2*sin((-3.8*pi/8+5*pi/4-0) r)} ) {};
	\node     at ($(kt)+(0.4,0.4)$) {$q^{\theta_t}t_{m}$};

	\node[star]     (qktm) at ({sqrt(sqrt(2))*2*3.5/2*cos((-3.8*pi/8+5*pi/4) r)},{-sqrt(sqrt(2))*2*3.5/2*sin((-3.8*pi/8+5*pi/4) r)} ) {};
	\node     at ($(qktm)+(0.48,-0.38)$) {$q^{-\theta_t}t_{m+1}$};
	\node[star]     (ktm) at ({2*2*3/2*cos((-3.8*pi/8+5*pi/4) r)},{-2*2*3/2*sin((-3.8*pi/8+5*pi/4) r)} ) {};
	\node     at ($(ktm)+(0.4,-0.35)$) {$q^{-\theta_t}t_{m}$};

 	\node[blue]     at ({5.8*cos((0.21*2*pi) r)+0.15},{5.8*sin((0.21*2*pi) r)}) {$\boldsymbol{\widehat{\gamma}_e^{\B{m}}}$};

 	\node[black]     at ({3.9*cos((0.21*2*pi) r)},{3.9*sin((0.21*2*pi) r)}) {$\boldsymbol{\gamma_{01}}$};

 	\node[blue]     at ({2.96*cos((0.21*2*pi) r)+0.15},{2.96*sin((0.21*2*pi) r)})  {$\boldsymbol{\widehat{\gamma}_i}$};

 	\node[blue]     at ({5.70*cos((0.3*2*pi) r)},{5.70*sin((0.3*2*pi) r)}) {${\scriptstyle\boldsymbol{+}}$};

 	\node[black]     at ({3.85*cos((0.3*2*pi) r)},{3.85*sin((0.3*2*pi) r)}) {${\scriptstyle\boldsymbol{+}}$};

 	\node[blue]     at ({2.85*cos((0.3*2*pi) r)},{2.85*sin((0.3*2*pi) r)})  {${\scriptstyle\boldsymbol{+}}$};

 	\node[blue]     at ({5.14*cos((0.3*2*pi) r)},{5.14*sin((0.3*2*pi) r)}) {${\scriptstyle\boldsymbol{-}}$};

 	\node[black]     at ({3.32*cos((0.3*2*pi) r)},{3.32*sin((0.3*2*pi) r)}) {${\scriptstyle\boldsymbol{-}}$};

 	\node[blue]     at ({2.4*cos((0.3*2*pi) r)},{2.4*sin((0.3*2*pi) r)})  {${\scriptstyle\boldsymbol{-}}$};

	\end{tikzpicture}
	\caption{Topological representation of Jordan curves $\widehat{\gamma}_i$, $\gamma_{01}$ and $\widehat{\gamma}_e^{\B{m}}$ relative to each other, the origin and points in the $q$-lines $q^{\mathbb{Z}}\cdot x$, $x\in\{q^{\pm \theta_t}t_0,q^{\pm \theta_1}\}$. For the sake of simplicity, the contours are displayed as circles, but we emphasise that they generally are not.}
	\label{fig:analytic_decomp_II}
\end{figure}
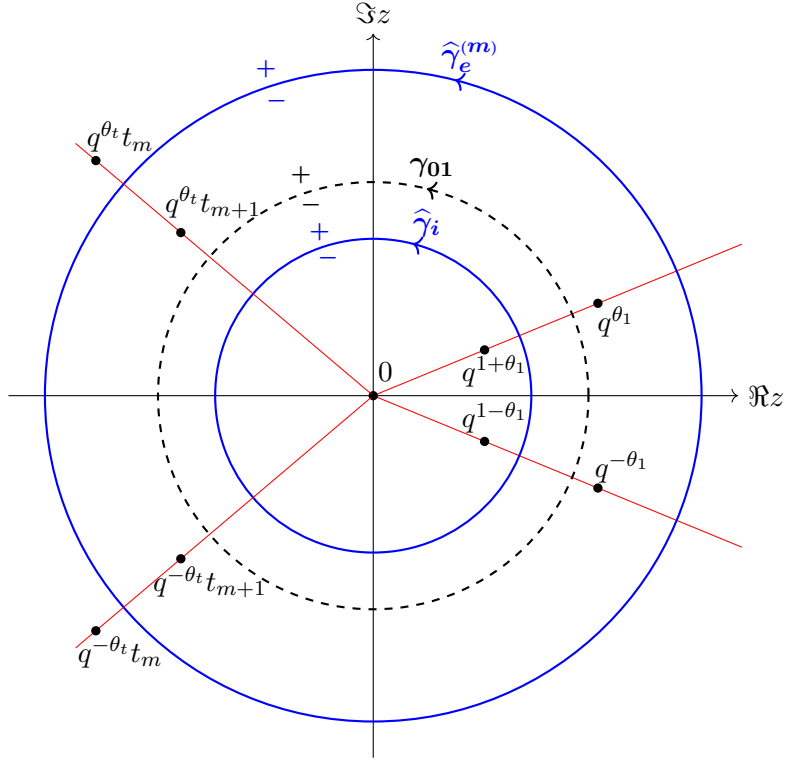

To formulate the factorised RHP corresponding to decomposition II, we take three analytic Jordan curves, $\widehat{\gamma}_i$, $\gamma_{01}$ and $\widehat{\gamma}_e^{\B{m}}=q^m\cdot \widehat{\gamma}_e^{(0)}$, as described pictorially in Figure \ref{fig:analytic_decomp_II}, for $m<0$.
Namely, $\widehat{\gamma}_i$ lies entirely on the inside of $\gamma_{01}$, $\gamma_{01}$ lies entirely in the inside of $\widehat{\gamma}_e^{\B{m}}$, and letting $\widehat{D}_0$ denote the inside of $\widehat{\gamma}_i$,
$\widehat{D}_\infty^{\B{m}}=q^m \widehat{D}_\infty^{\B{0}}$
 denote the outside of $\widehat{\gamma}_e^{\B{m}}$,  and $\widehat{D}_{01}^{\B{m}}$ denote the annulus like region in between $\widehat{\gamma}_i$ and $\widehat{\gamma}_e^{\B{m}}$, then
\begin{equation*}
    0\in \widehat{D}_0,\quad \widehat{D}_\infty^{\B{m}}\subseteq \mathbb{C}\setminus D_0,\qquad \widehat{D}_0\sqcup\widehat{\gamma}_i\sqcup \widehat{D}_{0t}^{\B{m}}\sqcup \widehat{\gamma}_e^{\B{m}}\sqcup \widehat{D}_\infty^{\B{m}}=\mathbb{C},\quad \gamma_{01}\subseteq \widehat{D}_{01}^{\B{m}},
\end{equation*}
and
\begin{align*}
q^{\mathbb{Z}_{>0}}\cdot \{q^{\theta_1},q^{-\theta_1}\}&\subseteq \widehat{D}_0,\\
q^{\mathbb{Z}_{\leq 0}}\cdot \{q^{\theta_1},q^{-\theta_1}\}&\subseteq \mathbb{C}\setminus \widehat{D}_0, \\
q^{\mathbb{Z}_{>0}}\cdot \{q^{\theta_t} t_m,q^{-\theta_t} t_m\}&\subseteq \mathbb{C}\setminus \widehat{D}_\infty^{\B{m}},\\
q^{\mathbb{Z}_{\leq 0}}\cdot \{q^{\theta_t} t_m,q^{-\theta_t} t_m\}&\subseteq \widehat{D}_\infty^{\B{m}},
\end{align*}
for all $m< 0$.

We now define
\begin{equation*}
    \widehat{\Phi}^{\B{m}}(z)=\begin{cases}
    \Psi_\infty^{\B{m}}(z) & \text{if }z\in \widehat{D}_\infty^{\B{m}},\\
    \Psi_{01}^{\B{m}}(z) & \text{if }z\in \widehat{D}_{01}^{\B{m}},\\
    \Psi_0^{\B{m}}(z) & \text{if }z\in \widehat{D}_0,\\
    \end{cases}
\end{equation*}
so that $\widehat{\Phi}^{\B{m}}(z)$ solves the following RHP.

\begin{rhprob}\label{rhp:decomII}
For $m<0$, find a $2\times 2$ matrix-valued function $\widehat{\Phi}(z)=\widehat{\Phi}^{\B{m}}(z)$, which satisfies the following conditions.
  \begin{enumerate}[label={{\rm (\roman *)}}]
  \item $\widehat{\Psi}^{\B{m}}(z)$ is analytic on $\mathbb{C}$ away from the curves $\widehat{\gamma}_e^{\B{m}}$ and $\widehat{\gamma}_i$.
    \item $\widehat{\Phi}^{\B{m}}(z)$ has continuous boundary values $\widehat{\Phi}_+^{\B{m}}(z)$ and $\widehat{\Phi}_-^{\B{m}}(z)$ as $z$ approaches $\widehat{\gamma}_e^{\B{m}}$ or $\widehat{\gamma}_i$ from their respective out and insides, related by
\begin{align*}
\widehat{\Phi}_+^{\B{m}}(z)&=\widehat{\Phi}_-^{\B{m}}(z)(-t_m)^{-\sigma_{01} \sigma_3}\widehat{C}^e(z/t_m)\widehat{D}_m & &(z\in \widehat{\gamma}_e^{\B{m}}),\\
\widehat{\Phi}_+^{\B{m}}(z)&=\widehat{\Phi}_-^{\B{m}}(z)\widehat{C}^i &   &(z\in \widehat{\gamma}_i).
  \end{align*}  
 \item $\widehat{\Phi}^{\B{m}}(z)$ satisfies
              \begin{equation*}
		\widehat{\Phi}^{\B{m}}(z)=I+\mathcal{O}\left(z^{-1}\right)\quad z\rightarrow \infty.
              \end{equation*}
              \end{enumerate}
\end{rhprob}

%% file: rhanalysis.tex
\section{Riemann-Hilbert analysis}\label{sec:rhanalysis}
In this section, we prove Theorem \ref{thm:generic_asymp_zero}, Proposition \ref{prop:asymptotics_lineLd2i} and Corollary \ref{coro:intersectionanalytic}.
We derive asymptotics of solutions near $t=0$ by analysing RHP \ref{rhp:decomI} in the large positive $m$ limit. We first consider the $0$-generic setting, see Definition \ref{def:generic0}. In Section \ref{sec:singlecontourI} we transform RHP \ref{rhp:decomI} to a RHP with a single jump, defined on the $m$-independent, closed contour $\gamma_{0t}$. In Section \ref{sec:rhp_perturb}, we derive a perturbative solution to this RHP for large positive $m$, and in Section \ref{sec:extract_asymp} we extract the asymptotics of the corresponding solution of $q\Psix$ as $t\rightarrow 0$, yielding a proof of Theorem \ref{thm:generic_asymp_zero}.

In Section \ref{sec:reducible_rhp}, we handle the non-generic case, when the monodromy coordinates lie on the line $\widetilde{\mathcal{L}}_2^\infty$, and prove Proposition \ref{prop:asymptotics_lineLd2i} and Corollary \ref{coro:intersectionanalytic}.

\subsection{RHP on a single closed contour}\label{sec:singlecontourI}
Assume that $\eta\in\mathcal{F}(\Theta,t_0)$ is $0$-generic, see Definition \ref{def:generic0}, so that we may specify the intermediate exponent $\sigma_{0t}$ uniquely by
\begin{equation}\label{eq:gen_int_exponent}
\mathcal{E}_0(\sigma_{0t})=\rho_{34},\quad  0< \Re\sigma_{0t}<\tfrac{1}{2},\quad \tfrac{1}{2}i\tau \leq \Im\sigma_{0t} <-\tfrac{1}{2}i\tau.
\end{equation}
By Proposition \ref{prop:algebraicdecomzero} and Lemma \ref{lem:twist_s0t}, the connection matrix $C(z)$ can be factorised as in equation \eqref{eq:factorisationm0}, where
\begin{equation}\label{eq:CeCidefi}
   C^e(z)=C_H(z;\Xi^e,w^e,s_1^e,s_2^e), \quad C^i(z)=C_H(z;\Xi^i,w^i,s_1^i,s_2^i),
\end{equation}
with critical data $\Xi^e$ and $\Xi^i$ as defined in Proposition \ref{prop:algebraicdecomzero}, for some gauge and scaling parameters $w^e,w^i,s_{1,2}^e,s_{1,2}^i\in\mathbb{C}^*$, which satisfy the single constraint
\begin{equation}\label{eq:gen_twist_parameter}
    \frac{s_1^e w^e}{s_2^e w^i}=r_{0t},
\end{equation}
where the twist parameter $s_{0t}=r_{0t}/c_{0t}$ is given by \eqref{eq:twist0t}.

Now, with this decomposition of the global connection matrix, consider RHP \ref{rhp:decomI}. In this section, we transform this RHP into a RHP with a single jump on the closed contour $\gamma_{0t}$. We achieve this by solving the two model RHPs, each defined by only considering one of the two jump conditions in RHP \ref{rhp:decomI}, and then quotienting the solution of RHP \ref{rhp:decomI} by the solutions to the model problems.

We start with the outer contour $\gamma_e$. So, we look for a pair of matrix-valued functions $\Psi_\infty^e(z)$ and $\Psi_0^e(z)$, analytic on the respective out and inside of $\gamma_e$, whose boundary values on $\gamma^e$ are well-defined and related by
\begin{equation*}
   \Psi_\infty^e(z)= \Psi_0^e(z)C^e(z)\quad (z\in \gamma_e),
\end{equation*}
with $\Psi_\infty^e(z)=(I+\mathcal{O}(z^{-1}))$ as $z\rightarrow \infty$.

This is precisely the model RHP \ref{rhp:model}, with $\gamma=\gamma^e$ and $C(z)=C^e(z)$.

We construct its unique solution, which we refer to as the external parametrix, by substituting the parameter values
    \begin{align*}
    \sigma_1^e&=-q^{-\sigma_{0t}}, &  x_1^e&=q^{+\theta_1}, &   \mu_1^e&=q^{+\theta_\infty},\\
    \sigma_2^e&=-q^{+\sigma_{0t}}, &  x_2^e&=q^{-\theta_1}, & \mu_2^e&=q^{-\theta_\infty},
    \end{align*} 
as well as gauge parameter $w^e$ and scaling parameters $s_{1,2}^e$, into equations \eqref{eq:hypersolinf} and \eqref{eq:hypersolzero}. This gives
\begin{align*}
    \Psi_\infty^e(z)&=\widehat{\Psi}_\infty^e(z) \begin{pmatrix} \big(q^{1+\theta_1}/z;q\big)_\infty & 0\\
    0 & \big(q^{1-\theta_1}/z;q\big)_\infty
    \end{pmatrix},\\    \Psi_0^e(z)&=\widehat{\Psi}_0^e(z)\begin{pmatrix} \big(q^{-\theta_1}z;q\big)_\infty^{-1} & 0\\
    0 & \big(q^{+\theta_1}z;q\big)_\infty^{-1}
    \end{pmatrix},
\end{align*}
where, recalling the definition of the function $F_q(\cdot)$ in equation \eqref{eq:defiFq},
\begin{align*}
\widehat{\Psi}_{\infty,{11}}^e(z)&=F_q\left(-\theta_1+\theta_\infty+\sigma_{0t},-\theta_1+\theta_\infty-\sigma_{0t};2 \theta_\infty;
\frac{q^{1+\theta_1}}{z}\right),\\
\widehat{\Psi}_{\infty,{12}}^e(z)&=\frac{r_1^e w^e}{z}F_q\left(1+\theta_1-\theta_\infty+\sigma_{0t},1+\theta_1-\theta_\infty-\sigma_{0t};2-2\theta_\infty;
\frac{q^{1-\theta_1}}{z}\right),\\
\widehat{\Psi}_{\infty,{21}}^e(z)&=\frac{r_2^e}{w^e \;z}F_q\left(1-\theta_1+\theta_\infty+\sigma_{0t},1-\theta_1+\theta_\infty-\sigma_{0t};2+2\theta_\infty;
\frac{q^{1+\theta_1}}{z}\right),\\
\widehat{\Psi}_{\infty,{22}}^e(z)&=F_q\left(\theta_1-\theta_\infty+\sigma_{0t},\theta_1-\theta_\infty-\sigma_{0t};-2\theta_\infty;
\frac{q^{1-\theta_1}}{z}\right),
\end{align*}
with
\begin{equation*}
    r_1^e=\frac{q\;\beta^e}{q^{\theta_\infty}-q^{-\theta_\infty}}, \qquad  r_2^e=\frac{q\;\gamma^e}{q^{-\theta_\infty}-q^{1+\theta_\infty}},
\end{equation*}
and
\begin{align*}
\widehat{\Psi}_{0,{11}}^e(z)&= h_{11}^eF_q\left(-\theta_1+\theta_\infty-\sigma_{0t},1-\theta_1-\theta_\infty-\sigma_{0t};1-2\sigma_{0t};
q^{+\theta_1}z\right),\\
\widehat{\Psi}_{0,{12}}^e(z)&= h_{12}^eF_q\left(+\theta_1+\theta_\infty+\sigma_{0t},1+\theta_1-\theta_\infty+\sigma_{0t};1+2\sigma_{0t};
q^{-\theta_1}z\right),\\
\widehat{\Psi}_{0,{21}}^e(z)&= h_{21}^eF_q\left(-\theta_1-\theta_\infty-\sigma_{0t},1-\theta_1+\theta_\infty-\sigma_{0t};1-2\sigma_{0t};
q^{+\theta_1}z\right),\\
\widehat{\Psi}_{0,{22}}^e(z)&= h_{22}^eF_q\left(+\theta_1-\theta_\infty+\sigma_{0t},1+\theta_1+\theta_\infty+\sigma_{0t};1+2\sigma_{0t};
q^{-\theta_1}z\right),
\end{align*}
with
\begin{equation}\label{eq:he}
h^e=\begin{pmatrix}
1-q^{\theta_1+\theta_\infty+\sigma_{0t}} & w^e(1-q^{\theta_1+\theta_\infty-\sigma_{0t}})\\
1/w^e(1-q^{\theta_1-\theta_\infty+\sigma_{0t}}) & 1-q^{\theta_1-\theta_\infty-\sigma_{0t}}
\end{pmatrix}\begin{pmatrix}1/s_1^e & 0\\
0 & 1/s_2^e\\
\end{pmatrix},
\end{equation}
and
\begin{align}
\beta^e&=\frac{(q^{\theta_1+\theta_\infty+\sigma_{0t}}-1)(q^{\theta_1+\theta_\infty-\sigma_{0t}}-1)}{q^{\theta_1}(q^{+\theta_\infty}-q^{-\theta_\infty})},\label{eq:betae}\\
\gamma^e&=\frac{(q^{\theta_1-\theta_\infty+\sigma_{0t}}-1)(q^{\theta_1-\theta_\infty-\sigma_{0t}}-1)}{q^{\theta_1}(q^{-\theta_\infty}-q^{+\theta_\infty})}.\nonumber
\end{align}

Furthermore, the corresponding linear system,
\begin{equation*}
    Y(qz)=A^e(z)Y(z),\quad A^e(z)=A_0^e+zA_1^e,
\end{equation*}
where
\begin{equation}\label{eq:A0e}
    A_0^e=\begin{pmatrix}
    \alpha^e & \beta^e\, w^e\\
    \gamma^e/w^e & \delta^e\\
    \end{pmatrix},\quad
      A_1=q^{-\theta_\infty \sigma_3},
\end{equation}
with
\begin{align*}
     \alpha^e&=\frac{q^{\theta_1}+q^{-\theta_1}-(q^{-\theta_\infty+\sigma_{0t}}+q^{-\theta_\infty-\sigma_{0t}})}{q^{-\theta_\infty}-q^{+\theta_\infty}},\\
     \delta^e&=\frac{q^{\theta_1}+q^{-\theta_1}-(q^{+\theta_\infty+\sigma_{0t}}+q^{+\theta_\infty-\sigma_{0t}})}{q^{+\theta_\infty}-q^{-\theta_\infty}},
\end{align*}
can be expressed in terms of $\Psi_\infty^e(z)$ and $\Psi_0^e(z)$ by the respective formulas,
\begin{align}\label{eq:Ae_psi}
    A^e(z)&=z\;\Psi_\infty^e(qz)q^{-\theta_\infty \sigma_3}\Psi_\infty^e(z)^{-1},\\
    &=-\Psi_0^e(qz)q^{-\sigma_{0t}\sigma_3}\Psi_0^e(z)^{-1},\nonumber
\end{align}
due to equations \eqref{eq:hypersolinfbranch} and \eqref{eq:hypersolzerobranch}. We refer to the pair $(\Psi_\infty^e(z),\Psi_0^e(z))$ as the exterior parametrix.

Next, we consider the model problem defined by the jump condition on the contour $\gamma_i^{\B{m}}$ in RHP \ref{rhp:decomI}. We first construct analytic matrix functions, $\Psi_\infty^i(\zeta)$ and $\Psi_0^i(\zeta)$, on the respective exterior and interior of the Jordan curve $t_0^{-1}\cdot \gamma_i^{\B{0}}$, whose boundary values on $t_0^{-1}\cdot \gamma_i^{\B{0}}$ are well-defined and related by
\begin{equation*}
   \Psi_\infty^i(\zeta)= \Psi_0^i(\zeta)C^i(\zeta)\qquad (\zeta\in t_0^{-1}\cdot \gamma_i^{\B{0}}),
\end{equation*}
and $\Psi_\infty^i(\zeta)=I+\mathcal{O}(\zeta^{-1})$ as $\zeta\rightarrow \infty$.

This is precisely model RHP \ref{rhp:model}, with $\gamma=t_0^{-1}\cdot\gamma^i$ and $C(z)=C^i(z)$.

As before, we construct its unique solution by substituting the parameter values
    \begin{align*}
    \sigma_1^i&=q^{+\theta_0}, &  x_1^i&=q^{+\theta_t}, &   \mu_1^i&=-q^{+\sigma_{0t}},\\
    \sigma_2^i&=q^{-\theta_0}, &  x_2^i&=q^{-\theta_t}, & \mu_2^i&=-q^{-\sigma_{0t}}.
    \end{align*}
as well as gauge parameter $w^i$ and scaling parameters $s_{1,2}^i$, into equations \eqref{eq:hypersolinf} and \eqref{eq:hypersolzero}. This gives
\begin{align}
    \Psi_\infty^i(\zeta)&=\widehat{\Psi}_\infty^i(\zeta) \begin{pmatrix} \big(q^{1+\theta_t}/\zeta;q\big)_\infty & 0\\
    0 & \big(q^{1-\theta_t}/\zeta;q\big)_\infty
    \end{pmatrix},\label{eq:internalpsiinf}\\    \Psi_0^i(\zeta)&=\widehat{\Psi}_0^i(\zeta)\begin{pmatrix} \big(q^{-\theta_t}\zeta;q\big)_\infty^{-1} & 0\\
    0 & \big(q^{+\theta_t}\zeta;q\big)_\infty^{-1}
    \end{pmatrix},\nonumber
\end{align}
where,
\begin{align*}
\widehat{\Psi}_{\infty,{11}}^i(\zeta)&=F_q\left(-\theta_t+\sigma_{0t}+\theta_0,-\theta_t+\sigma_{0t}-\theta_0;2 \sigma_{0t};
\frac{q^{1+\theta_t}}{\zeta}\right),\\
\widehat{\Psi}_{\infty,{12}}^i(\zeta)&=\frac{r_1^i w^i}{\zeta}F_q\left(1+\theta_t-\sigma_{0t}+\theta_0,1+\theta_t-\sigma_{0t}-\theta_0;2-2\sigma_{0t};
\frac{q^{1-\theta_t}}{\zeta}\right),\\
\widehat{\Psi}_{\infty,{21}}^i(\zeta)&=\frac{r_2^i}{w^i \;\zeta}F_q\left(1-\theta_t+\sigma_{0t}+\theta_0,1-\theta_1+\sigma_{0t}-\theta_0;2+2\sigma_{0t};
\frac{q^{1+\theta_t}}{\zeta}\right),\\
\widehat{\Psi}_{\infty,{22}}^i(\zeta)&=F_q\left(\theta_t-\sigma_{0t}+\theta_0,\theta_t-\sigma_{0t}-\theta_0;-2\sigma_{0t};
\frac{q^{1-\theta_t}}{\zeta}\right),
\end{align*}
with
\begin{equation*}
    r_1^i=\frac{q\;\beta^i}{q^{-\sigma_{0t}}-q^{\sigma_{0t}}}, \qquad  r_2^i=\frac{q\;\gamma^i}{q^{1+\sigma_{0t}}-q^{-\sigma_{0t}}},
\end{equation*}
and
\begin{align*}
\widehat{\Psi}_{0,{11}}^i(\zeta)&= h_{11}^iF_q\left(-\theta_t+\theta_0+\sigma_{0t},1-\theta_t+\theta_0-\sigma_{0t};1+2\theta_0;
q^{+\theta_t}\zeta\right),\\
\widehat{\Psi}_{0,{12}}^i(\zeta)&= h_{12}^iF_q\left(+\theta_t-\theta_0+\sigma_{0t},1+\theta_t-\theta_0-\sigma_{0t};1-2\theta_0;
q^{-\theta_t}\zeta\right),\\
\widehat{\Psi}_{0,{21}}^i(\zeta)&= h_{21}^iF_q\left(-\theta_t+\theta_0-\sigma_{0t},1-\theta_t+\theta_0+\sigma_{0t};1+2\theta_0;
q^{+\theta_t}\zeta\right),\\
\widehat{\Psi}_{0,{22}}^i(\zeta)&= h_{22}^iF_q\left(+\theta_t-\theta_0-\sigma_{0t},1+\theta_t+\sigma_{0t}-\theta_0;1-2\theta_0;
q^{-\theta_t}\zeta\right),
\end{align*}
with
\begin{equation*}
h^i=\begin{pmatrix}
1-q^{\theta_t-\theta_0+\sigma_{0t}} & w^i(1-q^{\theta_t+\theta_0+\sigma_{0t}})\\
1/w^i(1-q^{\theta_t-\theta_0-\sigma_{0t}}) & 1-q^{\theta_t+\theta_0-\sigma_{0t}}
\end{pmatrix}\begin{pmatrix}1/s_1^i & 0\\
0 & 1/s_2^i\\
\end{pmatrix},
\end{equation*}
and
\begin{align}
\beta^i&=\frac{(q^{\theta_t+\theta_0+\sigma_{0t}}-1)(q^{\theta_t-\theta_0+\sigma_{0t}}-1)}{q^{\theta_t}(q^{-\sigma_{0t}}-q^{+\sigma_{0t}})},\label{eq:betai}\\
\gamma^i&=\frac{(q^{\theta_t+\theta_0-\sigma_{0t}}-1)(q^{\theta_t-\theta_0-\sigma_{0t}}-1)}{q^{\theta_t}(q^{+\sigma_{0t}}-q^{-\sigma_{0t}})}.\nonumber
\end{align}

Furthermore, the corresponding linear system,
\begin{equation*}
    Y(q\zeta)=A^i(\zeta)Y(\zeta),\quad A^i(\zeta)=A_0^i+\zeta A_1^i,
\end{equation*}
where
\begin{equation*}
    A_0^i=\begin{pmatrix}
    \alpha^i & \beta^i\, w^i\\
    \gamma^i/w^i & \delta^i\\
    \end{pmatrix},\quad
      A_1=-q^{-\sigma_{0t} \sigma_3},
\end{equation*}
with
\begin{align*}
     \alpha^i&=\frac{q^{\theta_t}+q^{-\theta_t}-(q^{-\sigma_{0t}+\theta_0}+q^{-\sigma_{0t}-\theta_0})}{q^{+\sigma_{0t}}-q^{-\sigma_{0t}}},\\
     \delta^i&=\frac{q^{\theta_t}+q^{-\theta_t}-(q^{+\sigma_{0t}+\theta_0}+q^{+\sigma_{0t}-\theta_0})}{q^{-\sigma_{0t}}-q^{+\sigma_{0t}}},
\end{align*}
can be expressed in terms of $\Psi_\infty^i(\zeta)$ and $\Psi_0^i(\zeta)$ by the respective formulas,
\begin{align}
    A^i(\zeta)&=\zeta\;\Psi_\infty^i(q\zeta)q^{-\sigma_{0t} \sigma_3}\Psi_\infty^i(\zeta)^{-1},\nonumber\\
&=\Psi_0^i(q\zeta)q^{\theta_0\sigma_3}\Psi_0^i(\zeta)^{-1},\label{eq:aipsi0}
\end{align}
due to equations \eqref{eq:hypersolinfbranch} and \eqref{eq:hypersolzerobranch}.

We now set $\zeta=z/t_m$, so that
\begin{equation*}
   \Psi_\infty^i(z/t_m)= \Psi_0^i(z/t_m)C^i(z/t_m).
\end{equation*}
The jump matrix corresponding to the jump on $\gamma_i^{\B{m}}$ in RHP \ref{rhp:decomI}, is given by $D_mC^i(z/t_m) (-t_m)^{\sigma_{0t}\sigma_3}$, and we thus define
\begin{align*}
    U_\infty^{\B{m}}(z)&=(-t_m)^{-\sigma_{0t}\sigma_3}\Psi_\infty^i(z/t_m)(-t_m)^{\sigma_{0t}\sigma_3},\\
    U_0^{\B{m}}(z)&=(-t_m)^{-\sigma_{0t}\sigma_3}\Psi_0^i(z/t_m)D_m^{-1},
\end{align*}
so that
\begin{equation*}
  U_\infty^{\B{m}}(z)=U_0^{\B{m}}(z) D_mC^i(z/t_m) (-t_m)^{\sigma_{0t}\sigma_3},
\end{equation*}
and
\begin{equation*}
    U_\infty^{\B{m}}(z)=(I+\mathcal{O}(z^{-1}))\qquad (z\rightarrow \infty).
\end{equation*}
Furthermore, note that $U_\infty^{\B{m}}(z)$ and $U_0^{\B{m}}(z)$ are analytic and invertible on the respective exterior and interior of $\gamma_i^{\B{m}}$.
We refer to the pair $(U_\infty^{\B{m}}(z),U_0^{\B{m}}(z))$ as the interior parametrix.

We now return to RHP \ref{rhp:decomI}. We have constructed exterior and interior parametrices which solve respectively the jump conditions on $\gamma_e$ and $\gamma_i^{\B{m}}$ in RHP \ref{rhp:decomI}. We proceed to quotient the solution of RHP \ref{rhp:decomI} with these parametrices in different regions of the complex plane.

Let us write $D_{\text{ex}}(\gamma_{0t})$ and $D_{\text{in}}(\gamma_{0t})$ for the respective exterior and interior of the Jordan curve $\gamma_{0t}$. Since $\gamma_{0t}$ lies inside $D_{0t}^{\B{m}}$, it splits this annulus-like region into two, $D_{0t}^{\B{m}}\cap D_{\text{ex}}(\gamma_{0t})$ and $D_{0t}^{\B{m}}\cap D_{\text{in}}(\gamma_{0t})$,
for $m\geq 0$.

We now define quotient of the solution $\Phi^{\B{m}}(z)$ of RHP \ref{rhp:decomI}, with respect to the exterior and interior parametrices, as follows,
\begin{equation}\label{eq:omegapsi}
    \Omega^{\B{m}}(z)=\begin{cases}
    \Phi^{\B{m}}(z)\Psi_\infty^e(z)^{-1} & \text{if }z\in D_\infty,\\
    \Phi^{\B{m}}(z)\Psi_0^e(z)^{-1} & \text{if }z\in D_{0t}^{\B{m}}\cap D_{\text{ex}}(\gamma_{0t}),\\
    \Phi^{\B{m}}(z)U_\infty^{\B{m}}(z)^{-1}\Psi_0^e(z)^{-1} & \text{if }z\in D_{0t}^{\B{m}}\cap D_{\text{in}}(\gamma_{0t}),\\
    \Phi^{\B{m}}(z)U_0^{\B{m}}(z)^{-1}\Psi_0^e(z)^{-1} & \text{if }z\in D_0^{\B{m}}.
    \end{cases}
\end{equation}
It is clear that $\Omega^{\B{m}}(z)$ is analytic on $\mathbb{C}$ away from the Jordan curves $\gamma_e$, $\gamma_{0t}$ and $\gamma_i^{\B{m}}$, since each of the quotients in the above formula only involves matrix functions which are analytic and invertible on the relevant domains.

Furthermore, the quotients are chosen exactly such that the jumps of $\Omega^{\B{m}}(z)$ are trivial on $\gamma_e$ and $\gamma_i^{\B{m}}$. Indeed, for $z\in\gamma_e$,
\begin{align*}
  \Omega_+^{\B{m}}(z)&=\Phi_+^{\B{m}}(z)\Psi_\infty^e(z)^{-1} =\Phi_-^{\B{m}}(z)C^e(z)\Psi_\infty^e(z)^{-1}\\
  &=\Phi_-^{\B{m}}(z)\Psi_0^e(z)^{-1}=\Omega_-^{\B{m}}(z),
\end{align*}
and for $z\in\gamma_i^{\B{m}}$,
\begin{align*}
  \Omega_+^{\B{m}}(z)&=\Phi_+^{\B{m}}(z)U_\infty^{\B{m}}(z)^{-1}\Psi_0^e(z)^{-1}\\
  &=
  \Phi_-^{\B{m}}(z)D_mC^i(z/t_m) (-t_m)^{\sigma_{0t}\sigma_3}U_\infty^{\B{m}}(z)^{-1}\Psi_0^e(z)^{-1},\\
  &=\Phi_-^{\B{m}}(z)U_0^{\B{m}}(z)^{-1}\Psi_0^e(z)^{-1}=\Omega_-^{\B{m}}(z).
\end{align*}
It follows that $\Omega^{\B{m}}(z)$ is analytic on the full exterior and interior of $\gamma_{0t}$.

We compute the jump matrix, for the remaining jump on $\gamma_{0t}$, as follows. For $z\in \gamma_{0t}$,
\begin{align}
    J^{\B{m}}(z):&=\Omega_-^{\B{m}}(z)^{-1}\Omega_+^{\B{m}}(z)\nonumber\\
    &=\left(\Phi^{\B{m}}(z)U_\infty^{\B{m}}(z)^{-1}\Psi_0^e(z)^{-1}\right)^{-1} \Phi^{\B{m}}(z)\Psi_0^e(z)^{-1}\nonumber\\
    &=\Psi_0^e(z)U_\infty^{\B{m}}(z) \Psi_0^e(z)^{-1}\nonumber\\
    &=\Psi_0^e(z)(-t_m)^{-\sigma_{0t}\sigma_3}\Psi_\infty^i(z/t_m)(-t_m)^{\sigma_{0t}\sigma_3} \Psi_0^e(z)^{-1}.\label{eq:circlejumpmatrix}
\end{align}

Note, furthermore, that the asymptotic condition at $z=\infty$ in RHP \ref{rhp:decomI} is not altered by the quotienting, and thus $\Omega^{\B{m}}(z)$ is the unique solution of the following RHP.

\begin{figure}[ht]
	\centering
	\begin{tikzpicture}[scale=0.8]
	\draw[->] (-6,0)--(6,0) node[right]{$\Re{z}$};
	\draw[->] (0,-6)--(0,6) node[above]{$\Im{z}$};
	\tikzstyle{star}  = [circle, minimum width=3.5pt, fill, inner sep=0pt];
	\tikzstyle{starsmall}  = [circle, minimum width=3.5pt, fill, inner sep=0pt];
	\tikzstyle{dot}  = [circle, minimum width=2.5pt, fill, inner sep=0pt];

	\draw[domain=-1.3:6,smooth,variable=\x,red] plot ({exp(-\x*ln(2))*2*4/3*cos((pi/8) r)},{exp(-\x*ln(2))*2*4/3*sin((pi/8) r)});	
	\draw[domain=-1.3:6,smooth,variable=\x,red] plot ({exp(-\x*ln(2))*2*4/3*cos((pi/8) r)},{-exp(-\x*ln(2))*2*4/3*sin((pi/8) r)});

 \draw[black,thick,,dashed,decoration={markings, mark=at position 0.21 with {\arrow{>}}},
	postaction={decorate}] (0,0) ellipse (3.7cm and 3.7cm);

    \node[starsmall]     (or) at ({0},{0} ) {};
	\node     at ($(or)+(0.2,0.4)$) {$0$};

	\node[star]     (k1) at ({4.4*cos((-pi/8) r)},{-4.4*sin((-pi/8) r)} ) {};
	\node     at ($(k1)+(0.3,-0.3)$) {$q^{\theta_1}$};

	\node[star]     (km1) at ({4.4*cos((-pi/8) r)},{4.4*sin((-pi/8) r)} ) {};
	\node     at ($(km1)+(0.4,0.35)$) {$q^{-\theta_1}$};

	\node[star]     (qik1) at ({1.3*4.4*cos((-pi/8) r)},{-1.3*4.4*sin((-pi/8) r)} ) {};
	\node     at ($(qik1)+(0.3,-0.3)$) {$q^{-1+\theta_1}$};

	\node[star]     (qikm1) at ({1.3*4.4*cos((-pi/8) r)},{1.3*4.4*sin((-pi/8) r)} ) {};
	\node     at ($(qikm1)+(0.4,0.35)$) {$q^{-1-\theta_1}$};

	% \node[dot]     at ({(1+0.05)*1.3*4.4*cos((-pi/8) r)},{-(1+0.05)*1.3*4.4*sin((-pi/8) r)} ) {};

 % 	\node[dot]     at ({(1+0.085)*1.3*4.4*cos((-pi/8) r)},{-(1+0.09)*1.3*4.4*sin((-pi/8) r)} ) {};

 % 	\node[dot]     at ({(1+0.12)*1.3*4.4*cos((-pi/8) r)},{-(1+0.12)*1.3*4.4*sin((-pi/8) r)} ) {};

	% \node[dot]     at ({(1+0.05)*1.3*4.4*cos((-pi/8) r)},{(1+0.05)*1.3*4.4*sin((-pi/8) r)} ) {};

 % 	\node[dot]     at ({(1+0.085)*1.3*4.4*cos((-pi/8) r)},{(1+0.09)*1.3*4.4*sin((-pi/8) r)} ) {};

 % 	\node[dot]     at ({(1+0.12)*1.3*4.4*cos((-pi/8) r)},{(1+0.12)*1.3*4.4*sin((-pi/8) r)} ) {};

    \draw[domain=-0.95:6,smooth,variable=\x,red] plot ({exp(-\x*ln(2))*2*5/3*cos((-3.8*pi/8+5*pi/4+\x*0) r)},{-exp(-\x*ln(2))*2*5/3*sin((-3.8*pi/8+5*pi/4+\x*0) r)});	
    \draw[domain=-0.95:6,smooth,variable=\x,red] plot ({exp(-\x*ln(2))*2*5/3*cos((-3.8*pi/8+5*pi/4+\x*0) r)},{exp(-\x*ln(2))*2*5/3*sin((-3.8*pi/8+5*pi/4+\x*0) r)});

	\node[star]     (qkt) at ({0.9*sqrt(sqrt(2))*2*4/5*cos((-3.8*pi/8+5*pi/4-1/4*0) r)},{0.9*sqrt(sqrt(2))*2*4/5*sin((-3.8*pi/8+5*pi/4-1/4*0) r)} ) {};
	\node     at ($(qkt)+(0.53,0.4)$) {$q^{\theta_t}t_{2}$};	
	\node[star]     (kt) at ({0.9*2*2*0.83*cos((-3.8*pi/8+5*pi/4-0) r)},{0.9*2*2*0.83*sin((-3.8*pi/8+5*pi/4-0) r)} ) {};
	\node     at ($(kt)+(0.4,0.4)$) {$q^{\theta_t}t_{1}$};

	\node[star]     (qktm) at ({0.9*sqrt(sqrt(2))*2*4/5*cos((-3.8*pi/8+5*pi/4) r)},{-0.9*sqrt(sqrt(2))*2*4/5*sin((-3.8*pi/8+5*pi/4) r)} ) {};
	\node     at ($(qktm)+(0.48,-0.38)$) {$q^{-\theta_t}t_{2}$};
	\node[star]     (ktm) at ({0.9*2*2*0.83*cos((-3.8*pi/8+5*pi/4) r)},{-0.9*2*2*0.83*sin((-3.8*pi/8+5*pi/4) r)} ) {};
	\node     at ($(ktm)+(0.4,-0.35)$) {$q^{-\theta_t}t_{1}$};

	% \node[dot]     at ({(1-0.15)*0.9*sqrt(sqrt(2))*2*4/5*cos((-3.8*pi/8+5*pi/4) r)},{-(1-0.15)*0.9*sqrt(sqrt(2))*2*4/5*sin((-3.8*pi/8+5*pi/4) r)} ) {};

	% \node[dot]     at ({(1-0.25)*0.9*sqrt(sqrt(2))*2*4/5*cos((-3.8*pi/8+5*pi/4) r)},{-(1-0.25)*0.9*sqrt(sqrt(2))*2*4/5*sin((-3.8*pi/8+5*pi/4) r)} ) {};

	% \node[dot]     at ({(1-0.35)*0.9*sqrt(sqrt(2))*2*4/5*cos((-3.8*pi/8+5*pi/4) r)},{-(1-0.35)*0.9*sqrt(sqrt(2))*2*4/5*sin((-3.8*pi/8+5*pi/4) r)} ) {};

	% \node[dot]     at ({(1-0.15)*0.9*sqrt(sqrt(2))*2*4/5*cos((-3.8*pi/8+5*pi/4) r)},{(1-0.15)*0.9*sqrt(sqrt(2))*2*4/5*sin((-3.8*pi/8+5*pi/4) r)} ) {};

	% \node[dot]     at ({(1-0.25)*0.9*sqrt(sqrt(2))*2*4/5*cos((-3.8*pi/8+5*pi/4) r)},{(1-0.25)*0.9*sqrt(sqrt(2))*2*4/5*sin((-3.8*pi/8+5*pi/4) r)} ) {};

	% \node[dot]     at ({(1-0.35)*0.9*sqrt(sqrt(2))*2*4/5*cos((-3.8*pi/8+5*pi/4) r)},{(1-0.35)*0.9*sqrt(sqrt(2))*2*4/5*sin((-3.8*pi/8+5*pi/4) r)} ) {};

 	\node[black]     at ({4.07*cos((0.21*2*pi) r)+0.1},{4.07*sin((0.21*2*pi) r)}) {$\boldsymbol{\gamma_{0t}}$};

 	\node[black]     at ({4.06*cos((0.3*2*pi) r)},{4.06*sin((0.3*2*pi) r)}) {${\scriptstyle\boldsymbol{+}}$};

 	\node[black]     at ({3.45*cos((0.3*2*pi) r)},{3.45*sin((0.3*2*pi) r)}) {${\scriptstyle\boldsymbol{-}}$};

	\end{tikzpicture}
	\caption{Topological representation of Jordan curve $\gamma_{0t}$ with respect to the origin and points in the half $q$-lines 
 $q^{\mathbb{Z}_{\leq 0}} \cdot q^{\pm \theta_1}$ and
 $q^{\mathbb{Z}_{<0}} \cdot q^{\pm \theta_t}t_0$, with $t_m:=q^m t_0$.
 For the sake of simplicity, the contour is displayed as a circle, but we emphasise that it need not be one.}
	\label{fig:single_I}
\end{figure}
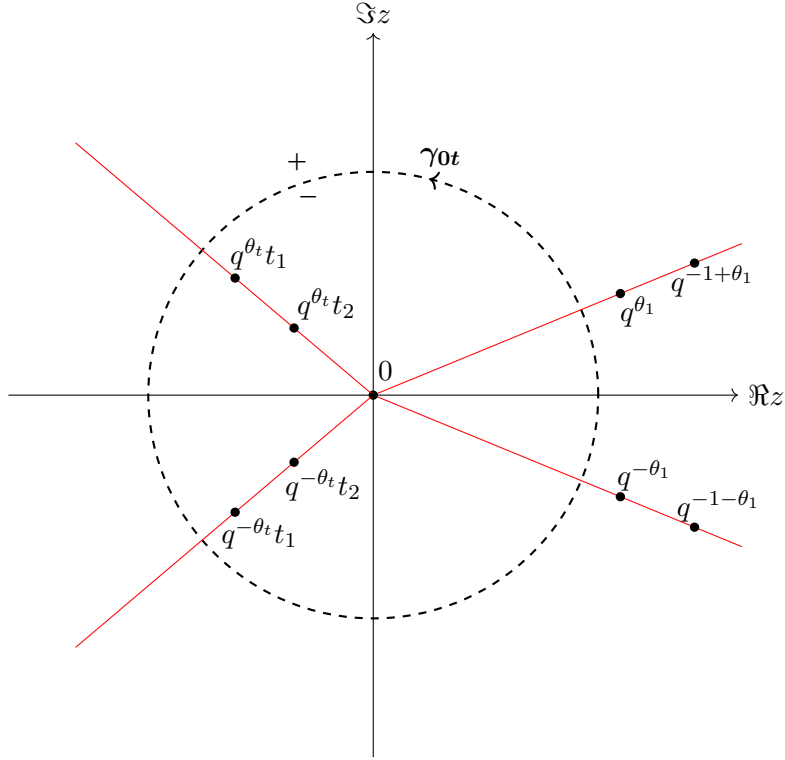

\begin{rhprob}\label{rhp:circleI}
For $m\geq 0$, find a matrix-valued function $\Omega(z)=\Omega^{\B{m}}(z)$, which satisfies the following conditions.
  \begin{enumerate}[label={{\rm (\roman *)}}]
  \item $\Omega^{\B{m}}(z)$ is analytic on $\mathbb{C}$ away from the Jordan curve $\gamma_{0t}$.
    \item $\Omega^{\B{m}}(z)$ has continuous boundary values $\Omega_+^{\B{m}}(z)$ and  $\Omega_-^{\B{m}}(z)$ as $z$ approaches $\gamma_{0t}$ from its respective exterior and interior, related by
\begin{equation*}
\Omega_+^{\B{m}}(z)=\Omega_-^{\B{m}}(z)J^{\B{m}}(z) \qquad (z\in \gamma_{0t}),
  \end{equation*}  
  where the jump matrix $J^{\B{m}}(z)$ is given in equation \eqref{eq:circlejumpmatrix}.
 \item $\Omega^{\B{m}}(z)$ satisfies
              \begin{equation*}
		\Omega^{\B{m}}(z)=I+\mathcal{O}\left(z^{-1}\right)\quad z\rightarrow \infty.
              \end{equation*}
              \end{enumerate}
\end{rhprob}
\begin{remark}\label{rem:gamma0t}
We note that, in the definition of the analytic Jordan curve $\gamma_{0t}$ in Section \ref{sec:dec_analytic}, the only topological conditions on $\gamma_{0t}$ are that the origin lies in the interior $D_{\text{in}}(\gamma_{0t})$ of $\gamma_{0t}$ and
\begin{align*}
q^{\mathbb{Z}_{>0}}\cdot \{q^{\theta_t}t_0,q^{-\theta_t}t_0\}&\subseteq D_{\text{in}}(\gamma_{0t}), \\ 
  q^{\mathbb{Z}_{\leq 0}}\cdot \{q^{\theta_1},q^{-\theta_1}\}&\subseteq D_{\text{ex}}(\gamma_{0t}),
\end{align*}
see Figure \ref{fig:single_I}. One can also consider RHP \ref{rhp:circleI} for negative $m$, say $n\leq m<0$, for some $n<0$, as long as one assures that 
\begin{equation*}
    q^{\mathbb{Z}_{>n}}\cdot \{q^{\theta_t}t_0,q^{-\theta_t}t_0\}\subseteq D_{\text{in}}(\gamma_{0t}).
\end{equation*}
\end{remark}

\subsection{A perturbative solution to RHP \ref{rhp:circleI}}\label{sec:rhp_perturb}
In this section, we will show that, for $m$ positive, large enough, RHP \ref{rhp:circleI} is solvable and its solution is given by a convergent, asymptotic expansion in $t_m=q^mt_0$ as $m\rightarrow+\infty$.

Consider the jump matrix $J^{\B{m}}(z)$. For $z$ in the compact curve $\gamma_{0t}$, we know that
\begin{equation*}
    \Psi_\infty^i(z/t_m)=I+\mathcal{O}(t_m),
\end{equation*}
as $m\rightarrow +\infty$. Recalling that  $0< \Re\sigma<\tfrac{1}{2}$, we thus find
\begin{align*}
    J^{\B{m}}(z)&=\Psi_0^e(z)(-t_m)^{-\sigma_{0t}\sigma_3}\Psi_\infty^i(z/t_m)(-t_m)^{\sigma_{0t}\sigma_3} \Psi_0^e(z)^{-1}\\
    &=\Psi_0^e(z)(-t_m)^{-\sigma_{0t}\sigma_3}((I+\mathcal{O}(t_m))(-t_m)^{\sigma_{0t}\sigma_3}
    \Psi_0^e(z)^{-1}\\
    &=\Psi_0^e(z)((I+\mathcal{O}((-t_m)^{1-2\sigma_{0t}}))
    \Psi_0^e(z)^{-1}\\
    &=I+\mathcal{O}((-t_m)^{1-2\sigma_{0t}}),    
\end{align*}
as $m\rightarrow +\infty$, uniform in $z\in\gamma_{0t}$. In other words, the jump matrix converges uniformly to the identity matrix as $m\rightarrow +\infty$. Using standard Riemann-Hilbert machinery, we will construct a corresponding perturbative solution $\Omega^{\B{m}}(z)$ around the identity to RHP \ref{rhp:circleI}.

Firstly, we give a complete, convergent series expansion of the jump matrix as $m\rightarrow +\infty$, in the following lemma.
\begin{lemma}\label{lem:jumpperturb} The jump matrix $J^{\B{m}}(z)$ admits a complete asymptotic expansion as $m\rightarrow +\infty$,
\begin{equation*}
    J^{\B{m}}(z)=I+\sum_{n=1}^\infty{r_{0t}^{-1}(-t_m)^{n-2\sigma_{0t}}R_n^-(z)+(-t_m)^{n}R_n^0(z)+r_{0t}^{+1}(-t_m)^{n+2\sigma_{0t}}R_n^+(z)},
\end{equation*}
which is absolutely convergent in the supremum norm on $\gamma_{0t}$, for $m\geq 0$ large enough, with coefficients $R_n^\diamond(z)$ that are analytic in a neighbourhood of $\gamma_{0t}$, for $n\geq 1$ and $\diamond\in\{-,0,+\}$. Furthermore, the coefficients only depend on the parameters $\Theta$, $\sigma_{0t}$ and $w^e$, with
\begin{equation*}
    R_n^\diamond(z)=(w^e)^{\frac{1}{2}\sigma_3}\widetilde{R}_n^\diamond(z)(w^e)^{-\frac{1}{2}\sigma_3},\quad
    \widetilde{R}_n^\diamond(z)=\widetilde{R}_n^\diamond(z;\Theta,\sigma_{0t}),
\end{equation*}
for $n\geq 1$ and $\diamond\in\{-,0,+\}$.
\end{lemma}
\begin{proof}
We start by considering the parameter dependence of the different factors in equation \eqref{eq:circlejumpmatrix}. we have
\begin{align*}
    \Psi_\infty^i(\zeta)&=(w^i)^{\frac{1}{2}\sigma_3}\widetilde{\Psi}_\infty^i(\zeta)
    (w^i)^{-\frac{1}{2}\sigma_3},  & \widetilde{\Psi}_\infty(\zeta)&=\widetilde{\Psi}_\infty(\zeta;\Theta,\sigma_{0t}),\\
    \Psi_0^e(z)&=(w^e)^{\frac{1}{2}\sigma_3} \widetilde{\Psi}_0^e(z)(w^e)^{-\frac{1}{2}\sigma_3}\begin{pmatrix}
        s_1^e & 0\\
        0 & s_2^e
    \end{pmatrix}^{-1},  &
\widetilde{\Psi}_0^e(z)&=\widetilde{\Psi}_0^e(z;\Theta,\sigma_{0t}).
\end{align*}
Recalling the definition of $r_{0t}$ in equation \eqref{eq:twistdefi}, it follows that
\begin{align}
    J^{\B{m}}(z)=&(w^e)^{\frac{1}{2}\sigma_3} \widetilde{\Psi}_0^e(z)(w^e)^{-\frac{1}{2}\sigma_3}\begin{pmatrix}
        s_1^e & 0\\
        0 & s_2^e
    \end{pmatrix}^{-1}(-t_m)^{\sigma_{0t}\sigma_3} 
    (w^i)^{\frac{1}{2}\sigma_3}\nonumber\\
    &\widetilde{\Psi}_\infty^i(z/t_m)
    (w^i)^{-\frac{1}{2}\sigma_3}
    (-t_m)^{\sigma_{0t}\sigma_3}
    \begin{pmatrix}
        s_1^e & 0\\
        0 & s_2^e
    \end{pmatrix}(w^e)^{\frac{1}{2}\sigma_3}
    \widetilde{\Psi}_0^e(z)^{-1}(w^e)^{-\frac{1}{2}\sigma_3},\nonumber\\
    =&(w^e)^{\frac{1}{2}\sigma_3} \widetilde{\Psi}_0^e(z)
    \left[r_{0t}(-t_m)^{2\sigma_{0t}}\right]^{-\frac{1}{2}\sigma_3}
    \widetilde{\Psi}_\infty^i(z/t_m)\nonumber\\
    &    \left[r_{0t}(-t_m)^{2\sigma_{0t}}\right]^{\frac{1}{2}\sigma_3}
    \widetilde{\Psi}_0^e(z)^{-1}(w^e)^{-\frac{1}{2}\sigma_3}.\label{eq:circle_jump_parameters}
\end{align}

Now, from the definition of $\Psi_\infty^i(\zeta)$ in equation \eqref{eq:internalpsiinf}, it follows that $\widetilde{\Psi}_\infty^i(\zeta)$ has a convergent power series expansion around $\zeta=\infty$,
\begin{equation}\label{eq:psi_power_series}
    \widetilde{\Psi}_\infty^i(\zeta)=I+\sum_{n=1}^\infty \zeta^{-k}\widetilde{u}_n,
\end{equation}
with coefficient matrices $\widetilde{u}_n=(\widetilde{u}_{n,ij})_{1\leq i,j\leq 2}$ that can be written explicitly in terms of $\theta_0,\theta_t,\sigma_{0t}$. By substituting this power series, with $\zeta=z/t_m$, in the expression \eqref{eq:circle_jump_parameters} for the jump matrix, we obtain the series expansion in the lemma, with
\begin{align*}
R_n^-(z)&=z^{-n}(w^e)^{\frac{1}{2}\sigma_3}\widetilde{\Psi}_0^e(z)\begin{pmatrix} 
    \hspace{2.8mm}0 & \hspace{2.8mm}\widetilde{u}_{n,12}\\
    \hspace{2.8mm}0 & \hspace{2.8mm}0
\end{pmatrix}
\widetilde{\Psi}_0^e(z)^{-1}(w^e)^{-\frac{1}{2}\sigma_3},\\
R_n^0(z)&=z^{-n}(w^e)^{\frac{1}{2}\sigma_3}\widetilde{\Psi}_0^e(z)\begin{pmatrix} 
    \widetilde{u}_{n,11} & 0\\
    0 & \widetilde{u}_{n,22}
\end{pmatrix}
\widetilde{\Psi}_0^e(z)^{-1}(w^e)^{-\frac{1}{2}\sigma_3},\\
R_n^+(z)&=z^{-n}(w^e)^{\frac{1}{2}\sigma_3}\widetilde{\Psi}_0^e(z)\begin{pmatrix} 
    0 & \hspace{3mm}0\hspace{1.2mm}{\color{white} 0}\\
    \widetilde{u}_{n,21} & \hspace{3mm}0\hspace{1.2mm}{\color{white} 0}
\end{pmatrix}
\widetilde{\Psi}_0^e(z)^{-1}(w^e)^{-\frac{1}{2}\sigma_3}.
\end{align*}
Let $M\geq 0$ be such that, for all $m\geq M$, $\sup_{z\in\gamma_{0t}}|t_m/z|$ is strictly bounded by the radius of convergence of the power series \eqref{eq:psi_power_series}, then the series expansion in the lemma is absolutely convergent in the supremum norm on $\gamma_{0t}$, for all $m\geq M$.
\end{proof}

The lemma above gives the complete perturbation of the jump matrix $J^{\B{m}}(z)$ around the identity matrix, with respect to the small parameter $t_m$. In the following proposition, we prove a similar result for the solution of RHP \ref{rhp:circleI}.

\begin{proposition}\label{prop:omega_perturb}
There exists an $M\geq 0$, such that the solution $\Omega^{\B{m}}(z)$ of RHP \ref{rhp:circleI} exists for all $m\geq M$ and has a complete asymptotic expansion as $m\rightarrow \infty$,
\begin{equation}\label{eq:omega_perturb}
\Omega^{\B{m}}(z)=I+\sum_{n=1}^\infty\sum_{k=-n}^n r_{0t}^k (-t_m)^{n+2k\sigma_{0t}}\Omega_{n,k}(z),
\end{equation}
which is uniformly absolutely convergent in $z\in\mathbb{CP}^1\setminus \gamma_{0t}$, with respect to the max norm, for some unique coefficients $\Omega_{n,k}(z)$, which are analytic matrix functions on $\mathbb{CP}^1\setminus \gamma_{0t}$. Furthermore, the coefficients have well-defined boundary values at $\gamma_{0t}$ from either side and the uniform absolute convergence of \eqref{eq:omega_perturb} extends to the boundary correspondingly. The coefficients only depend on the parameters $\Theta$, $\sigma_{0t}$ and $w^e$, with
\begin{equation}\label{eq:omega_parameter}
    \Omega_{n,k}(z)=(w^e)^{\frac{1}{2}\sigma_3}\widetilde{\Omega}_{n,k}(z)(w^e)^{-\frac{1}{2}\sigma_3},\quad
\widetilde{\Omega}_{n,k}(z)=\widetilde{\Omega}_{n,k}(z;\Theta,\sigma_{0t}),
\end{equation}
for $-n\leq k\leq n$ and $n\geq 1$.
\end{proposition}
\begin{proof}
We start by introducing the Cauchy transform on $\gamma_{0t}$. Let $h\in L^2(\gamma_{0t})$, then
	\begin{equation}\label{eq:cauchyoperator}
	\mathcal{C}(h)(z)=\frac{1}{2\pi i}\oint_{\gamma_{0t}}\frac{h(\zeta)}{\zeta-z}d\zeta \quad (z \in \mathbb{C}\setminus \gamma_{0t}),
	\end{equation}
	defines an analytic function on $\mathbb{C}\setminus \gamma_{0t}$, where we integrate along the anti-clockwise orientation of $\gamma_{0t}$, see Figure \ref{fig:single_I}. For $z\in \gamma_{0t}$, the limits
	\begin{equation*}
	\mathcal{C}_\pm(h)(z):=\lim_{z'\rightarrow z_{\pm}} \mathcal{C}(h)(z')
	\end{equation*}
	exists almost everywhere, where $z'\rightarrow z_{\pm}$ denotes any non-tangential limit from the $\pm$ side of $\gamma_{0t}$. Furthermore, $\mathcal{C}_\pm(h)$ are measurable on $\gamma_{0t}$ with finite $L^2$ norm, so $\mathcal{C}_\pm$ define linear operators on $L^2(\gamma_{0t})$, which are bounded since $\gamma_{0t}$ is analytic \cite{carleson}, and related by
 \begin{equation*}
     \mathcal{C}_+=\mathcal{C}_-+\mathcal{I},
 \end{equation*}
 where $\mathcal{I}$ denotes the identity operator.

Let $L_2^2(\gamma_{0t})$ denote the Banach space of $2\times 2$ matrix-valued functions, with entries in $L^2(\gamma_{0t})$, and norm given by the maximum of the $L^2$-norms of the entries.
As usual, we turn our attention to the singular integral equation,
 \begin{equation}\label{eq:singular_integral}
     \mu=\mathcal{C}_-\left[\mu\cdot(J^{\B{m}}-I)\right]+\mathcal{C}_-\left[J^{\B{m}}-I\right],
 \end{equation}
 for $\mu\in L_2^2(\gamma_{0t})$, which is well-defined since  $J^{\B{m}}-I \in L^2(\gamma_{0t})\cap L^\infty(\gamma_{0t})$. To study this singular integral equation, we introduce the operator,
\begin{equation*}
    \mathcal{L}^{\B{m}}(\mu):=\mu-\mathcal{C}_-\left[\mu \cdot(J^{\B{m}}-I)\right],
\end{equation*}
so that the singular integral equation becomes
\begin{equation*}
    \mathcal{L}^{\B{m}}(\mu) =\mathcal{C}_-[J^{m}-I].
\end{equation*}
We further write
\begin{equation*}
    \mathcal{L}^{\B{m}}=\mathcal{I}-\mathcal{R}^{\B{m}},\quad \mathcal{R}^{\B{m}}(\mu):=\mathcal{C}_-\left[\mu \cdot(J^{\B{m}}-I)\right].
\end{equation*}

Since $\gamma_{0t}$ is compact, the uniform expansion of the jump matrix in Lemma \ref{eq:circlejumpmatrix} also holds in $L_2^2(\gamma_{0t})$. Namely, there exists an $M_1\geq 0$, such that, for all $m\geq M_1$,
\begin{equation}\label{eq:Jexpansion}
    J^{\B{m}}-I=\sum_{n=1}^\infty{r_{0t}^{-1}(-t_m)^{n-2\sigma_{0t}}R_n^-+(-t_m)^{n}R_n^0+r_{0t}^{+1}(-t_m)^{n+2\sigma_{0t}}R_n^+},
\end{equation}
converges absolutely in $L_2^2(\gamma_{0t})$. It follows that $\mathcal{R}^{\B{m}}$ admits the following absolutely convergent expansion with respect to the operator norm from $L_2^2(\gamma_{0t})$ to $L_2^2(\gamma_{0t})$,
\begin{equation}\label{eq:Rexpansion}
\mathcal{R}^{\B{m}}=\sum_{n=1}^\infty{r_{0t}^{-1}(-t_m)^{n-2\sigma_{0t}}\mathcal{R}_n^-+(-t_m)^{n}\mathcal{R}_n^0+r_{0t}^{+1}(-t_m)^{n+2\sigma_{0t}}\mathcal{R}_n^+},
\end{equation}
for all $m\geq M_1$, where
\begin{equation*}
    \mathcal{R}_n^\diamond(\mu):=\mathcal{C}_-[\mu\cdot R_n^\diamond]\qquad (\mu\in L_2^2(\gamma_{0t})),
\end{equation*}
for $n\geq 1$ and $\diamond\in\{+,0,-\}$.

In particular, there exists an $M_2\geq M_1$ such that, for all $m\geq M_2$, the operator norm of $\mathcal{R}^{\B{m}}$ is strictly bounded by $1$, i.e. $||\mathcal{R}^{\B{m}}||<1$. This means that, for $m\geq M_2$, the operator $\mathcal{L}^{\B{m}}$ is invertible, with inverse given by
\begin{equation}\label{eq:Lexpansion}
(\mathcal{L}^{\B{m}})^{-1}=\mathcal{I}+\sum_{n=1}^\infty (\mathcal{R}^{\B{m}})^n,
\end{equation}
where $(\mathcal{R}^{\B{m}})^n$ denotes the $n$-fold composition of $\mathcal{R}^{\B{m}}$ with itself, as usual. This expansion is absolutely convergent in the operator norm. Substituting the absolutely convergent expansion of $\mathcal{R}^{\B{m}}$, given in equation \eqref{eq:Rexpansion}, into \eqref{eq:Lexpansion}, and collecting terms, we find that there exist unique, bounded operators $\mathcal{L}_{n,k}$, from $L_2^2(\gamma_{0t})$ to $L_2^2(\gamma_{0t})$, for $-n\leq k\leq n$, $n\geq 1$, such that
\begin{equation}\label{eq:Lexpansionfull}
(\mathcal{L}^{\B{m}})^{-1}=\mathcal{I}+\sum_{n=1}^\infty\sum_{k=-n}^n{r_{0t}^k (-t_m)^{n+2k\sigma_{0t}}\mathcal{L}_{n,k}},
\end{equation}
where the right-hand side converges absolutely in the operator norm, for all $m\geq M_2$. Furthermore, we note that the coefficients $\mathcal{L}_{n,k}$ only depend on the parameters $\Theta$, $\sigma_{0t}$ and $w^e$, with
\begin{equation*}
    \mathcal{L}_{n,k}=(w^e)^{\frac{1}{2}\sigma_3}\cdot \widetilde{\mathcal{L}}_{n,k}\cdot (w^e)^{-\frac{1}{2}\sigma_3},\quad
\widetilde{\mathcal{L}}_{n,k}=\widetilde{\mathcal{L}}_{n,k}(\Theta,\sigma_{0t}),
\end{equation*}
where the products are compositions of operators, and the matrices $(w^e)^{\pm\frac{1}{2}\sigma_3}$ are interpreted as operators acting by left-multiplication, for $-n\leq k\leq n$ and $n\geq 1$.

Since $\mathcal{L}^{\B{m}}$ is invertible, it follows that the singular integral equation \eqref{eq:singular_integral} has a unique solution, $\mu=\mu^{\B{m}}$, for $m\geq M_2$, given by
\begin{equation*}
    \mu^{\B{m}}=(\mathcal{L}^{\B{m}})^{-1}\left(\mathcal{C}_-[J^{m}-I]\right).
\end{equation*}
By substituting the absolutely convergent expansions \eqref{eq:Jexpansion} and \eqref{eq:Lexpansionfull} into this equation, and collecting terms, we find that there exist unique $\mu_{n,k}\in L_2^2(\gamma_{0t})$, for $-n\leq k\leq n$, $n\geq 1$, such that
\begin{equation}\label{eq:muexpansion}
\mu^{\B{m}}=I+\sum_{n=1}^\infty\sum_{k=-n}^n{r_{0t}^k (-t_m)^{n+2k\sigma_{0t}}\mu_{n,k}},
\end{equation}
where the right-hand side converges absolutely in $L_2^2(\gamma_{0t})$ for $m\geq M_2$. Furthermore, the coefficients $\mu_{n,k}$ only depend on the parameters $\Theta$, $\sigma_{0t}$ and $w^e$, with
\begin{equation*}
    \mu_{n,k}=(w^e)^{\frac{1}{2}\sigma_3}\cdot \widetilde{\mu}_{n,k}\cdot (w^e)^{-\frac{1}{2}\sigma_3},\quad
\widetilde{\mu}_{n,k}=\widetilde{\mu}_{n,k}(\Theta,\sigma_{0t}),
\end{equation*}
for $n\leq k\leq n$ and $n\geq 1$.

Finally, setting $M=M_2$, for $m\geq M$, we define
\begin{equation}\label{eq:omegadefi}
    \Omega^{\B{m}}(z)=I+\mathcal{C}[(\mu^{\B{m}}+I)(J^{\B{m}}-I)](z).
\end{equation}
Then, $\Omega^{\B{m}}(z)$ is analytic on $\mathbb{CP}^1\setminus \gamma_{0t}$, with $\Omega^{\B{m}}(\infty)=I$, and the boundary values
\begin{equation*}
    \Omega_\pm^{\B{m}}(z)=I+\mathcal{C}_{\pm}[(\mu^{\B{m}}+I)(J^{\B{m}}-I)](z),
\end{equation*}
satisfy the jump condition
\begin{equation*}
    \Omega_+^{\B{m}}(z)=\Omega_-^{\B{m}}(z)J^{\B{m}}(z),
\end{equation*}
due to the singular integral equation \eqref{eq:singular_integral}, for $z\in\gamma_{0t}$. In particular, $\Omega^{\B{m}}(z)$ is the unique solution to RHP \ref{rhp:circleI} for $m\geq M$.

Next, through substitution of the absolutely convergent expansions for $J^{\B{m}}-I$ and $\mu^{\B{m}}$ in equations \eqref{eq:Jexpansion} and \eqref{eq:muexpansion}, and collecting terms, we find the  expansion
\begin{equation}\label{eq:small_omega_expansion}
  (\mu^{\B{m}}+I)(J^{\B{m}}-I)=\sum_{n=1}^\infty\sum_{k=-n}^n{r_{0t}^k (-t_m)^{n+2k\sigma_{0t}}\omega_{n,k}},
\end{equation}
which is absolutely convergent in $L_2^2(\gamma_{0t})$ for all $m\geq M$, for some unique coefficients $\omega_{n,k}\in L_2^2(\gamma_{0t})$, for $-n\leq k\leq n$ and $n\geq 1$. Furthermore, the  dependence of the coefficients $\omega_{n,k}$ on the parameters $\Theta$, $\sigma_{0t}$ and $w^e$, takes the form
\begin{equation}\label{eq:omega_small_parameter}
    \omega_{n,k}=(w^e)^{\frac{1}{2}\sigma_3}\cdot\widetilde{\omega}_{n,k}\cdot (w^e)^{-\frac{1}{2}\sigma_3},\quad
\widetilde{\omega}_{n,k}=\widetilde{\omega}_{n,k}(\Theta,\sigma_{0t}),
\end{equation}
for $n\leq k\leq n$ and $n\geq 1$.

We now define
\begin{equation*}
    \Omega_{n,k}(z)=\mathcal{C}[\omega_{n,k}](z),
\end{equation*}
for $-n\leq k\leq n$ and $n\geq 1$, so that the parameter dependence in \eqref{eq:omega_parameter} follows immediately from \eqref{eq:omega_small_parameter}.

We proceed to prove the validity and uniform and absolute convergence of the expansion \eqref{eq:omega_perturb}, as detailed in the proposition. For $z \in \mathbb{CP}^1\setminus \gamma_{0t}$, define the function
\begin{equation*}
    k_z(\zeta)=\frac{1}{\zeta-z}\qquad (\zeta\in\gamma_{0t}),
\end{equation*}
with $k_z=0$ if $z=\infty$. Then $k_z$ lies in $L^2(\gamma_{0t})$ and, for $z \in \mathbb{CP}^1\setminus \gamma_{0t}$,
\begin{align*}
\Omega^{(m)}(z)&=I+\oint_{\gamma_{0t}}\frac{(\mu^{\B{m}}(\zeta)+I)(J^{\B{m}}(\zeta)-I)}{\zeta-z}d\zeta\\
&=I+\oint_{\gamma_{0t}}k_z(\zeta)\sum_{n=1}^\infty\sum_{k=-n}^n{r_{0t}^k (-t_m)^{n+2k\sigma_{0t}}\omega_{n,k}(\zeta)}d\zeta\\
&=I+\oint_{\gamma_{0t}}\sum_{n=1}^\infty\sum_{k=-n}^n{r_{0t}^k (-t_m)^{n+2k\sigma_{0t}}k_z(\zeta)\omega_{n,k}(\zeta)}d\zeta\\
&=I+\sum_{n=1}^\infty\sum_{k=-n}^n{r_{0t}^k (-t_m)^{n+2k\sigma_{0t}}\oint_{\gamma_{0t}}k_z(\zeta)\omega_{n,k}(\zeta)}d\zeta\\
&=I+\sum_{n=1}^\infty\sum_{k=-n}^n{r_{0t}^k (-t_m)^{n+2k\sigma_{0t}}\Omega_{n,k}(z)},
\end{align*}
for all $m\geq M$, where the interchanging of summation and integration is justified by the absolute convergence of expansion \eqref{eq:small_omega_expansion} in $L_2^2(\gamma_{0t})$.

Take any compact subset $K$ of the exterior of $\gamma_{0t}$ in $\mathbb{CP}^1$, so
\begin{equation*}
    K\subseteq D_{\text{ex}}(\gamma_{0t})\cup\{\infty\}.
\end{equation*}
Then, there exists a $d=d(K)$ such that the $L^2$ norm of $k_z$ is uniformly bounded by $d$ for $z\in K$.

Let $||\cdot||_{\text{max}}$ denote the max norm for matrices, $||\cdot||_p$ the standard $L^p$ norm for $L^p$ functions, and $||\cdot||_{p,\text{max}}$ the maximum of the $L^p$ norms of the entries of a matrix function. Then, for $z\in K$ and $m\geq M$,
\begin{align*}
    &\sum_{n=1}^\infty\sum_{k=-n}^n\big|\big|{r_{0t}^k (-t_m)^{n+2k\sigma_{0t}}\Omega_{n,k}(z)}\big|\big|_{\text{max}}=\\
    &\sum_{n=1}^\infty\sum_{k=-n}^n\big|\big|{\oint_{\gamma_{0t}}r_{0t}^k (-t_m)^{n+2k\sigma_{0t}}k_z(\zeta)\omega_{n,k}(\zeta)}d\zeta\big|\big|_{\text{max}}\leq\\
   &\sum_{n=1}^\infty\sum_{k=-n}^n\big|\big|{r_{0t}^k (-t_m)^{n+2k\sigma_{0t}}k_z\omega_{n,k}}\big|\big|_{1,\text{max}}\leq\\
   &\sum_{n=1}^\infty\sum_{k=-n}^n\big|\big|k_z\big|\big|_{2}\; \big|\big|{r_{0t}^k (-t_m)^{n+2k\sigma_{0t}}\omega_{n,k}}\big|\big|_{2,\text{max}}\leq\\
   &d\sum_{n=1}^\infty\sum_{k=-n}^n \big|\big|{r_{0t}^k (-t_m)^{n+2k\sigma_{0t}}\omega_{n,k}}\big|\big|_{2,\text{max}}<\infty,
\end{align*}
where H\"olders inequality was used in the second inequality and the last inequality followed from the absolute summability of the right-hand side of equation \eqref{eq:small_omega_expansion} in $L_2^2(\gamma_{0t})$.
It follows that the series expansion \eqref{eq:omega_perturb} is uniformly absolutely convergent on $K$.

We will now use a `deformation of contours' argument, to show that the series expansion on the right-hand side of \eqref{eq:omega_perturb} is 
uniformly absolutely convergent on the closure of the full exterior of $\gamma_{0t}$, within $\mathbb{CP}^1$.
To this end, we choose an analytic Jordan curve $\gamma_{0t}'$ which is contained in the interior of $\gamma_{0t}$ and satisfies the same conditions as $\gamma_{0t}$, specified in Remark \ref{rem:gamma0t}. The argument above, yields a sequence of coefficients $(\Omega_{n,k}'(z))_{n,k}$ analytic on $\mathbb{CP}^1\setminus \gamma_{0t}'$, such that all the statements above for the original coefficients relative to $\gamma_{0t}$, also hold for the $(\Omega_{n,k}'(z))_{n,k}$ relative to $\gamma_{0t}'$.

In particular, if we set 
\begin{equation*}
K=\gamma_{0t}\sqcup D_{\text{ex}}(\gamma_{0t})\sqcup\{\infty\},
\end{equation*} then $K$ is a compact subset of the exterior of $\gamma_{0t}'$ in $\mathbb{CP}^1$. It follows that equation \eqref{eq:omega_perturb}, with coefficients $(\Omega_{n,k}'(z))_{n,k}$, is valid and the right-hand side converges uniformly absolutely on $K$ with respect to the max norm. But then it must be true that $\Omega_{n,k}'(z)=\Omega_{n,k}(z)$ for $z\in K$, $-n\leq k\leq n$ and $n\geq 1$. Therefore, the right-hand side of expansion \eqref{eq:omega_perturb} (with the original coefficients) converges uniformly absolutely in the max norm on $K$.

A similar argument shows that the right-hand side of expansion \eqref{eq:omega_perturb} converges uniformly absolutely in the max norm on
\begin{equation*}
\gamma_{0t}\sqcup D_{\text{in}}(\gamma_{0t}).
\end{equation*}
The proposition follows.
\end{proof}

\subsection{Extracting asymptotics near $\boldsymbol{t_m=0}$}\label{sec:extract_asymp}
Recall that the solution $\Psi^{\B{m}}(z)$ of the main RHP \ref{rhp:main}, defines a unique corresponding linear system \eqref{eq:linear_system}, with coefficient matrix
\begin{equation}\label{eq:Acoefdefi}
    A(z,t_m)=A_0(t_m)+zA_1(t_m)+z^2 A_2,\quad A_2=q^{-\theta_\infty \sigma_3}.
\end{equation}
In this section, we will use the asymptotic expansion in Proposition \ref{prop:omega_perturb}, for $\Omega^{\B{m}}(z)$ as $m\rightarrow+\infty$, to obtain corresponding expansions for $A_0(t_m)$ and $A_1(t_m)$. In turn, we will use these to extract the asymptotics of the corresponding solution $(f,g)$ for small $t_m$, yielding the proof of Theorem \ref{thm:generic_asymp_zero}.

Firstly, we prove the following theorem.
\begin{theorem}\label{thm:linear_system_asymp_zero}
Let $\eta\in \mathcal{F}(\Theta,t_0)$ be $0$-generic, see Definition \ref{def:generic0}, and $C(z)\in\mathfrak{C}(\Theta,t_0)$ be a corresponding connection matrix. Let $\sigma_{0t}$, $s_{0t}$ and $w^e$ denote corresponding intermediate exponent, twist parameter and external gauge parameter respectively, with respect to the Mano decomposition I of $C(z)$, as in equations \eqref{eq:gen_int_exponent} and \eqref{eq:gen_twist_parameter}. Then, there exists an $M\geq 0$, such that, for all $m\geq M$, the solution $\Psi^{\B{m}}$ to RHP \ref{rhp:main} exists and the corresponding
 coefficients $A_0(t_m)$ and $A_1(t_m)$ in the coefficient matrix \eqref{eq:Acoefdefi}, admit complete asymptotic expansions as $m\rightarrow +\infty$,
\begin{align}
    A_1(t_m)&=A_0^e+\sum_{n=1}^\infty\sum_{k=-n}^n r_{0t}^k (-t_m)^{n+2k\sigma_{0t}}A_{1,n,k},\label{eq:a1exp}\\
    A_0(t_m)&=\sum_{n=1}^\infty\sum_{k=-n}^n r_{0t}^k (-t_m)^{n+2k\sigma_{0t}}A_{0,n,k},\label{eq:a0exp}
\end{align}
which are absolutely convergent with respect to the max norm for all $m\geq M$, for some unique $2\times 2$ matrices $A_{j,n,k}$, $j\in\{0,1\}$, $-n\leq k\leq n$ and $n\geq 0$.
Here $A_0^e$ is given in equation \eqref{eq:A0e} and the leading order coefficient $A_{0,-1,1}$ in \eqref{eq:a0exp} is given by
\begin{align}
A_{0,-1,1}=&\frac{(q^{\theta_t+\theta_0+\sigma_{0t}}-1)(q^{\theta_t-\theta_0+\sigma_{0t}}-1)}{q^{\theta_t+\theta_1}(q^{\theta_\infty}-q^{-\theta_\infty})(q^{\sigma_{0t}}-q^{-\sigma_{0t}})^2}\label{eq:A0m11}\\
&\begin{pmatrix}
    (q^{\theta_1+\sigma_{0t}+\theta_\infty}-1)(q^{\theta_1+\sigma_{0t}-\theta_\infty}-1) &
    w^e(q^{\theta_1+\sigma_{0t}+\theta_\infty}-1)^2 \\
    \frac{1}{w^e}(q^{\theta_1+\sigma_{0t}-\theta_\infty}-1)^2  & (q^{\theta_1+\sigma_{0t}+\theta_\infty}-1)(q^{\theta_1+\sigma_{0t}-\theta_\infty}-1)
\end{pmatrix}.\nonumber
\end{align}
The coefficients only depend on the parameters $\Theta$, $\sigma_{0t}$ and $w^e$, with
\begin{equation}\label{eq:Aparameterdependence}
    A_{j,n,k}=(w^e)^{\frac{1}{2}\sigma_3}\widetilde{A}_{j,n,k}(w^e)^{-\frac{1}{2}\sigma_3},\quad
\widetilde{A}_{j,n,k}=\widetilde{A}_{j,n,k}(\Theta,\sigma_{0t}),
\end{equation}
for $j\in\{0,1\}$, $-n\leq k\leq n$ and $n\geq 1$.
\end{theorem}
\begin{proof}
We start by applying Proposition \ref{prop:omega_perturb}, which gives us an $M\geq 0$, such that the solution $\Omega^{\B{m}}(z)$ of RHP \ref{rhp:circleI} exists for all $m\geq M$ and has a complete asymptotic expansion as $m\rightarrow +\infty$, given by equation \eqref{eq:omega_perturb}.

This means that RHP \ref{rhp:decomI} is solvable for all $m\geq M$, and its solution is given by
\begin{equation*}
    \Phi^{\B{m}}(z)=\begin{cases}
    \Psi_\infty^{\B{m}}(z) & \text{if }z\in D_\infty,\\
    \Psi_{0t}^{\B{m}}(z) & \text{if }z\in D_{0t}^{\B{m}},\\
    \Psi_0^{\B{m}}(z) & \text{if }z\in D_0^{\B{m}},\\
    \end{cases}
\end{equation*}
where
\begin{align*}
    \Psi_\infty^{\B{m}}(z)&=\Omega^{\B{m}}(z)\Psi_\infty^e(z) &  &(z\in D_\infty),\\
    \Psi_{0t}^{\B{m}}(z)&=\Omega^{\B{m}}(z)\Psi_0^e(z) &  &(z\in D_{0t}^{\B{m}}\cap D_\text{ex}(\gamma_{0t})),\\
    \Psi_{0t}^{\B{m}}(z)&=\Omega^{\B{m}}(z)\Psi_0^e(z)U_\infty^{\B{m}}(z) &  &(z\in D_{0t}^{\B{m}}\cap D_\text{in}(\gamma_{0t})),\\
    \Psi_{0}^{\B{m}}(z)&=\Omega^{\B{m}}(z)\Psi_0^e(z)U_0^{\B{m}}(z) &  &(z\in D_{0}^{\B{m}}),
\end{align*}
due to equation \eqref{eq:omegapsi}.

Next, we relate these functions to the coefficient matrix $A(z,t_m)$, using equations \eqref{eq:true_sol}, \eqref{eq:explicitsolyot} and \eqref{eq:linear_system}, giving
\begin{equation*}
    A(z,t_m)=\begin{cases}z^2 \Psi_\infty^{\B{m}}(qz)q^{-\theta_\infty\sigma_3}\Psi_\infty^{\B{m}}(z)^{-1}  &\text{if }z\in D_\infty\cap q^{-1}D_\infty,\\
    -z \Psi_{0t}^{\B{m}}(qz)q^{-\sigma_{0t}\sigma_3}\Psi_{0t}^{\B{m}}(z)^{-1}  &\text{if }z\in D_{0t}^{\B{m}}\cap q^{-1}D_{0t}^{\B{m}},\\
    t_m\Psi_{0}^{\B{m}}(qz)q^{\theta_0\sigma_3}\Psi_{0}^{\B{m}}(z)^{-1}  &\text{if }z\in D_{0}^{\B{m}}\cap q^{-1}D_{0}^{\B{m}}.
    \end{cases}
\end{equation*}
The first line shows that, for $z\in D_\infty\cap q^{-1}D_\infty$,
\begin{align*}
    A(z,t_m)&=z^2\; \Omega^{\B{m}}(qz)\Psi_\infty^e(qz)q^{-\theta_\infty\sigma_3}\Psi_\infty^e(z)^{-1}\Omega^{\B{m}}(z)^{-1}\\
    &=z\;\Omega^{\B{m}}(qz) A^e(z)\Omega^{\B{m}}(z)^{-1},
\end{align*}
where the second equality follows from equation \eqref{eq:Ae_psi}. In particular, we have the following exact formula for $A_1(t_m)$,
\begin{align}
    A_1(t_m)&=\frac{d}{dz}\left[A(z,t_m)-z^2q^{-\theta_\infty\sigma_3}\right]\nonumber\\
    &=\frac{d}{dz}\left[z\;\Omega^{\B{m}}(qz) A^e(z)\Omega^{\B{m}}(z)^{-1}-z^2q^{-\theta_\infty\sigma_3}\right]\nonumber\\
    &=\frac{d}{dz}\left[\Omega^{\B{m}}(qz)(z A_0^e+z^2q^{-\theta_\infty\sigma_3})\Omega^{\B{m}}(z)^{-1}-z^2q^{-\theta_\infty\sigma_3}\right],\label{eq:a1exact}
\end{align}
which holds for all $z\in D_\infty\cap q^{-1}D_\infty$. By substituting the (uniformly) absolutely convergent series expansion of $\Omega^{\B{m}}(z)$, given in equation \eqref{eq:omega_perturb}, into the last line, and collecting terms, we obtain the absolutely convergent series expansion \eqref{eq:a1exp}, for $A_1(t_m)$, in the theorem.

Similarly, for $z\in D_{0}^{\B{m}}\cap q^{-1}D_{0}^{\B{m}}$,
\begin{align*}
    A(z,t_m)=&t_m\Omega^{\B{m}}(qz)\Psi_0^e(qz)U_0^{\B{m}}(qz)q^{\theta_0\sigma_3}U_0^{\B{m}}(z)^{-1}\Psi_0^e(z)^{-1}\Omega^{\B{m}}(z)^{-1}\\
    =&t_m\Omega^{\B{m}}(qz)\Psi_0^e(qz)(-t_m)^{-\sigma_{0t}\sigma_3}A^i(z/t_m)(-t_m)^{\sigma_{0t}\sigma_3}\Psi_0^e(z)^{-1}\Omega^{\B{m}}(z)^{-1}\\
    =&\Omega^{\B{m}}(qz)\Psi_0^e(qz)(-t_m)^{-\sigma_{0t}\sigma_3}\left(t_mA_0^i-z\;q^{-\sigma_{0t}\sigma_3}\right)\\
    &(-t_m)^{\sigma_{0t}\sigma_3}\Psi_0^e(z)^{-1}\Omega^{\B{m}}(z)^{-1},
\end{align*}
where the second equality follows from equation \eqref{eq:aipsi0}. We now simply set $z=0$ to obtain an exact formula for $A_0(t_m)$,
\begin{equation}\label{eq:A0tm}
A_0(t_m)=\Omega^{\B{m}}(0)\Psi_0^e(0)(-t_m)^{-\sigma_{0t}\sigma_3}\left(t_mA_0^i\right)(-t_m)^{\sigma_{0t}\sigma_3}\Psi_0^e(0)^{-1}\Omega^{\B{m}}(0)^{-1}.
\end{equation}
Setting $z=0$ in equation \eqref{eq:omega_perturb}, gives an absolutely convergent series expansion for $\Omega^{\B{m}}(0)$. By substituting this expansion into equation \eqref{eq:A0tm}, we obtain the absolutely convergent series expansion \eqref{eq:a1exp}, for $A_1(t_m)$, in the theorem.

It follows from equation \eqref{eq:A0tm}, that the leading order coefficient, $A_{0,-1,1}$, in the asymptotic expansion \eqref{eq:a1exp}, is given by
\begin{equation}\label{eq:a0m11}
    A_{0,-1,1}=-r_{0t}\Psi_0^e(0)\begin{pmatrix}
        0 & (A_0^i)_{12}\\
        0 & 0
    \end{pmatrix}\Psi_0^e(0)^{-1}.
\end{equation}
Now, note that $\Psi_0^e(0)=h^e$, see equation \eqref{eq:he}, with
\begin{align*}
h^e&=(w^e)^{\frac{1}{2}\sigma_3} \widetilde{h}^e (w^e)^{-\frac{1}{2}\sigma_3} \begin{pmatrix}1/s_1^e & 0\\
0 & 1/s_2^e\\
\end{pmatrix},\\ 
\widetilde{h}^e&=\begin{pmatrix}
1-q^{\theta_1+\theta_\infty+\sigma_{0t}} & 1-q^{\theta_1+\theta_\infty-\sigma_{0t}}\\
1-q^{\theta_1-\theta_\infty+\sigma_{0t}} & 1-q^{\theta_1-\theta_\infty-\sigma_{0t}}
\end{pmatrix}.
\end{align*}
Further recalling that $(A_0^i)_{12}=\beta^iw^i$, equation \eqref{eq:a0m11} thus becomes
\begin{align*}
    A_{0,-1,1}&=-(w^e)^{\frac{1}{2}\sigma_3}\widetilde{h}^e(w^e)^{-\frac{1}{2}\sigma_3}\begin{pmatrix}
        0 & \beta^iw^i r_{0t}\frac{s_2^e}{s_1^e}\\
        0 & 0
    \end{pmatrix}(w^e)^{\frac{1}{2}\sigma_3}(\widetilde{h}^e)^{-1}(w^e)^{-\frac{1}{2}\sigma_3}\\
    &=-(w^e)^{\frac{1}{2}\sigma_3}\widetilde{h}^e\begin{pmatrix}
        0 & \beta^i r_{0t}\frac{w^is_2^e}{w_e s_1^e}\\
        0 & 0
    \end{pmatrix}(\widetilde{h}^e)^{-1}(w^e)^{-\frac{1}{2}\sigma_3}\\
    &=-(w^e)^{\frac{1}{2}\sigma_3}\widetilde{h}^e\begin{pmatrix}
        0 & \beta^i\\
        0 & 0
    \end{pmatrix}(\widetilde{h}^e)^{-1}(w^e)^{-\frac{1}{2}\sigma_3},
\end{align*}
where the last equality follows from equation \eqref{eq:gen_twist_parameter}. Direct substitution of the formula for $\widetilde{h}^e$ above, and the formula for $\beta^i$ in equation \eqref{eq:betai}, yield the explicit expression for the leading order coefficient $A_{0,-1,1}$ in the theorem.

Finally, the parameter dependence of the coefficients, as described in equation \eqref{eq:Aparameterdependence}, is a direct consequence of the exact formulas \eqref{eq:a1exact} and \eqref{eq:A0tm} for $A_1(t_m)$ and $A_0(t_m)$ respectively, and the parameter dependence of the coefficients in the asymptotic expansion for $\Omega^{\B{m}}(z)$ in Proposition \ref{prop:omega_perturb}, see equation \eqref{eq:omega_parameter}.
\end{proof}

We are now in a position to prove Theorem \ref{thm:generic_asymp_zero}.
\begin{proof}[Proof of Theorem \ref{thm:generic_asymp_zero}]
 To prove Theorem \ref{thm:generic_asymp_zero}, all that is left to do, is to extract the asymptotics of $(f,g)$ from the asymptotics of the matrix coefficients $A_0(t_m)$ and  $A_1(t_m)$ in Theorem \ref{thm:linear_system_asymp_zero}, as $m\rightarrow +\infty$.

To this end, note that, by equation \eqref{eq:coordinates_linear},
\begin{equation}\label{eq:f_in_terms_of_coef}
f(t_m)=-\frac{(A_0)_{12}(t_m)}{(A_1)_{12}(t_m)}.
\end{equation}
Now, by Theorem \ref{thm:linear_system_asymp_zero}, there exists an $M\geq 0$, such that
\begin{align}
    (A_1)_{12}(t_m)&=w^e\left(\beta^e+\sum_{n=1}^\infty\sum_{k=-n}^n r_{0t}^k (-t_m)^{n+2k\sigma_{0t}}(\widetilde{A}_{1,n,k})_{12}\right),\nonumber\\
    (A_0)_{12}(t_m)&=w^e \sum_{n=1}^\infty\sum_{k=-n}^n r_{0t}^k (-t_m)^{n+2k\sigma_{0t}}(\widetilde{A}_{0,n,k})_{12},\label{eq:a012}
\end{align}
where the right-hand sides are absolutely summable for all $m\geq M$.
Choose an $M'\geq M$ such that, for all $m\geq M',$
\begin{equation*}
\sum_{n=1}^\infty\sum_{k=-n}^n |r_{0t}^k (-t_m)^{n+2k\sigma_{0t}}(\widetilde{A}_{1,n,k})_{12}|\leq |\beta^e|\delta,
\end{equation*}
for some fixed $0<\delta<1$. Then, for $m\geq M'$, the reciprocal of $(A_1)_{12}(t_m)$, is given by
\begin{align*}
\frac{1}{(A_1)_{12}(t_m)}&=\frac{1}{w^e\beta^e}\sum_{N=0}^\infty\left(-\frac{1}{\beta^e}\sum_{n=1}^\infty\sum_{k=-n}^n r_{0t}^k (-t_m)^{n+2k\sigma_{0t}}(\widetilde{A}_{1,n,k})_{12}\right)^N\\
&=\frac{1}{w^e}\left(\frac{1}{\beta^e}+\sum_{n=1}^\infty\sum_{k=-n}^n r_{0t}^k (-t_m)^{n+2k\sigma_{0t}}\widetilde{p}_{n,k}\right),
\end{align*}
for some unique coefficients $\widetilde{p}_{n,k}=\widetilde{p}_{n,k}(\Theta,\sigma_{0t})$, $-n\leq k\leq n$, $n\geq 1$, where the series expansion on the last line is absolutely summable, for all $m\geq M'$.

It follows from this, as well as equations \eqref{eq:f_in_terms_of_coef} and \eqref{eq:a012}, that $f(t_m)$ has a asymptotic series expansion of the form
\begin{equation}\label{eq:fserieszero}
    f(t_m)=\sum_{n=1}^\infty\sum_{k=-n}^n F_{n,k}r_{0t}^k(- t_m)^{n+2k\sigma_{0t}},
\end{equation}
where the right-hand side is absolutely convergent for all $m\geq M'$, for some unique coefficients $F_{n,k}=F_{n,k}(\Theta,\sigma_{0t})$, $-n\leq k\leq n$, $n\geq 1$. Furthermore, the leading order coefficient $F_{1,-1}$ is given by
\begin{align}\label{eq:f1m1extract}
    F_{1,-1}=&-\frac{\widetilde{A}_{0,1,-1}}{\beta^e}\\
    =&q^{-\theta_t}\frac{\bigl(q^{\theta_t+\theta_0-\sigma_{0t}}-1\bigr)\bigl(q^{\theta_t-\theta_0-\sigma_{0t}}-1\bigr)\bigl(q^{\theta_1+\theta_\infty-\sigma_{0t}}-1\bigr)}{\bigl(q^{\theta_1+\theta_\infty- \sigma_{0t}}-1\bigr)\bigl(q^{\sigma_{0t}}-q^{-\sigma_{0t}}\bigr)^2},\nonumber
\end{align}
as follows directly from the explicit expressions for $\beta^e$ and $\widetilde{A}_{0,1,-1}$ in equations \eqref{eq:betae} and \eqref{eq:A0m11} respectively.

Similarly, by equation \eqref{eq:coordinates_linear}, we have the following explicit expression for $g(t_m)$,
\begin{equation}\label{eq:g_in_terms_of_coef}
    g(t_m)=\frac{(A_0)_{22}(t_m)+(A_1)_{22}(t_m)f(t_m)+q^{\theta_\infty} f(t_m)^2}{q(f(t_m)-q^{\theta_1})(f(t_m)-q^{-\theta_1})}.
\end{equation}
Considering the series expansion \eqref{eq:fserieszero} for $f(t_m)$, we choose an $M''\geq M'$, such that, for all $m\geq M''$,
\begin{equation*}
    \sum_{n=1}^\infty\sum_{k=-n}^n |F_{n,k}r_{0t}^k(- t_m)^{n+2k\sigma_{0t}}|\leq \delta' \min(|q^{\theta_1}|,|q^{-\theta_1}|),
\end{equation*}
for some fixed $0\leq \delta'<1$. Then, it follows that the reciprocal of $(f(t_m)-q^{\theta_1})(f(t_m)-q^{-\theta_1})$, has an asymptotic expansion, as $m\rightarrow +\infty$,
\begin{equation*}
    \frac{1}{(f(t_m)-q^{\theta_1})(f(t_m)-q^{-\theta_1})}=1+\sum_{n=1}^\infty\sum_{k=-n}^n r_{0t}^k (-t_m)^{n+2k\sigma_{0t}}r_{n,k},
\end{equation*}
which is absolutely convergent for all $m\geq M''$, for some unique coefficients $r_{n,k}=r_{n,k}(\Theta,\sigma_{0t})$, $-n\leq k\leq n$, $n\geq 1$. Combining this asymptotic expansion with the expansions for $(A_0)_{22}(t_m)$, $(A_1)_{22}(t_m)$ and $f(t_m)$ in Theorem \ref{thm:linear_system_asymp_zero} and equation \eqref{eq:fserieszero}, equation \eqref{eq:g_in_terms_of_coef} yields the following asymptotic expansion, as $m\rightarrow+\infty$, for $g(t_m)$,
\begin{equation}\label{eq:gserieszero}
    g(t_m)=\sum_{n=1}^\infty\sum_{k=-n}^n G_{n,k}r_{0t}^k(- t_m)^{n+2k\sigma_{0t}},
\end{equation}
which is absolutely convergent for all $m\geq M''$, for some unique coefficients $G_{n,k}=G_{n,k}(\Theta,\sigma_{0t})$, $-n\leq k\leq n$, $n\geq 1$. Furthermore, the leading order coefficient $G_{1,-1}$ is given by
\begin{align}\label{eq:G1m1extract}
G_{1,-1}&=\frac{1}{q}((A_{0,-1,1})_{22}+(A_0^e)_{22}F_{1,-1})\\
&=-q^{-1+\sigma_{0t}}F_{1,-1},
\end{align}
where the last line follows by direct substitution of the explicit formulas for $(A_0^e)_{22}$ and $(A_{0,-1,1})_{22}$ in equation \eqref{eq:A0e} and Theorem \ref{thm:linear_system_asymp_zero}.

Finally, the explicit expressions for the sub-leading coefficients $F_{1,0}$, $F_{1,1}$, $G_{1,0}$ and $G_{1,1}$ are computed by direct substitution of the series expansions for $f(t_m)$ and $g(t_m)$ into the $q\Psix$ equation and comparing terms. This finishes the proof of Theorem \ref{thm:generic_asymp_zero}.
\end{proof}

\begin{proof}[Proof of Statement in Remark \ref{remark:weakeningzero}]
In the derivation of the asymptotic expansions in Theorem \ref{thm:generic_asymp_zero}, we never used the fact that $\Re\sigma_{0t}\neq 0$. To be precise, the statement in Proposition \ref{prop:omega_perturb} continues to hold true without any modification when $\Re\sigma_{0t}=0$. The same is true for Theorem \ref{thm:linear_system_asymp_zero}, except that `leading order' has to be removed in the sentence above equation \eqref{eq:A0m11}, since the three terms in the first inner summation (i.e. $n=1$) on the right-hand side of equation \eqref{eq:a0exp} are of the same order. Similarly, the argument to extract the asymptotics of $(f,g)$, in the proof of Theorem \ref{thm:generic_asymp_zero}, continues to hold. We only need to remove `leading order' 
in the lines above equations \eqref{eq:f1m1extract} and \eqref{eq:G1m1extract}, and replace `sub-leading coefficients' by `other coefficients' in the second to last sentence of the proof. The statement in remark \ref{remark:weakeningzero} follows.
\end{proof}

\subsection{Asymptotics on the line $\widetilde{\mathcal{L}}_2^\infty$}\label{sec:reducible_rhp}
In this section, we prove the asymptotics of solutions corresponding to the line $\widetilde{\mathcal{L}}_2^\infty$ on the affine Segre surface $\mathcal{F}(\Theta,t_0)$, as detailed in Proposition \ref{prop:asymptotics_lineLd2i}. The reason for choosing this line, is that derivation for the generic asymptotics near $t=0$ goes through on this line without any modifications. We will derive the asymptotics on the other lines by employing symmetries in Section \ref{sec:lines_asymp}.

\begin{proof}[Proof of Proposition \ref{prop:asymptotics_lineLd2i} and Corollary \ref{coro:intersectionanalytic}]
Consider the general setup in Section \ref{sec:singlecontourI} and set
\begin{equation}\label{eq:sigma0t_special}
    \sigma_{0t}=\theta_t-\theta_0.
\end{equation}
We will show that the proof for the generic asymptotics around $t=0$, remains valid under weaker conditions than those in \eqref{eq:gen_int_exponent}. Namely, we only impose that
\begin{equation}\label{eq:realassump}
    \Re \sigma_{0t}=\Re(\theta_t-\theta_0)<\frac{1}{2},
\end{equation}
and we correspondingly get
\begin{equation}\label{eq:rho34}
    \rho_{34}=\frac{\vartheta_\tau(-\theta_0+\theta_t+\theta_1-\theta_\infty,-\theta_0+\theta_t-\theta_1+\theta_\infty)}{\vartheta_\tau(-\theta_0+\theta_t-\theta_1-\theta_\infty,-\theta_0+\theta_t+\theta_1+\theta_\infty)}.
\end{equation}
We note that this equality necessarily holds on the line $\widetilde{\mathcal{L}}_2^\infty$, as follows from equation \eqref{eq:rhoduality2} with $(i,j,k,l)=(3,4,2,1)$.

The main connection matrix is given by
\begin{equation*}
    C(z)=C^i\left(z/t_0\right)(-t_0)^{\sigma_{0t} \sigma_3}C^e(z),
\end{equation*}
where $C^i(z)$ and $C^e(z)$ are defined in equation \eqref{eq:CeCidefi}, and we are left with the free gauge and scaling parameters $w^i,w^e, s_{1,2}^i$ and $s_{1,2}^e$, some of whose relative values are controlled by the quantity $r_{0t}$, via equation \eqref{eq:gen_twist_parameter}.  For the sake of simplicity, we set
\begin{equation*}
s_1^i=s_2^i=s_1^e=s_2^e=1,\quad w^e=1,
\end{equation*}
so that the only free parameter left is
\begin{equation*}
    w^i=r_{0t}^{-1}.
\end{equation*}

The remainder of the proof consists of three parts. In the first, we show that, by varying $r_{0t}$, we trace out the line $\widetilde{\mathcal{L}}_2^\infty$. In the second, we derive the asymptotics when $r_{0t}\in\mathbb{C}^*$. In the third, we consider the case $r_{0t}=\infty$ and prove Corollary \ref{coro:intersectionanalytic}.

Firstly, we show that the monodromy coordinates $\eta$ corresponding to the connection matrix $C(z)$, indeed lie on the line $\widetilde{\mathcal{L}}_2^\infty$. To this end, we note that, by the choice \eqref{eq:sigma0t_special}, the internal connection matrix, $C^i(z)$, is reducible. By equation \eqref{eq:explicit_connection}, it reads
\begin{equation*}
    C^i(z)=\begin{pmatrix} s_1^i & 0\\
    0 & s_2^i\end{pmatrix}
    \begin{pmatrix}
   c_{11}^i\,\theta_q(q^{-\theta_t}z) & r_{0t}^{-1}\, c_{12}^i\,\theta_q(q^{\theta_t-2\theta_0}z)\\
    0& c_{22}^i\,\theta_q(q^{\theta_t}z)\\
    \end{pmatrix},
\end{equation*}
where
\begin{equation*}
    c_{11}^i=\frac{1}{1-q^{2(\theta_t-\theta_0)}},\quad
    c_{12}^i=\frac{\Gamma_q(-2\theta_0,2\theta_0-2\theta_t)}{(1-q)\Gamma_q(-2\theta_t)},\quad
    c_{22}^i=\frac{1}{1-q^{2\theta_0}}.
\end{equation*}
In particular, the dual Tyurin parameters of $C^i(z)$ equal
\begin{align*}
\widetilde{\rho}_1^i&=\pi[C^i(q^{+\theta_t})^T]=\frac{u}{r_{0t}},\qquad u:=\frac{c_{12}^i\theta_q(q^{2(\theta_t-\theta_0)})}{c_{22}^i\theta_q(q^{2\theta_t})},\\
\widetilde{\rho}_2^i&=\pi[C^i(q^{-\theta_t})^T]=\infty,
\end{align*}
Therefore, see equation \eqref{eq:tyurin_identification}, the dual Tyurin parameters $\widetilde{\rho}_{1,2}$ of the global connection matrix $C(z)$, are given by the same equations,
\begin{equation}\label{eq:dualtyurinproof}
    \widetilde{\rho}_{1}=\frac{u}{r_{0t}},\qquad
    \widetilde{\rho}_{2}=\infty.
\end{equation}
 In particular, the corresponding monodromy coordinates $\eta$ lie on the line $\widetilde{\mathcal{L}}_2^\infty$.

Next, we compute how $\eta$ depends on $r_{0t}$. To accomplish this, we first compute an explicit parametrisation for the triplet of Tyurin ratios
$\{\rho_{14},\rho_{24},\rho_{34}\}$. We already found that $\rho_{34}$ is constant and given by
\eqref{eq:rho34}. Similarly, it follows from equation \eqref{eq:rhoduality2} with $(i,j,k,l)=(1,4,2,3)$ that $\rho_{14}$ is constant and given by
\begin{equation}\label{eq:rho14}
    \rho_{14}=\frac{\vartheta_\tau(-\theta_0+\theta_t+\theta_1-\theta_\infty)\theta_q(q^{-\theta_0+\theta_\infty}t_0^{-1})}{\vartheta_\tau(-\theta_0+\theta_t+\theta_1+\theta_\infty)\theta_q(q^{-\theta_0-\theta_\infty}t_0^{-1})}.
\end{equation}
Finally, using the formula for the twist parameter in Lemma \ref{lem:twist_s0t}, we find
\begin{equation*}
    \rho_{24}=\frac{\vartheta_\tau(\theta_0-\theta_t+\theta_1-\theta_\infty)\theta_q(q^{\theta_0+\theta_\infty}t_0^{-1})-Z\vartheta_\tau(-\theta_0+\theta_t+\theta_1-\theta_\infty)\theta_q(q^{-\theta_0+\theta_\infty+2\theta_t}t_0^{-1})}
    {\vartheta_\tau(\theta_0-\theta_t+\theta_1+\theta_\infty)\theta_q(q^{\theta_0-\theta_\infty}t_0^{-1})-Z\vartheta_\tau(-\theta_0+\theta_t+\theta_1+\theta_\infty)\theta_q(q^{-\theta_0-\theta_\infty+2\theta_t}t_0^{-1})},
\end{equation*}
where
\begin{equation*}
    Z=-\frac{(-t_0)^{-2\sigma_{0t}}}{s_{0t}}=-\frac{c_{0t}(-t_0)^{-2\sigma_{0t}}}{r_{0t}},
\end{equation*}
and
\begin{equation*}
    c_{0t}=\frac{\Gamma_q(1+2\theta_t)\Gamma_q(1+2\theta_0-2\theta_t)^2}{\Gamma_q(1+2\theta_0)\Gamma_q(1-2\theta_0+2\theta_t)^{\color{white} 2}}\prod_{\epsilon=\pm 1}\frac{\Gamma_q(1-\theta_0+\theta_t+\theta_1+\epsilon \theta_\infty)}{\Gamma_q(1+\theta_0-\theta_t+\theta_1+\epsilon \theta_\infty)}.
\end{equation*}

By direct substitution of the expressions for $\{\rho_{14},\rho_{24},\rho_{34}\}$ into the defining equations \eqref{eq:eta_defi} of the $\eta$-coordinates, we find that each
\begin{equation*}
    \eta_{ij}=\eta_{ij}(Z)\qquad (1\leq i<j\leq 4),
\end{equation*}
is a M\"obius transform in $Z$. For example,
\begin{equation*}
    \eta_{12}(Z)=\frac{T_{12}\rho_{14} \rho_{24}(Z)\vartheta_\tau(+\frac{1}{2},-\frac{1}{2})/\vartheta_\tau(+\theta_0,-\theta_0)}{T_{12}'\rho_{14} \rho_{24}(Z)+T_{13}'\rho_{14} \rho_{34}+T_{14}'\rho_{34} +T_{23}'\rho_{24}(Z) \rho_{34}+T_{24}'\rho_{24}(Z) +T_{34}'\rho_{34}},
\end{equation*}
where we recall that $\rho_{14}$ and $\rho_{34}$ are constant and $\rho_{24}=\rho_{24}(Z)$ is a M\"obius transform in $Z$.

Furthermore, since the connection matrix $C(z)$ is well-defined for any choice of $r_{0t}^{-1}\in \mathbb{C}$, the $\eta$-coordinates must be regular in $Z$ on $\mathbb{C}$. In other words, each of these coordinates is linear in $Z$ and $\eta=\eta(Z)$ parametrises a line in $\mathcal{F}(\Theta,t_0)$ as $Z$ varies in $\mathbb{C}$. This line must be $\widetilde{\mathcal{L}}_2^\infty$, as we have already shown that $\widetilde{\rho}_2=\infty$, regardless of the value of $r_{0t}$.

When $r_{0t}=\infty$, it follows from equation \eqref{eq:dualtyurinproof} that also $\widetilde{\rho}_1=0$. In other words, this is the intersection point of the lines $\widetilde{\mathcal{L}}_2^\infty$ and $\widetilde{\mathcal{L}}_1^0$, for which the internal connection matrix $C^i(z)$ is diagonal. On the other hand, the connection matrix is ill-defined at $r_{0t}=0$ and in terms of the $\eta$-coordinates this corresponds to the intersection point of $\widetilde{\mathcal{L}}_2^\infty$ with the hyperplane section at infinity of the Segre surface.

Next, we derive the asymptotics of the solution $(f,g)$ corresponding to any choice of $r_{0t}\in \mathbb{C}^*$. Returning to the setting in Section \ref{sec:singlecontourI}, each of the matrix functions $\Psi_\infty^e(z)$, $\Psi_0^e(z)$ and $A^e(z)$ remains well-defined under the parameter specification $\sigma_{0t}=\theta_t-\theta_0$. Similarly, each of the matrix functions $\Psi_\infty^i(\zeta)$, $\Psi_0^i(\zeta)$ and $A^i(\zeta)$ remain well-defined and they all three become upper-triangular. The same statement holds true for the matrix functions $U_\infty^{\B{m}}(z)$ and $U_0^{\B{m}}(z)$.

We now consider the quotient \eqref{eq:omegapsi}, which defines the solution $\Omega^{\B{m}}(z)$ to RHP \ref{rhp:circleI}. Due to the fact that $\Psi_\infty^i(\zeta)$ is upper-triangular, the jump matrix $J^{\B{m}}(z)$ of RHP \ref{rhp:circleI} simplifies slightly. Namely,
\begin{align*}
    J^{\B{m}}(z)&=\Psi_0^e(z)(-t_m)^{-\sigma_{0t}\sigma_3}\Psi_\infty^i(z/t_m)(-t_m)^{\sigma_{0t}\sigma_3} \Psi_0^e(z)^{-1}\\
&=\Psi_0^e(z)\begin{pmatrix}
   (\Psi_\infty^i)_{11}(z/t_m) & (-t_m)^{-2\sigma_{0t}}(\Psi_\infty^i)_{12}(z/t_m) \\
   0 & (\Psi_\infty^i)_{22}(z/t_m) 
\end{pmatrix}
\Psi_0^e(z)^{-1},
\end{align*}
which means that the coefficients $R_n^+(z)$, $n\geq 1$, in the asymptotic expansion for $J^{\B{m}}(z)$ in Lemma \ref{lem:jumpperturb} all vanish. Namely, the asymptotic expansion as $m\rightarrow +\infty$, simplifies to
\begin{equation*}
    J^{\B{m}}(z)=I+\sum_{n=1}^\infty{r_{0t}^{-1}(-t_m)^{n-2\sigma_{0t}}R_n^-(z)+(-t_m)^{n}R_n^0(z)}.
\end{equation*}

We now move to the proof of Proposition \ref{prop:omega_perturb}. By the above simplification, the coefficients $\mathcal{R}_n^+$, $n\geq 1$, in the expansion \eqref{eq:Rexpansion} for $\mathcal{R}^{\B{m}}$ vanish. The expansion is therefore absolutely convergent with respect to the operator norm from $L_2^2(\gamma_{0t})$ to $L_2^2(\gamma_{0t})$, under the weakened condition $\Re \sigma_{0t}<\tfrac{1}{2}$.

It follows that the norm of $\mathcal{R}^{\B{m}}$ is strictly smaller than $1$ for $m\geq 0$ large enough and $\mathcal{L}^{\B{m}}$ is thus invertible and given by the absolutely convergent sum \eqref{eq:Lexpansion}. It follows that the series representation of $\mathcal{L}^{\B{m}}$ in equation \eqref{eq:Lexpansionfull} remains to holds true, with all coefficients $\mathcal{L}_{k,n}=0$, for $0<k\leq n$, $n\geq 1$. A similar truncation holds true for the solution $\mu^{\B{m}}$ of the singular integral equation. Namely, the series representation \eqref{eq:muexpansion}, becomes
\begin{equation*}
\mu^{\B{m}}=I+\sum_{n=1}^\infty\sum_{k=-n}^0{r_{0t}^k (-t_m)^{n+2k\sigma_{0t}}\mu_{n,k}},
\end{equation*}
and remains absolutely convergent in $L2^2(\gamma_{0t})$ for large enough $m\geq 0$.

Equation \eqref{eq:omegadefi} then defines the unique solution $\Omega^{\B{m}}(z)$ to RHP \ref{rhp:circleI}, for large enough $m\geq 0$. In particular, the statement of Proposition \ref{prop:omega_perturb} continues to hold true. That is, there exists an $M\geq 0$, such that the solution $\Omega^{\B{m}}(z)$ of RHP \ref{rhp:circleI} exists for all $m\geq M$ and has a complete asymptotic expansion as $m\rightarrow \infty$,
\begin{equation}\label{eq:omega_perturb_truncate}
\Omega^{\B{m}}(z)=I+\sum_{n=1}^\infty\sum_{k=-n}^0 r_{0t}^k (-t_m)^{n+2k\sigma_{0t}}\Omega_{n,k}(z),
\end{equation}
which is uniformly absolutely convergent in $z\in\mathbb{CP}^1\setminus \gamma_{0t}$, with respect to the max norm, for some unique coefficients $\Omega_{n,k}(z)$, which are analytic matrix functions on $\mathbb{CP}^1\setminus \gamma_{0t}$. Furthermore, the coefficients have well-defined boundary values at $\gamma_{0t}$ from either side and the uniform absolute convergence of \eqref{eq:omega_perturb_truncate} extends to the boundary correspondingly.

Similarly, the proof of Theorem \ref{thm:linear_system_asymp_zero} remains valid, and we obtain series expansions
\begin{align}
    A_1(t_m)&=A_0^e+\sum_{n=1}^\infty\sum_{k=-n}^0 r_{0t}^k (-t_m)^{n+2k\sigma_{0t}}A_{1,n,k},\label{eq:a1exptrunc}\\
    A_0(t_m)&=\sum_{n=1}^\infty\sum_{k=-n}^0 r_{0t}^k (-t_m)^{n+2k\sigma_{0t}}A_{0,n,k},\label{eq:a0exptrunc}
\end{align}
for the coefficient matrices, which are absolutely convergent with respect to the max norm, for large enough $m\geq 0$. The only thing to note, is that $A_{0,-1,1}$ need no longer be the leading order coefficient in \eqref{eq:a0exptrunc}. since the real part of $\sigma_{0t}$ can be negative.

Finally, extracting the asymptotics of $(f,g)$ is done exactly as in the proof of Theorem \ref{thm:generic_asymp_zero}. The series expansion for $f$ in equation \eqref{eq:fserieszero} becomes
\begin{equation}\label{eq:fserieszerotrunc}
    f(t_m)=\sum_{n=1}^\infty\sum_{k=-n}^0 F_{n,k}r_{0t}^k(- t_m)^{n+2k\sigma_{0t}},
\end{equation}
and the value of $F_{1,-1}$, which need no longer be the leading order coefficient, is extracted using equation \eqref{eq:f1m1extract}. Similarly, the series expansion for $g$ in equation \eqref{eq:gserieszero} becomes
\begin{equation*}
    g(t_m)=\sum_{n=1}^\infty\sum_{k=-n}^0 G_{n,k}r_{0t}^k(- t_m)^{n+2k\sigma_{0t}},
\end{equation*}
and the value of $G_{1,-1}$ is extracted using equation \eqref{eq:G1m1extract}.
This then yields the series expansions for $(f,g)$ in Proposition\ref{prop:asymptotics_lineLd2i}, when $r_{0t}\neq \infty$.

Finally, we consider the case $r_{0t}=\infty$ and prove Corollary \ref{coro:intersectionanalytic}. We may now drop the assumption on the real part of $\sigma_{0t}$ in equation \eqref{eq:realassump}. All the series expansions in the proof above further simplify, since $\Psi^i(z)$ is now diagonal. Firstly, the asymptotic expansion for $J^{\B{m}}(z)$ in Lemma \ref{lem:jumpperturb}, becomes a power series around $t_m=0$,
\begin{equation*}
    J^{\B{m}}(z)=I+\sum_{n=1}^\infty{-t_m)^{n}R_n^0(z)},
\end{equation*}
uniformly convergent for $z\in\gamma_{0t}$.
Following the steps above, we find that the statement of Proposition \ref{prop:omega_perturb} continues to hold true and the series expansion simplifies drastically. Namely, we find that there exists an $M\geq 0$, such that the solution $\Omega^{\B{m}}(z)$ of RHP \ref{rhp:circleI} exists for all $m\geq M$ and has a power series expansion around $t_m=0$,
\begin{equation*}
\Omega^{\B{m}}(z)=I+\sum_{n=1}^\infty (-t_m)^{n}\Omega_{n,0}(z),
\end{equation*}
which is uniformly absolutely convergent in $z\in\mathbb{CP}^1\setminus \gamma_{0t}$, with respect to the max norm, for large enough $m\geq 0$, for some unique coefficients $\Omega_{n,0}(z)$, which are analytic matrix functions on $\mathbb{CP}^1\setminus \gamma_{0t}$.

The proof of Theorem \ref{thm:linear_system_asymp_zero} remains valid, and we obtain power series expansions
\begin{align*}
    A_1(t_m)&=A_0^e+\sum_{n=1}^\infty (-t_m)^{n}A_{1,n,0},\\
    A_0(t_m)&=\sum_{n=1}^\infty(-t_m)^{n}A_{0,n,k},
\end{align*}
for the coefficient matrices, which are absolutely convergent with respect to the max norm, for large enough $m\geq 0$. We then obtain the power series expansions for $(f,g)$ using equations \eqref{eq:f_in_terms_of_coef} and \eqref{eq:g_in_terms_of_coef}, making use of the formula
\begin{equation*}
    A_{0,1,0}=\Psi_0^e(0)A_0^i\Psi_0^e(0)^{-1}.
\end{equation*}
This concludes the proofs of Theorem \ref{thm:generic_asymp_zero} and Corollary \ref{coro:intersectionanalytic}.
\end{proof}

%% file: symmetries.tex
\section{Symmetries and asymptotics}\label{sec:symmetries}

In this section, we use symmetries of the $q\Psix$ equation to derive all the remaining asymptotic formulas from the results in Theorem \ref{thm:generic_asymp_zero} and Proposition \ref{prop:asymptotics_lineLd2i}.
To this end, we require six symmetries, denoted by $r_0,r_t,r_1,r_\infty, v_{t1}$ and $v_{0\infty}$, described in Table \ref{table:symmetries}. We recall the two relations between them in equation \eqref{eq:conjugation_sym}.

For each of the symmetries, we show how it can be lifted to a symmetry of the linear problem and then work out how it acts on monodromy. The results are given in Table \ref{table:symmetries}.

The symmetry $v_{t1}$ allows us to translate asymptotic results around $t=0$ to asymptotics results around $t=\infty$. Similarly, the symmetry $v_{0\infty}$ allows us to find a dual asymptotic formula to a given asymptotic formula. For example, as we will see, application of this symmetry to the results in Theorem \ref{thm:generic_asymp_zero}, gives the dual asymptotic formulas in Theorem \ref{thm:generic_asymp_zero_dual}. We remark that this trick is essentially due to Guzzetti \cite{guzzetti_solving}, who used it to find some of the last missing asymptotic formulas for $\Psix$.

We will not give a full account of the symmetry group of $q\Psix$, see for example Kajiwara et al. \cite{kajiwarareview} for a complete description. We will just note that $v_{t1}$ and $v_{0\infty}$ are related to the two Dynkin diagram automorphisms denoted by $\pi_1$ and $\pi_2$, and $r_0$, $r_t$, $r_1$ and $r_\infty$ correspond to $s_5$, $s_0$, $s_1$ and $s_4$ respectively in \cite{kajiwarareview}*{\S 8.4.4}.

In Section \ref{sec:symr0t1} we derive the action of $r_0$, $r_t$ and $r_1$ on monodromy data. In Sections \ref{sec:symrt1} and \ref{sec:symroinf} we do the same for the respective symmetries $v_{t1}$ and $v_{0\infty}$. The action of $r_\infty$ now follows by conjugation, and is deduced at the end of Section \ref{sec:symroinf}.

Then, in Section \ref{sec:extract_asymp_zero_dual}, we derive the dual asymptotics near $t=0$, using the symmetry $v_{0\infty}$, yielding  Theorem \ref{thm:generic_asymp_zero_dual}. in Section \ref{sec:extract_asymp_infty} we derive Theorems \ref{thm:generic_asymp_infty} and \ref{thm:generic_asymp_infty_dual} respectively, by application of the symmetry $v_{t1}$. 

Finally, in Section \ref{sec:lines_asymp}, the asymptotics on the remaining lines are derived as well as their points of intersection. To this end, we make use of the actions of the symmetries on the lines, which are illustrated graphically in Figure \ref{figure:action_on_lines}.

\subsection{Symmetries $\boldsymbol{r_0}$, $\boldsymbol{r_t}$ and $\boldsymbol{r_1}$}\label{sec:symr0t1}
We start with the symmetry
\begin{equation*}
    r_t:\theta_t\rightarrow -\theta_t,\quad (f,g)\mapsto (f,g).
\end{equation*}
It lifts to a trivial symmetry of the linear system,
\begin{equation*}
    r_t:A(z,t)\mapsto A(z,t),
\end{equation*}
and thus also acts trivially on the connection matrix,
\begin{equation*}
    r_t:C(z,t)\mapsto C(z,t).
\end{equation*}
The reflection $\theta_t\mapsto -\theta_t$ interchanges the zeros $x_1=q^{+\theta_t}t_0$ and $x_2=q^{-\theta_t}t_0$ of the determinant of $C(z,t_0)$, and thus acts on the Tyurin parameters as a permutation,
\begin{equation*}
   r_t: \rho_k\mapsto \rho_{\alpha_t(k)},\quad \widetilde{\rho}_k\mapsto \widetilde{\rho}_{\alpha_t(k)},\quad (1\leq k\leq 4),\qquad \alpha_t:=(1\;2)\in S_4.
\end{equation*}
We now  work out how $r_t$ acts on the $\eta$-coordinates. To this end, note that
\begin{equation*}
    T_{jk}|_{\theta_t\mapsto -\theta_t}=-q^{-2\theta_t}T_{\alpha_t(j)\alpha_t(k)}\quad (1\leq j<k\leq 4),
\end{equation*}
where we identify $T_{jk}=T_{kj}$ as before. It follows that $r_t$ acts as a permutation on the $\eta$-coordinates,
\begin{equation*}
  r_t:  \eta_{jk}\mapsto \eta_{\alpha_t(j)\alpha_t(k)}\quad (1\leq j<k\leq 4).
\end{equation*}
In particular, $r_t$ extends to an invertible linear mapping on the ambient space $\mathbb{CP}^6$, mapping $\widehat{\mathcal{F}}(\Theta,t_0)$ isomorphically onto $\widehat{\mathcal{F}}(\widehat{\Theta},t_0)$, where $\widehat{\Theta}=(\theta_0,-\theta_t,\theta_1,\theta_\infty)$. It permutes lines as follows,
\begin{equation*}
r_t: \mathcal{L}_k^\diamond\mapsto \mathcal{L}_{\alpha_t(k)}^\diamond,\quad
\widetilde{\mathcal{L}}_k^\diamond\mapsto \widetilde{\mathcal{L}}_{\alpha_t(k)}^\diamond\qquad
(\diamond\in\{0,\infty\},1\leq k\leq 4).
\end{equation*}

\subsubsection{Symmetry $r_1$}
Similarly, we have the symmetry
\begin{equation*}
    r_1:\theta_1\rightarrow -\theta_1,\quad (f,g)\mapsto (f,g).
\end{equation*}
which acts on the Tyurin parameters by
\begin{equation*}
   r_1: \rho_k\mapsto \rho_{\alpha_1(k)},\quad \widetilde{\rho}_k\mapsto \widetilde{\rho}_{\alpha_1(k)},\quad (1\leq k\leq 4),\qquad \alpha_1:=(3\;4)\in S_4.
\end{equation*}
It acts on the coefficients of $T(\rho)$ by
\begin{equation*}
    T_{jk}|_{\theta_1\mapsto -\theta_1}=-q^{-2\theta_t}T_{\alpha_1(j)\alpha_1(k)}\quad (1\leq j<k\leq 4),
\end{equation*}
and thus acts as a permutation on the $\eta$-coordinates,
\begin{equation*}
   r_1: \eta_{jk}\mapsto \eta_{\alpha_1(j)\alpha_1(k)}\quad (1\leq j<k\leq 4).
\end{equation*}
In particular, $r_1$ extends to an invertible linear mapping on the ambient space $\mathbb{CP}^6$, mapping $\mathcal{F}(\Theta,t_0)$ isomorphically onto $\mathcal{F}(\widehat{\Theta},t_0)$, where $\widehat{\Theta}=(\theta_0,\theta_t,-\theta_1,\theta_\infty)$.
It permutes lines as follows,
\begin{equation*}
r_1: \mathcal{L}_k^\diamond\mapsto \mathcal{L}_{\alpha_1(k)}^\diamond,\quad
\widetilde{\mathcal{L}}_k^\diamond\mapsto \widetilde{\mathcal{L}}_{\alpha_1(k)}^\diamond\qquad
(\diamond\in\{0,\infty\},1\leq k\leq 4).
\end{equation*}

\subsubsection{Symmetry $r_0$}
Next, we consider the symmetry
\begin{equation*}
    r_0:\theta_0\rightarrow -\theta_0,\quad (f,g)\mapsto (f,g).
\end{equation*}
It has a trivial lift to the linear system,
\begin{equation*}
    r_0:A(z,t)\mapsto A(z,t),
\end{equation*}
but acts non-trivally on the canonical solution at $z=0$, defined in equation \eqref{eq:linear_sys_solutions}. To see this, let us first note that $\Psi_0(z,t)$ and $\Psi_\infty(z,t)$ are uniquely characterised by
\begin{subequations}\label{eq:psi0infqdif}
\begin{align}
\Psi_\infty(qz,t)&=z^{-2}A(z,t)\Psi_\infty(z,t)q^{\theta_\infty \sigma_3} & \Psi_\infty(z,t)=I+\mathcal{O}(z^{-1})\quad (z\rightarrow \infty),\\
\Psi_0(qz,t)&=t^{-1}\;A(z,t)\Psi_0(z,t)q^{-\theta_0 \sigma_3} & \Psi_0(z,t)=H(t)+\mathcal{O}(z)\quad (z\rightarrow 0),
\end{align}
\end{subequations}
where $H(t)\in GL_2(\mathbb{C})$ diagonalises $A(0,t)$ and is such that the connection matrix $C(z,t)=\Psi_0(z,t)^{-1}\Psi_\infty(z,t)$ satisfies
\begin{equation}\label{eq;pure_isomonodromy}
    C(z,qt)=z\;C(z,t).
\end{equation}

We thus have the following lift of $r_0$, acting on the canonical solutions and connection matrix, as
\begin{equation*}
    r_0:\Psi_{0}(z,t)\mapsto \Psi_{0}(z,t)\sigma_1,\quad 
    \Psi_{\infty}(z,t)\mapsto \Psi_{\infty}(z,t),
\end{equation*}
and
\begin{equation*}
    r_0: C(z,t)\mapsto \sigma_1 C(z,t),
\end{equation*}
where $\sigma_1$ is the Pauli matrix
\begin{equation*}
    \sigma_1:=\begin{pmatrix}
        0 & 1\\
    1 & 0
    \end{pmatrix}.
\end{equation*}

It thus acts on the Tyurin parameters by
\begin{equation*}
   r_0: \rho_k\mapsto \rho_{k},\quad \widetilde{\rho}_k\mapsto 1/\widetilde{\rho}_{k}\qquad (1\leq k\leq 4).
\end{equation*}
It acts trivially on the coefficients of $T(\rho)$ and also acts trivially
on the $\eta$-coordinates,
\begin{equation*}
    r_0:\eta\mapsto \eta.
\end{equation*}
In particular, $r_0$ simply acts as the identity from $\mathcal{F}(\Theta,t_0)$ to $\mathcal{F}(\widehat{\Theta},t_0)$, where 
$\widehat{\Theta}=(-\theta_0,\theta_t,\theta_1,\theta_\infty)$. It permutes lines as follows,
\begin{equation*}
r_0: \mathcal{L}_k^\diamond\mapsto \mathcal{L}_{k}^\diamond,\quad
\widetilde{\mathcal{L}}_k^0\mapsto \widetilde{\mathcal{L}}_k^\infty,\quad 
\widetilde{\mathcal{L}}_k^\infty\mapsto \widetilde{\mathcal{L}}_k^0\quad
(\diamond\in\{0,\infty\},1\leq k\leq 4).
\end{equation*}

\subsection{The symmetry $\boldsymbol{v_{t1}}$} \label{sec:symrt1}
Next, we consider a slightly more involved symmetry. For the sake of simplicity, we write $f^{\B{m}}=f(q^mt_0)$, $g^{\B{m}}=g(q^mt_0)$, 
$A^{\B{m}}(z)=A(z,t_m)$ and so on. The symmetry then reads
\begin{equation*}
    v_{t1}:\Theta=(\theta_0,\theta_t,\theta_1,\theta_\infty)\mapsto \widehat{\Theta}=(\theta_0,\theta_1,\theta_t,\theta_\infty), \quad (f,g)\mapsto (\widehat{f},\widehat{g}),\quad t_0\mapsto t_0^{-1},
\end{equation*}
where
\begin{equation}\label{fgrt1}
    \widehat{f}^{\B{m}}=\frac{f^{\B{-m}}}{t_{-m}},\quad \widehat{g}^{\B{m}}=\frac{1}{q g^{\B{1-m}}}\qquad (m\in\mathbb{Z}).
\end{equation}

A lift of this symmetry to the linear system is given by
\begin{equation*}
    v_{t1}: A\mapsto \widehat{A},\quad \widehat{A}^{\B{m}}(z)=t_{-m}^{-2} A^{\B{-m}}(t_{-m}z).
\end{equation*}
To check this assertion, note that $\widehat{A}^{\B{m}}(z)$ is a degree two matrix polynomial with
\begin{equation*}
\widehat{A}^{\B{m}}(z)=t_{-m}^{-2}(t_{-m}z)^2A_2+\mathcal{O}(z)=z^2q^{-\theta_\infty \sigma_3}+\mathcal{O}(z) \quad (z\rightarrow \infty),
\end{equation*}
and
\begin{align*}
  \widehat{A}^{\B{m}}(0)=t_{-m}^{-2}A_0^{\B{-m}}=&t_{-m}^{-2}H^{\B{-m}}(t_{-m} q^{\theta_0\sigma_3})(H^{\B{-m}})^{-1}\\
  =&H^{\B{-m}}(q^mt_0^{-1} q^{\theta_0\sigma_3})(H^{\B{-m}})^{-1}.
\end{align*}
Furthermore, the determinant of $\widehat{A}^{\B{m}}(z)$ reads
\begin{align*}
    |\widehat{A}^{\B{m}}(z)|=&t_{-m}^{-4}(t_{-m}z-q^{+\theta_t}t_{-m})(t_{-m}z-q^{-\theta_t}t_{-m})(t_{-m}z-q^{+\theta_1})(t_{-m}z-q^{-\theta_1})\\
    =&(z-q^{+\theta_t})(z-q^{-\theta_t})(z-q^{+\theta_1}q^mt_0^{-1})(z-q^{-\theta_1}q^mt_0^{-1}).
\end{align*}
This shows that the action of $v_{t1}$ on the linear system is consistent with the way $v_{t1}$ acts on the parameters $\Theta$ and $t_0$. In order to check consistency with equation \eqref{fgrt1}, we use equations \eqref{eq:coordinates_linear}, giving
\begin{equation*}
\widehat{f}^{\B{m}}=-\frac{(\widehat{A}_{0}^{\B{m}})_{12}}{\widehat{A}_{1}^{\B{m}})_{12}}=
\frac{t_{-m}^{-2}(A_{0}^{\B{-m}})_{12}}{t_{-m}^{-1}A_{1}^{\B{-m}})_{12}}=\frac{f^{\B{-m}}}{t_{-m}},
\end{equation*}
and
\begin{align*}
\widehat{g}^{\B{m}}&=\frac{\widehat{A}_{22}^{\B{m}}(\widehat{f}^{\B{m}})}{q(\widehat{f}^{\B{m}}-q^{+\theta_t}) (\widehat{f}^{\B{m}}-q^{-\theta_t})}\\
&=\frac{t_{-m}^{-2}A_{22}^{\B{-m}}(t_{-m}\widehat{f}^{\B{m}})}{q(\widehat{f}^{\B{m}}-q^{+\theta_t}) (\widehat{f}^{\B{m}}-q^{-\theta_t})}\\
&=\frac{A_{22}^{\B{-m}}(f^{\B{-m}})}{q(f^{\B{-m}}-q^{+\theta_t}t_{-m}) (f^{\B{-m}}-q^{-\theta_t}t_{-m})}\\
&=\frac{(f^{\B{-m}}-q^{+\theta_1})(f^{\B{-m}}-q^{-\theta_1})g^{\B{-m}}}{(f^{\B{-m}}-q^{+\theta_t}t_{-m}) (f^{\B{-m}}-q^{-\theta_t}t_{-m})}\\
&=\frac{1}{q g^{\B{1-m}}}.
\end{align*}

Next, we consider the action of $v_{t1}$ on monodromy. To this end, we note that equations \eqref{eq:psi0infqdif} imply the following lift of $v_{t1}$ to the canonical solutions near $z=0$ and $z=\infty$,
\begin{align*}
    v_{t1}: \Psi_0\mapsto \widehat{\Psi}_0, \quad \widehat{\Psi}_0^{\B{m}}(z)=\Psi_0^{\B{-m}}(t_{-m}z),\quad
    \Psi_\infty\mapsto \widehat{\Psi}_\infty, \quad \widehat{\Psi}_\infty^{\B{m}}(z)=\Psi_\infty^{\B{-m}}(t_{-m}z).
\end{align*}
Therefore,
\begin{equation*}
     v_{t1}: C\mapsto \widehat{C},\quad \widehat{C}^{\B{m}}(z)=C^{\B{-m}}(t_{-m}z).
\end{equation*}

We remark that the system
\begin{equation*}
    Y(qz)=\widehat{A}^{\B{m}}(z)Y(z),
\end{equation*}
is not purely isomonodromic in $t$, since, in comparison to equation \eqref{eq;pure_isomonodromy},
\begin{equation*}
    \widehat{C}^{\B{m+1}}(z)=q^{-1}zq^{-\theta_0\sigma_3}\widehat{C}^{\B{m}}(z)q^{-\theta_\infty\sigma_3}.
\end{equation*}
To fix this, one could apply a further scaling
\begin{equation*}
   \widehat{\Psi}_\infty^{\B{m}}\mapsto q^{-m\theta_\infty\sigma_3} \widehat{\Psi}_\infty^{\B{m}} q^{m\theta_\infty\sigma_3},\quad
    \widehat{\Psi}_0^{\B{m}}\mapsto q^{-m}q^{-m\theta_\infty\sigma_3} \widehat{\Psi}_0^{\B{m}} q^{-m\theta_0\sigma_3},
\end{equation*}
so that
\begin{equation*}
    \widehat{A}^{\B{m}}\mapsto q^{-m\theta_\infty\sigma_3} \widehat{A}^{\B{m}} q^{m\theta_\infty\sigma_3},\quad
    \widehat{C}^{\B{m}}\mapsto q^m q^{m\theta_0\sigma_3} \widehat{C}^{\B{m}} q^{m\theta_\infty\sigma_3}.
\end{equation*}
This, however, has no impact on the Tyurin data. The action of $v_{t1}$ on them is easily computed, for example
\begin{equation*}
\widehat{\rho}_1=\pi\big[\widehat{C}^{\B{0}}\big(q^{\theta_1}t_0^{-1}\big)\big]=
    \pi[C^{\B{0}}(q^{\theta_1})]=\rho_3,
\end{equation*}
and
\begin{equation*}
\widehat{\rho}_3=\pi\big[\widehat{C}^{\B{0}}\big(q^{\theta_t}\big)\big]=
    \pi[C^{\B{0}}(q^{\theta_t}t_0)]=\rho_1.
\end{equation*}
In general, we find that $v_{t1}$ permutes them as follows,
\begin{equation}\label{eq:symmetryrt1tyurin}
   v_{t1}: \rho_k\mapsto \rho_{\alpha_{t1}(k)},\quad \widetilde{\rho}_k\mapsto \widetilde{\rho}_{\alpha_{t1}(k)},\quad (1\leq k\leq 4),\qquad \alpha_{t1}:=(1\;3)\;(2\; 4)\in S_4.
\end{equation}
It acts on the coefficients of $T(\rho)$ by
\begin{equation*}
    v_{t1}(T_{jk})=-t_0^{-2}T_{\alpha_{t1}(j)\alpha_{t1}(k)}\quad (1\leq j<k\leq 4),
\end{equation*}
and consequently we find that
\begin{equation}\label{eq:symmetryrt1eta}
   v_{t1}: \eta_{jk}\mapsto \eta_{\alpha_{t1}(j)\alpha_{t1}(k)}\quad (1\leq j<k\leq 4).
\end{equation}
In particular, $v_{t1}$ extends to an invertible linear mapping on the ambient space $\mathbb{CP}^6$, mapping $\mathcal{F}(\Theta,t_0)$ isomorphically onto $\widehat{\mathcal{F}}(\widehat{\Theta},t_0^{-1})$.
It permutes lines as follows,
\begin{equation*}
v_{t1}: \mathcal{L}_k^\diamond\mapsto \mathcal{L}_{\alpha_{t1}(k)}^\diamond,\quad
\widetilde{\mathcal{L}}_k^\diamond\mapsto \widetilde{\mathcal{L}}_{\alpha_{t1}(k)}^\diamond\qquad
(\diamond\in\{0,\infty\},1\leq k\leq 4).
\end{equation*}

\subsection{The symmetry $\boldsymbol{v_{0\infty}}$}\label{sec:symroinf}
 The final symmetry we need, reads
\begin{equation*}
    v_{0\infty}:\Theta=(\theta_0,\theta_t,\theta_1,\theta_\infty)\mapsto \widehat{\Theta}=(\theta_\infty-\tfrac{1}{2},\theta_1,\theta_t,\theta_0+\tfrac{1}{2}), \quad (f,g)\mapsto (\widehat{f},\widehat{g}),
\end{equation*}
where
\begin{equation}\label{eq:fgr0i}
    \widehat{f}^{\B{m}}=\frac{t_m}{f^{\B{m}}},\quad \widehat{g}^{\B{m}}=\frac{t_m}{q^{\frac{3}{2}} g^{\B{m}}}\qquad (m\in\mathbb{Z}).
\end{equation}

Our starting point to compute a lift of this symmetry to the linear system, is equations \eqref{eq:psi0infqdif}. Firstly, we are going to swap the role of $z=0$ and $z=\infty$, by introducing the new variable
\begin{equation*}
    \zeta=\frac{qt_m}{z},
\end{equation*}
and setting
\begin{equation*}
\widetilde{\Psi}_\infty^{\B{m}}(\zeta)=g_\infty(\zeta)\Psi_0^{\B{m}}\bigg(\frac{qt_m}{\zeta}\bigg),\quad 
\widetilde{\Psi}_0^{\B{m}}(\zeta)=g_0(\zeta)\Psi_\infty^{\B{m}}\bigg(\frac{qt_m}{\zeta}\bigg),
\end{equation*}
so that
\begin{subequations}\label{eq:psitildeqdif}
\begin{align}
\widetilde{\Psi}_\infty^{\B{m}}(q\zeta)&=\frac{g_\infty(q\zeta)}{g_\infty(\zeta)}t_m A^{\B{m}}(t_m/\zeta)^{-1}\widetilde{\Psi}_\infty^{\B{m}}(\zeta)q^{\theta_0\sigma_3},\\
   \widetilde{\Psi}_0^{\B{m}}(q\zeta)&=\frac{g_0(q\zeta)}{g_0(\zeta)}(t_m/\zeta)^2 A^{\B{m}}(t_m/\zeta)^{-1}\widetilde{\Psi}_0^{\B{m}}(\zeta)q^{-\theta_\infty\sigma_3}.
\end{align}
\end{subequations}
Next, we choose $g_0(\zeta)$ and $g_\infty(\zeta)$ precisely such that they cancel out poles of $A^{\B{m}}(t_m/\zeta)^{-1}$ and are analytic at $\zeta=0$ and $\zeta=\infty$ respectively, namely
\begin{align*}
g_\infty(\zeta)&=\big( q^{1-\theta_t}/\zeta, q^{1+\theta_t}/\zeta, q^{1-\theta_1}t_m/\zeta, q^{1+\theta_1}t_m/\zeta;q\big)_\infty,\\
g_0(\zeta)&=\big(\zeta/q^{-\theta_t},\zeta/q^{+\theta_t},\zeta/(q^{-\theta_1}t_m),\zeta/(q^{+\theta_1}t_m);q\big)_\infty^{-1}.
\end{align*}
Equations \eqref{eq:psitildeqdif} then simplify to
\begin{subequations}\label{eq:psihatqdif}
\begin{align}
\widetilde{\Psi}_\infty^{\B{m}}(q\zeta)&=\zeta^{-2} \widetilde{A}^{\B{m}}(\zeta)\widetilde{\Psi}_\infty^{\B{m}}(\zeta)q^{\theta_0\sigma_3},\\
   \widetilde{\Psi}_0^{\B{m}}(q\zeta)&=t_m^{-1}\widetilde{A}^{\B{m}}(\zeta)\widetilde{\Psi}_0^{\B{m}}(\zeta)q^{-\theta_\infty\sigma_3},
\end{align}
\end{subequations}
where
\begin{equation*}
    \widetilde{A}^{\B{m}}(\zeta)=\frac{\zeta^2}{t_m}\operatorname{adj}(A^{\B{m}}(t_m/\zeta)).
\end{equation*}

The matrix function $\widetilde{A}^{\B{m}}(\zeta)$ is a polynomial in $\zeta$,
\begin{equation*}
  \widetilde{A}^{\B{m}}(\zeta)=  \widetilde{A}_0^{\B{m}}+\zeta \widetilde{A}_1^{\B{m}}+\zeta^2 \widetilde{A}_2^{\B{m}},
\end{equation*}
with coefficients
\begin{equation*}
    \widetilde{A}_0^{\B{m}}=t_m q^{\theta_\infty\sigma_3},\quad
    \widetilde{A}_1^{\B{m}}=\operatorname{adj}(A_1^{\B{m}}),\quad
    \widetilde{A}_2^{\B{m}}=H^{\B{m}} q^{-\theta_0\sigma_3} (H^{\B{m}})^{-1}.
\end{equation*}
One could now proceed to normalise $\widetilde{A}^{\B{m}}(\zeta)$ by conjugating it by $H^{\B{m}}$, but computing the corresponding bi-rational action on $(f,g)$ is involved. Instead, we are going to apply a gauge transformation,
\begin{align*}
    \widehat{\Psi}_\infty^{\B{m}}(\zeta)&=G^{\B{m}}(\zeta)\widetilde{\Psi}_\infty^{\B{m}}(\zeta)\zeta^{\frac{1}{2}\sigma_3},\\
    \widehat{\Psi}_0^{\B{m}}(\zeta)&=G^{\B{m}}(\zeta)\widetilde{\Psi}_0^{\B{m}}(\zeta)\zeta^{\frac{1}{2}\sigma_3},
\end{align*}
where $G^{\B{m}}(\zeta)$ is lower-triangular and takes the form
\begin{equation*}
    G^{\B{m}}(\zeta)=\zeta^{\frac{1}{2}}\begin{pmatrix}
        0 & 0\\
        g_1^{\B{m}} & g_2^{\B{m}}
    \end{pmatrix}+\zeta^{-\frac{1}{2}}\begin{pmatrix}
        g_3^{\B{m}} & 0\\
        g_4^{\B{m}} & 0
    \end{pmatrix}.
\end{equation*}
The main idea behind this gauge transformation is that it will simultaneously normalise $\widetilde{A}_2^{\B{m}}$ and shift the monodromy exponents at $\zeta=\infty$ and $\zeta=0$ by $+\tfrac{1}{2}$ and $-\tfrac{1}{2}$ respectively.

To see the shift in exponents, note that, under the gauge transformation, equations \eqref{eq:psihatqdif} become
\begin{align*}
\widehat{\Psi}_\infty^{\B{m}}(q\zeta)&=\zeta^{-2} \widehat{A}^{\B{m}}(\zeta)\widehat{\Psi}_\infty^{\B{m}}(\zeta)q^{(\theta_0+\frac{1}{2})\sigma_3},\\
   \widehat{\Psi}_0^{\B{m}}(q\zeta)&=t_m^{-1}\widehat{A}^{\B{m}}(\zeta)\widehat{\Psi}_0^{\B{m}}(\zeta)q^{-(\theta_\infty-\frac{1}{2})\sigma_3},
\end{align*}
where
\begin{align}\label{eq:Ahatdefi}
    \widehat{A}^{\B{m}}(\zeta)&=G^{\B{m}}(q\zeta)\widetilde{A}^{\B{m}}(\zeta)G^{\B{m}}(\zeta)^{-1}\\
 &=\frac{\zeta^2}{t_m}G^{\B{m}}(q\zeta)\operatorname{adj}(A^{\B{m}}(t_m/\zeta))G^{\B{m}}(\zeta)^{-1}.\label{eq:Ahatdefi2}
\end{align}

Furthermore, the general shape of $G^{\B{m}}(\zeta)$ was chosen such that $\widehat{\Psi}_0^{\B{m}}(\zeta)$ is automatically analytic around $\zeta=0$, indeed
\begin{align*}
  \widehat{\Psi}_0^{\B{m}}(\zeta)&=G^{\B{m}}(\zeta)g_0(\zeta)\Psi_\infty^{\B{m}}\bigg(\frac{qt_m}{\zeta}\bigg)\zeta^{\frac{1}{2}\sigma_3}\\
  &=G^{\B{m}}(\zeta)(I+\mathcal{O}(\zeta))\zeta^{\frac{1}{2}\sigma_3}\\
  &=\left(\zeta\begin{pmatrix}
        0 & 0\\
        g_1^{\B{m}} & g_2^{\B{m}}
    \end{pmatrix}+\begin{pmatrix}
        g_3^{\B{m}} & 0\\
        g_4^{\B{m}} & 0
    \end{pmatrix}\right)(I+\mathcal{O}(\zeta))\begin{pmatrix}
        1 & 0\\
        0 & \zeta^{-1}
    \end{pmatrix},\\
    &=\begin{pmatrix}
        g_3^{\B{m}} & 0\\
        g_4^{\B{m}} & 0
    \end{pmatrix}(I+\mathcal{O}(\zeta))\begin{pmatrix}
        0 & 0\\
        0 & \zeta^{-1}
    \end{pmatrix}+\mathcal{O}(1)\\
    &=\mathcal{O}(1),
\end{align*}
as $\zeta\rightarrow 0$.

The remaining condition that we have to impose, which will complete determine the coefficients of $G^{\B{m}}(\zeta)$, is that
\begin{equation}\label{eq:asymptotic_condition}
    \widehat{\Psi}_\infty^{\B{m}}(\zeta)=I+\mathcal{O}(\zeta^{-1})\qquad (\zeta\rightarrow \infty).
\end{equation}
Let us write
\begin{equation*}
    \Psi_0^{\B{m}}(z)=H^{\B{m}}(I+z\;U^{\B{m}}+\mathcal{O}(z^2)),
\end{equation*}
then condition \eqref{eq:asymptotic_condition} is equivalent to
\begin{equation*}
g_\infty(\zeta) \left(\zeta\begin{pmatrix}
        0 & 0\\
        g_1^{\B{m}} & g_2^{\B{m}}
    \end{pmatrix}+\begin{pmatrix}
        g_3^{\B{m}} & 0\\
        g_4^{\B{m}} & 0
    \end{pmatrix}\right) H^{\B{m}}\left(I+\frac{qt_m}{\zeta}U^{\B{m}}\right)\begin{pmatrix}
        1 & 0\\
        0 & \zeta^{-1}
    \end{pmatrix}=I+\mathcal{O}(\zeta^{-1}),
\end{equation*}
as $\zeta\rightarrow \infty$.
Using the fact that
\begin{equation*}
    g_\infty(\zeta)=1-\frac{q^{\theta_t}+q^{-\theta_t}+t_m(q^{\theta_1}+q^{-\theta_1})}{(q-1)\zeta},
\end{equation*}
this asymptotic condition is equivalent to the following four algebraic equations,
\begin{align*}
    &h_{11}^{\B{m}} g_1^{\B{m}}+h_{21}^{\B{m}}g_2^{\B{m}}=0,\\
    &h_{21}^{\B{m}} g_1^{\B{m}}+h_{22}^{\B{m}}g_2^{\B{m}}=1,\\
    &h_{11}^{\B{m}}g_3^{\B{m}}=1,\\
    &h_{11}^{\B{m}}g_4^{\B{m}}=-qt_m U_{21}^{\B{m}}.
\end{align*}
We  thus find
\begin{equation*}
    g_1^{\B{m}}=-\frac{h_{21}^{\B{m}}}{|H^{\B{m}}|},\quad
    g_2^{\B{m}}=\frac{h_{11}^{\B{m}}}{|H^{\B{m}}|},\quad
    g_3^{\B{m}}=\frac{1}{h_{11}^{\B{m}}},\quad
    g_4^{\B{m}}=-\frac{qt_m U_{21}^{\B{m}}}{h_{11}^{\B{m}}},
\end{equation*}
and in particular $|G^{\B{m}}(\zeta)|\equiv 1$.

Next, we calculate the action of $v_{0\infty}$ on $(f,g)$ using formulas \eqref{eq:coordinates_linear} and check that they match with \eqref{eq:fgr0i}. Firstly, we note that $\widehat{f}^{\B{m}}$ is defined as the unique value of $\zeta$ for which $\widehat{A}^{\B{m}}(\zeta)$ is lower-triangular. Since $G^{\B{m}}(\zeta)$ is lower-triangular, it follows from equation \eqref{eq:Ahatdefi} that $\widehat{A}^{\B{m}}(\zeta)$ is lower triangular if and only if $\widetilde{A}^{\B{m}}(\zeta)$ is lower-triangular, for $\zeta=\widehat{f}^{\B{m}}$. In turn, the latter holds true if and only if $\operatorname{adj}(A^{\B{m}}(t_m/\zeta))$ is lower-triangular, which is equivalent to $t_m/\zeta=f^{\B{m}}$. It follows $\widehat{f}^{\B{m}}$ is need given by the formula in \eqref{eq:fgr0i}.

To compute $\widehat{g}^{\B{m}}$, we note that, by equation \eqref{eq:Ahatdefi2},
\begin{equation*}
     \widehat{A}^{\B{m}}(\widehat{f}^{\B{m}})
 =\frac{t_m}{(f^{\B{m}})^2}G^{\B{m}}(qf^{\B{m}})\begin{pmatrix}
     A_{22}^{\B{m}}(f^{\B{m}}) & 0\\
     -A_{21}^{\B{m}}(f^{\B{m}}) & A_{11}^{\B{m}}(f^{\B{m}})
 \end{pmatrix}G^{\B{m}}(f^{\B{m}})^{-1}.
\end{equation*}
and thus
\begin{equation*}
    \widehat{A}_{22}^{\B{m}}(\widehat{f}^{\B{m}})=q^{\frac{1}{2}}\frac{t_m}{(f^{\B{m}})^2}A_{11}^{\B{m}}(f^{\B{m}}).
\end{equation*}
Using this identity, we compute $\widehat{g}^{\B{m}}$ as follows,
\begin{align*}
    \widehat{g}^{\B{m}}&=\frac{\widehat{A}_{22}^{\B{m}}(\widehat{f}^{\B{m}})}{q(\widehat{f}^{\B{m}}-q^{\theta_t})(\widehat{f}^{\B{m}}-q^{-\theta_t})}\\
    &=q^{\frac{1}{2}}\frac{t_m}{(f^{\B{m}})^2}\frac{A_{11}^{\B{m}}(f^{\B{m}})}{q(t_m/f^{\B{m}}-q^{\theta_t})(t_m/f^{\B{m}}-q^{-\theta_t})}\\
    &=\frac{|A^{\B{m}}(f^{\B{m}})|/A_{22}^{\B{m}}(f^{\B{m}})}{q^{\frac{1}{2}}(f^{\B{m}}-q^{\theta_t}t_m)(f^{\B{m}}-q^{-\theta_t}t_m)}\\
    &=\frac{|A^{\B{m}}(f^{\B{m}})|}{q^{\frac{3}{2}}(f^{\B{m}}-q^{\theta_t}t_m)(f^{\B{m}}-q^{-\theta_t}t_m)(f^{\B{m}}-q^{\theta_1})(f^{\B{m}}-q^{-\theta_1})g^{\B{m}}}\\
    &=\frac{1}{\frac{3}{2}g^{\B{m}}},
\end{align*}
where in the last equality we used equation \eqref{eq:detA} for the determinant of $A^{\B{m}}(z)$.

We have now shown that
\begin{equation*}
    v_{0\infty}: A\mapsto \widehat{A},\quad \widehat{A}^{\B{m}}(\zeta)=\frac{\zeta^2}{t_m}G^{\B{m}}(q\zeta)\operatorname{adj}(A^{\B{m}}(t_m/\zeta))G^{\B{m}}(\zeta)^{-1},
\end{equation*}
lifts the action of $v_{0\infty}$ to the linear problem.
Next, we consider the action of $v_{0\infty}$ on monodromy. We have already constructed a corresponding lift of $v_{0\infty}$ to the canonical solutions,
\begin{align*}
    &v_{0\infty}: \Psi_0\mapsto \widehat{\Psi}_0, & \widehat{\Psi}_0^{\B{m}}(\zeta)&=g_0(\zeta)G^{\B{m}}(\zeta)\Psi_\infty^{\B{m}}(qt_m/\zeta)\zeta^{\frac{1}{2}\sigma_3},\\
    &v_{0\infty}:
    \Psi_\infty\mapsto \widehat{\Psi}_\infty, & \widehat{\Psi}_\infty^{\B{m}}(\zeta)&=g_\infty(\zeta)G^{\B{m}}(\zeta)\Psi_0^{\B{m}}(qt_m/\zeta)\zeta^{\frac{1}{2}\sigma_3},
\end{align*}
and thus $v_{0\infty}$ acts on the connection matrix as,
\begin{equation*}
     v_{0\infty}: C\mapsto \widehat{C},
\end{equation*}
where
\begin{align*}
    \widehat{C}^{\B{m}}(\zeta)&=\widehat{\Psi}_0^{\B{m}}(\zeta)^{-1}\widehat{\Psi}_\infty^{\B{m}}(\zeta)\\
    &=\zeta^{-\frac{1}{2}\sigma_3}\frac{g_\infty(\zeta)}{g_0(\zeta)} \Psi_\infty^{\B{m}}(qt_m/\zeta)^{-1}\Psi_0^{\B{m}}(qt_m/\zeta)\zeta^{\frac{1}{2}\sigma_3}\\
    &=\zeta^{-\frac{1}{2}\sigma_3}\frac{g_\infty(\zeta)}{g_0(\zeta)}C^{\B{m}}(qt_m/\zeta)^{-1}\zeta^{\frac{1}{2}\sigma_3}\\
    &=\zeta^{-\frac{1}{2}\sigma_3}\frac{|C^{\B{m}}(qt_m/\zeta)|}{c_m}C^{\B{m}}(qt_m/\zeta)^{-1}\zeta^{\frac{1}{2}\sigma_3}\\
    &=c_m^{-1}\zeta^{-\frac{1}{2}\sigma_3}\operatorname{adj}(C^{\B{m}}(qt_m/\zeta))\zeta^{\frac{1}{2}\sigma_3},\\
    &=\frac{\zeta^2}{c_mt_m}\zeta^{-\frac{1}{2}\sigma_3}q^{-\theta_\infty\sigma_3}\operatorname{adj}(C^{\B{m}}(t_m/\zeta))q^{-\theta_0\sigma_3}\zeta^{\frac{1}{2}\sigma_3},
\end{align*}
for some $c_m\in\mathbb{C}^*$, where in the fourth equality we used equation \eqref{eq:connectiondet} for the determinant of $C^{\B{m}}(z)$.

We remark that the system
\begin{equation*}
    Y(q\zeta)=\widehat{A}^{\B{m}}(\zeta)Y(\zeta),
\end{equation*}
is not purely isomonodromic in $t$, since
\begin{equation*}
    \widehat{C}^{\B{m+1}}(\zeta)=\frac{c_m}{c_{m+1}}\zeta q^{-\theta_\infty\sigma_3}\widehat{C}^{\B{m}}(\zeta)q^{-\theta_0\sigma_3}.
\end{equation*}
We could resolve this by a further scaling, which does not affect the Tyurin data. 

We proceed to compute the action of $v_{0\infty}$ on the Tyurin parameters. To compute $\widehat{\rho}_1$, we proceed as follows,
\begin{align*}
    \widehat{\rho}_1&=\pi[\widehat{C}^{\B{0}}(q^{\widehat{\theta}_t}t_0)]\\
    &=\pi[\widehat{C}^{\B{0}}(q^{\theta_1}t_0)]\\
    &=\pi\left[\frac{(q^{\theta_1}t_0)^2}{c_m t_m}(q^{\theta_1}t_0)^{-\frac{1}{2}\sigma_3}q^{-\theta_\infty\sigma_3}\operatorname{adj}(C^{\B{0}}(q^{-\theta_1}))q^{-\theta_0\sigma_3}(q^{\theta_1}t_0)^{\frac{1}{2}\sigma_3}\right]\\
    &=\pi\left[\operatorname{adj}(C^{\B{0}}(q^{-\theta_1}))q^{-\theta_0\sigma_3}(q^{\theta_1}t_0)^{\frac{1}{2}\sigma_3}\right]\\
    &=\pi\left[\operatorname{adj}(C^{\B{0}}(q^{-\theta_1}))\right]q^{-2\theta_0} q^{\theta_1}t_0\\
    &=-q^{-2\theta_0}\frac{q^{\theta_1}t_0}{\widetilde{\rho}_4},
\end{align*}
where, in the last equality, we used that
\begin{equation*}
    \pi[\operatorname{adj}(R)]=-1/\pi(R^T),
\end{equation*}
for any rank one $2\times 2$ matrix.

The action of $v_{0\infty}$ on the other Tyurin parameters is computed similarly, yielding
\begin{equation*}
    \widehat{\rho}_k=-q^{-2\theta_0} \frac{t_0}{x_{\beta(k)}\widetilde{\rho}_{\beta(k)}},\quad \widehat{\widetilde{\rho}}_k=-q^{-2\theta_\infty} \frac{x_{\beta(k)}}{t_0\rho_{\beta(k)}}\qquad (1\leq k\leq 4),
\end{equation*}
where $\beta=(1\;4)\;(2\;3)\in S_4$. Now, recall that the Tyurin parameters are only determined up to scalar multiplication by the corrersponding solution. In other words, we can get rid of the common factors $-q^{-2\theta_0}$ and $-q^{-2\theta_\infty}$ and multiply by $t_{0}^{\pm \frac{1}{2}}$ and write the action of $v_{0\infty}$ on the Tyurin parameters as
\begin{equation}\label{eq:tyurin_parametersr0infty}
 v_{0\infty}:   \rho_k\mapsto t_0^{\frac{1}{2}}\frac{1}{x_{\beta(k)}\widetilde{\rho}_{\beta(k)}},\quad \widetilde{\rho}_k\mapsto t_0^{-\frac{1}{2}}\frac{x_{\beta(k)}}{\rho_{\beta(k)}}\qquad (1\leq k\leq 4),
\end{equation}
Alternatively, we could have accomplished this by a further scaling of $\widehat{A}$, $\widehat{\Psi}_0$, $\widehat{\Psi}_\infty$ and $\widehat{C}$. The normalisation is chosen such that applying $ v_{0\infty}$ twice maps $\rho_k$ to $1\cdot \rho_k$, $1\leq k\leq 4$, rather than a non-trivial scalar multiple. For example,
\begin{align*}
v_{0\infty}^2 \left[\rho_1\right]&=v_{0\infty}\left[t_0^{\frac{1}{2}}\frac{1}{x_4\widetilde{\rho}_{4}}\right]=v_{0\infty}\left[t_0^{\frac{1}{2}}q^{\theta_1}\frac{1}{\widetilde{\rho}_{4}}\right]\\
&=t_0^{\frac{1}{2}}q^{\theta_t}v_{0\infty}[\widetilde{\rho}_{4}]^{-1}
=t_0^{\frac{1}{2}}q^{\theta_t}\big(t_0^{-\frac{1}{2}} x_1/\rho_1\big)^{-1}=\rho_1.
\end{align*}

The action of $v_{0\infty}$ on the $\eta$-coordinates can now be computed, yielding
\begin{equation*}
   v_{0\infty}: \eta\mapsto \widehat{\eta},
\end{equation*}
where $\widehat{\eta}=\widehat{\eta}(\eta,\Theta,t_0)$ a rational function in $\eta$. Furthermore, since $v_{0\infty}$ maps the solution space of $q\Psix(\Theta,t_0)$ bijectively onto the solution space of $q\Psix(\widehat{\Theta},t_0)$, it must be regular on $\mathcal{F}(\Theta,t_0)$ and map it biholomorphically onto $\mathcal{F}(\widehat{\Theta},t_0)$. Furthermore, it permutes lines as follows,
\begin{equation*}
\mathcal{L}_k^0\mapsto \widetilde{\mathcal{L}}_{\beta(k)}^\infty,\quad
\mathcal{L}_k^\infty\mapsto \widetilde{\mathcal{L}}_{\beta(k)}^0,\quad
\widetilde{\mathcal{L}}_k^0\mapsto \mathcal{L}_{\beta(k)}^\infty,\quad
\widetilde{\mathcal{L}}_k^\infty\mapsto \mathcal{L}_{\beta(k)}^0\quad (1\leq k\leq 4),
\end{equation*}
where $\beta=(1\;4)\;(2\;3)\in S_4$.

A unique linear extension of $v_{0\infty}$ to the ambient space $\mathbb{CP}^6$ can now be computed, by simply evaluating the action of $v_{0\infty}$ on $\eta$ at $7$ explicit points in $\mathcal{F}(\Theta,t_0)$ and consequently determining the coefficients of the linear mapping. We will, however, not need an explicit expression for the linear extension of $v_{0\infty}$.

We finish this section with computing the action of the symmetry $r_\infty$,
\begin{equation*}
    r_\infty:\theta_\infty\rightarrow 1-\theta_\infty,\quad (f,g)\mapsto (f,g).
\end{equation*}
Its action on monodromy follows by conjugation, due to equation \eqref{eq:conjugation_sym}.
In particular, $r_\infty$ acts on the Tyurin parameters by
\begin{equation*}
   r_\infty: \rho_k\mapsto \frac{x_k^2}{t_0\rho_{k}},\quad \widetilde{\rho}_k\mapsto \widetilde{\rho}_{k}\qquad (1\leq k\leq 4).
\end{equation*}
It acts on the coefficients of $T(\rho)$ by
\begin{equation*}
    r_{\infty}(T_{ij})=\frac{t_0^2}{x_ix_j}T_{kl}\qquad (\{i,j,k,l\}=\{1,2,3,4\}),
\end{equation*}
so that
\begin{equation*}
    r_{\infty}: T_{ij}\rho_i\rho_j\mapsto T_{kl}\frac{1}{\rho_i\rho_j},
\end{equation*}
and, consequently, we find that
\begin{equation*}
   r_{\infty}: \eta_{ij}\mapsto \eta_{kl}\qquad (\{i,j,k,l\}=\{1,2,3,4\}).
\end{equation*}
In particular, $r_{\infty}$ extends to an invertible linear mapping on the ambient space $\mathbb{CP}^6$, mapping $\mathcal{F}(\Theta,t_0)$ isomorphically onto $\widehat{\mathcal{F}}(\widehat{\Theta},t_0)$,
where $\widehat{\Theta}=(\theta_0,\theta_t,\theta_1,1-\theta_\infty)$.
It permutes lines as follows,
\begin{equation*}
r_{\infty}: \mathcal{L}_k^0\mapsto \mathcal{L}_k^\infty,\quad \mathcal{L}_k^\infty\mapsto \mathcal{L}_k^0,\quad
\widetilde{\mathcal{L}}_k^\diamond\mapsto \widetilde{\mathcal{L}}_k^\diamond \qquad
(\diamond\in\{0,\infty\},1\leq k\leq 4).
\end{equation*}

\subsection{Dual asymptotics near $\boldsymbol{t=0}$}\label{sec:extract_asymp_zero_dual}
In this section, we derive the dual asymptotic expansions for solutions of $q\Psix$ around $t=0$, essentially by applying the symmetry $v_{0\infty}$ to the results in Theorem \ref{thm:generic_asymp_zero}.
\begin{proof}[Proof of Theorem \ref{thm:generic_asymp_zero_dual} and Remark \ref{remark:weakeningzerodual}]
To obtain the dual asymptotic expansions of solutions around $t=0$, we use the symmetry $v_{0\infty}$, following Guzzetti \cite{guzzetti_solving}. Let us take a $0$-generic point $\eta\in\mathcal{F}(\Theta,t_0)$, see Definition \ref{def:generic0}. Let $(f,g)$ denote the corresponding solution of $q\Psix(\Theta,t_0)$. By applying $v_{0\infty}$, we obtain a corresponding solution $(\widehat{f},\widehat{g})$ of $q\Psix(\widehat{\Theta},t_0)$, with parameters
\begin{equation*}
    \widehat{\Theta}=(\theta_\infty-\tfrac{1}{2},\theta_1,\theta_t,\theta_0+\tfrac{1}{2}).
\end{equation*}
To this solution corresponds a unique point $\widehat{\eta}=v_{0\infty}(\eta)\in \mathcal{F}(\widehat{\Theta},t_0)$ and we know the values of the correponding (dual) Tyurin ratios, due to equations \eqref{eq:tyurin_parametersr0infty}. Of particular relevance for the current proof, are the values of
\begin{equation*}
\widehat{\rho}_{34}=v_{0\infty}(\rho_{34})=q^{2\theta_t}\widetilde{\rho}_{12},\qquad 
\widehat{\rho}_{24}=v_{0\infty}(\rho_{24})=q^{\theta_t-\theta_1}t_0\widetilde{\rho}_{13}.
\end{equation*}

The solutions $(f,g)$ and $(\widehat{f},\widehat{g})$ are related by
\begin{align}\label{eq:fgandtheirhats}
    \frac{t_m}{f(t_m)}&=\widehat{f}(t_m),\\
    \frac{t_m}{q^{\frac{3}{2}}g(t_m)}&=\widehat{g}(t_m),
\end{align}
and the crux of the proof lies in applying Theorem \ref{thm:generic_asymp_zero} to the right-hand sides of the above equations, i.e. to the solution $(\widehat{f},\widehat{g})$. In order to apply the theorem, the first thing we have to do is check that $\widehat{\eta}$ is $0$-generic. To this end, we consider the equation
\begin{equation}\label{eq:exponent_equation}
    \frac{\vartheta_\tau(\widehat{\sigma}-\widehat{\theta}_1+\widehat{\theta}_\infty,\widehat{\sigma}+\widehat{\theta}_1-\widehat{\theta}_\infty)}{\vartheta_\tau(\widehat{\sigma}+\widehat{\theta}_1+\widehat{\theta}_\infty,\widehat{\sigma}-\widehat{\theta}_1-\widehat{\theta}_\infty)}=\widehat{\rho}_{34}.
\end{equation}
Through substitution of parameter values and application of the basic identities of $\theta_\tau(\cdot)$, this equation becomes
\begin{equation*}
    \frac{\vartheta_\tau(\sigma-\theta_t+\theta_0,\sigma+\theta_t-\theta_0)}{\vartheta_\tau(\sigma+\theta_t+\theta_0,\sigma-\theta_t-\theta_0)}=\widetilde{\rho}_{12},\quad \sigma=\tfrac{1}{2}-\widehat{\sigma}.
\end{equation*}
Through the identification of dual Tyurin parameters $
(\widetilde{\rho}_1,\widetilde{\rho}_2)=(\widetilde{\rho}_1^i,\widetilde{\rho}_2^i)$, see equation \eqref{eq:tyurin_identification}, and application of equation \eqref{eq:tyurinqoutientdual}, we find precisely
\begin{equation}\label{eq:rhodual12}
    \widetilde{\rho}_{12}=\frac{\vartheta_\tau(\sigma_{0t}-\theta_t+\theta_0,\sigma_{0t}+\theta_t-\theta_0)}{\vartheta_\tau(\sigma_{0t}+\theta_t+\theta_0,\sigma_{0t}-\theta_t-\theta_0)}.
\end{equation}
This tells us that $\widehat{\sigma}_{0t}=\tfrac{1}{2}-\sigma_{0t}$ solves equation \eqref{eq:exponent_equation}. Furthermore, note that
\begin{equation}\label{eq:realityofsigma}
    0<\Re \widehat{\sigma}_{0t}<\tfrac{1}{2}, 
\end{equation}
and the conditions \eqref{eq:line_conditions} are preserved exactly under $v_{0\infty}$, with $v_{0\infty}:\sigma_{0t}\mapsto \widehat{\sigma}_{0t}$. We conclude that $\widehat{\eta}$ is $0$-generic and we may thus apply Theorem \ref{thm:generic_asymp_zero} to the solution $(\widehat{f},\widehat{g})$. This gives the asymptotic expansions in the theorem, where
\begin{equation*}
    \widehat{r}_{0t}=\widehat{c}_{0t}\widehat{s}_{0t},\qquad \widehat{c}_{0t}:=c_{0t}|_{\Theta\mapsto \widehat{\Theta},\sigma_{0t}\mapsto \widehat{\sigma}_{0t}},
\end{equation*}
and
\begin{equation*}
 \widehat{s}_{0t}=-(-t_0)^{-2\widehat{\sigma}_{0t}}M_{0t}|_{\Theta\mapsto \widehat{\Theta},\sigma_{0t}\mapsto \widehat{\sigma}_{0t}}(\widehat{\rho}_{24}).   
\end{equation*}

There are now two ways to proceed. The easy way would be to compare asymptotic behaviours on both sides of equation \eqref{eq:fgandtheirhats} to obtain
\begin{equation}\label{eq:productofranditsdual}
  r_{0t}\widehat{r}_{0t}=-F_{1,-1}(\Theta,\sigma_{0t})F_{1,-1}(\widehat{\Theta},\widehat{\sigma}_{0t}),
\end{equation}
from which the remaining equations in the theorem follow. However, this approach breaks down when $\Re \widehat{\sigma}_{0t}=0$, corresponding to the cases in Remark \ref{remark:weakeningzerodual}.

Instead, we have to go the hard way. Firstly, by application of equation \eqref{eq:rhoduality2} with $(i,j,k,l)=(2,3,1,4)$ we obtain
\begin{equation*}
\rho_{24}=\rho_{34}\rho_{23}=\rho_{34}\frac{x_3}{x_2}M_{(2,3,1,4)}\left( \widetilde{\rho}_{12};\theta_\infty,\theta_0\right)M_{(3,4,1,2)}\left(\widetilde{\rho}_{13},\theta_\infty,\theta_0\right).
\end{equation*}
Through substitution of this equation for $\rho_{24}$ into the formula for the twist parameter, equation \eqref{eq:twist0t}, as well as the equations for $\widetilde{\rho}_{12}$ and $\rho_{34}$ in equations \eqref{eq:rhodual12} and \eqref{eq:elliptic_equation0},
and simplification using standard manipulations of $q$-theta functions, we arrive at the alternative formula for the twist parameter \eqref{eq:twist0talt}. It follows from this alternative formula that
\begin{equation*}
    s_{0t}\widehat{s}_{0t}=-q^{1-\theta_t-\theta_1}.
\end{equation*}
Furthermore, a direct computation yields
\begin{equation*}
    c_{0t}\widehat{c}_{0t}=q^{-1+\theta_t+\theta_1}F_{1,-1}(\Theta,\sigma_{0t})F_{1,-1}(\widehat{\Theta},\widehat{\sigma}_{0t}),
\end{equation*}
which yields the formula \eqref{eq:productofranditsdual} and Theorem \ref{thm:generic_asymp_zero_dual} follows.
\end{proof}

\subsection{Generic asymptotics at $\boldsymbol{t=\infty}$}\label{sec:extract_asymp_infty}
In this section, we make use of the symmetry $v_{t1}$ to derive the generic asymptotics around $t=\infty$ of solutions, proving Theorem \ref{thm:generic_asymp_infty}. We then employ $v_{0\infty}$ to obtain the dual asymptotic expansions around $t=\infty$ in Theorem \ref{thm:generic_asymp_infty_dual}.
\begin{proof}[Proof of Theorem \ref{thm:generic_asymp_infty} and Remark \ref{remark:coefficients}]
To obtain the asymptotic expansions of solutions around $t=\infty$, we make use of the symmetry $v_{t1}$.
 Let us take a $\infty$-generic point $\eta\in\mathcal{F}(\Theta,t_0)$, see Definition \ref{def:genericinfty}. Let $(f,g)$ denote the corresponding solution of $q\Psix(\Theta,t_0)$. By applying $v_{t1}$, we obtain a corresponding solution $(\widehat{f},\widehat{g})$ of $q\Psix(\widehat{\Theta},t_0^{-1})$, with parameters
\begin{equation*}
    \widehat{\Theta}=(\theta_0,\theta_1,\theta_t,\theta_\infty).
\end{equation*}
To this solution corresponds a unique point $\widehat{\eta}=v_{0\infty}(\eta)\in \mathcal{F}(\widehat{\Theta},t_0^{-1})$, given by
\begin{equation*}
    \widehat{\eta}_{ij}=\eta_{\alpha_{t1}(i)\alpha_{t1}(j)}\qquad (1\leq i<j\leq 4),
\end{equation*}
see equation \eqref{eq:symmetryrt1eta}. By equation \eqref{eq:symmetryrt1tyurin}, the corresponding Tyurin parameters are given by
\begin{equation*}
    \widehat{\rho}=(\rho_3,\rho_4,\rho_1.\rho_2).
\end{equation*}
In particular, we have the following important identities among the Tyurin ratios,
\begin{equation*}
\widehat{\rho}_{34}=v_{t1}(\rho_{34})=\rho_{12},\qquad 
\widehat{\rho}_{24}=v_{t1}(\rho_{24})=1/\rho_{24}.
\end{equation*}

The solutions $(f,g)$ and $(\widehat{f},\widehat{g})$ are related by
\begin{align}\label{eq:fgandtheirhatsinfty}
    \frac{f(t_m)}{t_m}&=\widehat{f}(1/t_m),\\
    \frac{1}{g(t_m)}&=q\widehat{g}(q/t_m),
\end{align}
and the crux of the proof lies in applying Theorem \ref{thm:generic_asymp_zero} to the right-hand sides of the above equations, i.e. to the solution $(\widehat{f},\widehat{g})$. To this end, we start by considering the corresponding exponent $\widehat{\sigma}_{0t}$, which by definition satisfies
\begin{equation}\label{eq:exponent_equation_infty}
    \frac{\vartheta_\tau(\widehat{\sigma}-\widehat{\theta}_1+\widehat{\theta}_\infty,\widehat{\sigma}+\widehat{\theta}_1-\widehat{\theta}_\infty)}{\vartheta_\tau(\widehat{\sigma}+\widehat{\theta}_1+\widehat{\theta}_\infty,\widehat{\sigma}-\widehat{\theta}_1-\widehat{\theta}_\infty)}=\widehat{\rho}_{34}.
\end{equation}
Substituting the values of the parameters and recalling that $\widehat{\rho}_{34}=\rho_{12}$, this is precisely equation \eqref{eq:elliptic_equationinf}. In particular, a solution is given by $\widehat{\sigma}_{0t}=\sigma_{01}$. Furthermore, it follows that $\widehat{\eta}$ is $0$-generic, precisely because $\eta$ is $\infty$-generic. In particular, we can apply Theorem \ref{thm:generic_asymp_zero} to $(\widehat{f},\widehat{g})$. This yields the asymptotic expansions in Theorem \ref{thm:generic_asymp_infty}, with coefficients given in Remark \ref{remark:coefficients}. The remaining formulas in the theorem now simply follow by applying $v_{t1}$ to the corresponding formulas in Theorem \ref{thm:generic_asymp_zero}. In other words, we just replace
\begin{equation*}
\Theta\mapsto \widehat{\Theta},\quad \sigma_{0t}\mapsto \sigma_{01}, \quad t_0\mapsto t_0^{-1},\quad  \rho_{24}\mapsto 1/\rho_{24},
\end{equation*}
everywhere, so that $c_{0t}\mapsto c_{01}$ and $s_{0t}\mapsto s_{01}$, and Theorem \ref{thm:generic_asymp_infty} follows.
\end{proof}

\begin{proof}[Proof of Theorem \ref{thm:generic_asymp_infty_dual}]
The proof is analogous to the proof of Theorem \ref{thm:generic_asymp_zero_dual}, following again Guzzetti \cite{guzzetti_solving}.
We start by taking a $\infty$-generic point $\eta\in\mathcal{F}(\Theta,t_0)$ and let $(f,g)$ denote the corresponding solution of $q\Psix(\Theta,t_0)$. By applying $v_{0\infty}$, we obtain a corresponding solution $(\widehat{f},\widehat{g})$ of $q\Psix(\widehat{\Theta},t_0)$, with parameters
\begin{equation*}
    \widehat{\Theta}=(\theta_\infty-\tfrac{1}{2},\theta_1,\theta_t,\theta_0+\tfrac{1}{2}).
\end{equation*}
To this solution corresponds a unique point $\widehat{\eta}=v_{0\infty}(\eta)\in \mathcal{F}(\widehat{\Theta},t_0)$ and we know the values of the correponding (dual) Tyurin ratios, due to equations \eqref{eq:tyurin_parametersr0infty}. Of particular relevance for the current proof, are the values of
\begin{equation*}
\widehat{\rho}_{12}=v_{0\infty}(\rho_{12})=q^{2\theta_1}\widetilde{\rho}_{34},\qquad 
\widehat{\rho}_{24}=v_{0\infty}(\rho_{24})=q^{\theta_t-\theta_1}t_0\widetilde{\rho}_{13}.
\end{equation*}

The solutions $(f,g)$ and $(\widehat{f},\widehat{g})$ are related by
\begin{align*}
    \frac{1}{f(t_m)}&=\frac{\widehat{f}(t_m)}{t_m},\\
    \frac{q^{\frac{3}{2}}g(t_m)}{t_m}&=\frac{1}{\widehat{g}(t_m)},
\end{align*}
and the crux of the proof lies in applying Theorem \ref{thm:generic_asymp_infty} to the right-hand sides of the above equations, i.e. to the solution $(\widehat{f},\widehat{g})$, following the same steps as in the proof Theorem \ref{thm:generic_asymp_zero_dual}, yielding Theorem \ref{thm:generic_asymp_infty_dual}.
\end{proof}

\subsection{Degenerate asymptotics on lines}\label{sec:lines_asymp}

In this section, we consider truncated asymptotics on the remaining lines and in particular prove Theorem \ref{thm:line_truncations}.

\begin{proof}[Proof of Theorem \ref{thm:line_truncations}]
Note that the content of Proposition \ref{prop:asymptotics_lineLd2i} is precisely the truncation on line $\widetilde{\mathcal{L}}_2^\infty$ indicated in Theorem \ref{thm:line_truncations}. We prove the truncations on the other lines through application of the symmetries in Table \ref{table:symmetries}. Here, Figure \ref{figure:action_on_lines}, which displays the action of the symmetries on the lines, can be helpful.

Firstly, we apply $r_0$ to the truncated asymptotics on the line $\widetilde{\mathcal{L}}_2^\infty$. By simply replacing $\theta_0\mapsto -\theta_0$ in every formula in Proposition \ref{prop:asymptotics_lineLd2i}, we obtain the truncated asymptotics on $\widetilde{\mathcal{L}}_2^0$. 

Next, we apply $r_t$ to the truncated asymptotics on the lines $\widetilde{\mathcal{L}}_2^\infty$ and $\widetilde{\mathcal{L}}_2^0$, to obtain the truncated asymptotics on the lines $\widetilde{\mathcal{L}}_1^\infty$ and $\widetilde{\mathcal{L}}_1^0$ respectively. Practically, this means, for example, simply replacing
\begin{equation*}
    \theta_t\mapsto -\theta_t,\quad \rho_{24}\mapsto \rho_{14},
\end{equation*}
in every formula in Proposition \ref{prop:asymptotics_lineLd2i}, to obtain the truncated asymptotics on $\widetilde{\mathcal{L}}_1^\infty$.

This gives all the statements in Theorem \ref{thm:line_truncations} on truncations of the asymptotic expansions around $t=0$ in Theorem \ref{thm:generic_asymp_zero}.
By applying $v_{0\infty}$ to each of these cases, we obtain all the statements in Theorem \ref{thm:line_truncations} on truncations of the dual asymptotic expansions around $t=0$ in Theorem \ref{thm:generic_asymp_zero_dual}.
Similarly, by applying $v_{t1}$ to each of the cases, we obtain all the statements in Theorem \ref{thm:line_truncations} on truncations of the asymptotic expansions around $t=\infty$ in Theorem \ref{thm:generic_asymp_infty}. Finally, by applying $v_{t1}v_{0\infty}$  to each of the cases, we obtain all the statements in Theorem \ref{thm:line_truncations} on truncations of the dual asymptotic expansions around $t=\infty$ in Theorem \ref{thm:generic_asymp_infty_dual}. This proves theorem \ref{thm:line_truncations}.
\end{proof}

For convenience of the reader, a full statement of the truncated asymptotics on one of the other lines of the Segre surface is given in the following proposition.
\begin{proposition}[Asymptotics on $\widetilde{\mathcal{L}}_1^0$] \label{prop:asymptotics_line_tilde01}
Assume $\theta_t-\theta_0\notin \frac{1}{2}\Lambda_\tau$ and $\Re (\theta_t-\theta_0)>-\tfrac{1}{2}$.
Take a point $\eta$ in the affine Segre surface $\mathcal{F}(\Theta,t_0)$ which lies on the line $\widetilde{\mathcal{L}}_1^0$ and let $(f,g)$ denote the corresponding solution of $q\Psix(\Theta,t_0)$. Denote $t_m=q^mt_0$, $m\in\mathbb{Z}$, then $f$ and $g$ have complete and convergent asymptotic expansions as $t_m\rightarrow 0$ of the form,
\begin{align*}
    f(t_m)&=\sum_{n=1}^\infty\sum_{k=0}^n F_{n,k}r_{0t}^k(- t_m)^{n+2k(\theta_t-\theta_0)},\\
    g(t_m)&=\sum_{n=1}^\infty\sum_{k=0}^n G_{n,k}r_{0t}^k(- t_m)^{n+2k(\theta_t-\theta_0)},
\end{align*}
where
\begin{align*}
    F_{1,1}&=\frac{\bigl(q^{-\theta_t}-q^{\theta_t}\bigr)\bigl(q^{\theta_1+\theta_\infty+\theta_0-\theta_t}-1\bigr)}{\bigl(q^{\theta_1+\theta_\infty-\theta_0+\theta_t}-1\bigr)\bigl(q^{2\theta_t-2\theta_0}-1\bigr)},  & G_{1,1}&=-q^{-1-\theta_t+\theta_0}F_{1,-1},\\
    F_{1,0}&=\frac{q^{\theta_0}-q^{-\theta_0}}{q^{\theta_t-\theta_0}-q^{\theta_0-\theta_t}}, &
    G_{1,0}&=\frac{q^{\theta_t}-q^{-\theta_t}}{q^{\theta_0-\theta_t}-q^{\theta_t-\theta_0}},
\end{align*}
and the higher order coefficients may be computed recursively via the $q\Psix$ equation.
Here the branches of the complex powers are principal and the integration constant $r_{0t}$ is given by
\begin{equation*}
    r_{0t}=c_{0t}\times s_{0t},
\end{equation*}
with
\begin{align*}
c_{0t}&=\frac{\Gamma_q(1-2(\theta_t-\theta_0))^2\Gamma_q(1-2\theta_0)}{\Gamma_q(1+2(\theta_t-\theta_0))^2\Gamma_q(1-2\theta_t)}\prod_{\epsilon=\pm1}\frac{ \Gamma_q(1+\theta_1+\epsilon\,\theta_\infty+\theta_t-\theta_0)}{ \Gamma_q(1+\theta_1+\epsilon\,\theta_\infty-\theta_t+\theta_0)},\\
s_{0t}&=-(-t_0)^{-2(\theta_t-\theta_0)}M_0(\rho_{14}),
\end{align*}
where $M_0(\cdot)$ is the M\"obius transformation
\begin{equation*}
    M_0(Z)=\frac{\vartheta_\tau(\theta_1-\theta_\infty+\theta_t-\theta_0)\theta_q(q^{\theta_\infty-\theta_0}t_0^{-1})-Z
\vartheta_\tau(\theta_1+\theta_\infty+\theta_t-\theta_0)\theta_q(q^{-\theta_0-\theta_\infty}t_0^{-1})}{\vartheta_\tau(\theta_1-\theta_\infty-\theta_t+\theta_0)\theta_q(q^{\theta_0+\theta_\infty-2\theta_t}t_0^{-1})-Z
\vartheta_\tau(\theta_0+\theta_\infty-\theta_t+\theta_0)\theta_q(q^{\theta_0-\theta_\infty-2\theta_t}t_0^{-1})}.
\end{equation*}
\end{proposition}
\begin{proof}
 The proposition is obtained by the method described in the proof of Theorem \ref{thm:line_truncations}. Explicitly, it is obtained by replacing
 \begin{equation*}
     \theta_0\mapsto-\theta_0,\quad \theta_t\mapsto-\theta_t,\quad \rho_{24}\mapsto \rho_{14}
 \end{equation*}
 in Proposition \ref{prop:asymptotics_lineLd2i}, and then changing the inner index $k\mapsto -k$ and correspondingly changing $c_{0t}\mapsto 1/c_{0t}$ and $s_{0t}\mapsto 1/s_{0t}$.
\end{proof}
\begin{remark}
If we consider $r_{0t}$ as a free variable, then $\eta=\eta(r_{0t})$ traces out the line $\widetilde{\mathcal{L}}_1^0$ in the Segre surface $\widehat{\mathcal{F}}$. The value $r_{0t}=0$ corresponds to the intersection point of this line with the line $\widetilde{\mathcal{L}}_2^\infty$, see Corollary \ref{coro:intersectionanalytic}. The value $r_{0t}=\infty$ corresponds to the intersection point of the line with the hyperplane section at infinity, $X=\widehat{\mathcal{F}}\setminus \mathcal{F}$, and is thus excluded.
\end{remark}

\subsection{Double and doubly degenerate asymptotics on intersection points}\label{sec:asymptotics_intersection}
In this section, we study the asymptotics of solutions corresponding to the intersection points of lines in the Segre surface. In particular, we prove Theorems \ref{thm:kanekoohyama} and \ref{thm:orbitII}.  For convenience of the reader, the three orbits of the intersection points under the action of the symmetries in Table \ref{table:symmetries} are given in Figures \ref{figure:kaneko_ohyama_orbit}, \ref{figure:orbit_I} and \ref{figure:orbit_II}.

\begin{figure}[t]
\begin{tikzpicture}
  \matrix (m) [matrix of math nodes, row sep=3em,
    column sep=3em]{
    & {\color{red}\widetilde{\mathcal{L}}_3^0\cap \widetilde{\mathcal{L}}_4^\infty}& & {\color{red}\widetilde{\mathcal{L}}_3^\infty\cap \widetilde{\mathcal{L}}_4^0}  & \\
    {\color{blue}\widetilde{\mathcal{L}}_1^0\cap \widetilde{\mathcal{L}}_2^\infty} & & {\color{blue}\widetilde{\mathcal{L}}_1^\infty\cap \widetilde{\mathcal{L}}_2^0}  &  &  \\
    & {\color{red}\mathcal{L}_1^0\cap \mathcal{L}_2^\infty} & & {\color{red}\mathcal{L}_1^\infty\cap \mathcal{L}_2^0}  &   \\
    {\color{blue}\mathcal{L}_3^0\cap\mathcal{L}_4^\infty} & & {\color{blue}\mathcal{L}_3^\infty\cap\mathcal{L}_4^0}  &  &    \\};
  \path[-stealth]
    (m-1-2) edge[<->] node[above,xshift=-7mm] {$r_t\text{/}r_{\infty}$} (m-1-4) edge[<->] node[above left,xshift=1mm,yshift=-1mm] {$v_{t1}$} (m-2-1)
            edge[<->]   [densely dotted] node[right,yshift=8mm] {$v_{0\infty}$} (m-3-2)
    (m-1-4) edge[<->] node[right,yshift=8mm] {$v_{0\infty}$} (m-3-4) edge[<->] node[above left,xshift=1mm,yshift=-1mm] {$v_{t1}$} (m-2-3)
    (m-2-1) edge[<->] [-,line width=6pt,draw=white] (m-2-3)
            edge[<->] node[above,xshift=-7mm] {$r_t\text{/}r_0$} (m-2-3) edge[<->] node[right,yshift=8mm] {$v_{0\infty}$} (m-4-1)
    (m-3-2) edge[<->] [densely dotted] node[above,xshift=-7mm] {$r_t\text{/}r_\infty$} (m-3-4)
            edge[<->] [densely dotted] node[above left,xshift=1mm,yshift=-1mm] {$v_{t1}$} (m-4-1)
    (m-4-1) edge[<->] node[above,xshift=-7mm] {$r_1\text{/}r_\infty$} (m-4-3)
    (m-3-4) edge[<->] node[above left,xshift=1mm,yshift=-1mm] {$v_{t1}$} (m-4-3)
    (m-2-3) edge[<->]  [-,line width=6pt,draw=white] (m-4-3)
            edge[<->] node[right,yshift=8mm] {$v_{0\infty}$} (m-4-3);
\end{tikzpicture}
\caption{Orbit of Kaneko-Ohyama points under action of symmetries in Table \ref{table:symmetries}. We omit loops and write e.g. $r_t\text{/}r_\infty$ when either symmetry interchanges the two vertices.}
\label{figure:kaneko_ohyama_orbit}
\end{figure}
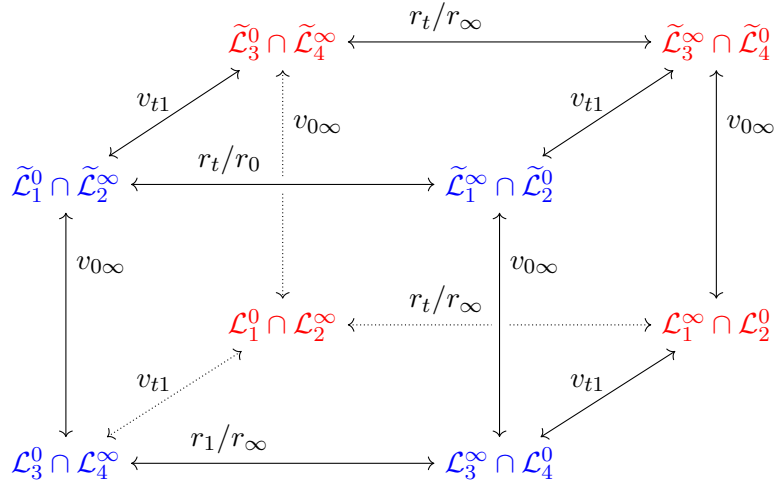

\begin{proof}[Proof of Theorem \ref{thm:kanekoohyama}]
Due to Corollary \ref{coro:intersectionanalytic}, we already know that the solution $(f,g)$ corresponding to the intersection point of $\widetilde{\mathcal{L}}_1^0$ and $\widetilde{\mathcal{L}}_2^\infty$, has asymptotics around $t=0$ given by a convergent power series.

Now, we note that the Kaneko-Ohyama points, as defined in the beginning of Section \ref{subsection:intersect}, form a closed orbit under the symmetries $r_0$, $r_t$, $r_1$, $r_\infty$, $v_{t1}$ and $v_{0\infty}$, displayed in Figure \ref{figure:kaneko_ohyama_orbit}. An easy way to derive this figure, is to note that each Kaneko-Ohyama point corresponds to a diagonal on one the front or back faces of the two cubes in Figure \ref{figure:action_on_lines}. Application of any of the symmetries sends such diagonals to other such diagonals. The actions of the symmetries on Kaneko-Ohyama points in Figure \ref{figure:kaneko_ohyama_orbit} is now easily derived.

At the same time, each symmetry sends Kaneko-Ohyama solutions to Kaneko-Ohyama solutions.
Proving the theorem is now a matter of direct calculations. We discuss one such calculation in detail. Our starting point is the Kaneko-Ohyama solution $(f_0,g_0)$ of  $q\Psix(\Theta,t_0)$ corresponding to the intersection point of $\widetilde{\mathcal{L}}_1^0$ and $\widetilde{\mathcal{L}}_2^\infty$, whose power series expansion around $t=0$ is given in Corollary \ref{coro:intersectionanalytic}. We apply the symmetry $v_{0\infty}$, yielding the solution
\begin{equation*}
    f_*(t_m)=\frac{t_m}{f_0(t_m)},\quad 
    g_*(t_m)=\frac{t_m}{q^{\frac{3}{2}}g(t_m)},
\end{equation*}
of $q\Psix(\widehat{\Theta},t_0)$, where $\widehat{\Theta}=(\theta_\infty-\tfrac{1}{2},\theta_1,\theta_t,\theta_0+\tfrac{1}{2})$, corresponding to the Kaneko-Ohyama point $\mathcal{L}_3^0\cap\mathcal{L}_4^\infty$ in $\mathcal{F}(\widehat{\Theta},t_0)$.
It has a convergent power series expansion around $t=0$, of the form
\begin{align*}
    f_*(t_m)&=\sum_{n=0}^\infty F_{n}^*(- t_m)^{n}, &
    g_*(t_m)&=\sum_{n=0}^\infty G_{n}^*(- t_m)^{n},\\
    F_{0}^*&=\frac{q^{\theta_0-\theta_t}-q^{-\theta_0+\theta_t}}{q^{\theta_0}-q^{-\theta_0}}, &
    G_{0}^*&=-q^{-\frac{1}{2}}\frac{q^{\theta_0-\theta_t}-q^{-\theta_0+\theta_t}}{q^{\theta_t}-q^{-\theta_t}}.
\end{align*}
Next, we undo the change of parameters, by setting
\begin{align*}
    f_1(t_m)&=f_*(t_m)|_{\Theta\mapsto \widehat{\Theta}},\\
    g_1(t_m)&=g_*(t_m)|_{\Theta\mapsto \widehat{\Theta}},
\end{align*}
so that $(f_1,g_1)$ is the solution of $q\Psix(\Theta,t_0)$ corresponding to the Kaneko-Ohyama point $\mathcal{L}_3^0\cap\mathcal{L}_4^\infty$ in $\mathcal{F}(\Theta,t_0)$. The power series expansion of $(f_1,g_1)$ around $t=0$ is then given in (2) of Theorem \ref{thm:kanekoohyama} with $\pm=-$.

Furthermore, the existence of the original Kaneko-Ohyama solution $(f_0,g_0)$ and corresponding point is predicated on $\theta_0-\theta_t\notin \tfrac{1}{2}\Lambda_\tau$. Under the change of parameter values, this condition becomes $\theta_\infty-\theta_1-\tfrac{1}{2}\notin \tfrac{1}{2}\Lambda_\tau$, which is equivalent to $\theta_\infty-\theta_1\notin \tfrac{1}{2}\Lambda_\tau$.

This proves (2) of Theorem \ref{thm:kanekoohyama} with $\pm=-$. The other ones are derived by similar applications of the symmetries, yielding the theorem.
\end{proof}

\begin{figure}[t]
\begin{tikzpicture}
  \matrix (m) [matrix of math nodes, row sep=3em,
    column sep=1em]{
    & {\color{blue}\widetilde{\mathcal{L}}_1^\infty}\cap {\color{red}\mathcal{L}_1^0}& & {\color{blue}\widetilde{\mathcal{L}}_1^\infty}\cap {\color{red}\mathcal{L}_1^\infty}& & {\color{blue}\mathcal{L}_3^0}\cap {\color{red}\widetilde{\mathcal{L}}_3^\infty} &  & {\color{blue}\mathcal{L}_3^\infty}\cap {\color{red}\widetilde{\mathcal{L}}_3^\infty}\\
    {\color{blue}\widetilde{\mathcal{L}}_1^0}\cap {\color{red}\mathcal{L}_1^0} & & {\color{blue}\widetilde{\mathcal{L}}_1^0}\cap {\color{red}\mathcal{L}_1^\infty} &  & {\color{blue}\mathcal{L}_3^0}\cap {\color{red}\widetilde{\mathcal{L}}_3^0} &  & {\color{blue}\mathcal{L}_3^\infty}\cap {\color{red}\widetilde{\mathcal{L}}_3^0} &\\
    & {\color{blue}\widetilde{\mathcal{L}}_2^\infty}\cap {\color{red}\mathcal{L}_2^0} & & {\color{blue}\widetilde{\mathcal{L}}_2^\infty}\cap {\color{red}\mathcal{L}_2^\infty}  &  & {\color{blue}\mathcal{L}_4^0}\cap {\color{red}\widetilde{\mathcal{L}}_4^\infty} &  & {\color{blue}\mathcal{L}_4^\infty}\cap {\color{red}\widetilde{\mathcal{L}}_4^\infty}\\
    {\color{blue}\widetilde{\mathcal{L}}_2^0}\cap {\color{red}\mathcal{L}_2^0} & & {\color{blue}\widetilde{\mathcal{L}}_2^0}\cap {\color{red}\mathcal{L}_2^\infty} &  & {\color{blue}\mathcal{L}_4^0}\cap {\color{red}\widetilde{\mathcal{L}}_4^0} &  & {\color{blue}\mathcal{L}_4^\infty}\cap {\color{red}\widetilde{\mathcal{L}}_4^0} & \\};
  \path[-stealth]
    (m-1-2) edge[<->] node[above,xshift=-5mm] {$r_\infty$} (m-1-4) edge[<->] node[above left,xshift=1mm,yshift=-1mm] {$r_0$} (m-2-1)
            edge[<->]   [densely dotted] node[right,yshift=8mm] {$r_t$} (m-3-2)
    (m-1-4) edge[<->] node[right,yshift=8mm] {$r_t$} (m-3-4) edge[<->] node[above left,xshift=1mm,yshift=-1mm] {$r_0$} (m-2-3)
    (m-2-1) edge[<->] [-,line width=6pt,draw=white] (m-2-3)
            edge[<->] node[above,xshift=-5mm] {$r_\infty$} (m-2-3) edge[<->] node[right,yshift=8mm] {$r_t$} (m-4-1)
    (m-3-2) edge[<->] [densely dotted] node[above,xshift=-5mm] {$r_\infty$} (m-3-4)
            edge[<->] [densely dotted] node[above left,xshift=1mm,yshift=-1mm] {$r_0$} (m-4-1)
    (m-4-1) edge[<->] node[above,xshift=-5mm] {$r_\infty$} (m-4-3)
    (m-3-4) edge[<->] node[above left,xshift=1mm,yshift=-1mm] {$r_0$} (m-4-3)
    (m-2-3) edge[<->]  [-,line width=6pt,draw=white] (m-4-3)
            edge[<->] node[right,yshift=8mm] {$r_t$} (m-4-3);
    \path[-stealth]
    (m-1-6) edge[<->] node[above,xshift=-5mm] {$r_\infty$} (m-1-8) edge[<->] node[above left,xshift=1mm,yshift=-1mm] {$r_0$} (m-2-5)
            edge[<->]   [densely dotted] node[right,yshift=8mm] {$r_1$} (m-3-6)
    (m-1-8) edge[<->] node[right,yshift=8mm] {$r_1$} (m-3-8) edge[<->] node[above left,xshift=1mm,yshift=-1mm] {$r_0$} (m-2-7)
    (m-2-5) edge[<->] [-,line width=6pt,draw=white] (m-2-7)
            edge[<->] node[above,xshift=-5mm] {$r_\infty$} (m-2-7) edge[<->] node[right,yshift=8mm] {$r_1$} (m-4-5)
    (m-3-6) edge[<->] [densely dotted] node[above,xshift=-5mm] {$r_\infty$} (m-3-8)
            edge[<->] [densely dotted] node[above left,xshift=1mm,yshift=-1mm] {$r_0$} (m-4-5)
    (m-4-5) edge[<->] node[above,xshift=-5mm] {$r_\infty$} (m-4-7)
    (m-3-8) edge[<->] node[above left,xshift=1mm,yshift=-1mm] {$r_0$} (m-4-7)
    (m-2-7) edge[<->]  [-,line width=6pt,draw=white] (m-4-7)
            edge[<->] node[right,yshift=8mm] {$r_1$} (m-4-7);

\node (c1) at (m-1-2) {};
\node (c2) at (m-2-3) {};
\node (c12) at ($(c1)!0.5!(c2)$) {};

\node (d1) at (m-1-6) {};
\node (d2) at (m-2-7) {};
\node (d12) at ($(d1)!0.5!(d2)$) {};

\path[<->]  (c12)  edge   [bend left=38]   node[above] {$v_{t1}$} (d12);
\end{tikzpicture}
\caption{The orbit of the intersection point $\widetilde{\mathcal{L}}_1^0\cap\mathcal{L}_1^0$ under the symmetries in Table \ref{table:symmetries}. Loops are omitted and the action of $v_{0\infty}$ is not displayed. The symmetry $v_{t1}$ interchanges the vertices of the two cubes, realised by simply swapping them. The symmetry $v_{t1}\circ v_{0\infty}$ can be realised in the figure as rotating each cube $180^{\circ}$ along the axis of symmetry that passes through the midpoints of the right vertical edge on the front face and the left vertical edge on the back face.}
\label{figure:orbit_I}
\end{figure}
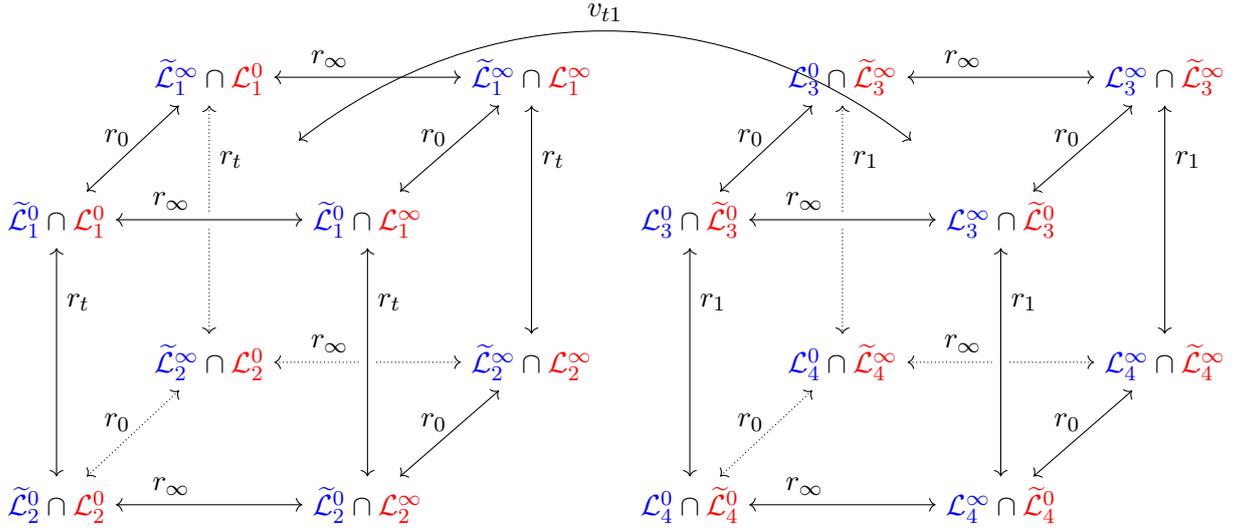

\begin{figure}[th!]
\begin{tikzpicture}
  \matrix (m) [matrix of math nodes, row sep=3em,
    column sep=1em]{
    & {\color{blue}\widetilde{\mathcal{L}}_2^0}\cap {\color{red}\widetilde{\mathcal{L}}_3^\infty}& & {\color{blue}\widetilde{\mathcal{L}}_2^0}\cap {\color{red}\widetilde{\mathcal{L}}_4^\infty}& & {\color{blue}\mathcal{L}_3^\infty}\cap {\color{red}\mathcal{L}_2^0} &  & {\color{blue}\mathcal{L}_3^\infty}\cap {\color{red}\mathcal{L}_1^0}\\
    {\color{blue}\widetilde{\mathcal{L}}_1^0}\cap {\color{red}\widetilde{\mathcal{L}}_3^\infty} & & {\color{blue}\widetilde{\mathcal{L}}_1^0}\cap {\color{red}\widetilde{\mathcal{L}}_4^\infty} &  & {\color{blue}\mathcal{L}_4^\infty}\cap {\color{red}\mathcal{L}_2^0} &  & {\color{blue}\mathcal{L}_4^\infty}\cap {\color{red}\mathcal{L}_1^0} &\\
    & {\color{blue}\widetilde{\mathcal{L}}_2^\infty}\cap {\color{red}\widetilde{\mathcal{L}}_3^0} & & {\color{blue}\widetilde{\mathcal{L}}_2^\infty}\cap {\color{red}\widetilde{\mathcal{L}}_4^0}  &  & {\color{blue}\mathcal{L}_3^0}\cap {\color{red}\mathcal{L}_2^\infty} &  & {\color{blue}\mathcal{L}_3^0}\cap {\color{red}\mathcal{L}_1^\infty}\\
    {\color{blue}\widetilde{\mathcal{L}}_1^\infty}\cap {\color{red}\widetilde{\mathcal{L}}_3^0} & & {\color{blue}\widetilde{\mathcal{L}}_1^\infty}\cap {\color{red}\widetilde{\mathcal{L}}_4^0} &  & {\color{blue}\mathcal{L}_4^0}\cap {\color{red}\mathcal{L}_2^\infty} &  & {\color{blue}\mathcal{L}_4^0}\cap {\color{red}\mathcal{L}_1^\infty} & \\};
  \path[-stealth]
    (m-1-2) edge[<->] node[above,xshift=-5mm] {$r_1$} (m-1-4) edge[<->] node[above left,xshift=1mm,yshift=-1mm] {$r_t$} (m-2-1)
            edge[<->]   [densely dotted] node[right,yshift=8mm] {$r_0$} (m-3-2)
    (m-1-4) edge[<->] node[right,yshift=8mm] {$r_0$} (m-3-4) edge[<->] node[above left,xshift=1mm,yshift=-1mm] {$r_t$} (m-2-3)
    (m-2-1) edge[<->] [-,line width=6pt,draw=white] (m-2-3)
            edge[<->] node[above,xshift=-5mm] {$r_1$} (m-2-3) edge[<->] node[right,yshift=8mm] {$r_0$} (m-4-1)
    (m-3-2) edge[<->] [densely dotted] node[above,xshift=-5mm] {$r_1$} (m-3-4)
            edge[<->] [densely dotted] node[above left,xshift=1mm,yshift=-1mm] {$r_t$} (m-4-1)
    (m-4-1) edge[<->] node[above,xshift=-5mm] {$r_1$} (m-4-3)
    (m-3-4) edge[<->] node[above left,xshift=1mm,yshift=-1mm] {$r_t$} (m-4-3)
    (m-2-3) edge[<->]  [-,line width=6pt,draw=white] (m-4-3)
            edge[<->] node[right,yshift=8mm] {$r_0$} (m-4-3);
    \path[-stealth]
    (m-1-6) edge[<->] node[above,xshift=-5mm] {$r_t$} (m-1-8) edge[<->] node[above left,xshift=1mm,yshift=-1mm] {$r_1$} (m-2-5)
            edge[<->]   [densely dotted] node[right,yshift=8mm] {$r_\infty$} (m-3-6)
    (m-1-8) edge[<->] node[right,yshift=8mm] {$r_\infty$} (m-3-8) edge[<->] node[above left,xshift=1mm,yshift=-1mm] {$r_1$} (m-2-7)
    (m-2-5) edge[<->] [-,line width=6pt,draw=white] (m-2-7)
            edge[<->] node[above,xshift=-5mm] {$r_t$} (m-2-7) edge[<->] node[right,yshift=8mm] {$r_\infty$} (m-4-5)
    (m-3-6) edge[<->] [densely dotted] node[above,xshift=-5mm] {$r_t$} (m-3-8)
            edge[<->] [densely dotted] node[above left,xshift=1mm,yshift=-1mm] {$r_1$} (m-4-5)
    (m-4-5) edge[<->] node[above,xshift=-5mm] {$r_t$} (m-4-7)
    (m-3-8) edge[<->] node[above left,xshift=1mm,yshift=-1mm] {$r_1$} (m-4-7)
    (m-2-7) edge[<->]  [-,line width=6pt,draw=white] (m-4-7)
            edge[<->] node[right,yshift=8mm] {$r_\infty$} (m-4-7);

\node (c1) at (m-1-2) {};
\node (c2) at (m-2-3) {};
\node (c12) at ($(c1)!0.5!(c2)$) {};

\node (d1) at (m-1-6) {};
\node (d2) at (m-2-7) {};
\node (d12) at ($(d1)!0.5!(d2)$) {};

\path[<->]  (c12)  edge   [bend left=38]   node[above] {$v_{0\infty}$} (d12);
\end{tikzpicture}
\caption{The orbit of the intersection point $\widetilde{\mathcal{L}}_1^0\cap\widetilde{\mathcal{L}}_3^\infty$ under the symmetries in Table \ref{table:symmetries}. Loops are omitted and the action of $v_{t1}$ is not displayed. The symmetry $v_{0\infty}$ interchanges the vertices of the two cubes, realised by simply swapping them. The symmetry $v_{t1}$ can be realised in the figure as rotating each cube $180^{\circ}$ along the axis of symmetry that passes through the midpoints of the left vertical edge on the front face and the right vertical edge on the back face.}
\label{figure:orbit_II}
\end{figure}

\begin{proof}[Proof of Theorem \ref{thm:orbitII}]
We start with the observation that the intersection point
\begin{equation*}
\{N_*(t_0)\}=\widetilde{\mathcal{L}}_1^0\cap\widetilde{\mathcal{L}}_3^\infty,
\end{equation*}
defines an analytic function
\begin{equation*}
    N_*:\mathbb{C}^*\rightarrow \mathbb{CP}^6,
\end{equation*}
which is $q$-periodic, such that $N_*(t_0)\in \widehat{\mathcal{F}}(\Theta,t_0)$ for all $t_0\in\mathbb{C}^*$.
Furthermore, due to item (4)' in the proof of Theorem \ref{thm:intersection} in Section \ref{sec:intersectionpoints}, $N_*(t_0)$ is infinite if and only if $t_0\in q^{\mathbb{Z}+2\theta_0-\theta_t+\theta_1}$. Define the set
\begin{equation*}
    V=\mathbb{C}^*\setminus (q^{\mathbb{Z}-2\theta_0+\theta_t-\theta_1}\cup S),
\end{equation*}
where $S$ is the union of $q$-lines on the right-hand side of equation \eqref{eq:param_assumptions_2}.
Then, for any $t_0\in V'$, the Riemann-Hilbert correspondence defines a unique corresponding solution $(f(q^mt_0),g(q^m t_0))_{m\in\mathbb{Z}}$ on the $q$-line $q^{\mathbb{Z}}t_0$. We denote by $(f,g)$ the corresponding uniformisation in $t_0$, so that both $f$ and $g$ define complex functions on $V$ which satisfy $q\Psix$.

Next, we apply Proposition \ref{prop:asymptotics_line_tilde01} to obtain the asymptotics of $(f,g)$ near $t=0$. This yields the following asymptotic expansion as $m\rightarrow +\infty$, absolutely convergent for large enough $m$,
\begin{align}\label{eq:zzzasymptoticzero}
    f(t_m)&=\sum_{n=1}^\infty\sum_{k=0}^n F_{n,k}(- t_m)^{n}e_0(m;t_0)^k,\\
    g(t_m)&=\sum_{n=1}^\infty\sum_{k=0}^n G_{n,k}(- t_m)^{n}e_0(m;t_0)^k,\nonumber
\end{align}
where the coefficients are as in Proposition \ref{prop:asymptotics_line_tilde01} and in particular only depend on $\Theta$, and
\begin{equation}\label{eq:defie0}
    e_0(m;t_0):=r_{0t}(t_0)(-q^m t_0)^{2(\theta_t-\theta_0)}.
\end{equation}
Here $r_{0t}=c_{0t}\times s_{0t}$ and explicit formulas for the constant $c_{0t}$ and twist parameter $s_{0t}=s_{0t}(t_0)$ are given in Proposition \ref{prop:asymptotics_line_tilde01}. It will be convenient to consider the dual expression for the twist parameter, given in general in equation \eqref{eq:twist0talt}. By setting $\sigma_{0t}=\theta_t-\theta_0$ in equation \eqref{eq:twist0talt}, and consequently replacing
\begin{equation*}
     \theta_0\mapsto-\theta_0,\quad \theta_t\mapsto-\theta_t,\quad \rho_{24}\mapsto \rho_{14},
 \end{equation*}
and taking reciprocal of the final result, we obtain the following dual expression for the twist parameter,
\begin{equation}\label{eq:twist0taltorbitII}
s_{0t}=-(-t_0)^{-2(\theta_t-\theta_0)}L(\widetilde{\rho}_{23};t_0)^{-1},
\end{equation}
where $L=L(Z;t_0)$ is the following linear map,
\begin{equation*}
    L(Z;t_0)=\frac{\vartheta_\tau(-2\theta_t)\theta_q(q^{2\theta_0-\theta_t+\theta_1}t_0^{-1})-Z
\vartheta_\tau(2\theta_0-2\theta_t)\theta_q(q^{\theta_1-\theta_t}t_0^{-1})}{\vartheta_\tau(-2\theta_0)\theta_q(q^{\theta_t+\theta_1}t_0^{-1})}.
\end{equation*}

Now, we use the fact that $N_*\in \widetilde{\mathcal{L}}_3^\infty$, so that the dual Tyurin ratio $\rho_{23}$ is zero and thus we obtain the following simplified expression for $r_{0t}$,
\begin{align*}
    r_{0t}&=-(-t_0)^{-2(\theta_t-\theta_0)}c_{0t} \frac{\vartheta_\tau(-2\theta_0)\theta_q(q^{\theta_t+\theta_1}t_0^{-1})}{\vartheta_\tau(-2\theta_t)\theta_q(q^{2\theta_0-\theta_t+\theta_1}t_0^{-1})}\\
    &=-(-t_0)^{-2(\theta_t-\theta_0)}c_{0t} \frac{\vartheta_\tau(-1-2\theta_0)\theta_q(q^{-\theta_t-\theta_1}t_0)}{\vartheta_\tau(-1-2\theta_t)\theta_q(q^{-2\theta_0+\theta_t-\theta_1}t_0)}.
\end{align*}
From equation \eqref{eq:defie0}, we then get the following expression for $e_0(m;t_0)$,
\begin{align*}
    e_0(m;t_0)&=-(-q^m t_0)^{2(\theta_t-\theta_0)}(-t_0)^{-2(\theta_t-\theta_0)}c_{0t} \frac{\vartheta_\tau(-1-2\theta_0)\theta_q(q^{-\theta_t-\theta_1}t_0)}{\vartheta_\tau(-1-2\theta_t)\theta_q(q^{-2\theta_0+\theta_t-\theta_1}t_0)}\\
    &=q^{2(\theta_t-\theta_0)m}c_{0t} \frac{\vartheta_\tau(-1-2\theta_0)\theta_q(q^{-\theta_t-\theta_1}t_0)}{\vartheta_\tau(-1-2\theta_t)\theta_q(q^{-2\theta_0+\theta_t-\theta_1}t_0)}\\
    &=c_{0t} \frac{\vartheta_\tau(-1-2\theta_0)\theta_q(q^{-\theta_t-\theta_1}t_m)}{\vartheta_\tau(-1-2\theta_t)\theta_q(q^{-2\theta_0+\theta_t-\theta_1}t_m)},
\end{align*}
where the last equality follows by $m$-fold application of the identity $\theta_q(qz)=-z^{-1}\theta_q(z)$ in both numerator and denominator. The last equation gives the desired uniformisation in the asymptotic expansion. Namely, we now define
\begin{equation*}
    E_0(t):=-F_{1,1}c_{0t} \frac{\vartheta_\tau(-1-2\theta_0)\theta_q(q^{-\theta_t-\theta_1}t)}{\vartheta_\tau(-1-2\theta_t)\theta_q(q^{-2\theta_0+\theta_t-\theta_1}t)},
\end{equation*}
so that $E_0(q^mt_0)=-F_{1,1}e_0(m;t_0)$. Then $E_0(t)$ is given by equation \eqref{eq:e0infdefi} in the theorem and we obtain the following asymptotic expansions, 
\begin{align}\label{eq:zzzzasympzero}
    f(t)&=\frac{q^{\theta_0}-q^{-\theta_0}}{q^{\theta_0-\theta_t}-q^{\theta_t-\theta_0}}\;t+t E_0(t)+\sum_{n=2}^\infty\sum_{k=0}^n f_{n,k}t^n E_0(t)^k,\\
    g(t)&=-q^{-1}\frac{q^{\theta_t}-q^{-\theta_t}}{q^{\theta_0-\theta_t}-q^{\theta_t-\theta_0}}\;t-q^{-1+\theta_0-\theta_t}t E_0(t)+\sum_{n=2}^\infty\sum_{k=0}^n g_{n,k}t^n E_0(t)^k,\nonumber
\end{align}
as $t\rightarrow 0$ along any fixed $q$-line, i.e. $t\in q^{\mathbb{Z}}t_0$, for any fixed $t_0\in V$, which are absolutely convergent for small enough $t\in q^{\mathbb{Z}}t_0$.
Here the coefficients are related to those in expansions \eqref{eq:zzzasymptoticzero} by
\begin{equation*}
    f_{n,k}=(-1)^{n+k}F_{n,k}/F_{1,-1}^k,\quad g_{n,k}=(-1)^{n+k}G_{n,k}/F_{1,-1}^k\qquad (0\leq k \leq n,n\geq 2).
\end{equation*}
The function $E_0(t)$ has poles on the $q$-line $q^{\mathbb{Z}+2\theta_0-\theta_t+\theta_1}$. By choosing $t_0$ close enough to $q^{2\theta_0-\theta_t+\theta_1}$, we may ensure that $|E_0(t_0)|$ becomes arbitrarily large, whilst \eqref{eq:zzzzasympzero} remains absolutely convergent for small enough $t\in q^\mathbb{Z}t_0$. Since
\begin{equation}\label{eq:qdif}
    E_0(qt)=E_0(t)q^{2\theta_0-2\theta_t},
\end{equation}
it follows that, for arbitrarily large fixed $C>0$, the series expansions \eqref{eq:zzzzasympzero}, with $E_0(t)$ replaced by $C|t|^{\Re(2\theta_0-2\theta_t)}$ everywhere, is uniformly absolutely convergent for small enough $t\in\mathbb{C}^*$.

Now, take any compact set $K\subseteq \mathbb{CP}^1\setminus q^{\mathbb{Z}-2\theta_0+\theta_t-\theta_1}$, such that $K\cap\mathbb{C}^*$ is non-empty, with $qK=K$. Then, the function $E_0(t)$ is bounded on $\{z\in K: |q|\leq |z|\leq 1\}$ and so we can find a $C>$ such that \begin{equation}\label{eq:boundonE0}
    |E_0(t)|\leq C|t|^{\Re(2\theta_0-2\theta_t)},
\end{equation} on this region. But then, by equation \eqref{eq:qdif}, the bound \eqref{eq:boundonE0} holds on the whole of $K\cap\mathbb{C}^*$. Since the series expansions \eqref{eq:zzzzasympzero}, with $E_0(t)$ replaced by $C|t|^{\Re(2\theta_0-2\theta_t)}$ everywhere, is absolutely convergent for small enough $t\in\mathbb{C}^*$, it follows that the original series expansions \eqref{eq:zzzzasympzero} must be uniformly absolutely convergent for small enough $t\in K\setminus\{0\}$.

This also shows that $f$ and $g$ are meromorphic functions for small enough $t\in \mathbb{C}^*\setminus q^{\mathbb{Z}-2\theta_0+\theta_t-\theta_1}$. Then, by analytic continuation via the $q\Psix$ equation, it follows that $f$ and $g$ are meromorphic functions on all of $\mathbb{C}^*\setminus q^{\mathbb{Z}-2\theta_0+\theta_t-\theta_1}$. This concludes the asymptotic analysis of the solution near $t=0$. An analogous analysis around $t=\infty$ gives the asymptotic expansion around $t=\infty$ and Theorem \ref{thm:orbitII} follows.
\end{proof}

%% file: appendix.tex
\section{Proof that RH is a biholomorpism}\label{appendix:rhdif}

\begin{proof}[Proof of Proposition \ref{prop:rhdiffeo}.]
Let $\mathcal{V}(\Theta,t_0)$ denote the initial value space of $q\Psix(\Theta)$ at $t=t_0$, constructed by blowing up
the eight points
\begin{equation}
    \begin{aligned}
    &b_1=(0,q^{-1+\theta_0}t_0), & &b_3=(q^{+\theta_t}t_0,0), & &b_5=(q^{+\theta_1},\infty), & &b_7=(\infty,q^{-\theta_\infty}),\\
    &b_2=(0,q^{-1-\theta_0}t_0), & &b_4=(q^{-\theta_t}t_0,0), & &b_6=(q^{-\theta_1},\infty), & &b_8=(\infty,q^{-1+\theta_\infty}).
\end{aligned}\label{eq:basepoints}
\end{equation}
in $\{(f,g)\in\mathbb{CP}^1\times \mathbb{CP}^1\}$ and subsequently removing the inaccessible divisor, given by the union of the proper transforms of $\{f=0\}, \{f=\infty\},\{g=0\}$ and  $\{g=\infty\}$, see \cite{s:01} for details.
The time-evolution of $q\Psix$ now defines an isomorphism
\begin{equation}\label{eq:time-evolution}
    \overline{\cdot}:\mathcal{V}(\Theta,t_0)\rightarrow \mathcal{V}(\Theta,qt_0),
\end{equation}
which is also a biholomorphism.
Through the identification of the initial value space $\mathcal{V}(\Theta,t_0)$ with the solution space $\mathcal{I}(\Theta,t_0)$, the mapping $\textrm{RH}$ in equation \eqref{eq:rhmapping} becomes a bijective mapping
\begin{equation}\label{eq:rhmappinginit}
    \widetilde{\textrm{RH}}:\mathcal{V}(\Theta,t_0)\rightarrow \mathcal{F}(\Theta,t_0),
\end{equation}
with the fundamental property that it `commutes' with the time-evolution in equation \eqref{eq:time-evolution}. To be precise, for any $p\in \mathcal{V}(\Theta,t_0)$,
\begin{equation}\label{eq:commutativediagram}
    \widetilde{\textrm{RH}}(\overline{p})=\widetilde{\textrm{RH}}(p),
\end{equation}
where we note that $\mathcal{F}(\Theta,qt_0)=\mathcal{F}(\Theta,t_0)$. The verification of this statement goes as follows. We know that the time-evolution in equation \eqref{eq:time-evolution} has a lift to the monodromy manifold, see equation \eqref{eq:connection_deformation}, 
\begin{equation*}
    \mathcal{M}(\Theta,t_0)\rightarrow \mathcal{M}(\Theta,q t_0),C(z)\mapsto \overline{C}(z)=z \; C(z).
\end{equation*}
The Tyurin parameters $\rho_{3}$ and $\rho_{4}$ remain unchanged under the time-evolution, 
\begin{equation*}
\overline{\rho}_3=\pi\left[\overline{C}(q^{\theta_1})\right]=\pi\left[q^{\theta_1} C(q^{\theta_1})\right]=\pi\left[C(q^{\theta_1})\right]=\rho_3,
\end{equation*}
and similarly we have $\overline{\rho}_4=\rho_4$. On the other hand, the points $x_1=q^{+\theta_t}t_0$
and $x_2=q^{-\theta_t}t_0$ move under the time-evolution, and correspondingly,
\begin{align*}
   \overline{\rho}_1&= \pi\left[\overline{C}(q^{\theta_t+1}t_0)\right]= \pi\left[q^{\theta_t+1}t_0C(q^{\theta_t+1}t_0)\right]\\
   &=\pi\left[q^{\theta_t+1}t_0 \frac{t_0}{(q^{\theta_t}t_0)^2}q^{\theta_0 \sigma_3}C(q^{\theta_t}t_0)q^{\theta_\infty \sigma_3}\right]=\pi\left[C(q^{\theta_t}t_0)q^{\theta_\infty \sigma_3}\right]\\
   &=\pi\left[C(q^{\theta_t}t_0)\right]q^{2\theta_\infty}=\rho_1 q^{2\theta_\infty},
\end{align*}
where in the third equality we used the $q$-difference equation that $C(z)$ satisfies, equation \eqref{eq:connectionqdif} with $t=t_0$. The computation for $\rho_2$ is identical and we find
\begin{equation*}
    \overline{\rho}=(q^{2\theta_\infty}\rho_1,q^{2\theta_\infty}\rho_2,\rho_3,\rho_4).
\end{equation*}

Now, the time-evolution acts on the coefficients of $T(\rho)$ by
\begin{align*}
  \overline{T}_{12}&=q^{-2\theta_\infty}t_0^{-2} T_{12},  & 
   \overline{T}_{13}&=t_0^{-2} T_{13}, & 
   \overline{T}_{14}&=t_0^{-2} T_{14},\\
  \overline{T}_{34}&=q^{+2\theta_\infty}t_0^{-2} T_{34},  & 
   \overline{T}_{24}&=t_0^{-2} T_{24}, & 
   \overline{T}_{23}&=t_0^{-2} T_{23}.
\end{align*}
It is now straightforward to check that the time-evolution acts trivially on the $\eta$-coordinates. We give the details of the computation for $\eta_{12}$ in inhomogeneous coordinates,
\begin{align*}
\overline{\eta}_{12}&=\frac{\vartheta_\tau(+\frac{1}{2},-\frac{1}{2})}{\vartheta_\tau(+\theta_0,-\theta_0)}\frac{\overline{T}_{12}\overline{\rho}_1\overline{\rho}_2}{\overline{T}_{12}'\overline{\rho}_1\overline{\rho}_2+\overline{T}_{13}'\overline{\rho}_1\overline{\rho}_3+\ldots+\overline{T}_{34}'\overline{\rho}_3\overline{\rho}_4}\\
&=\frac{\vartheta_\tau(+\frac{1}{2},-\frac{1}{2})}{\vartheta_\tau(+\theta_0,-\theta_0)}\frac{T_{12}\rho_1\rho_2\; q^{2\theta_\infty}t_0^{-2}}{(T_{12}'\rho_1\rho_2+T_{13}'\rho_1\rho_3+\ldots+T_{34}'\rho_3\rho_4) q^{2\theta_\infty}t_0^{-2}}\\
&=\eta_{12}.
\end{align*}

Next, we are going to show that $\widetilde{\textrm{RH}}$ is analytic. Firstly, a definition. We say that a point $p\in \mathcal{V}(\Theta,t_0)$ is  in generic position, if it does not lie on any of the eight exceptional lines above the base points given in equation \eqref{eq:basepoints}. 

We first prove that $\widetilde{\textrm{RH}}$ is analytic at points in generic position. So, take a point $p_*=(f_*,g_*)\in \mathcal{V}(\Theta,t_0)$  in generic position. Then we can construct an open environment $U\subseteq \mathcal{V}(\Theta,t_0)$ of $(f_*,g_*)$ such that any any point in it is in generic position. Consider the corresponding coefficient matrix $A(z)=A(z;f,g)$, defined via equations \eqref{eq:coordinates_linear}, where we put $w\equiv 1$, for $(f,g)\in U$.

By schrinking $U$ if necessary, we can diagonalise $A_0=A(0)$ by a matrix $H=H(f,g)$ which is analytic on $U$. It is straightforward to incorporate parameters in Carmichael's existence theorems in \cite{carmichael1912}, from which it follows that $\Psi_\infty(z;f,g)$ and $\Psi_0(z;f,g)$ are analytic in $(f,g)$ in an open environment $U'\subseteq U$ of $(f_*,g_*)$. It follows that
\begin{equation*}
    C(z;f,g)=\Psi_0(z;f,g)^{-1}\Psi_\infty(z;f,g),
\end{equation*}
is analytic  in $(f,g)\in U'$, and thus so are the Tyurin parameters
\begin{equation*}
    \rho_k(f,g)=\pi(C(x_k;f,g))\qquad (1\leq k\leq 4).
\end{equation*}
Then, by equation \eqref{eq:eta_defi},
\begin{equation*}
    \eta_{ij}=\eta_{ij}(f,g)=\frac{\vartheta_\tau(+\frac{1}{2},-\frac{1}{2})}{\vartheta_\tau(+\theta_0,-\theta_0)}\frac{T_{ij}\rho_i(f,g)\rho_j(f,g)}{T'(\rho(f,g))}\qquad (1\leq i<j\leq 4),
\end{equation*}
from which it follows that $\eta=\eta(f,g)$ is analytic on $U'$.

% We conclude that $\widetilde{RH}$ is analytic on $U'$.  
% Suppose that $\widetilde{RH}$ is not locally biholomorpic at $(f_0,g_0)$. Upon constructing some local coordinates $\{n^{(1)},n^{(2)}\}$ on $\mathcal{F}(\Theta,t_0)$ around $\eta(f_0,g_0)$, the Jacobian
% \begin{equation*}
%  J=\begin{pmatrix}
%      \frac{\partial n^{(1)}}{\partial f} & \frac{\partial n^{(1)}}{\partial g}\\
%      \frac{\partial n^{(2)}}{\partial f} & \frac{\partial n^{(2)}}{\partial g}\\
%  \end{pmatrix},
% \end{equation*}
% must have rank less than $2$. It follows that there exists a curve

%  Now, suppose that $\widetilde{RH}$ is not locally biholomorpic at $(f_0,g_0)$. 

% Now, suppose that 

We conclude that $\widetilde{\textrm{RH}}$ is analytic at any point $p\in \mathcal{V}(\Theta,t_0)$ in generic position.  Since $\widetilde{\textrm{RH}}$ commutes with the time-evolution \eqref{eq:time-evolution}, see equation \eqref{eq:commutativediagram}, it follows that $\widetilde{\textrm{RH}}$ is analytic everywhere if, for any point $p\in \mathcal{V}(\Theta,t_0)$, there exists an $m\in\mathbb{Z}$ such that $m$-fold application of the time-evolution \eqref{eq:time-evolution} gives a point $\overline{p}^{(m)}\in \mathcal{V}(\Theta,q^m t_0)$ in generic position. The last statement is proven in a straightforward but laborious manner, as follows.

Take a point $p\in \mathcal{V}(\Theta,t_0)$ and suppose it is not in generic position. Then it must lie in one of the eight exceptional lines above $b_1,\ldots,b_8$. Suppose it lies in the exceptional line above $b_1$. We introduce local coordinates
\begin{equation*}
f=u,\quad g-q^{-1+\theta_0}t_0=uv,
\end{equation*}
so that $\{u=0,v\in\mathbb{C}\}$ parametrises the exceptional line above $b_1$, minus one point which lies on the inaccessible divisor. It follows that the point $p$ is given by a unique value $v_p\in\mathbb{C}$ with $u=u_p=0$. Application of the time-evolution, yields
\begin{align*}
\overline{g}&=q^{-\theta_0}t_0,\\
\overline{f}&=\frac{q(q^{\theta_0}-q^{-\theta_0})t_0}{(q^{\theta_0-\theta_\infty}-t_0)(q^{\theta_0+\theta_\infty-1}-t_0)}
(q^{-1+\theta_0}(q^{\theta_t}+q^{-\theta_t}-t_0(q^{\theta_1}+q^{-\theta_1})+v).
\end{align*}
This point is in generic position, unless $\overline{f}$ happens to vanish, i.e. when
\begin{equation*}
    v=-q^{-1+\theta_0}(q^{\theta_t}+q^{-\theta_t}-t_0(q^{\theta_1}+q^{-\theta_1})).
\end{equation*}
For this special value of $v$, we find
\begin{align*}
\overline{\overline{g}}&=q^{1+\theta_0}t_0,\\
\overline{\overline{f}}&=\frac{q^{-1-\theta_0}(q^{\theta_0}-q^{-\theta_0})t_0}{(q^{-\theta_0-\theta_\infty-1}-t_0)(q^{-\theta_0+\theta_\infty-2}-t_0)}
(q(q^{\theta_1}+q^{-\theta_1})t_0-(q^{\theta_t}+q^{-\theta_t})).
\end{align*}
This point is in generic position, unless $t_0$ happens to be such that $\overline{\overline{f}}$ vanishes, that is,
\begin{equation*}
    t_0=q^{-1}(q^{\theta_t}+q^{-\theta_t})/(q^{\theta_1}+q^{-\theta_1}).
\end{equation*}
Assuming this special value for $t_0$, we find
\begin{align*}
\overline{\overline{\overline{g}}}&=q^{2-\theta_0}t_0,\\
\overline{\overline{\overline{f}}}&=\frac{q^{\theta_0-3}(1-q)(q^{\theta_0}-q^{-\theta_0})(q^{\theta_t}-q^{-\theta_t})^2}{(q^{\theta_1}+q^{\theta_1})(q^{\theta_0-\theta_\infty-2}-t_0)(q^{\theta_0+\theta_\infty-3}-t_0)},
\end{align*}
which is necessarily in generic position, due to parameter conditions \eqref{eq:param_assumptions_1} and \eqref{eq:param_assumptions_2}.

By a similar argument for points on the seven other exceptional lines, we find that, for any point $p\in \mathcal{V}(\Theta,t_0)$, there exists an $m\in\mathbb{Z}$ such that $m$-fold application of the time-evolution \eqref{eq:time-evolution} gives a point $\overline{p}^{(m)}\in \mathcal{V}(\Theta,q^m t_0)$ in generic position. Since $\widetilde{\textrm{RH}}$ is analytic at points in generic position, it  follows that $\widetilde{\textrm{RH}}$ is analytic
on the whole of $\mathcal{V}(\Theta,t_0)$. So $\widetilde{\textrm{RH}}$ is an analytic bijection from $\mathcal{V}(\Theta,t_0)$ to $\mathcal{F}(\Theta,t_0)$.

What it left, is to prove that $\widetilde{\textrm{RH}}^{-1}$ is analytic. Take any point $\eta_*\in \mathcal{F}(\Theta,t_0)$ and choose a small open environment $V\subseteq \mathcal{F}(\Theta,t_0)$ of $\eta_*$. By shrinking $V$ if necessary, we can construct a quadruplet of Tyurin parameters $\rho=\rho(\eta)\in\mathbb{CP}^4$, analytic in $\eta\in V$, such that equations \eqref{eq:Thom}, \eqref{eq:Thom0} and \eqref{eq:eta_defi} are satisfied for $\eta\in V$. We can now construct a corresponding connection matrix $C(z;\eta)\in \mathfrak{C}(\Theta,t_0)$, analytic in $\eta$,
such that 
\begin{equation}\label{eq:rho_uniform}
    \pi(C(x_k;\eta))=\rho_k(\eta)\qquad (1\leq k\leq 4),
\end{equation}
for $\eta\in V'$, for some open environment $V'\subseteq V$ of $\eta_*$. In other words, $C(z;\eta)$ is an analytic matrix function on $\mathbb{C}^*\times V'$. To prove this rigorously, one chooses a basis of the vector space of $2\times 2$ analytic matrix functions satisfying \eqref{eq:connectionqdif}, with $t=t_0$, and rewrites equations \eqref{eq:rho_uniform} as a system of eight homogeneous equations among the eight entries of the coordinate vector for $C(z;\eta)$ on this basis (see \cite{roffelsenjoshiqpvi}*{Equations (5.9) and (5.10)}. The homogeneous system has rank less than $8$, and we construct a non-trivial solution of the homogeneous system, analytically in $\eta$, around $\eta=\eta_*$. This then defines a connection matrix $C(z;\eta)$ with the desired properties.

Define $p_*\in \mathcal{V}(\Theta,t_0)$ by $\widetilde{\textrm{RH}}(p_*)=\eta_*$, and let $m\in\mathbb{Z}$ be such that $\overline{p}_*^{(m)}\in \mathcal{V}(\Theta,q^m t_0)$ is in generic position. Then, we know for a fact, that Riemann-Hilbert problem \ref{rhp:main}, with connection matrix $C(z;\eta_*)$, is solvable for that value of $m$. Choose open sets $\overline{p}_*^{(m)}\in U''\subseteq\mathcal{V}(\Theta,q^m t_0)$ and $\eta_*\in V''\subseteq V'\subseteq \mathcal{F}(\Theta,q^m t_0)$, such that $\widetilde{\textrm{RH}}$ maps $U''$ homeomorphically onto $V''$ and every point in $U''$ is in generic position. Then, Riemann-Hilbert problem \ref{rhp:main}, with connection matrix $C(z;\eta)$, has a unique solution $Y^{\B{m}}(z;\eta)$, for all $\eta\in V''$. Furthermore, by the analytic Fredholm alternative, see \cite{zhourhp}*{Proposition 4.3} and also \cite{fokasitskapaev}*{Corollary 3.1}, this solution  is analytic in $\eta\in V''$.

Therefore, the corresponding coefficient matrix $A(z,t_m,\eta)$, defined via equation \eqref{eq:linear_system}, is analytic in $\eta\in V''$. In particular, the coordinates $f=f(\eta)$ and $g=g(\eta)$, defined via equations \eqref{eq:coordinates_linear}, are also analytic in $\eta\in V''$, and thus so is 
\begin{equation*}
    \overline{p}^{(m)}(\eta):=(f(\eta),g(\eta))\in \mathcal{V}(\Theta,q^m t_0).
\end{equation*}
But then $p(\eta)\in \mathcal{V}(\Theta,t_0)$, defined as the $m$-fold down-shift of $\overline{p}^{(m)}(\eta)$, is analytic in $\eta\in V''$.
On the other hand, by definition,
\begin{equation*}
   \widetilde{\textrm{RH}}( p(\eta))=\eta,
\end{equation*}
and we conclude that $\widetilde{\textrm{RH}}^{-1}$ is analytic at $\eta_*\in V''$.

It follows that $\widetilde{\textrm{RH}}^{-1}$ is analytic and since we have already shown that the forward mapping is analytic, Proposition \ref{prop:rhdiffeo} follows.
\end{proof}